\documentclass[doublespacing]{lsuthesis}

\usepackage{graphicx}
\usepackage{mathrsfs}
\usepackage[usenames,dvipsnames]{color}
\usepackage{hyperref}
\hypersetup{colorlinks = true, linkcolor = [rgb]{0,0,0}, urlcolor = [rgb]{0.125490,0.29542,0.1647058}, citecolor = [rgb]{0.75882,0.37411,0.14117}}
\usepackage{comment}
\usepackage{amsmath, amssymb}
\usepackage{mathtools}
\usepackage{notes2bib}
\usepackage{float}
\usepackage[normalem]{ulem}
\usepackage{soul}
\setstcolor{red}
\usepackage{newtxtext,newtxmath}

\newcommand{\bra}[1]{\langle#1|}
\newcommand{\ket}[1]{|#1\rangle}

\newcommand{\tr}{\operatorname{Tr}}
\newcommand{\qed}{\openbox}

\allowdisplaybreaks

\usepackage{lsutitle}
\title{A MENAGERIE OF SYMMETRY TESTING QUANTUM ALGORITHMS}
\thesistype{Dissertation}
\department{The Department of Physics and Astronomy}
\soughtdegree{Doctor of Philosophy}
\author{Margarite Lynn LaBorde}
\degrees{B.S., Louisiana State University, 2018\\
  B.S., Louisiana State University, 2018 \\
  M.S., Louisiana State University, 2022}
\graduationdate{May 2023}

\begin{document}

\providecommand{\U}[1]{\protect\rule{.1in}{.1in}}
\newtheorem{theorem}{Theorem}
\newtheorem{acknowledgement}[theorem]{Acknowledgement}
\newtheorem{Algorithm}{Algorithm}
\newtheorem{axiom}[theorem]{Axiom}
\newtheorem{claim}[theorem]{Claim}
\newtheorem{conclusion}[theorem]{Conclusion}
\newtheorem{condition}[theorem]{Condition}
\newtheorem{conjecture}[theorem]{Conjecture}
\newtheorem{corollary}[theorem]{Corollary}
\newtheorem{criterion}[theorem]{Criterion}
\newtheorem{definition}{Definition}
\newtheorem{example}{Example}
\newtheorem{exercise}[theorem]{Exercise}
\newtheorem{lemma}[theorem]{Lemma}
\newtheorem{notation}[theorem]{Notation}
\newtheorem{problem}[theorem]{Problem}
\newtheorem{proposition}{Proposition}
\newtheorem{remark}{Remark}
\newtheorem{solution}[theorem]{Solution}
\newtheorem{summary}[theorem]{Summary}
\newenvironment{proof}[1][Proof]{\noindent\textbf{#1.} }{\ \rule{0.5em}{0.5em}}
\numberwithin{theorem}{section}
\numberwithin{proposition}{section}
\numberwithin{definition}{section}
\numberwithin{example}{section}

\frontmatter

\maketitle

\pagebreak

\doublespacing
\vspace{0.55ex}
To my parents, Rebecca and Raymond LaBorde. My greatest privilege has been being your child.
\vspace{12pt}
\pagebreak

\chapter*{Acknowledgments}
\doublespacing
\vspace{0.55ex}
I would like to take this opportunity to thank everyone who supported me throughout my experiences as a graduate student. I would like to thank Dr. Jonathan Dowling for his mentorship and for welcoming me into quantum physics. He was a great mentor and friend and will be greatly missed. I would especially like to thank Dr. Mark Wilde for his mentorship and guidance. He stepped in to offer support during a tumultuous time---during a pandemic and the loss of an advisor---and I am extremely grateful to have been his student. I would also like to thank Dr. Stephen Shipman, Dr. Ivan Agullo, Dr. Mette Gaarde, and Dr. Feng Chen for taking the time to be a part of my graduate committee.

     Additionally, I would like to acknowledge my husband, Zachary Bradshaw. Many restaurant napkins and whiteboards have seen our discussions, and there isn't an idea in math or physics that I find fascinating that isn't immediately shared with him. He is not only my best friend but my favorite mathematician. 
     
     Special thanks go to my fellow students of the Quantum Science and Technologies group. In particular, I would like to thank Soorya Rethinasamy, Aliza Siddiqui, and Roy Pace for their friendship and conversation. Discussing physics and research with them made being at LSU all the more fulfilling.
     
     All of these people were instrumental in making my graduate studies enjoyable and fruitful. It is through their support that I was able to explore so many opportunities and choose to pursue my doctoral research in physics.
     
     Funding for this work was provided by the Department of Defense SMART Scholarship. Funding for travel to conferences and research visits was provided by

\addcontentsline{toc}{chapter}{\hspace{0em} {Acknowledgements} \vspace{12pt}}
\pagebreak

\tableofcontents

\chapter*{Abstract}
\addcontentsline{toc}{chapter}{\hspace{0em} Abstract}

\vspace{0.55ex}
\doublespacing
     In Chapter 1, we establish the mathematical background used throughout this thesis. We review concepts from group and representation theory. We further establish fundamental concepts from quantum information. This will allow us to then define the different notions of symmetry necessary in the following chapters.

     In Chapter 2, we investigate Hamiltonian symmetries. We propose quantum algorithms capable of testing whether a Hamiltonian exhibits symmetry with respect to a group. Furthermore, we show that this algorithm is that this algorithm is DQC1-Complete. Finally, we execute one of our symmetry-testing algorithms on existing quantum computers for simple examples.
     
     In Chapter 3, we discuss tests of symmetry for quantum states. For the case of testing Bose symmetry of a state, there is a simple and efficient quantum algorithm, while the tests for other kinds of symmetry rely on the aid of a quantum prover. We prove that the acceptance probability of each algorithm is equal to the maximum symmetric fidelity of the state being tested Finally, we establish various generalizations of the resource theory of asymmetry, with the upshot being that the acceptance probabilities of the algorithms are resource monotones.
     
     In Chapter 4, we begin by showing that the analytical form of the acceptance probability of such a test is given by the cycle index polynomial of the symmetric group $S_k$. We derive a family of quantum separability tests, each of which is generated by a finite group; for all such algorithms, we show that the acceptance probability is determined by the cycle index polynomial of the group. Finally, we produce and analyze explicit circuit constructions for these tests, showing that the tests corresponding to the symmetric and cyclic groups can be executed with $O(k^2)$ and $O(k\log(k))$ controlled-SWAP gates, respectively, where $k$ is the number of copies of the state being tested.

     In Chapter 5, we include additional results not previously published; in particular, we give a test for symmetry of a quantum state using density matrix exponentiation, a further result of Hamiltonian symmetry measurements when using Abelian groups, and an alternate Hamiltonian symmetry test construction for a block-encoded Hamiltonian.

\mainmatter

\chapter{Mathematical Background and Introduction to Quantum Symmetry}\label{ch:intro}
\doublespacing
\section{Introduction}

Before we begin discussing in any detail the finer aspects of symmetry-testing algorithms on quantum computers, we must necessarily establish for ourselves the language and tools we will employ to do so. The work described in this thesis rests primarily on three pillars---a basic understanding of group theory with an accompanying idea of group representations, a strong familiarity with quantum information, and concepts of symmetry with regards to both of the former. 

We begin in Section~\ref{sec:mathbg} by introducing only so much math as an uninitiated reader may need to follow along with the more abstract notions in this text. This includes an elementary review of groups, group representations, and a smattering of examples of common groups. Someone well-versed with group theory may very easily skip this section and expect no repercussions in their understanding of the work. It is presented in a spirit of compassion to those with an interest in quantum symmetry without the benefit of some \textit{a priori} knowledge of abstract algebra.

Section~\ref{sec:qibg} continues with introducing concepts of quantum information that will be integral to the comprehension of this thesis. This section considers common terms and concepts such as density matrices, quantum channels, and common norms. Additionally, we give some derivations for important lemmas that will open the gateway to future proofs. 

In Section~\ref{sec:symdef}, we come to the primary function of this chapter, which is to introduce notions of symmetry. We concern ourselves with the symmetries of states, Hamiltonians, and channels. The definitions herein are remarkably similar, but the nuances are great. In this section, we endeavor to delineate them here in appropriate detail as a reference for the rest of the thesis. 
     
\section{Mathematical Background}   \label{sec:mathbg}  
\subsection{A Gentle Overview of Group Theory} \label{sec:grouptheory}
For each section, we will provide a text which guides the knowledge and information reviewed therein. Our inaugural effort is Nicholson's \textit{Introduction to Abstract Algebra} \cite{nicholson2012introduction}, which is a standard undergraduate text on the matter. Here, the focus remains on finite, discrete groups as that is the primary concern of this body of work. (As an aside, continuous, compact groups are also within the capabilities of our future effort, as per the result of \cite{girard2021twirling}.) This section endeavors to introduce elementary aspects of group theory while introducing and defining groups of interest for this work as we proceed.

Let us start by defining a group. A set $G$ equipped with a binary operation, $*$, is called a group if
\begin{enumerate}
    \item $\forall g_1,\, g_2 \in G$, we have $g_1*g_2 \in G$. 
    \item The operation $*$ is associative, i.e., $\forall g_1,\, g_2,\, g_3 \in G$, $g_1*(g_2*g_3)=(g_1*g_2)*g_3$.
    \item There exists a unity element $\epsilon$ such that $\forall g \in G$, $g*\epsilon=\epsilon*g =g$.
    \item $\forall g \in G$, there exists an inverse $g^{-1} \in G$ such that $g*g^{-1}=\epsilon$.
\end{enumerate}
For finite groups, $|G|$ denotes the order of the group, which is the number of unique elements of the group. Additionally, note that we often suppress the binary operation ``$*$" in favor of using juxtaposition. 

This list of axioms encapsulates a powerful but simple mathematical object, best studied via examples. An illustrative but simple example is the integers modulo $n$, where $n$ is some integer. This group is denoted $\mathbb{Z}_n$ and comes equipped with addition as its binary operator. We already know that addition is associative and has an identity, $\epsilon=0$. For inverses, consider
\begin{equation}
    z + z^{-1} = 0\, \textrm{mod}\, n \cong n \,.
\end{equation}
This equation shows that any integer $z$ in $\mathbb{Z}_n$ has an inverse $n-z$ also in $\mathbb{Z}_n$. Elementary investigations will also demonstrate that this group is closed under this modulo addition. We can convince ourselves that many other sets form groups under addition in just such a way, be they the reals, complex numbers, etc., but $\mathbb{Z}_n$ is a nice, finite plaything to illustrate these traits. It will also reappear quite often as we progress throughout this work.

Those familiar with quantum mechanics (or mathematics) likely won't be surprised to learn that not all groups are Abelian (or commutative, if you prefer the term), but we caution the reader anyway. For an intuitive picture of this, consider the dihedral group $D_n$ which has order $2n$. Picture a regular polygon of $n$ sides; $D_n$ is defined to be the set of all transformations that leave the shape invariant. For instance, an equilateral triangle is invariant under flips and rotations of $1/3$ the way around the shape. Obviously, combinations of these actions will also leave the shape invariant. Yet, if we identify each vertex of a triangle, it becomes immediately clear that a flip followed by a rotation is \textit{not} the same as a rotation followed by a flip. Figure~\ref{fig:triangle} visualizes this discrepancy.
 
\begin{figure}[h!]
\begin{center}
\includegraphics[width=\columnwidth]{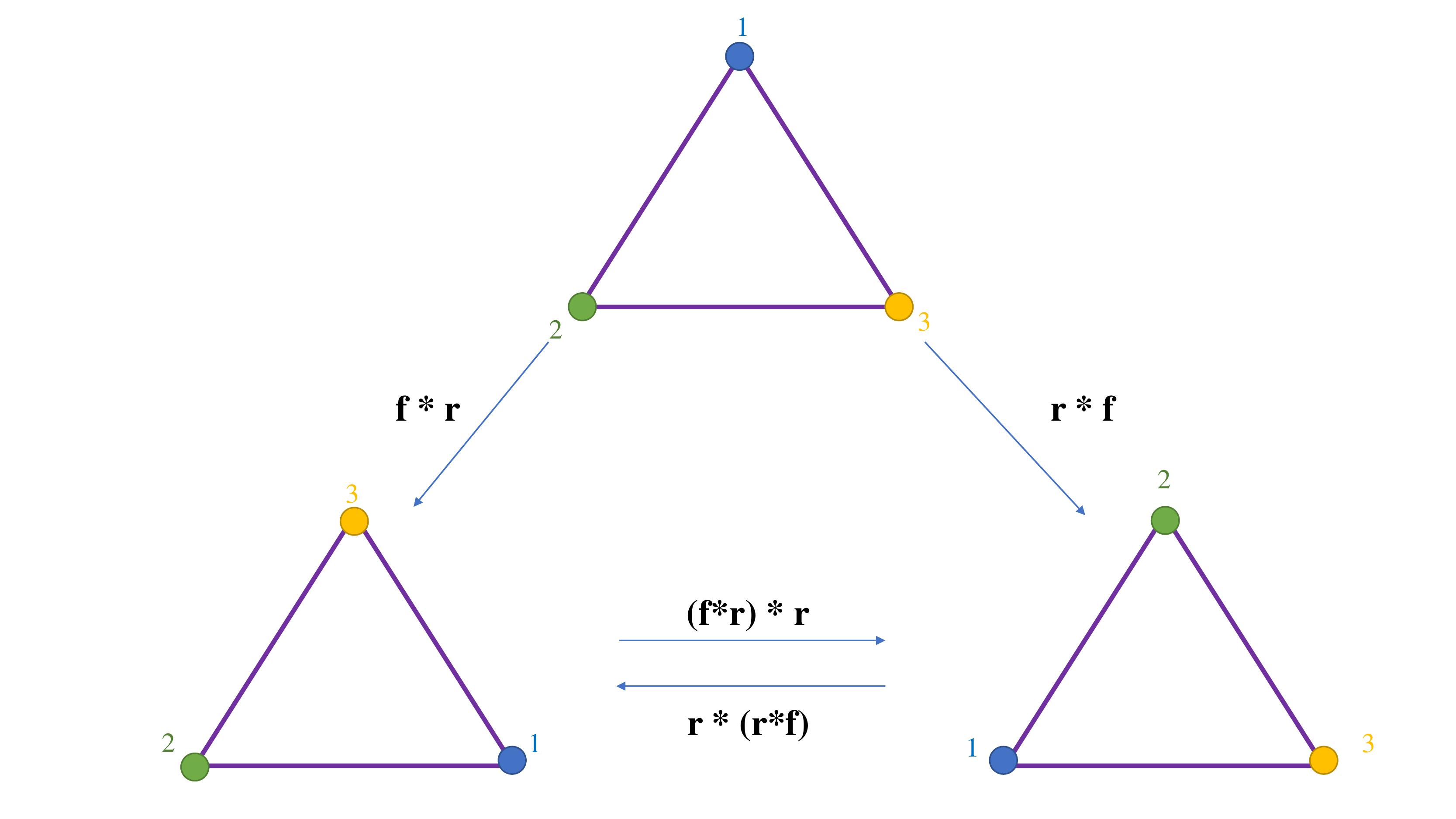} 
\caption{\label{fig:triangle} A cartoon demonstrating that the dihedral group $D_3$ is non-Abelian. Colored circles identify the three vertices, and it becomes visually obvious that the rotation and flip operations do not commute.
}
\end{center}
\end{figure}

Any definition of the groups should encapsulate these observed rules. First, identify the rotations with the symbol $r$ and flips with $f$. Then define the group as
\begin{equation}\label{eq:d3}
    D_3 \coloneqq \langle f,r | f^2 = r^3 = \epsilon\, \textrm{and } rf = f r^2 \rangle \, .
\end{equation}
A quick sanity check indicates that two flips or three rotations give us back the original triangle exactly as expected. Figure~\ref{fig:triangle} demonstrates the last rule regarding the commutation of flips and rotations. Often, non-Abelian groups like the dihedral group will have relationships like this explicitly specified.

Now, given some group $G$, we might also be interested in any subgroup contained within it. We call a subset $H$ a \textit{subgroup} of $G$ if it is also a group under the same operation as $G$. Every group has at least two subgroups---the identity and the entirety of $G$. Since these cases are ubiquitous, they are typically referred to as \textit{trivial} subgroups. Any other valid subgroup is a \textit{proper} subgroup. 

Subgroups feature prominently in the set of permutations on $k$ letters, denoted $S_k$.  Elements of this group take the form $(1\ 2\ \cdots\ k)$ (this as an example; not all permutations involve all $k$ elements) where they are read from left to right within the parentheses but right-to-left without. For example, consider $S_4$ and its elements. The element $(1\ 2)$ says to permute the first object with the second. Now, consider the product $(1\ 2)(2\ 3)= (1\ 2\ 3)$. 

So why is this group interesting? Well, Cayley's theorem states that all finite groups are isomorphic to a subgroup of a permutation group\cite{jordan1870traite}. That is a hefty assertion! However, the totality of Cayley's theorem is even stronger than that. It further states that any group is a subset of a symmetric group even for infinite groups, and it can be given a more general definition as well in terms of cosets. For our needs, the restricted case of finite groups and the standard symmetric group of $S_k$ will suit just fine. 

Consider the example of the alternating group of degree $k$, or $A_k$. This is the set of all even permutations of $S_k$. We can employ the subgroup test \cite[Section 2.3]{nicholson2012introduction} to ensure that $A_k$ is indeed a subgroup. First, it contains the identity permutation $(\ )$. Second, the product of two even permutations $g_1$ and $g_2$ is also even; therefore, $g_1*g_2 \in A_k$. Third, if some permutation $g$ is even, its inverse is also even. This last fact can be verified by writing $g$ as a product of transpositions. Then $g^{-1}$ is just those same transpositions in the opposite order; therefore, if there are an even number of transpositions in $g$, there is also an even number in $g^{-1}$ and $g^{-1}\in A_k$. These three facts tell us $A_k$ is a subgroup of $S_k$. Furthermore, for $k \geq 2$, we know $A_k$ is a proper subgroup because $\forall k \geq 2$, the symmetric group $S_k$ has at least one odd transposition.

Throughout, we have been defining groups using intuitive concepts or words. A diligent student may be discontent with this approach, as it lacks mathematical elegance and succinctness. This is where generators make their appearance. The generating set of a group $G$ is a subset of elements by which every other group element can be produced as a combination of these elements, called generators, and their inverses. 

In \eqref{eq:d3}, we actually made use of the generating set to define $D_3$, but let's consider a simpler example as well. Define the cyclic group of order $k$ to be the groups generated by cyclic shifts of $k$ elements. For example, consider the group $C_4 =\{( \ ),\, (1\ 2\ 3\ 4),\, (1\ 3)(2\ 4),\, (1\ 4\ 3\ 2)  \}$. The elements of this group read as shifting four things by nothing, one spot, two spots, and three spots respectively. However, an equivalent interpretation is to simply apply the element $(1\ 2\ 3\ 4)$ once, twice, thrice, or four times for the identity. Therefore, $(1\ 2\ 3\ 4)$ generates $C_4$. Mathematically, this is denoted by
\begin{equation}
    C_4 = \langle (1\ 2\ 3\ 4) \rangle \,,
\end{equation}
and in general we have
\begin{equation}
    C_k = \langle (1\ 2\ ... \, k) \rangle \, .
\end{equation}
Generating sets are not necessarily unique in the sense that one group may be generated equivalently by different generating sets. This can be understood through some rudimentary representation theory, which we defer until the next section.

There is another important property of groups we should define, which is the notion of a normal subgroup. This we may not use, but it is good to know in general. To do so, we must first describe what it means for a group to be self-conjugate. Suppose that $H$ is a subgroup of $G$. Then if $g\in G$, we can define another subgroup called the \textit{conjugate} of $H$ in $G$ by
\begin{equation}
    gHg^{-1} = \{ ghg^{-1} | h \in H \}\,.
\end{equation}
Of course, since the identity element is always in $G$, every group is a conjugate of itself. A group $H$ is only called self-conjugate or \textit{normal} if it is its only conjugate in $G$. If a group has no proper, normal subgroups, it is called \textit{simple}.

Alas, we come to the final aspect of basic group theory with which we will concern ourselves in this section. As with many fields in mathematics, group theory must be able to address a fundamental question: when is one group different from another? To this purpose, we will discuss the group homomorphism and isomorphism.

A group homomorphism is a map $\phi : G \to H$ from one group $G$ to another $H$ such that the group action is preserved. That is, $\forall g_1,g_2 \in  G$,
\begin{equation}\label{eq:grhomo}
    \phi(g_1 * g_2) = \phi(g_1) * \phi (g_2) \,.
\end{equation}
For $\phi$ to be a group homomorphism it must preserve the identity, inverses, and powers. All these criteria result from the above statement. If $\phi$ is additionally one-to-one and onto (thus a bijection) then $\phi$ is called an isomorphism. If two groups $G$ and $H$ are isomorphic, denoted $G \cong H$, then they are the same up to notation. 

To exemplify this, let's look at two previous examples---$C_k$ and $Z_n$. We will show that $C_k \cong Z_n$ when $k=n$. First, recall that $C_k$ is generated by $(1\ 2\ ...\ k)$. Then each element of $C_k$ is just this element multiplied by itself some $m$ times. Define the homomorphism $\phi$ such that $\phi ((1\ 2\ ...\ k)) = 1 $. Then by \eqref{eq:grhomo}, 
\begin{align}
    \phi ((1\ 2\ ...\ k)^m) &=\phi ((1\ 2\ ...\ k))+...+\phi ((1\ 2\ ...\ k)) = m \, .
\end{align}
Thus we have identified each element of $C_k$ with an element of $Z_k$. Note that the order of both groups is $k$, so  we can restrict $m$ such that $m \in \{0,1,...,k-1\}=Z_k$. Clearly, this map is one-to-one and onto by construction. Thus $\phi$ is a group isomorphism and $C_k \cong Z_k$.

With this, we close our review of group theory in order to progress necessarily to representations of groups. Know that group theory as a topic in mathematics is rich and well worth future study beyond the elementary concepts reviewed here. We merely hope that readers completely unfamiliar with its jargon will now be equipped to follow along with relevant discussions pertinent to the results presented in future sections. 

\subsection{Just Enough Representation Theory to be Dangerous}\label{sec:representin}

Groups are all well and good, but we need to know how to apply them to physical quantities to say anything meaningful. After all, it may not be immediately obvious how a quantum state may exhibit dihedral symmetry if we have only defined $D_n$ geometrically. Representation theory saves the day here. Exactly as might be expected, representation theory lets us \textit{represent} a group on a mathematical space that we care about. For this work, we need only consider finite-dimensional group representations. We offer as further reference the books \textit{Representation Theory of Finite Groups} by Benjamin Steinberg \cite[Chapters~2-4]{steinberg2009representation} and \textit{Representing finite groups: a semisimple introduction} by Ambar Sengupta  \cite{sengupta2011representing} for those interested in learning representation theory beyond the context given here.

We note, as a preface, that representations are usually defined in reference to a vector space (typically $V$ or $W$ in this context) over a field ($F$). In quantum computing and indeed quantum physics at large, we almost always consider $F$ to be the complex numbers and $V$ our Hilbert space. This is application dependent, however, and we will suppress such assumptions for now. 

That being said, let's begin. A \textit{representation} $\phi$ of a group $G$ on a vector space $V$ is defined to be a group homomorphism that preserves the action of the group. That is to say, $\phi$ associates to every group element a linear map $\phi :  G \to \text{GL} (V)$ such that
\begin{align*}
    \phi (g) \phi (h) &= \phi (g h) \ \ \  \textrm{and}\\
    \phi(\epsilon) &= \mathbb{I} \,,
\end{align*}   
$\forall g,h \in G$ and $\epsilon$ the identity element. (As an aside, we have chosen to use $\text{GL}(V)$ instead of $\textrm{End}_{F}(V)$ but the general idea remains the same. The map $\phi$ sends elements of the group $G$ to the set of automorphisms on the vector space $V$. In the context of quantum mechanics, matrix representations will be of utmost importance, and this notation fits nicely with that need.) Indeed, these conditions play well with the definition of group homomorphism introduced in the previous section.

We say that a representation is faithful if $\phi$ is a group isomorphism. This can be succinctly enforced by the rule that $\phi(g) \neq \mathbb{I}$ unless $g = \epsilon$. If the identity uniquely maps to the identity and the action of the group is obeyed, it can be quickly determined that $\phi (G)$ is isomorphic to $G$.

What if this is not the case? Suppose we have a map $\phi_t : G \to 1$ where every element of $G$ is mapped to the number 1. Is this even a representation? Well, yes, it obeys both of the above rules in our definition of a representation, although it certainly isn't a faithful one. This is called the trivial representation, and it is literally a textbook example of not being faithful.

There are two other example representations which it will behoove us to cover. These are the standard representation and matrix representations. Consider first the latter. Suppose that $V = F^n$; then if $\phi : G \to \text{GL}_n (F)$, $\phi$ is called a matrix representation.

Now suppose we have a matrix representation of the group $S_n$ such that $\phi : S_n \to GL_n (\mathbb{C})$. Intuitively, such a representation can be obtained by letting the permutations in $S_n$ permute the basis vectors in $\mathbb{C}^n$. Thus the standard representation of a permutation group is defined to be 
\begin{equation}
    \phi(\pi) \coloneqq [e_{\pi(1)} \ ... \ e_{\pi(n)}]\, ,
\end{equation}
where $e_i$ is the $i$-th basis vector and $\pi \in S_n$ is some permutation. Note that since every permutation group is a subgroup of $S_n$ for some $n$, this defines the standard representation for all of them simultaneously. As an example, consider the standard representation of $( 1\ 3)$ acting on $V = \mathbb{C}^3$:
\begin{align}
    \phi((1\ 3)) = 
    \begin{pmatrix} 
    0 & 0 & 1 \\
    0 & 1 & 0 \\
    1 & 0 & 0\\
    \end{pmatrix}\, ,
\end{align}
where it is clear the first and third basis vectors have been swapped.

Again, we ask ourselves that ever-pertinent question: how do we know when two representations are equivalent to each other? Well, two representations $\phi: G \to \text{GL}(V_1)$ and $\psi: G \to \text{GL}(V_2)$ are considered equivalent if there exists an isomorphism or equivalence $T: V_1 \to V_2$ such that $T$ is a linear map satisfying
\begin{equation}
    \psi (g) \circ T = T \circ \phi (g)\, ,
\end{equation}
for all $g \in G$.

Now that we have a way to identify equivalent representations, we can observe a particularly useful fact; all finite dimensional representations are equivalent to matrix representations. To see this, \cite{sengupta2011representing} notes that whenever a basis in $V$ is chosen, if $\textrm{dim}(V) = n < \infty$, then $\phi(g)$ is encoded in the following matrix:
\begin{align}
    \begin{pmatrix}
    \phi (g)_{11} & \phi (g)_{12} & ... & \phi (g)_{1n} \\
    \phi (g)_{21} & \phi (g)_{22} & ... & \phi (g)_{2n} \\
    \vdots & \vdots  & \vdots & \vdots \\
    \phi (g)_{n1} & \phi (g)_{n2} & ... & \phi (g)_{nn} \\
    \end{pmatrix}\, .
\end{align}
This matrix essentially encodes an equivalence between a representation $\phi$ and a matrix group once a basis is specified. Thus we are guaranteed to have a matrix representation for any finite group. If $F$ is the complex numbers and $V$ is a Hilbert space, as is the case in quantum physics, then we can further guarantee that a unitary representation of $G$ exists, which we typically denote $\{U(g)\}_{g\in G}$. 

As an aside, guaranteed unitary representations do not necessarily mean these representations will be useful for quantum computing purposes. Often, we augment these with the requirement that they are \textit{projective} representations, which means they project into the set space of unitary operators modded out by the identity. In other words, projective unitary representations enforce the equivalence
\begin{equation}
    U \cong \lambda U\, ,
\end{equation}
where $\lambda \in \mathbb{C}$. This reflects the irrelevance of global phases in quantum mechanics; we do not consider $e^{i \phi}\ket{\psi}$ to be different from $\ket{\psi}$, and so this should be reflected in our operators.

We move now to concepts that, while not immediately apparent in later chapters, underpin the intuition behind choosing certain representations over others. It may well be surmised that, in the course of examining group symmetries on quantum computers, we will take a group $G$ and a unitary matrix representation of it to act on our space. However, there are many choices of representation that each have their own benefit. When examples of representations are proffered in future chapters, these considerations inevitably played a role in why and how these examples were formed: Is it faithful? What are the invariant subspaces? Is it irreducible---if not, what are the irreducible representations? 

We have already discussed faithfulness, so let us begin with invariant subspaces. The vector space $V$ on which the elements $\phi(g)$ act is called the representation space of $\phi$. (Sometimes it is simply called the representation $V$, but this terminology can quickly become confusing.) If there exists a subspace $W \subseteq V$ such that, $\forall g \in G$,
\begin{equation*}
    \phi (g) W = W\,,
\end{equation*}
then $W$ is called an invariant subspace of $V$. If the only invariant subspaces of $V$ are itself and 0, then the representation $\phi$ of $G$ on $V$ is said to be \textit{irreducible} as long as $V \neq {0}$. As a side note, every one-dimensional representation is an irreducible representation. 

Invariant subspaces give rise to another representation---the subrepresentation. A subrepresentation $\phi |_{W} : G \to \text{GL}(W)$ formed by the restriction of the representation $\phi$ to $W$ is defined as
\begin{equation}
    (\phi |_{W})(w) \coloneqq \phi(w)\, ,
\end{equation}
for all elements $w \in W$. In the same manner that $V$ is sometimes called the representation, $W$ is occasionally referenced as the subrepresentation itself. Unfortunately, this confusing terminology cannot be abolished once established in mathematical texts, so it pays to be aware of it. 

It turns out that every non-trivial representation of a finite group is either irreducible or decomposable, where decomposable means it can be written as a direct sum of two or more proper, nontrivial subrepresentations. Note that if $\phi_1$ and $\phi_2$ are representations of $G$ on $V_1$ and $V_2$ respectively, then the direct sum $\phi_1 \oplus \phi_2$ is a representation of $G$ on $V_1 \oplus V_2$, and similarly for tensor products. 

Here are the fun statements for which all of this machinery has been building: every finite-dimensional representation on a Hilbert space $V$ is the direct sum of irreducible representations. Furthermore, for unitary representations on Hilbert spaces relevant to quantum physics, all irreducible representations of $G$ are one dimensional if and only if $G$ is Abelian. Identifying the relevant decomposition proves greatly useful for investigations into symmetry in the Hilbert space in question. Furthermore, the intuition used to generate examples in future chapters heavily relies on these facts. 

Throughout this section, we have hopped from concept to concept haphazardly to cover the bare essentials required to comprehend the remainder of the text. To be fair, this review was prefaced to be dangerous in nature and should not be considered the ultimate authority on representation theory. Far from it, we have merely hoped to convey succinctly concepts that are not necessarily part of a standard physics curriculum but will recur within this text nonetheless.

\section{Quantum Information and Computation Background}\label{sec:qibg}
This work is written, first and foremost, as a thesis of work in the field of physics. As such, we take for granted a certain modicum of knowledge in any potential readers. A review must always start somewhere, and while defining the notion of a quantum mechanical state or ideas of entanglement and superposition would prove for a grand and thorough background, such explanations would necessarily extend well beyond our scope of interest. Therefore, in this section we review only so much as a typical physics student might require to acquaint themselves with topics in quantum information and computation. Nonetheless, we shall find no lack of topics to discuss.

In Section~\ref{sec:qcbg}, we undertake the Herculean task of reviewing all relevant terms in quantum computing. We begin at a relatively elementary level by defining basic building blocks such as qubits, density matrices, and quantum channels. Additionally, we go over how to read a quantum circuit intuitively and give reference to how they can be viewed as tensor networks. From this point, we advance to relevant norms and lemmas used throughout the literature on quantum information and computing.

In Section~\ref{sec:ctbg}, we take a moment to review a smattering of complexity theory, only so much as to contextualize later results. This will include a cursory explanation of how complexity classes contextualize computational capability, as well as introduce some relevant complexity classes. In particular, we define as background P versus NP, QIP(n), and DQC1. If all of those acronyms prove impenetrable now, fear not, for they will be contextualized in what follows.

\subsection{Necessarily Thorough Review of Relevant Terms}\label{sec:qcbg}
\subsubsection{1.3.1.1. Qubits, Quantum States, and Quantum Gates}
To begin, we will review the primary building blocks of quantum computing. For a thorough education in quantum computing for someone totally unfamiliar with the field, the textbook by Nelson and Chuang \cite{nielsenchuang} is a fantastic and standard reference. However, we can quickly recount some highlights here for any readers less acquainted with the material. (Quickly being a subjective term, we suggest a more experienced reader may well skip this subsection entirely.)

The natural starting point for quantum computation is the qubit or ``quantum bit". Qubits are the most basic method for encoding quantum states into a form suitable for computation, analogous to the usual bit from classical computing. Consider a pure state, $\ket{\psi}$, a vector that lives in some two-dimensional Hilbert space $\mathcal{H}_A$. (We will often write $\ket{\psi}_A$ to indicate explicitly which Hilbert space our state is contained in.) Then a qubit represents this state as
\begin{equation}
    \ket{\psi}= \alpha \ket{0} + \beta \ket{1} \, ,
\end{equation}
where $|\alpha|^2 + |\beta|^2 =1$, and $\ket{0}= \begin{bmatrix} 1\\ 0 \end{bmatrix}$ and $\ket{1}= \begin{bmatrix} 0\\ 1 \end{bmatrix}$ are referred to as the computational basis vectors. Unless otherwise indicated, we will always use the computational basis to express our quantum states.

The actual physical meaning behind the computational basis vectors is heavily dependent on computational architecture. For instance, in an optical scenario, $\ket{0}$ could correspond to ``no photons" and $\ket{1}$ to ``one photon" in a clear analogy to how classical bits are often representative of a presence or lack of voltage. However, quantum computers vary greatly in their current implementations, so architecture is not considered here. Instead, it is helpful to view a qubit as a vector on the Bloch sphere, as in Figure~\ref{fig:bloch}. For this visualization, we take $\alpha = \cos{\theta/2}$ and $\beta = e^{i \phi}\sin{\theta/2}$.

\begin{figure}[h!]
\begin{center}
\includegraphics[width=3.5 in]{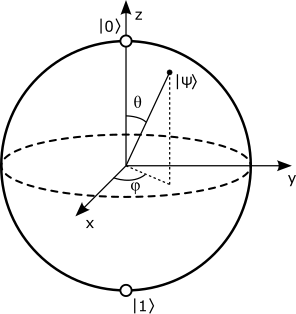} 
\caption{\label{fig:bloch} Illustration of a qubit as a vector on the Bloch Sphere. Image credit: Smite-Meister - Own work, CC BY-SA 3.0, \href{https://commons.wikimedia.org/w/index.php?curid=5829358}{url link} 
}
\end{center}
\end{figure}

Now, given some qubit, how do we act on it to perform computations? This is achieved using unitary operators. Unitary operators can be described intuitively as ``preserving probability" and are pictured as rotations on the Bloch sphere. In mathematics, this idea of preserving probability is more properly referred to as unitaries preserving the inner product. Nevertheless, we can see that acting on a single qubit by a two-by-two square unitary will result in another properly normalized quantum state. Such unitary operators are often referred to as quantum logic gates or simply quantum gates. 

Let us now take the time to define some common quantum gates. First, the Hadamard gate, normally denoted simply as $H$, is given by 
\begin{align}
    H\coloneqq\frac{1}{\sqrt{2}}
    \begin{pmatrix} 
    1 & 1 \\
    1 & -1 \\
    \end{pmatrix}\, .
\end{align}
Note each single qubit gate can be characterized by how it acts on the basis states. The Hadamard gate acts by sending $\ket{0}$ to the state $\ket{+}\coloneqq \frac{1}{\sqrt{2}}(\ket{0}+\ket{1})$ and $\ket{1}$ to the state $\ket{-}\coloneqq \frac{1}{\sqrt{2}}(\ket{0}+\ket{1})$.

Next, the NOT gate is given by
\begin{align}
    X\coloneqq
    \begin{pmatrix} 
    0 & 1 \\
    1 & 0 \\
    \end{pmatrix}\, .
\end{align}
As one might expect, the NOT gate sends each basis state to the other just as the classical NOT does. Why, then, do we denote it with the letter $X$? That is because this is exactly the Pauli-X matrix---or $\sigma_1$--- that we know and love from quantum mechanics. In fact, all of the Pauli matrices are realized exactly as quantum gates for computation and adopt the moniker of ``Pauli gates" to suggest this.

The next set of gates that any review of quantum computing should acknowledge are the rotation gates, of which there are three:
\begin{align}
    R_x(\theta) \coloneqq
    \begin{pmatrix} 
    \cos{\theta/2} & -i \sin{\theta/2} \\
    -i \sin{\theta/2} & \cos{\theta/2} \\
    \end{pmatrix}\, ,
\end{align}
\begin{align}
    R_y(\theta) \coloneqq
    \begin{pmatrix} 
    \cos{\theta/2} & - \sin{\theta/2} \\
     \sin{\theta/2} & \cos{\theta/2} \\
    \end{pmatrix}\, ,
\end{align}
\begin{align}
    R_z(\theta) \coloneqq
    \begin{pmatrix} 
    e^{-i\theta/2} & 0 \\
     0 & e^{i\theta/2}  \\
    \end{pmatrix}\, .
\end{align}

Of course, there are many other gates of interest in quantum computing, infinitely so. However, as we have only finitely many pages to discuss such things, there are two more gates that should be defined for the uninitiated. These two gates are the CNOT and SWAP gates which differ from our previous examples in that they are two qubit gates. 

Now, we have yet to discuss what it means to have a system of more than one qubit, but the construction is simple enough. Suppose one qubit lives in the Hilbert space $\mathcal{H}_A$ and another in the Hilbert space $\mathcal{H}_B$. Then the total state of both qubits must be contained in the tensor product Hilbert space $\mathcal{H}_{AB} \coloneqq \mathcal{H}_A \otimes \mathcal{H}_B$. However, it is not necessarily the case that the total state of the system, call it $\ket{\psi}_{AB}$, can be decomposed into a tensor product of states on each individual Hilbert space in such a manner. That is, it is not necessarily true that $\ket{\psi}_{AB}= \ket{\phi}_A \otimes \ket{\sigma}_B$ for $\ket{\phi}_A,\, \ket{\sigma}_B$ pure states. Such product states are necessarily separable and thus cannot describe entangled states on the two systems. Naturally, quantum gates are simply unitaries that act on this larger system and any further addition of qubits is accounted for in much the same way.

Now we are well-equipped to discuss the CNOT and SWAP gates. The CNOT or controlled-NOT gate acts on two qubits. One qubit is the ``control", and the other is the ``target". When the state of the control qubit is $\ket{0}$, nothing happens. When the state of the control qubit is $\ket{1}$, however, a NOT gate is performed on the target qubit. When combined with superpositions of states, this gate is immensely powerful. Denoting this action as a matrix is a bit tricky, because it depends on the orientation of the qubits, but supposing the control is the first qubit and the target the second, convention gives the following matrix:
\begin{align}
    \text{CNOT} \coloneqq
    \begin{pmatrix} 
    1 & 0 & 0 & 0\\
     0 & 1 & 0 & 0  \\
     0 & 0 & 0 & 1 \\
     0 & 0 & 1 & 0 \\
    \end{pmatrix}\, .
\end{align}
Note that this is not the only controlled gate, although it is likely the most common. Any unitary gate can, in principle, have a control appended onto it.

The final gate I will be defining explicitly is the SWAP gate. This gate does exactly what it says on the label---it takes two qubits and swaps them. It is given by
\begin{align}
    \text{SWAP} \coloneqq
    \begin{pmatrix} 
    1 & 0 & 0 & 0\\
     0 & 0 & 1 & 0  \\
     0 & 1 & 0 & 0 \\
     0 & 0 & 0 & 1 \\
    \end{pmatrix}\, .
\end{align}

As mentioned before, this is far from an exhaustive list. There are infinitely many gates that could be called for in any given algorithm. Actually constructing these gates for use on real hardware is an issue often referred to as quantum compiling (see, e.g., \cite{moro2021quantum, khatri2019quantum, sharma2020noise}) which is a rich research topic all of its own. For the scope of this work, it should merely be noted that such a thing is possible.

Realizing these quantum gates on a quantum computer can be done in a number of ways, but the general approach is typically the same. Each system will have some set of ``native gates" that can be performed on the physical qubits. In order to act with some arbitrary unitary gate, it must be decomposed into these elementary gates as shown in \cite{Elementary}. There are two main takeaways pertinent to this discussion: all single unitary gates can be approximated by a minimal native gate set satisfying some conditions as per the results of the Solovay-Kitaev algorithm \cite{dawson2005solovay,bouland2021efficient} and all multi-qubit unitary gates and controlled-unitary gates can be continuously decomposed into only single qubit unitaries and CNOT gates \cite{Elementary}.

At such a point, we have reviewed nearly all primary building blocks of quantum computing save one: mixed states and density matrices. Everything stated so far about qubits has been in terms of pure states---but mixed states should not be neglected. A mixed state $\rho$ is a probabilistic mixture of pure states of the form
\begin{equation}
    \rho = \sum_i p_i \ket{\psi_i}\bra{\psi_i}\, ,
\end{equation}
where $\ket{\psi_i}\bra{\psi_i}$ is the density matrix associated the pure state $\ket{\psi_i}$. (This can be thought of as the outer product between $\ket{\psi}$ and its dual.) In order to ensure that it describes a quantum state, $\rho$ must be positive semi-definite $\rho \geq 0$ and have trace equal to one, $\tr[\rho]=1$.
Obviously, a pure state can also be denoted this way. Thus this is the more general formalism to discuss quantum states.

So the more general quantum state is a mixed state, but we have formulated everything in terms of pure states. Why? Well, there is a result in quantum information that all mixed states can be viewed as a pure state on a higher dimensional Hilbert space that has had some part of it traced out \cite[Chapter 5]{wildebook}. This is called the purification of the mixed state. The pure state $\ket{\psi_\rho}_{AB}$ is a purification of $\rho_A$ if
\begin{equation}
    \rho_A = \tr_B[\ket{\psi_\rho}\bra{\psi_\rho}_{AB}]\, .
\end{equation}

With this fact, we can confidently construct algorithms with pure states in mind and know that there will still be a way to implement them on mixed states as well.

\subsubsection{1.3.1.2. Quantum Channels}

Quantum channels are another prevalent term in quantum information. We separate this term out as, while it is certainly a prolific and important concept, it is often neglected in introductory approaches to quantum computing. Texts that focus on quantum information such as \cite{wildebook} will have no lack of channels and thus serve as a good resource for those unfamiliar with the formalism.

Quantum channels are completely positive, trace-preserving maps that take one quantum state to another. A wise man once said, ``Everything is a quantum channel." This might be overstating things a tad, but nonetheless, the insight remains somewhat true. A unitary gate is a quantum channel. The identity is a quantum channel. Teleportation is a quantum channel. The process of the environment stealing coherence from our experiments is a quantum channel. So on and so forth we continue. 

Quantum channels display their usefulness best, in this author's humble opinion, in two circumstances. The first is when describing actions on density matrices. The second is for non-unitary state evolutions. The latter may grab attention more strongly than the former but its physicality is demonstrated rather simply. Suppose Alice is communicating some quantum state of $d$ qubits to Bob when a nefarious eavesdropper intercedes. This bad actor takes all of Alice's states and replaces them with the maximally-mixed state $\pi \coloneqq \mathbb{I}/d$. This `trace-and-replace' system may not be unitary, but it is a quantum channel. A less dramatic example is the loss of information in the form of noise introduced by the environment. 

Consider for now a unitary channel $\mathcal{N}$ acting on the mixed state $\rho$ via the unitary $U$. Then we can define its action as
\begin{equation}
    \mathcal{N} = U \rho U^\dag = \sum_i p_i U \ket{\psi_i}\bra{\psi_i} U^\dag \, .
\end{equation}
This is a good picture to have in mind when wondering how unitary gates affect mixed states. In general, a quantum channel (often denoted $\mathcal{N}$) acts as a map
\begin{equation}
    \mathcal{N}_{A\to B}( \cdot ): D_{\mathcal{H}_A} \to D_{\mathcal{H}_B}
\end{equation}
where $D_{\mathcal{H}}$ is the space of density matrices on some Hilbert space.

This concludes our toolkit of the essentials. Now, we move on to visualizations.

\subsubsection{1.3.1.3. Reading a Quantum Circuit Diagram for Beginners}

Now that we have acquainted ourselves with the relevant building blocks, we may well wish to begin constructing algorithms. To facilitate this, we turn to pictorial depictions of circuits implementable on a quantum computer. Quantum circuit diagrams are an invaluable tool that allow us to visualize sometimes immensely complex mathematics. They can be understood through the lens of tensor networks---for which Biamonte and Bergholm's ``Tensor Networks in a Nutshell" \cite{biamonte2017tensor} is a fantastic roadmap. We will leave all of the intricacies of tensor networks in their wise hands and cover only the very basics of quantum circuit diagram literacy. 

First, we quickly run down the basic rules. Each circuit diagram is composed of wires or lines that are directly correlated to a specific Hilbert space, sometimes labelled but often not. Time progresses from left to right along these lines. Gates are usually indicated by labelled boxes that intersect these lines. Gates can be understood to act only on the Hilbert spaces that they intersect. Inputs to the circuit are kept to the left-most side and outputs to the right. A single-line wire is understood to be quantum in nature whilst double lines indicate classical communication. A dotted line, however, is usually meant to indicate that there is a space-like separation (or at least, some sort of temporal separation).

Out of breath yet? Don't worry, because the rest of this explanation is accompanied by pictures upon which the above rules can be practiced until understood. Figure~\ref{fig:gates} gives a list of commonly used logic gates and their relevant symbols. In the interest of space and relative brevity, we suggest referencing the list of common gates on pages xxvii and xxviii of the tenth anniversary edition of Nielsen and Chuang \cite{nielsenchuang} for a slightly more extensive list.

\begin{figure}[h!]
\begin{center}
\includegraphics[width=\columnwidth]{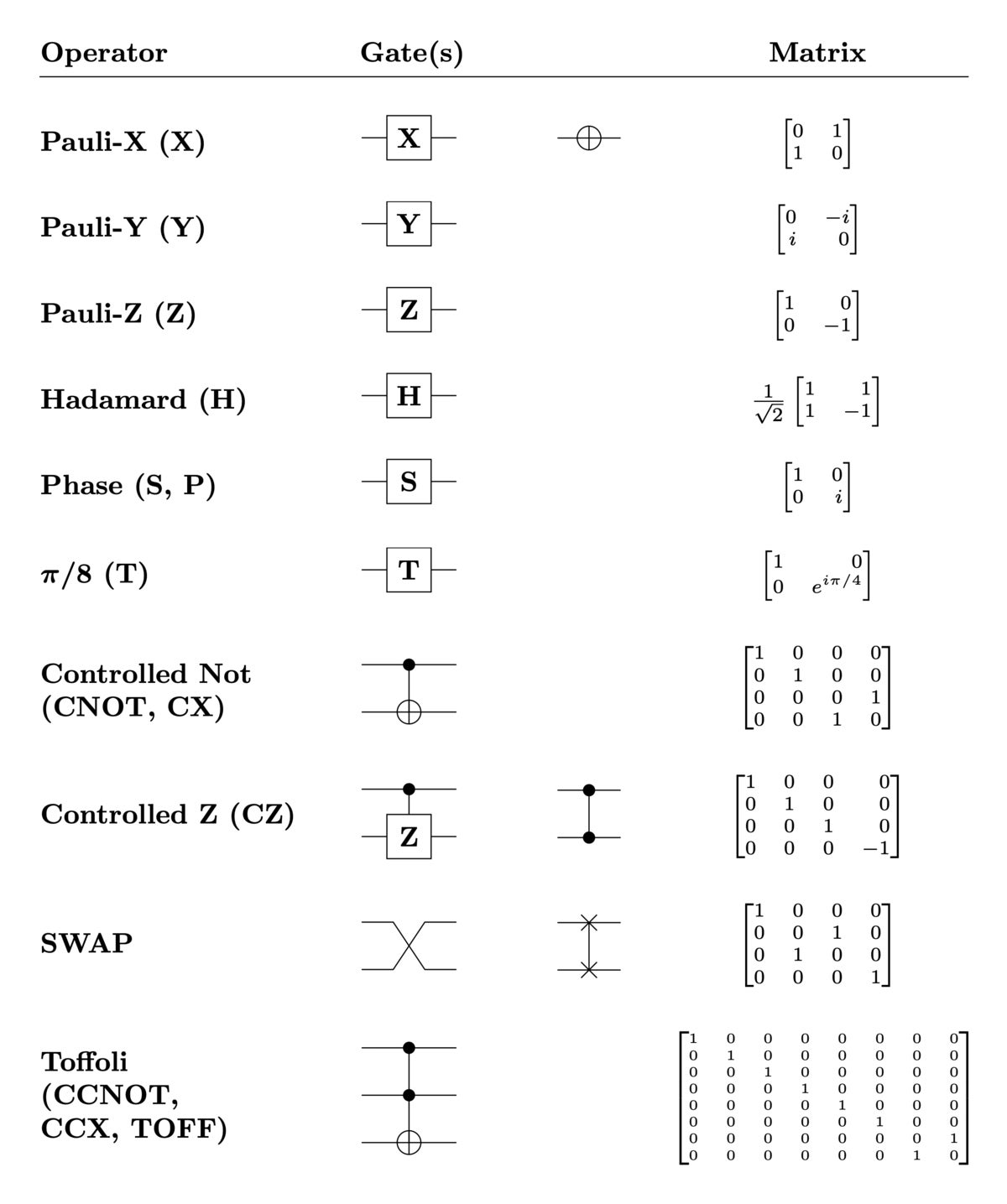} 
\caption{\label{fig:gates} Commonly used quantum logic gates, their circuit diagram representations, and their typical matrix representation.
}
\end{center}
\end{figure}

Some of the symbols in Figure~\ref{fig:gates} might be slightly confusing---namely, the controlled gates. We know that controlled gates have specified control and target qubits, and we need to be able to differentiate between them. A control is usually indicated by a closed circle with a vertical wire connecting it to the unitary on the target qubit. A gate can have more than one control. Additionally, gates can also be conditioned on the $\ket{0}$ state rather than the $\ket{1}$ state, and this is indicated by an open circle instead.

\begin{figure}[h!]
\begin{center}
\includegraphics[width=3.5 in]{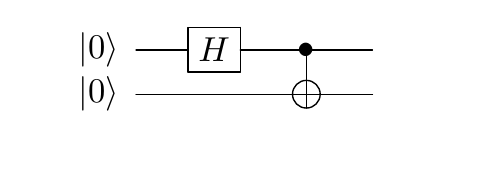} 
\caption{\label{fig:bell} The Bell state circuit. Two qubits in the state $\ket{00}$ are acted on first by a Hadamard gate and then by a CNOT gate. The output will be the maximally-entangled state on two qubits.
}
\end{center}
\end{figure}

Let's look now at a simple example circuit. Figure~\ref{fig:bell} shows the Bell state circuit, so named because depending on the input, it can generate any of the four maximally-entangled Bell states. Reading the circuit from left to right, we see the input state go from $\ket{00}\ \to H \otimes \mathbb{I} \ket{00}\ \to \text{CNOT}(H \otimes \mathbb{I} \ket{00})$. Working through this we get that the output state is given by
\begin{align}
    \text{CNOT}(H \otimes \mathbb{I} \ket{00}) &= \text{CNOT}(\frac{1}{\sqrt{2}} (\ket{00} + \ket{10})) \\
    &=\frac{1}{\sqrt{2}} (\ket{00} + \ket{11}\, .
\end{align}

This may seem too simple an example, but even more complicated circuits are read the same way. Just progress from left-to-right implementing each gate in its turn. 

\begin{figure}[h!]
\begin{center}
\includegraphics[width=3.5 in]{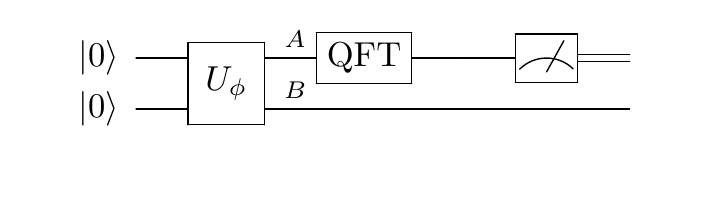} 
\caption{\label{fig:abstract} A more abstract circuit. A multiqubit gate acts on $\ket{00}$ to create the state $\ket{\phi}_{AB}$. Then a quantum Fourier transform (QFT) is performed only on the first qubit, which is then measured.  
}
\end{center}
\end{figure}

Next, let's consider a more abstract example. Figure~\ref{fig:abstract} depicts a much more abstract circuit for us to work through. First, define $U_{\phi}$ to be the multi-qubit unitary such that $U_{\phi} \ket{00} = \ket{\phi}_{AB}$ for $ \ket{\phi}_{AB}$ some pure state. Then act on only on the $A$ system with a quantum Fourier transform (QFT), and measure that system. The output of any measurement is classical, so this is denoted with doubled wires. The output of this circuit is not important, per se, but it allows us to discuss various common occurrences in circuit diagrams. For instance, the second qubit of the system is thrown out and never measured. Essentially, our measurement only considers the reduced system $\rho_A = \tr_B[(\textrm{QFT} \otimes \mathbb{I}) \ket{\phi}\bra{\phi}_{AB} ]$. This demonstrates how mixed states are often created and denoted in circuit diagrams. Additionally, we introduce the QFT as a black-box operator that, in principle, puts a state into a superposition of some new set of basis states. (For a single qubit, the QFT is actually equal to a Hadamard, but we just wanted to have it in the circuit for discussion purposes.) For a good definition and example of the use of the QFT in quantum algorithms, see Shor's algorithm \cite{shor1994algorithms} or review Nielsen and Chuang \cite{nielsenchuang}.

A note of caution: although typically each wire identifies a single qubit, these diagrams aren't always so friendly. Pay attention to context clues to determine if each wire indicates a qubit or a Hilbert space. In this work, we will identify our lines with Hilbert spaces more so than qubits, but this should not lead to much confusion for a careful reader.

\subsubsection{1.3.1.4. Advanced Terms: Norms and Important Facts}\label{sec:normsandlemmas}

This section will consist heavily of commonly used norms and some choice lemmas that will be employed in this work. As these terms are more technical, we will simply be tersely listing them off for best reference. See \cite{wildebook} as the source of many of these definitions unless otherwise specified.

First, we define some common norms. As a preliminary, we take the absolute value of an operator $M$ in $\mathcal{L}(\mathcal{H},\mathcal{H'})$ to be $|M| \coloneqq \sqrt{M^\dag M}$.

\begin{definition}\textbf{The Trace Norm}---The trace norm, or Schatten-1 norm, $\left \Vert M \right \Vert_1$ of an operator $M$ is defined to be
\begin{equation}
    \left \Vert M \right \Vert_1 \coloneqq \tr[\sqrt{M^\dag M}] = \tr[|M|]\,.
\end{equation}
\end{definition}

\begin{definition}
\textbf{The Hilbert--Schmidt Norm}---The Hilbert--Schmidt norm, or Schatten 2-norm, of an operator $M$ is defined to be 
\begin{equation}
    \left\| M \right\|_2 \coloneqq \sqrt{\operatorname{Tr}[|M|^2]}\, .
\end{equation}
\end{definition}

\begin{definition}
\textbf{The Schatten $p$-Norm}---The Schatten $p$-norm for $p \geq 1$ of an operator $M$ is defined to be 
\begin{equation}
    \left\| M \right\|_p \coloneqq \operatorname{Tr}[|M|^p]^{1/p} \, .
\end{equation}
\end{definition}

Note that the Schatten $p$-Norm encompasses the first two definitions as well. Still, both are typically identified by name in quantum computing as both are common in the field. 

It will be beneficial to now introduce some relevant mathematical tools. The following three (four) lemmas will be used to prove relevant results in later chapters. (It is pure coincidence that these statements are all lemmas; they come from independent works and do not build to any theorems here. As is often the case, lemmas simply tend to present amazing and useful facts.)

First, we will tackle Schur's Lemma. Indeed, there is some ambiguity here. Despite both often being referenced simply as ``Schur's Lemma" there is both a concept of Schur's \textit{first} lemma and Schur's \textit{second} lemma. Both will be combined here into a single definition. To add to the confusion, sometimes Schur's lemma is defined in textbooks as a theorem! Will the notation abuse never cease? Below, we give both Schur's Lemma as is used often used in quantum computing \cite{RevModPhys.79.555} as well as Schur's lemma as presented in representation theory \cite{sengupta2011representing}.

Now, we give the first and second lemmas as presented in \cite{RevModPhys.79.555}.
\begin{theorem}[Schur's Lemma] 
Schur's first and second lemmas are as follows:

\begin{enumerate}
        \item If $T(g)$ is an irreducible representation of the group $G$ on the Hilbert Space $\mathcal{H}$, then any operator $A$ satisfying $T(g) A T^\dag (g) = A$, $\forall g \in G$, is a multiple of the identity on $\mathcal{H}$.
    \item If $T_1(g)$ and $T_2(g)$ are inequivalent representations of $G$, then $T_1(g)AT^\dag_2(g) =A$, $\forall g \in G$, implies $A =0$.

\end{enumerate}
\end{theorem}

In \textit{Representing Finite Groups: a Semisimple Introduction}, Sengupta calls Schur's lemma ``the Incredible Hulk of representation theory"---certainly, a ringing endorsement. In honor of this, we will restate Schur's lemma and further give the proof of it as given in \cite{sengupta2011representing}. After all, Schur's lemma \cite{schur1905neue} was originally a result of representation theory itself, so proving it in this context is perhaps the best use of the tools given in Section~\ref{sec:representin}.

\begin{theorem}[Schur's Lemma (Again)]
A morphism between irreducible representations is either an isomorphism or 0. That is, if $\phi_1$ and $\phi_2$ are irreducible representations of a group $G$ on vector spaces $V_1$ and $V_2$, over an arbitrary field $F$, and is $T:V_1 \to V_2$ is a linear map where $\forall g \in G$,
\begin{equation}\label{eq:schur1}
    T \phi_1 (g) = \phi_2 (g) T \, ,
\end{equation}
then $T$ is either invertible or $0$.

If $\phi$ is an irreducible representation of a group $G$ on a finite-dimensional vector space $V$ over an algebraically closed field $F$ and $S: V \to V$ is a linear map where $\forall g \in G$,
\begin{equation}\label{eq:schur2}
    S\phi(g) = \phi(g)S \, ,
\end{equation}
then $S=c\mathbb{I}$ for some scalar $c \in F$.
\end{theorem}
\begin{proof}
Suppose $\phi_1,\phi_2,$ and $T$ are as given above. From \eqref{eq:schur1}, we can determine that the kernel of $T$, $\operatorname{ker}(T)$, is invariant under the action of the group $G$. Since $\phi_1$ is irreducible, and thus the only invariant subspaces are $0$ and $V$, it follows that $\operatorname{ker}(T)$ is either $0$ or $V$ itself. Then if $T \neq 0$, it must be injective. Similarly, the image of $T$, $\operatorname{Im}(T) \subset V_2$ is also invariant under the action of the group, and thus if $T \neq 0$, it must be surjective. Therefore, either $T=0$ or $T$ is an isomorphism.

For the second part, suppose $F$ is algebraically closed. (This is the case for the complex numbers $\mathbb{C}$.) Further suppose that $V$ and $S$ are as stated, and $S$ is an intertwining operator from the irreducible representation $\phi$ on $V$ to itself. Note that $S - c\mathbb{I} \in \text{GL}_n (F)$ is not invertible, as 
\begin{equation*}
    \operatorname{det}(S-\lambda \mathbb{I} )=0
\end{equation*}
has a solution for $\lambda = c \in F$. Observe that \eqref{eq:schur2} holds even if when replacing $S$ with $S- c\mathbb{I}$. Thus, by \eqref{eq:schur1}, $S- c\mathbb{I} = 0$.
\end{proof}

By inspection, this second statement of Schur's lemma encapsulated the first. This may well be expected, as the use case in quantum computing is much more limited. Suffice to say, this lemma-come-theorem strikes heavy significance in any application where representations of groups occur.

Our heavy hitter out of the way, we give two more lemmas simply because they will be used, and it is always helpful to have such things explained in house. The gentle measurement and gentle operator lemmas found in \cite{Davies1969,itit1999winter,ON07} are recreated below.

\begin{lemma}[Gentle Measurement Lemma]\label{gentleM} Given a density operator $\rho$ and a measurement operator $\Lambda$ with $0 \leq \Lambda \leq \mathbb{I}$, suppose that $\Lambda$ has a high probability of detecting the state $\rho$ such that
\begin{equation}
    \tr[\Lambda \rho] \geq 1 - \epsilon\,,
\end{equation}
where $\epsilon \in [0,1]$. (The probability is considered high if $\epsilon$ is close to zero.) Then the post-measurement state 
\begin{equation}
    \rho ' \equiv \frac{\sqrt{\Lambda} \rho \sqrt{\Lambda}}{\tr[\Lambda \rho]} \, ,
\end{equation}
is $2\sqrt{\epsilon}$-close to the original state $\rho$ in trace distance. That is,
\begin{equation}
    \left \Vert \rho - \rho' \right \Vert_1 \leq 2\sqrt{\epsilon} \, .
\end{equation}
Thus, the measurement does not disturb the state $\rho$ much if $\epsilon$ is small.
\end{lemma}

\begin{lemma}[Gentle Operator Lemma]\label{gentleO}
Given a density operator $\rho$ and a measurement operator $\Lambda$ with $0 \leq \Lambda \leq \mathbb{I}$, suppose that $\Lambda$ has a high probability of detecting the state $\rho$ such that
\begin{equation}
    \tr[\Lambda \rho] \geq 1 - \epsilon\,,
\end{equation}
where $\epsilon \in [0,1]$. (The probability is considered high if $\epsilon$ is close to zero.) Then $\sqrt{\Lambda}\rho\sqrt{\Lambda}$ is $2\sqrt{\epsilon}$-close to the original state in trace distance. That is,
\begin{equation}
    \left \Vert \rho -  \sqrt{\Lambda}\rho\sqrt{\Lambda} \right\Vert_1 \leq 2\sqrt{\epsilon}\, .
\end{equation}
\end{lemma}

Proofs of both of these lemmas can be found in Chapter 1 of \cite{wildebook} alongside other related quantities. 

\subsubsection{1.3.1.5. Parting Words}

Thus we conclude our review of terms in quantum information theory, including both elementary concepts and more advanced ones. Readers fully unfamiliar with the topics introduced here are encouraged not to rely on this work alone, but instead to investigate any of the references given. With such context out of the way, the way forward holds more interesting and niche information that would be of greater interest to more advanced audiences. The next subsection, for instance, introduces quantum complexity theory, which should rouse anyone bored from the material reviewed so far.

\subsection{Complexity Theory but Only So Far as We are Concerned}\label{sec:ctbg}

The aim of this section is to acquaint any unwary reader with the smattering of complexity theory terms that will inevitably make themselves known in this work. For an enjoyable sojourn on quantum computing that relates the topic to complexity theory, we proffer Scott Aaronson's ``Quantum Computing Since Democritus" \cite{aaronson2013quantum}. Do not think, however, that upon completing this section alone a reader would be fit to debate hot topics in complexity theory with Aaronson himself. We cover only the classes of P, NP, QIP(n), and DQC1, and even then we do so in a rather cursory manner.

Generally, complexity classes concern themselves with one primary question: ``As the problem grows larger in scale, how many more resources does the solution require?" These can be either time or space resources, but generally speaking most limit themselves to time. Space rarely proves to be the limiting factor. As such, many complexity classes define themselves primarily with time taken to solve some problem. The class P is a prime example of this: it is the class of problems that are solvable in polynomial time on a classical computer.

We believe---indeed, the field in general believes---that any quantum computer worth its salt will be able to solve problems in the class P efficiently. Why? Because problems in P are already efficiently solvable on classical computers, and few people interest themselves in making a worse machine than we already have. The interesting question comes about when we discuss P versus NP.

Here is the rub with complexity classes--their distinctions are rather fuzzy. Indeed, disguising if P is equal to NP is one of the Millennium Prize Problems \cite{carlson2006millennium}, so it must be a difficult feat. For our purposes, we will assume P $\neq$ NP in the grand tradition of many complexity theorists before us. But what is NP? Again, before defining these terms, we will meander to a few preliminary concepts first.

Every computational class has a set of problems that define a class. That is, every problem in the class can be mapped to these problems in polynomial steps. These are referred to as being complete, e.g., NP-Complete. Often, decision problems hold this seat. A decision problem is an algorithm that returns `yes' if some verifiable question is true and `no' otherwise. NP is short for ``nondeterministic polynomial time", and it is the class of decision problems that a probabilistic classical computer could solve in polynomial time. As that is rather esoteric and we would rather not delve too deeply into the implications of probabilistic Turing machines, we will adopt a more colloquial definition of NP. NP is often thought of as the class of decision problems whose solutions are ``easy to check, hard to find." It is a widespread belief that NP problems are difficult for classical computers to solve efficiently.

We have introduced this concept of completeness in our definition so that we may make some rather tenuous statements about relative hardness of algorithms. (Complexity theory is always rather tenuous as it is very hard to prove that one class of problems is definitively more difficult than another.) The current attitude of the quantum computing community is that quantum computers should also be able to solve some problems in NP efficiently. NP-Complete or NP-Hard problems are not generally included in this number. As evidence for this claim, most present Shor's algorithm \cite{shor1994algorithms} which allows quantum computers to factor large numbers in polynomial time on a quantum computer. Factoring is classically an NP problem---that is, there is no currently known classical algorithm to efficiently factor large numbers---but not NP-Complete. Thus, this result indicates that there exists some sweet spot of computational problems that are only efficient on quantum computers.

Now that we have established that there exists a class of easy problems (P), a class of difficult problems (NP-Complete), and a class of interesting problems (NP), we can restrict ourselves to definitions of quantum complexity classes. There are two which we will need for further reading: DQC1 and QIP(n).

\subsubsection{1.3.2.1. DQC1}\label{sec:dqc1def}

DQC1 stands for ``deterministic quantum computation with one clean qubit", or ``one clean qubit" colloquially. The one clean qubit problem models systems with one pure state qubit and the rest of the system consisting of maximally-mixed qubits at the input. As one clean qubit is complete for DQC1, is can be considered emblematic of the class. In what follows, we review the definition of this class \cite{KL98,SJ08}. An important note on the class, however, is that its problems are deemed classically intractable, thus making it a complexity class of interest in quantum computing.

Suppose we have $n$ qubits. The basic model involves preparing one qubit in a pure state $\ket{0}\!\bra{0}$ and all other qubits in the maximally mixed state $\pi \coloneqq \mathbb{I}/d$ where $d\coloneqq 2^n$. Further suppose we have a quantum circuit generating a unitary $U$ and perform this on all of the qubits. Then the first qubit is measured in the computational basis $\{|0\rangle\!\langle 0|,|1\rangle\!\langle 1|\}$. The algorithm accepts if the outcome $|1\rangle\!\langle 1|$ occurs. The problem of estimating the acceptance probability $\operatorname{Tr}\!\left[(|1\rangle\!\langle 1|\otimes I)U(|0\rangle\!\langle 0|\otimes I/d)U^\dag\right]$ to within additive error is a DQC1-complete problem by definition.

An insight of \cite[Section~1]{SJ08} is that the problem of estimating $\operatorname{Re}[\operatorname{Tr}[U]]/d$ to within additive error, where $U$ is the unitary realized by a quantum circuit acting on $n$ qubits and consisting of polynomially many gates, is a DQC1-complete problem. This means that the problem can be solved within the computational model mentioned above, and it is also just as hard as every other problem that can be solved in the model. Thus, this problem of normalized trace estimation characterizes the class DQC1. Another key observation of \cite[Section~1]{SJ08} is that the complexity class DQC1 does not change if there are a constant or even logarithmic number of pure qubits, where here we mean logarithmic in $n$.

\subsubsection{1.3.2.2. QIP(n)}\label{qipn}

In classical computational complexity theory, there is a concept of something called a ``Merlin-Arthur Proof". This class proffers a decision problem wherein an all-powerful prover with infinite computational resources interacts with a practically-limited verifier. This prover is analogized as the wizard Merlin and the verifier as the human king Arthur, who must take the information Merlin provides him and make a final decision. 

If this all sounds like fantasy, let us recontextualize. In this formalism, the verifier must accurately decide whether or not some input satisfies a decision problem. The prover acts as an oracle who wishes to maximize the acceptance probability of the algorithm. The prover is not limited by time or space but is limited by the laws of nature. Colloquially, if a solution exists, then the prover will find it. If not, the prover will attempt to fool the verifier by choosing an input as close to the acceptance condition as possible. When all parties involved are quantum mechanical (meaning, for our purposes, that the verifier has a resource-limited quantum computer and the all-powerful wizard has an unlimited quantum computer) then this situation describes a quantum interactive proof (QIP) \cite{W09,VW15}.

\begin{figure}[h!]
\begin{center}
\includegraphics[width=3.5 in]{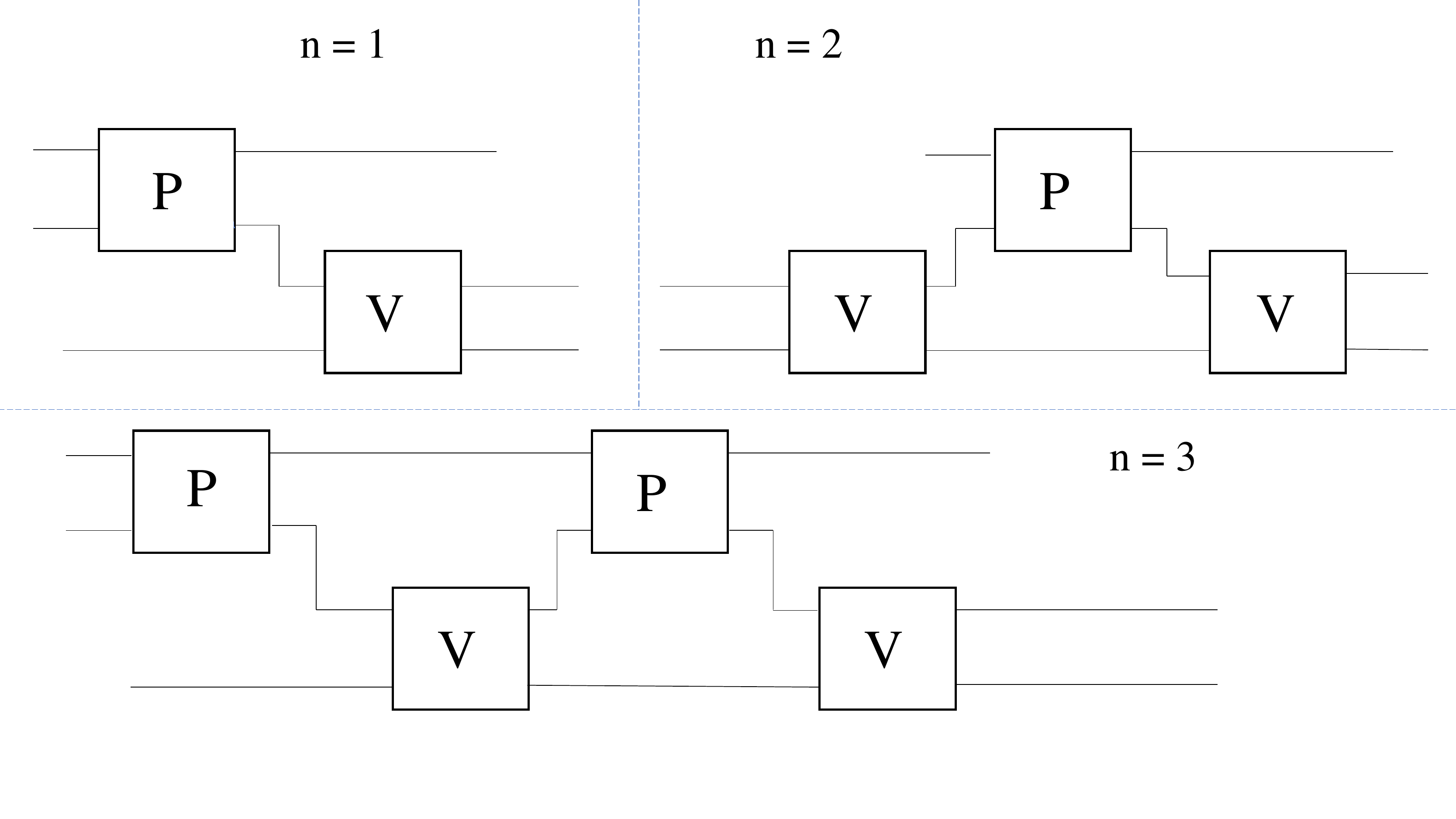} 
\caption{\label{fig:qip} A cartoon demonstrating the different numbers of messages exchanged between a prover system $P$ and a verifier system $V$. This figure shows the basic form of QIP(n) algorithms for n=1, n=2, and n=3.
}
\end{center}
\end{figure}

These quantum interactive proofs are characterized by a parameter $n$, which refers to the number of messages exchanged between the verifier and the prover, QIP(n). A system in which no messages are exchanged, QIP(0), is considered to be efficient on a quantum computer (equivalent to BQP, the so-called quantum analogy of P). If the prover acts as an oracle sending only one message, that is QIP(1), so on and so forth. A result of Kitaev and Watrous \cite{kitaev2000parallelization} showed that the upper limit is three messages; any exchange beyond that reduces to QIP(3). This exchange of messages is shown pictorially in Figure~\ref{fig:qip}.

\section{Defining Notions of Symmetry}\label{sec:symdef}

We may now finally progress from background to original contributions rather than review. The primary focus of this thesis is evaluating symmetry of various quantities of interest in quantum mechanics. However, the concept of symmetry can be ambiguous and abstract. Defining what symmetries we plan to encounter is thus a necessary preliminary, and we have amassed within this work quite a menagerie.  In Section~\ref{sec:statenotions}, we give a series of four definitions of symmetries of quantum states originally presented in \cite{laborde2021testing}, the contents of which will be utilized in Chapters~\ref{ch:symmstates} and \ref{ch:gensep}. In Section~\ref{sec:channelnotions}, we review the definitions of covariance of a quantum channel and Hamiltonian symmetries, which will be paramount in Chapters~\ref{ch:symmham} and \ref{ch:symmstates}.  While the latter section is well-studied in the literature, the former should be considered minutely as part of the research contributions collected herein.

For further reading on symmetry in quantum information, we point to various resources which have influenced our own work and education. The work of Iman Marvian, particularly his PhD thesis \cite{marvian2012symmetry}, gives a fantastic picture of the topic. Additionally, Aram Harrow's ``Church of the Symmetric Subspace" \cite{harrow2013church} provides a great starting place and overview of the topic. Both of these works we cite in later chapters where their influence has made its mark, but we provide the citations here in a continuous effort to give additional reference texts throughout the introduction.

\subsection{Various Symmetries for Quantum States}\label{sec:statenotions}
We introduce the notions of $G$-symmetric extendibility and $G$-Bose symmetric extendibility of a state, as generalizations of the notions of $G$-symmetry \cite[Section~2]{MS13} and extendibility \cite{W89a,DPS02,DPS04}. Later on in Chapter~\ref{ch:symmstates}, we devise quantum algorithms to test for these symmetries.

Let $\rho_{S}$ be a quantum state of system $S$ with corresponding Hilbert space $\mathcal{H}_{S}$. Let $G$ be a finite group, and let $U_{RS}(g)$ be a unitary representation \cite[Section~2]{MS13}\ of the group element $g\in G$, where $R$ indicates another Hilbert space, so that $U_{RS}(g)$ acts on the tensor-product Hilbert space $\mathcal{H}_{R} \otimes\mathcal{H}_{S}$. Let $\Pi_{RS}^{G}$ denote the following projection operator:
\begin{equation}\label{eq:groupprojectorog}
\Pi_{RS}^{G}\coloneqq\frac{1}{\left\vert G\right\vert }\sum_{g\in G}U_{RS}(g).
\end{equation}
Observe that
\begin{equation}
\Pi_{RS}^{G}=U_{RS}(g)\Pi_{RS}^{G}=\Pi_{RS}^{G}U_{RS}(g),
\label{eq:unitaries-and-projs}
\end{equation}
for all $g\in G$.

We now define $G$-symmetric-extendible and $G$-Bose-symmetric-extendible states.
\begin{definition}
[$G$-symmetric-extendible]\label{def:g-sym-ext}A state $\rho_{S}$ is
$G$-symmetric-extendible if there exists a state $\omega_{RS}$ such that
\begin{enumerate}
\item the state $\omega_{RS}$ is an extension of $\rho_{S}$, i.e.,
\begin{equation}
\operatorname{Tr}_{R}[\omega_{RS}]=\rho_{S}, \label{eq:G-ext-1}
\end{equation}
\item the state $\omega_{RS}$ is $G$-invariant, in the sense that
\begin{equation}
\omega_{RS}=U_{RS}(g)\omega_{RS}U_{RS}(g)^{\dag}\qquad\forall g\in G.
\label{eq:G-ext-2}
\end{equation}
\end{enumerate}
\end{definition}

\begin{definition}
[$G$-Bose-symmetric-extendible]\label{def:g-bose-sym-ext}A state $\rho_{S}$ is $G$-Bose-symmetric-extendible (G-BSE) if there exists a state $\omega_{RS}$ such that

\begin{enumerate}
\item the state $\omega_{RS}$ is an extension of $\rho_{S}$, i.e.,
\begin{equation}
\operatorname{Tr}_{R}[\omega_{RS}]=\rho_{R},
\end{equation}
\item the state $\omega_{RS}$ satisfies
\begin{equation}
\omega_{RS}=\Pi_{RS}^{G}\omega_{RS}\Pi_{RS}^{G}.
\label{eq:bose-G-sym-ext-cond}
\end{equation}
\end{enumerate}
\end{definition}

Note that the condition in \eqref{eq:bose-G-sym-ext-cond} is equivalent to $\omega_{RS}=\Pi_{RS}^{G}\omega_{RS}$ or $\omega_{RS}=U_{RS}(g)\omega_{RS}$ for all $g\in G$. Observe also that $\rho_{S}$ is $G$-symmetric-extendible if it is $G$-Bose-symmetric-extendible, but the opposite implication does not necessarily hold.

We have made no assumptions about the unitary representation used thus far. It is important to mention the case of projective unitary representations, due to their physical relevance in the case of symmetries of density operators. See, e.g., Eqs.~(1.2) and (1.3) of \cite{marvian2012symmetry} for a definition of a projective unitary representation. Restricting to projective unitary representations helps in avoiding trivial representations, and when considering symmetries of density operators, they necessarily arise. Furthermore, when considering  implementations of groups in later chapters, we limit ourselves to faithful representations of the group. In principle, neither faithfulness nor a projective representation are required unless stated otherwise. (The choice of representation will matter when considering the symmetry of a state, for instance; however, in the manner of existing literature, we describe all symmetries with respect to the group and omit the reliance on the representation in notation.)

Although the concepts of $G$-symmetric extendibility and $G$-Bose-symmetric extendibility, in Definitions~\ref{def:g-sym-ext} and~\ref{def:g-bose-sym-ext} respectively, are generally different, we can relate them by purifying a $G$-symmetric-extendible state to a larger Hilbert space. The ability to do so plays a critical role in the algorithms proposed in Chapter~\ref{ch:symmstates}.

\begin{theorem}
\label{thm:Bose-sym-purify}A state $\rho_{S}$ is $G$-symmetric-extendible if and only if there exists a purification $\psi_{RS\hat{R}\hat{S}}^{\rho}$ of $\rho_{S}$ satisfying the following:
\begin{equation}
|\psi^{\rho}\rangle_{RS\hat{R}\hat{S}}=\left(  U_{RS}(g)\otimes\overline
{U}_{\hat{R}\hat{S}}(g)\right)  |\psi^{\rho}\rangle_{RS\hat{R}\hat{S}}
\quad\forall g\in G, \label{eq:g-sym-pur-cond}
\end{equation}
where the overbar denotes the complex conjugate. The condition in \eqref{eq:g-sym-pur-cond} is equivalent to
\begin{equation}
|\psi^{\rho}\rangle_{RS\hat{R}\hat{S}}=\Pi_{RS\hat{R}\hat{S}}^{G}|\psi^{\rho
}\rangle_{RS\hat{R}\hat{S}}, \label{eq:g-sym-pur-cond-proj}
\end{equation}
where
\begin{equation}
\Pi_{RS\hat{R}\hat{S}}^{G}\coloneqq \frac{1}{\left\vert G\right\vert } \sum_{g\in G} U_{RS}(g) \otimes \overline{U}_{\hat{R}\hat{S}}(g).
\label{eq:projector-ref-unitaries}
\end{equation}
\end{theorem}

\begin{proof}
We give the proof for completeness, and we note here that it is very close to the proof of \cite[Lemma~II.5]{CKMR08} (see also \cite[Lemma~3.6]{KW20book}).

We begin with the forward implication. Suppose that $\rho_{S}$ is $G$-symmetric extendible. By definition, this means that there exists a state $\omega_{RS}$ satisfying \eqref{eq:G-ext-1} and \eqref{eq:G-ext-2}. Suppose that $\omega_{RS}$ has the following spectral decomposition:
\begin{equation}
\omega_{RS}=\sum_{k}\lambda_{k}\Pi_{RS}^{k},
\end{equation}
where $\lambda_{k}$ is an eigenvalue and $\Pi_{RS}^{k}$ is a spectral projection. We can write $\Pi_{RS}^{k}$ as
\begin{equation}
\Pi_{RS}^{k}=\sum_{\ell}|\phi_{\ell}^{k}\rangle\!\langle\phi_{\ell}^{k}|_{RS},
\end{equation}
where $\{|\phi_{\ell}^{k}\rangle_{RS}\}_{\ell}$ is an orthonormal basis. Now define
\begin{align}
|\Gamma^{k}\rangle_{RS\hat{R}\hat{S}}  &  \coloneqq \sum_{\ell}|\phi_{\ell}^{k}\rangle_{RS} \otimes \overline{|\phi_{\ell}^{k}\rangle}_{\hat{R}\hat{S}},\\
|\psi^{\rho}\rangle_{RS\hat{R}\hat{S}}  &  \coloneqq \sum_{k}\sqrt{\lambda_{k}}|\Gamma^{k}\rangle_{RS\hat{R}\hat{S}},
\end{align}
where $\overline{|\phi_{\ell}^{k}\rangle}_{\hat{R}\hat{S}}$ is the complex conjugate of $|\phi_{\ell}^{k}\rangle_{RS}$ with respect to the standard basis. Observe that $|\psi^{\rho}\rangle\!\langle\psi^{\rho}|_{RS\hat{R} \hat{S}}$ is a purification of $\omega_{RS}$. Now let us establish \eqref{eq:g-sym-pur-cond}. Given that $\omega_{RS}$ satisfies \eqref{eq:G-ext-2}, it follows that
\begin{equation}
U_{RS}(g)^{\dag}\omega_{RS}U_{RS}(g)|\phi_{\ell}^{k}\rangle_{RS}  =\omega_{RS}|\phi_{\ell}^{k}\rangle_{RS}  =\lambda_{k}|\phi_{\ell}^{k}\rangle_{RS},
\end{equation}
for all $k$, $\ell$, and $g$. Left multiplying by $U_{RS}(g)$ implies that
\begin{equation}
\omega_{RS}U_{RS}(g)|\phi_{\ell}^{k}\rangle_{RS}=\lambda_{k}U_{RS}(g)|\phi_{\ell}^{k}\rangle_{RS},
\end{equation}
so that $U_{RS}(g)|\phi_{\ell}^{k}\rangle_{RS}$ is an eigenvector of $\omega_{RS}$ with eigenvalue $\lambda_{k}$. We conclude that the $k$-th eigenspace corresponding to eigenvalue $\lambda_{k}$ is invariant under the action of $U_{RS}(g)$ because $|\phi_{\ell}^{k}\rangle_{RS}$ and $U_{RS}(g)|\phi_{\ell}^{k}\rangle_{RS}$ are eigenvectors of $\omega_{RS}$ with eigenvalue $\lambda_{k}$. This implies that the restriction of $U_{RS}(g)$ to the $k$th eigenspace is equivalent to a unitary $U_{RS}^{k}(g)$. Then it follows that
\begin{align}
&  (U_{RS}(g)\otimes\overline{U}_{\hat{R}\hat{S}}(g))|\Gamma^{k}
\rangle_{RS\hat{R}\hat{S}}\nonumber\\
&  =(U_{RS}^{k}(g)\otimes\overline{U}_{\hat{R}\hat{S}}^{k}(g))|\Gamma
^{k}\rangle_{RS\hat{R}\hat{S}}\\
&  =|\Gamma^{k}\rangle_{RS\hat{R}\hat{S}},
\end{align}
for all $g\in G$. The first equality follows from the fact stated just above. 
The second equality follows from the invariance of the maximally entangled vector $|\Gamma^{k}\rangle_{RS\hat
{R}\hat{S}}$ under unitaries of the form $V\otimes\overline{V}$. Thus, it
follows by linearity that
\begin{equation}
|\psi^{\rho}\rangle_{RS\hat{R}\hat{S}} = (U_{RS}(g)\otimes\overline{U}_{\hat{R}\hat{S}}(g))|\psi^{\rho}\rangle_{RS\hat{R}\hat{S}},
\label{eq:unitary-inv-purified}
\end{equation}
for all $g\in G$, which is the statement of \eqref{eq:g-sym-pur-cond}.

Let us now consider the opposite implication; suppose that
$\psi_{RS\hat{R}\hat{S}}^{\rho}$ is a purification of $\rho_{S}$ and
$\psi_{RS\hat{R}\hat{S}}^{\rho}$ satisfies \eqref{eq:g-sym-pur-cond}. Set
\begin{equation}
\omega_{RS}=\operatorname{Tr}_{\hat{R}\hat{S}}[\psi_{RS\hat{R}\hat{S}}^{\rho
}].
\end{equation}
Then $\omega_{RS}$ is an extension of $\rho_{S}$. Furthermore, employing the
shorthand $U_{RS}\equiv U_{RS}(g)$ and $\overline{U}_{\hat{R}\hat{S}}
\equiv\overline{U}_{\hat{R}\hat{S}}(g)$, we find that $\omega_{RS}
=U_{RS}(g)\omega_{RS}U_{RS}(g)^{\dag}$ for all $g\in G$ because
\begin{align}
 \omega_{RS}\nonumber &=\operatorname{Tr}_{\hat{R}\hat{S}}[\psi_{RS\hat{R}\hat{S}}^{\rho}]\\
&  =\operatorname{Tr}_{\hat{R}\hat{S}}[(U_{RS}\otimes\overline{U}_{\hat{R}
\hat{S}})\psi_{RS\hat{R}\hat{S}}^{\rho}(U_{RS}\otimes\overline{U}_{\hat{R}
\hat{S}})^{\dag}]\\
&  =U_{RS}(g)\operatorname{Tr}_{\hat{R}\hat{S}}[\overline{U}_{\hat{R}\hat{S}
}(g)\psi_{RS\hat{R}\hat{S}}^{\rho}\overline{U}_{\hat{R}\hat{S}}(g)^{\dag
}]U_{RS}(g)^{\dag}\\
&  =U_{RS}(g)\operatorname{Tr}_{\hat{R}\hat{S}}[\overline{U}_{\hat{R}\hat{S}
}(g)^{\dag}\overline{U}_{\hat{R}\hat{S}}(g)\psi_{RS\hat{R}\hat{S}}^{\rho
}]U_{RS}(g)^{\dag}\\
&  =U_{RS}(g)\operatorname{Tr}_{\hat{R}\hat{S}}[\psi_{RS\hat{R}\hat{S}}^{\rho
}]U_{RS}(g)^{\dag}\\
&  =U_{RS}(g)\omega_{RS}U_{RS}(g)^{\dag}.
\end{align}
Thus, it follows that $\rho_{S}$ is $G$-symmetric extendible.

We now justify the equivalence of \eqref{eq:g-sym-pur-cond} and \eqref{eq:g-sym-pur-cond-proj}. Using the result in \eqref{eq:unitary-inv-purified}, observe that
\begin{equation}
|\psi^{\rho}\rangle_{RS\hat{R}\hat{S}}= \frac{1}{|G|}\sum_{g\in G}(U_{RS}(g)\otimes\overline{U}_{\hat{R}\hat{S}}(g))|\psi^{\rho}\rangle_{RS\hat{R}\hat{S}},
\end{equation}
which simplifies to \eqref{eq:g-sym-pur-cond-proj} by substituting in \eqref{eq:projector-ref-unitaries}. Now starting with \eqref{eq:projector-ref-unitaries}, let us apply the property in \eqref{eq:unitaries-and-projs}, and we have that
\begin{equation}
    |\psi^{\rho}\rangle_{RS\hat{R}\hat{S}} = (U_{RS}(g)\otimes\overline{U}_{\hat{R}\hat{S}}(g))\Pi^G_{RS\hat{R}\hat{S}} |\psi^{\rho}\rangle_{RS\hat{R}\hat{S}},
\end{equation}
for all $g \in G$.
This reduces to \eqref{eq:g-sym-pur-cond} by applying \eqref{eq:g-sym-pur-cond-proj}.
\end{proof}

Now, both of these definitions of symmetry contain a condition of extendibility, but that criterion may not always be necessary for particular applications. In those cases, simply let the extension be trivial to regain a related, simpler notion of symmetry. We present those cases here as examples.

\begin{example}
[$G$-symmetric]\label{ex:usual-symmetry}Let $G$ be a group with projective
unitary representation $\{U_{S}(g)\}_{g\in G}$, and let $\rho_{S}$ be a
quantum state of system $S$. A state $\rho_{S}$ is symmetric with respect to
$G$ \cite{MS13,MS14}\ if
\begin{equation}
\rho_{S}=U_{S}(g)\rho_{S}U_{S}(g)^{\dag}\quad\forall g\in G.
\end{equation}
Thus, the established notion of symmetry of a state $\rho_{S}$ with respect to
a group $G$ is a special case of $G$-symmetric extendibility in which the
system $R$ is trivial.
\end{example}

\begin{example}
[$G$-Bose-symmetric]\label{ex:usual-Bose-symmetry}A state $\rho_{S}$ is
Bose-symmetric with respect to $G$ if
\begin{equation}
\rho_{S}=U_{S}(g)\rho_{S}\quad\forall g\in G,
\end{equation}
which is equivalent to the condition
\begin{equation}
\rho_{S}=\Pi_{S}^{G}\rho_{S}\Pi_{S}^{G},
\end{equation}
where the projector $\Pi_{S}^{G}$ is defined as
\begin{equation}
\label{eq:group_proj_GBS}
\Pi_{S}^{G}\coloneqq \frac{1}{\left\vert G\right\vert }\sum_{g\in G}U_{S}(g).
\end{equation}
Thus, the established notion of Bose symmetry of a state~$\rho_{S}$ with
respect to a group $G$ is a special case of $G$-Bose symmetric extendibility in
which the system $R$ is trivial.
\end{example}

Thus we have defined the notions of $G$-symmetric extendibility, $G$-Bose-symmetric extendibility, 
and the related notions of $G$-symmetry and $G$-Bose symmetry. These concepts will be revisited in Chapter~\ref{ch:symmstates} and are described in great detail in \cite{laborde2021testing}.

\subsection{Hamiltonian Symmetry and Covariance of a Quantum Channel}\label{sec:channelnotions}

Finally, we arrive at Hamiltonian symmetry and covariance of quantum channels. The former will be integral to the discussions in Chapter~\ref{ch:symmham} and the latter both in Chapters~\ref{ch:symmham} and \ref{ch:symmstates}. We will begin by reviewing very briefly the notion of Hamiltonian symmetry and then spend much more time recalling the definition of covariance symmetry.

In quantum mechanics, a Hamiltonian $H$ is symmetric with respect to an operator $O$ if it commutes with $H$:
\begin{equation*}
    [O,H] =0\,.
\end{equation*}
We say that $H$ commutes with a unitary representation of a group $G$, $\{U(g)\}_{g \in G}$, if
\begin{equation}
    [H,U(g)] = 0 \ \ \ \ \ \ \ \ \ \  \forall g \in G\, .
\end{equation}

Given this basic concept of Hamiltonian symmetry, let us move on to the covariance symmetry of quantum channels. Let $G$ be a group with projective unitary representations $\{U(g)_A\}_{g \in G}$ and $\{V(g)_B\}_{g \in G}$ on the $A$ and $B$ subsystems respectively. Then the channel $\mathcal{N}_{A\rightarrow B}$ is covariant if the following $G$-covariance symmetry condition holds
\begin{equation}
\mathcal{N}_{A\rightarrow B}\circ\mathcal{U}_{A}(g)=\mathcal{V}_{B}(g) \circ \mathcal{N}_{A\rightarrow B} \qquad \forall g\in G,
\label{eq:ch-cov-def}
\end{equation}
where the unitary channels $\mathcal{U}_{A}(g)$ and $\mathcal{V}_{B}(g)$ are
respectively defined from $U_{A}(g)$ and $V_{B}(g)$ as
\begin{align}\label{eq:uandv}
\mathcal{U}_{A}(g)(\omega_{A}) &  \coloneqq  U_{A}(g)\omega_{A}U_{A}(g)^{\dag
},\\
\mathcal{V}_{B}(g)(\tau_{B}) &  \coloneqq  V_{B}(g)\tau_{B}V_{B}(g)^{\dag}.
\end{align}

Furthermore, a channel is covariant in the sense above if and only if its Choi state is invariant in the following sense \cite[Eq.~(59)]{CDP09}:
\begin{equation} \label{eq:choicondition}
\Phi_{RB}^{\mathcal{N}}=(\overline{\mathcal{U}}_{R}(g)\otimes\mathcal{V}
_{B}(g))(\Phi_{RB}^{\mathcal{N}})\quad\forall g\in G,
\end{equation}
where
\begin{equation}
\overline{\mathcal{U}}_{R}(g)(\omega_{R})\coloneqq  \overline{U}_{R}
(g)\omega_{R}U_{R}(g)^{T}.
\end{equation}
Note that the Choi state of the channel $\mathcal{N}_{A\rightarrow B}$, denoted $\Phi_{RB}^{\mathcal{N}}$,  is defined to be 
\begin{align}\label{eq:choistate}
\Phi_{RB}^{\mathcal{N}}  & \coloneqq  \mathcal{N}_{A\rightarrow B}(\Phi_{RA}),\\
\Phi_{RA}  & \coloneqq  \frac{1}{\left\vert A\right\vert }\sum_{i,j}|i\rangle\!\langle
j|_{R}\otimes|i\rangle\!\langle j|_{A}.
\end{align}
We can note a different condition by specifying a projector over the space of states symmetric with respect to group $G$ \cite{harrow2013church}. Denote the projector of this group as
\begin{equation}\label{eq:projector_app}
    \Pi^G =\frac{1}{|G|} \sum_{g \in G} \overline{U}_R(g) \otimes U_B (g) .
\end{equation}
Then the Choi state is symmetric with respect to $G$ if
\begin{equation}\label{eq:bosechoi}
    \Phi_{RB}^{\mathcal{N}} = \Pi^G \Phi_{RB}^{\mathcal{N}}.
\end{equation}
The above definition is a stronger condition of symmetry; thus, if the Choi state obeys \eqref{eq:bosechoi} then \eqref{eq:choicondition} follows.

If a channel describing Hamiltonian dynamics exhibits $G$-covariance symmetry, then the underlying Hamiltonian is symmetric with respect to $G$. This will become paramount to the results in Chapter~\ref{ch:symmham}, but can be verified easily.

\section{Conclusion}
With this, we have defined all of the relevant background knowledge necessary to comprehend the work beyond a basic background in physics. We have equipped ourselves with some meager knowledge of group theory and representation theory, and we have a handy guide for the formalities of quantum computing. Furthermore, we have established the utmost important concept of symmetry, upon which all the work in this thesis is based. Hopefully, this chapter will serve as a lighthouse to guide anyone unfamiliar with these concepts through the results we wish to communicate in later chapters.

\pagebreak
\singlespacing
\chapter{Hamiltonian Symmetry}\label{ch:symmham}
\doublespacing
\section{\label{sec:level1}Introduction}
\let\thefootnote\relax\footnote{Sections 2.2-2.6 of this chapter were previously published in \textit{Physical Review Letters} in "Quantum Algorithms for Testing Hamiltonian Symmetry" by Margarite L. LaBorde and Mark M. Wilde, Phys. Rev. Lett. 129, 160503.}
Symmetry is a key facet of nature that plays a fundamental role in physics \cite{Gross96,FR96}, and for many physics students this is first stressed via Noether's theorem \cite{Noether1918}, which states that symmetries in Hamiltonians correspond with conserved quantities in the related physical systems. In much the same manner, we too begin with a discussion of symmetry tests for a Hamilitonian. Hamiltonian symmetries have a number of immediate and profound effects. For instance, the symmetries of a Hamiltonian indicate the presence of superselection rules \cite{PhysRev.155.1428, Wick1952}. In quantum computing and information, symmetry can indicate the presence of resources or lack thereof \cite{marvian2012symmetry}, and it can be useful for improving the performance of variational quantum algorithms \cite{SSY20,Gard2020,BGAMBE21,LXYB22}. Identification of symmetries can simplify calculations by eliminating degrees of freedom associated with conserved quantities---this is at the heart of Noether's theorem. This makes symmetries, and especially Hamiltonian symmetries, extraordinarily useful in the context of physics.

In comparison, quantum computing is a significantly younger field of study. First introduced as a quantum-mechanical model of a Turing machine \cite{benioff1980computer}, the intrigue of quantum computers lies in their potential to outperform their classical counterparts. The most obvious asset of quantum computers is the inherent physics behind the calculation, utilizing non-classical features such as superposition and entanglement. Classical simulations of quantum systems quickly become intractable as the size of the Hilbert space grows, needing exponentially many bits to explore the state space which multiple qubits naturally occupy. Intuitively, the quantum mechanical nature of these computers allows for simulations of quantum systems in a forthright way (see \cite{childs2018toward} and references therein). 

A pertinent example of this, Hamiltonian simulation \cite{lloyd1996universal}, garners high interest in the field \cite{blatt2012quantum,somma2016trotter,cubitt2018universal,clinton2021hamiltonian}. Much work has been done to understand how to simulate these dynamics on quantum hardware such that they can be efficiently realized; however, to the best of our knowledge, before our work in \cite{laborde2022quantum}, there were no algorithms that test Hamiltonian symmetries on a quantum computer, even though simulating Hamiltonians in this manner and identifying the symmetries of said Hamiltonian are both deemed to be of utmost importance.

In this chapter, we give quantum algorithms to test whether a Hamiltonian evolution is symmetric with respect to the action of a discrete, finite group. This property is often referred to as the covariance \cite{CDP09} of the evolution. If the evolution is symmetric, then the Hamiltonian itself is also symmetric, and so our algorithms thus test for Hamiltonian symmetry. Furthermore, we show that for a Hamiltonian with an efficiently realizable unitary evolution and a group with an efficiently realizable unitary representation, we can perform our first test efficiently on a quantum computer \cite{clinton2021hamiltonian}. ``Efficiently'' here means that our algorithm is in the complexity class DQC-1. We give a second quantum algorithm for testing Hamiltonian symmetry which can be implemented by means of a variational approach \cite{CABBEFMMYCC20,bharti2021noisy}. The acceptance probabilities of both algorithms can be elegantly expressed in terms of familiar expressions of Hamiltonian symmetry. We further consider physically relevant examples demonstrating the capabilities of our algorithms.

The consequences of such results extend throughout many areas of physics. Any study of a physical Hamiltonian can benefit from finding its symmetries, and our algorithms allow for an efficient check for these symmetries. With this knowledge, dynamics can be simplified by excluding symmetry-breaking transitions, calculations can be reduced into fewer dimensions, and intuition can be gained about the system of interest. Our first algorithm also scales well, meaning that systems too large and cumbersome to be studied by hand or classical computation can instead be investigated in a practical time scale. Our quantum tests offer meaningful insight into physical dynamics.

In Section~\ref{sec:hamcov}, we begin by describing covariance symmetry of a unitary quantum channel---of which Hamiltonian dynamics are a special case. This section revisits the notion of symmetry given in Section~\ref{sec:channelnotions} in a more abstract, conversational frame. This description will directly motivate the algorithms proffered. 
Next, in Section~\ref{sec:hamsim} we briefly review how Hamiltonian dynamics can be simulated on a quantum computer through the Trotter--Suzuki approximation \cite{suzuki1976generalized}. We describe the assumptions of Trotterization and the resultant evolution approximation.  
Section~\ref{sec:efficient} presents our main result of this chapter. We give a quantum algorithm to test the covariance symmetry of Hamiltonian dynamics, and we show that this algorithm is DCQ1-complete.  
Section~\ref{sec:hamaccept} gives the derivation of the acceptance probability for our algorithm. The result shows demonstrable reliance on the quantum mechanical notion of Hamiltonian symmetry. 
Following this, in Section~\ref{sec:hamvar}, we further give another, related algorithm achievable with the aid of a quantum variational approach. This related test assumes a maximization over all input states, and thus is less efficient, yet realizes an interesting bound on Hamiltonian symmetry via the commutator norm and the twirl.
Finally, in Section~\ref{sec:hamexamples} we demonstrate examples of symmetry tests on currently available quantum computers. We consider the transverse-field Ising model, the Heisenberg~XY model \cite{LSM61}, and the weakly $J$-coupled NMR Hamiltonian \cite{van1996multidimensional},  whose evolution we test for various symmetry cases. 

\section{\label{sec:hamcov}Covariance of a Quantum Channel}

Before describing the symmetries of a Hamiltonian, we first recall the notion of covariance symmetry of a quantum channel \cite{Hol02}. Quantum channels transform one quantum state to another and are described by completely positive, trace-preserving maps. They serve as a convenient mathematical description of the dynamics induced by a Hamiltonian. The symmetries of a Hamiltonian naturally correspond to a covariance symmetry in the channel given by its evolution, and we exploit this in our algorithms.

As the mathematical description has been given in greater detail previously, let us instead take a high-level, intuitive review of the matter. (We recall the established concept of covariance symmetry in more detail in Chapter~\ref{ch:intro} in Section~\ref{sec:channelnotions}.) Suppose there is a channel sending Alice's quantum system to Bob's. For simplicity, we consider their systems to have the same dimension, though this is not required in general. Further suppose that we wish to determine if this channel is symmetric with respect to some finite, discrete group $G$, which has a projective unitary representation (as usual, denoted $\{U(g)\}_{g\in G}$). Then the channel is covariant if Alice acting with her representation $U(g)$ before sending the system through the channel is completely equivalent to Bob acting on his system with his representation of $g$ after the state has been sent through the channel. In this sense, the channel commutes with the action of the group.

One method for testing this property given some channel involves using its Choi state, formally defined in \eqref{eq:choistate}. The Choi state is generated by sending one half of a maximally-entangled state through the channel, which we now assume to be unitary. Given the same group and its unitary representation, we define a projector
\begin{equation}\label{eq:projector}
    \Pi^G \coloneqq\frac{1}{|G|} \sum_{g \in G} \overline{U}_R(g) \otimes U_B (g) ,
\end{equation}
onto the space of states of a composite system $RB$ that are symmetric with respect to the group $G$, where the overline denotes complex conjugation. (Here we use $R$ to refer to a reference system and $B$ to refer to Bob's system after the channel, a notion we use throughout.) The Choi state of the channel is equal to its projection onto the symmetric space if and only if the Choi state is symmetric with respect to $G$, given unitary representations of the system. If the Choi state of a channel exhibits this symmetry, then the channel itself is covariant \cite{CDP09}, and the converse is true as well. 

This last notion of symmetry allows us to directly prescribe an algorithm to test for Hamiltonian symmetries. If we can emulate the dynamics of a Hamiltonian efficiently, we can test for the symmetry of its Choi state. The symmetry of the Choi state then directly implies symmetry of the Hamiltonian being tested.

\section{\label{sec:hamsim}Quantum Simulations of Hamiltonians}

A necessary preliminary to testing Hamiltonian dynamics in any respect is first making sure they can be imported to computational framework to begin with. Quantum simulation techniques directly answer this need by providing a method for implementing Hamiltonian dynamics on quantum computers. Usually, this approach involves approximating them as sequences of quantum logic gates \cite{lloyd1996universal,childs2018toward}. Much work has been conducted in this field, including work on implementations on near-term hardware \cite{clinton2021hamiltonian,CCHCCS20}, simulation by qubitization \cite{low2019hamiltonian}, simulation of operator spread \cite{geller2021quantum}, and more. Here, we review a rather popular example implementation, though be advised that this is not the only method available for this purpose. 

One common approach \cite{lloyd1996universal} employs the Trotter--Suzuki approximation \cite{trotter1959product,suzuki1976generalized}. This method allows for decomposition into local Hamiltonian evolutions with some specified error. In this approximation, we suppose that the Hamiltonian $H$ is of the form
$ H = \sum_{i=1}^m H_i $,
where each $H_i$ is a $k$-local Hamiltonian, which means $H_i$ affects at most $k$ systems simultaneously. Then we can describe its evolution by
\begin{equation}
\label{eq:trotter}
    e^{-i H t} = \left(\prod_{j=1}^m e^{-i H_j t/r}\right)^r + \mathcal{O}\!\left (\frac{m^2 t^2}{r^2}\right ) \, ,
\end{equation}
where the correction term is negligible for $mt/r \ll 1$ and vanishes when the terms in the decomposition commute. (Here and throughout, we take $\hbar=1$.) By other methods, the error can be reduced to higher orders in~$t$~\cite{BACS07}.

\section{\label{sec:efficient}An Efficient Quantum Algorithm to Test Hamiltonian Symmetries}

Given the notion of covariance recalled above and a way to simulate the applicable Hamiltonian, we now propose a quantum algorithm to test a Hamiltonian for covariance symmetry. We begin by supposing that we have a Hamiltonian composed of a finite sum of $k$-local Hamiltonians, as described previously, with dynamics realized by higher-order methods such that the simulation error is $\mathcal{O}(t^4)$. Then we claim a test for symmetries of this Hamiltonian with respect to a group $G$ with a projective unitary representation $\{U(g)\}_{g\in G}$ can be performed efficiently on a quantum computer.

\begin{figure}[t]
\begin{center}
\includegraphics[width=3.5in]{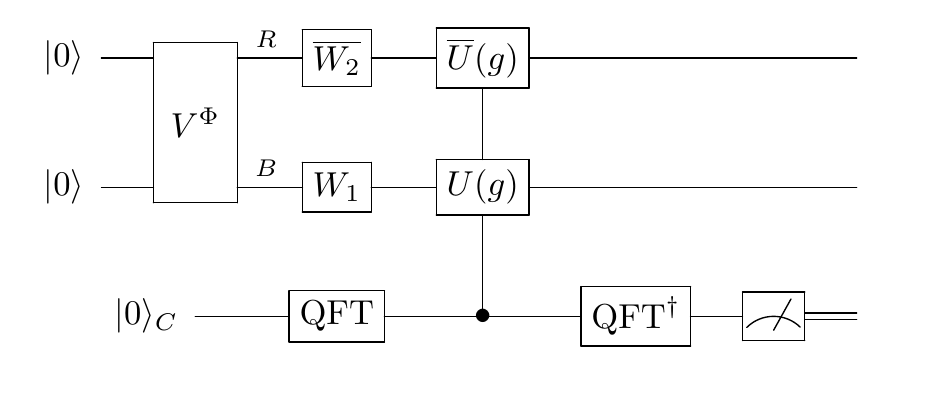}
\end{center}
\caption{Quantum circuit to test for the covariance of a unitary Hamiltonian evolution. The unitary $V^\Phi$ generates the state $\ket{\Phi}_{RA}$, the maximally-entangled state on $RA$. The evolution of the system is given by $e^{-i H t} = W_1 W_2^\dag$ and the $U(g)$ gates are controlled on a superposition over all of the elements $g \in G$, as in \eqref{eq:superpose-gr-elems}.}
\label{fig:originalcircuit}
\end{figure}

The circuit presented in Figure~\ref{fig:originalcircuit} implements just such a test, and we sketch its action here. Let the input state to the circuit be the maximally-entangled state $\Phi_{RA}$. Then act on the $A$ subsystem with the unitary Hamiltonian dynamics. As indicated in Figure~\ref{fig:originalcircuit}, the depth of the circuit to realize this algorithm can be cut in half by taking advantage of the transpose trick $(X \otimes I) \ket{\Phi} = (I\otimes X^T)\ket{\Phi}$ and the decomposition $e^{-i H t} = W_1 W_2^\dag$, which is clearly possible for Hamiltonian simulations of the form in~\eqref{eq:trotter} or from \cite{BACS07}. (The transpose trick, for those unfamiliar, can be understood best pictographically. A common shorthand for quantum circuits is to represent a maximally-entangled state as a `cap' or a Bell measurement as a `cup'. When there is a cup or cap present in the circuit diagram, a tensor can be moved around to the other leg of the cup/cap if the transpose of the tensor is used instead. See \cite{biamonte2017tensor} for further illustration.)

The state of the system after such tricks have been applied is given by
\begin{equation}
    \Phi_{RB}^{t} \coloneqq (\mathbb{I}_R \otimes e^{-i H t} ) \Phi_{RA} (\mathbb{I}_R \otimes e^{i H t}),
\end{equation}
which is exactly the Choi state of the channel generated by~$e^{-i H t}$. 
We then use the quantum Fourier transform (QFT) to generate a control register in the following superposed state:
\begin{equation}
    \ket{+}_C \coloneqq \frac{1}{\sqrt{|G|}}\sum_{g\in G}\ket{g}.
    \label{eq:superpose-gr-elems}
\end{equation}
Implementing the controlled  $\overline{U}(g)$ and $U(g)$ gates using the above control register yields the state
\begin{equation}
    \frac{1}{|G|}\sum_{g,g'\in G}(\overline{U}_R(g) \otimes U_B(g)) (\Phi_{RB}^{t}\otimes \ket{g} \! \bra{g'}_C) (\overline{U}^\dag_R(g') \otimes U^\dag_B(g')).
\end{equation}
Finally, we perform the measurement $\mathcal{M}=\{\ket{+}\!\bra{+}_C,\mathbb{I}-\ket{+}\!\bra{+}_C\}$ on the control register and accept if and only if the outcome $\ket{+}\!\bra{+}_C$ is observed. With this condition, the acceptance probability is given by
\begin{align}
     P_{\text{acc}} &= \operatorname{Tr}[\Pi^G \Phi_{RB}^{t}], \label{eq:originalacc}
\end{align}
where we have used the projector defined in \eqref{eq:projector} (see Appendix~\ref{app:accept} for a quick derivation of \eqref{eq:originalacc}, which we do not present here in an effort to preserve coherentness). 
As a limiting case of the gentle measurement lemma (see Section~\ref{sec:normsandlemmas}, and references \cite{Davies1969,itit1999winter,ON07}), we have that
\begin{equation}
\operatorname{Tr}[\Pi^{G}\Phi_{RB}^{t}]=1\quad\Leftrightarrow\quad \Phi_{RB}^{t}
=\Pi^{G}\Phi_{RB}^{t}\Pi^{G}, \label{eq:Ham-Bose-symmetric-equiv-cond}
\end{equation}
where the second statement is equivalent to the condition on the Choi state given in \eqref{eq:bosechoi} in Section~\ref{sec:channelnotions}. Therefore, by implementing this algorithm, we can determine whether a Hamiltonian exhibits a symmetry under a group $G$ with some projective unitary representation $\{U(g)\}_{g\in G}$. 
 
Now, indulge us in a minor aside to demonstrate approximate equivalence in~\eqref{eq:Ham-Bose-symmetric-equiv-cond}; specifically, to show that the acceptance probability is near to one if and only if the Choi state is approximately Bose symmetric. This will endow our test with a sense of continuity and robustness, and thus merits taking the time to ascertain.

First, consider the situation where the Choi state is approximately Bose symmetric, and set $\epsilon$ such that
\begin{equation}
    \epsilon \coloneqq \left\Vert \Phi_{RB}^{t} - \Pi^{G}\Phi_{RB}^{t}\Pi^{G} \right\Vert_1 \, .
    \label{eq:approx-sym-def}
\end{equation}
Here we employ the trace distance as a standard metric between states or subnormalized states. From this point, we use the reverse triangle inequality to conclude that
\begin{equation}
    \left\Vert \Phi_{RB}^{t} - \Pi^{G}\Phi_{RB}^{t}\Pi^{G} \right\Vert_1 \geq \left\Vert \Phi_{RB}^{t} \right\Vert_1 - \left\Vert \Pi^{G}\Phi_{RB}^{t}\Pi^{G} \right\Vert_1,
    \label{eq:rev-tri-app}
\end{equation}
where the first term on the right-hand side is equal to one for every quantum state, and the left-hand side is equal to $\epsilon$ by definition. Meanwhile, recall that our acceptance probability is given by
\begin{equation}
    P_{\textrm{acc}}=\operatorname{Tr}[ \Pi^{G}\Phi_{RB}^{t}]=\operatorname{Tr}[ \Pi^{G}\Phi_{RB}^{t}\Pi^{G}] = \left \Vert \Pi^{G}\Phi_{RB}^{t}\Pi^{G}  \right \Vert_1,
\end{equation}
where the second equality follows from cyclicity of trace and the last equality follows because $\Pi^{G}\Phi_{RB}^{t}\Pi^{G}$ is positive semi-definite.
This is the final term on the right-hand side of \eqref{eq:rev-tri-app}. Thus, by substituting in these terms and conducting some simple algebra, we conclude that
\begin{equation}
    P_{\textrm{acc}} \geq  1 - \epsilon\, .
\end{equation}

Consequently, whenever the state is approximately symmetric (in the sense that $\epsilon \approx 0$ in \eqref{eq:approx-sym-def}), the acceptance probability is near to one. Thus, our acceptance probability demonstrates a continuity property.

Next, we will show that the reverse direction is also true. This relationship can be demonstrated via the gentle operator lemma (see, again, Sec~\ref{sec:normsandlemmas} or references \cite{Davies1969,itit1999winter,ON07}). Let $\Phi$ be a density operator  and $\Lambda$ a measurement operator satisfying $0 \leq \Lambda \leq I$.  If $\operatorname{Tr}[\Lambda \Phi] \geq 1 - \epsilon$ for $\epsilon \in [0,1]$, then by Lemma~\ref{gentleO}, the following inequality holds 
\begin{equation}
    \left \Vert \Phi - \sqrt{\Lambda}\Phi \sqrt{\Lambda} \right \Vert_1 \leq 2 \sqrt{\epsilon}\,.
\end{equation}

In our case, set $\Lambda = \Pi ^G$. (Note that $\Pi^G = \sqrt{\Pi^G}$ for a projector.) Then suppose our acceptance probability is close to one, as in
\begin{equation}
    P_{\textrm{acc}}=\operatorname{Tr}[\Pi^G \Phi^t_{RB}] \geq 1 - \epsilon\, ,
\end{equation}
for some $\epsilon \in [0,1]$. Then Lemma~\ref{gentleO} implies that the state is approximately symmetric:
\begin{equation}
   \left \Vert \Phi^t_{RB} - \Pi^G \Phi^t_{RB} \Pi^G \right \Vert_1 \leq 2  \sqrt{\epsilon} \, .
\end{equation}

\begin{figure}[t!]
\begin{center}
\includegraphics[width=\linewidth]{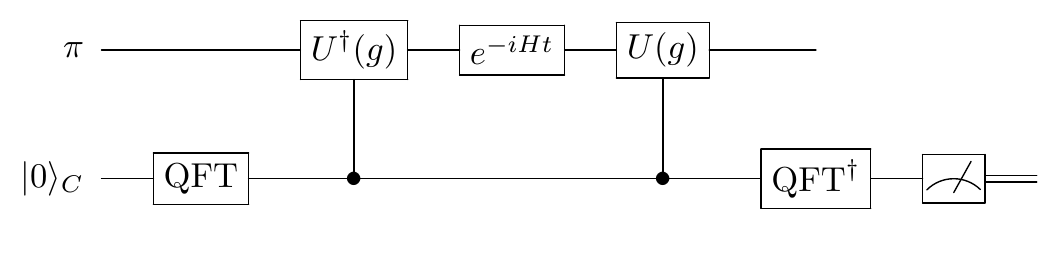}
\end{center}
\caption{Quantum circuit to test for the covariance of a unitary Hamiltonian evolution. Here, $\pi$ denotes the maximally-mixed state $\mathbb{I}/d$.
}
\label{fig:altcircuit}
\end{figure}

Returning to the algorithm with renewed confidence, upon investigation, it can be further simplified. By invoking the transpose trick (see, e.g., \cite{biamonte2017tensor}), we can identify the unitary on the reference system, $\overline{U}_R(g)$, with an equivalent action on $A$ given by $U^\dag_A(g)$. Since the action of the circuit would then take place solely on the subsystem $A$, the reference system $R$ is traced out. This is equivalent to preparing the maximally-mixed state (denoted by $\pi$) on $A$, such that this variation of our algorithm bears some resemblance to a one-clean-qubit algorithm \cite{KL98}, or a DQC1 algorithm, as discussed in Section~\ref{sec:dqc1def}. The only exception is that it requires $\log_2 |G|$ clean qubits for the control register. Figure~\ref{fig:altcircuit} shows this simplification. The acceptance probability corresponding to this described situation is 
\begin{equation}
\label{eq:modifiedacc}
    P_{\text{acc}}=\frac{1}{d |G|}\sum_{g \in G}\operatorname{Tr}[U^\dag (g) e^{i H t} U(g) e^{-i H t}],
\end{equation}
where $d$ is the dimension of the system being tested. Appendix~\ref{app:accept} gives a proof that the expression in \eqref{eq:modifiedacc} is equal to the acceptance probability of the circuit in Figure~\ref{fig:originalcircuit} using only elementary methods.

The proposed circuit in Figure~\ref{fig:altcircuit} is limited in complexity only by the implementation of the Hamiltonian and unitary representation. Thus, our first quantum algorithm is efficiently realizable. Furthermore, we have shown that entanglement resources usually necessary for characterizing the Choi operator of a quantum channel are not necessary here. We also note that the statistics accumulated for the maximally-mixed state can be equivalently found in a sampling manner using computational basis state inputs.

\subsection{DQC1-Completeness of Acceptance Probability} \label{sec:DQC1-completeness}

Having established the algorithm, it is now necessary to provide evidence that our algorithm cannot generally be simulated efficiently by classical computers---else, that would certainly be the preferred calculation method by most. For this purpose, we turn to established notions of computational complexity we so fortuitously described in Section~\ref{sec:dqc1def}. In this section, we prove that estimating the acceptance probability in \eqref{eq:modifiedacc} to within additive error is a DQC1-complete problem. This means that \eqref{eq:modifiedacc} can be estimated within a restricted model of quantum computing (via our algorithm and by an observation of \cite[Section~1]{SJ08}), and thus can be estimated efficiently on a quantum computer. Furthermore, this demonstrates that estimating~\eqref{eq:modifiedacc} is just as computationally hard as any problem in this complexity class. Strong evidence exists that classical computers cannot solve DQC1-complete problems efficiently \cite{MFF14,FKMNTT18}, thus ruling out any possibility of estimating the acceptance probability in \eqref{eq:modifiedacc} by a classical sampling approach. It is the end goal of this section, then, to convince any wary reader that this approach belongs firmly to the class of problems best handled by a quantum machine.

Before establishing our primary claim, we must first prove that estimating $\operatorname{Re}[\operatorname{Tr}[U^{2}]]/d$, where $U$ is the unitary generated by a quantum circuit, is a DQC1-complete problem. This will be a useful tool in asserting our claim. To do this, we will make use of the established fact that estimating $\operatorname{Re}[\operatorname{Tr}[U]]/d$ is a DQC1-complete problem. First, consider the usual construction with $U$ substituted by $U^{2}$. (The usual construction, as it were, is that which is succintly given in Section~\ref{sec:dqc1def}, but more thoroughly presented \cite[Section~1]{SJ08}.)  Similarly, we prepare a control qubit in the $|+\rangle$ state and all other qubits in the maximally-mixed state, act with controlled-$U^{2}$ (easily realized as two applications of controlled-$U$), and then measure in the Hadamard basis. Assigning the values $+1$ and $-1$ to the measurement outcomes, the expected value of the measurement outcomes is equal to $\operatorname{Re}[\operatorname{Tr}[U^{2}]]/d$. This implies that the problem is in DQC1. 

To show hardness, suppose that we have a way of estimating $\operatorname{Re}[\operatorname{Tr}[U^{2}]]/d$ for every $U$, where $U$ is the unitary generated by a quantum circuit. We can then show that it possible to use such an algorithm to estimate $\operatorname{Re}[\operatorname{Tr}[U]]/d$. The key idea behind the reduction is the following unitary
\begin{equation}
V\coloneqq 
\begin{bmatrix}
0 & I\\
U & 0
\end{bmatrix}
=|0\rangle\!\langle1|\otimes I+|1\rangle\!\langle0|\otimes U,
\end{equation}
which can be realized in terms of a quantum circuit as a $\sigma_{X}$ acting on a control qubit, followed by a controlled-$U$, i.e.,
\begin{equation}
\left(  |0\rangle\!\langle0|\otimes I+|1\rangle\!\langle1|\otimes U\right)
\left(  \sigma_{X}\otimes I\right)  .
\end{equation}
We then observe that
\begin{align}
\operatorname{Tr}\left[  V^{2}\right]   &  =\operatorname{Tr}\left[
\begin{bmatrix}
0 & I\\
U & 0
\end{bmatrix}
\begin{bmatrix}
0 & I\\
U & 0
\end{bmatrix}
\right]  \\
&  =\operatorname{Tr}\left[
\begin{bmatrix}
U & 0\\
0 & U
\end{bmatrix}
\right]  \\
&  =2\operatorname{Tr}[U].
\end{align}
Thus, by using the method to estimate $\operatorname{Re}[\operatorname{Tr}[V^{2}]]/d$, we estimate $\operatorname{Re}[\operatorname{Tr}[U]]/d$ up to a constant factor of $2$. This completes the proof that estimating $\operatorname{Re}[\operatorname{Tr}[U^{2}]]/d$ is a DQC1-complete problem.

Now let us turn to the main goal of this section: proving that estimating the acceptance probability in \eqref{eq:modifiedacc} to within additive error is a DQC1-complete problem. Let $H$ be a local Hamiltonian, and let $\left\{ U(g)\right\}_{g\in G}$ be a unitary representation of a group $G$, such that each $U(g)$ can be realized by a quantum circuit. Furthermore, suppose that the size $|G|$ of the group $G$ is no larger than linear in the number of qubits on which each circuit for $U(g)$ acts (so that $\log|G|$ is logarithmic in the number of qubits).  Thus, the size of the computational problem depends solely on the classical description of the Hamiltonian $H$ and the circuit descriptions of each unitary $U(g)$. Then we assert that the task of estimating the value
\begin{equation}
\frac{1}{d\left\vert G\right\vert }\sum_{g\in G}\operatorname{Tr}[U^{\dag
}(g)e^{iHt}U(g)e^{-iHt}]\label{eq:acceptance-prob-alg}
\end{equation}
is a DQC1-complete problem. To show this is in DQC1, we use the quantum
algorithm presented above and in \cite{laborde2022quantum}, as well as the observation from \cite[Section~1]{SJ08}, that the class DQC1 does not change if there are a logarithmic number of pure qubits, which is the case under our description of the problem stated above. 

To show hardness, let $U$ be a unitary realized by an arbitrary quantum circuit. We will now use the fact that estimating $\operatorname{Re}[\operatorname{Tr}[U^{2}]]/d$ is a DQC1-complete problem by showing that an algorithm for estimating the value in
\eqref{eq:acceptance-prob-alg} can estimate the value $\operatorname{Re} [\operatorname{Tr}[U^{2}]]/d$. To this end, we let the group $G$ be $\mathbb{Z}_{2}$ with representation $\left\{  I,V\right\}  $, where 
\begin{equation}
V=|0\rangle\!\langle1|\otimes U+|1\rangle\!\langle0|\otimes U^{\dag}.
\label{eq:V_def}
\end{equation}
(A quick calculation indicates that $V^2=I$, so that indeed $\{I,V\}$ is a representation of $\mathbb{Z}_2$.) A circuit for realizing $V$ can be efficiently generated from the circuit for realizing $U$. Indeed, we can construct a $0$-controlled-$U$ from~$U$:
\begin{equation}
|0\rangle\!\langle0|\otimes U+|1\rangle\!\langle1|\otimes I,
\end{equation}
and a $1$-controlled-$U^{\dag}$ from $U$, by reversing the order of the gates used to construct $U$, essentially running it backwards, with each circuit gate controlled on $1$, leading to
\begin{equation}
|0\rangle\!\langle0|\otimes I+|1\rangle\!\langle1|\otimes U^{\dag}.
\end{equation}
The overall circuit consists of $X\otimes I$, and then the above controlled
gates, so that
\begin{align}
&  \left(  |0\rangle\!\langle0|\otimes U+|1\rangle\!\langle1|\otimes I\right)
\left(  |0\rangle\!\langle0|\otimes I+|1\rangle\!\langle1|\otimes U^{\dag
}\right)  \left(  X\otimes I\right)  \nonumber\\
&  =\left(  |0\rangle\!\langle0|\otimes U+|1\rangle\!\langle1|\otimes U^{\dag
}\right)  \left(  X\otimes I\right)  \\
&  =|0\rangle\!\langle1|\otimes U+|1\rangle\!\langle0|\otimes U^{\dag}\\
&  =V.
\end{align}
We take the Hamiltonian to be one that realizes $H_{2}\otimes I$ via Hamiltonian evolution, where $H_{2}$ is a $2\times2$ Hadamard gate. Indeed, such a Hamiltonian is local and acts non-trivially on only one qubit. Then we have that $\left\vert G\right\vert =2$, its unitary representation is $\left\{  I,V\right\}  $, and $e^{-iHt}=H_{2}\otimes I$. Plugging into \eqref{eq:acceptance-prob-alg}, we find that
\begin{multline}
\frac{1}{d\left\vert G\right\vert }\sum_{g\in G}\operatorname{Tr}[U^{\dag
}(g)e^{iHt}U(g)e^{-iHt}]\label{eq:step-0}
=\frac{1}{2d}\operatorname{Tr}\left[  I\left(  H_{2}\otimes I\right)  I\left(
H_{2}\otimes I\right)  \right]  \\
+\frac{1}{2d}\operatorname{Tr}\left[  V\left(  H_{2}\otimes I\right)  V\left(
H_{2}\otimes I\right)  \right]  .
\end{multline}
Consider that
\begin{equation}
\operatorname{Tr}\left[  I\left(  H_{2}\otimes I\right)  I\left(  H_{2}\otimes
I\right)  \right]  =\operatorname{Tr}[I],\label{eq:step-1}
\end{equation}
and
\begin{align}
&  \operatorname{Tr}\left[  V\left(  H_{2}\otimes I\right)  V\left(
H_{2}\otimes I\right)  \right]  \nonumber\\
&  =\operatorname{Tr}\left[  \left(  |0\rangle\!\langle1|\otimes
U+|1\rangle\langle0|\otimes U^{\dag}\right)  \left(  H_{2}\otimes I\right)
\left(  |0\rangle\!\langle1|\otimes U+|1\rangle\!\langle0|\otimes U^{\dag
}\right)  \left(  H_{2}\otimes I\right)  \right]  \label{eq:step-2}\\
&  =\operatorname{Tr}\left[  \left(  |0\rangle\!\langle1|\otimes
U+|1\rangle\langle0|\otimes U^{\dag}\right)  \left(  |+\rangle\!\langle
-|\otimes U+|-\rangle\!\langle+|\otimes U^{\dag}\right)  \right]  \\
&  =\operatorname{Tr}\left[  |0\rangle\!\langle1|+\rangle\!\langle-|\otimes
U^{2}+|1\rangle\!\langle0|+\rangle\!\langle-|\otimes I+|0\rangle
\!\langle1|-\rangle\!\langle+|\otimes I+|1\rangle\!\langle0|-\rangle
\!\langle+|\otimes\left(  U^{\dag}\right)  ^{2}\right]  \\
&  =\langle1|+\rangle\!\langle-|0\rangle\operatorname{Tr}\left[  U^{2}\right]
+\langle0|+\rangle\!\langle-|1\rangle\operatorname{Tr}\left[  I\right]
+\langle1|-\rangle\!\langle+|0\rangle\operatorname{Tr}\left[  I\right]
+\langle0|-\rangle\!\langle+|1\rangle\operatorname{Tr}\left[  \left(  U^{\dag
}\right)  ^{2}\right]  \\
&  =\frac{1}{2}\left(  \operatorname{Tr}\left[  U^{2}\right]
-\operatorname{Tr}\left[  I\right]  -\operatorname{Tr}\left[  I\right]
+\operatorname{Tr}\left[  \left(  U^{\dag}\right)  ^{2}\right]  \right)
.\label{eq:step-2b}
\end{align}
Putting together \eqref{eq:step-0}, \eqref{eq:step-1}, and
\eqref{eq:step-2}--\eqref{eq:step-2b}, and noting the dimensions of the identities in \eqref{eq:step-1} and \eqref{eq:step-2b}, this implies for the above choices that
\begin{equation}
\frac{1}{d\left\vert G\right\vert }\sum_{g\in G}\operatorname{Tr}[U^{\dag
}(g)e^{iHt}U(g)e^{-iHt}]=\frac{1}{4} + \frac{\operatorname{Re}\left[  \operatorname{Tr}
\left[  U^{2}\right]  \right]  }{4d}
\end{equation}
Thus, the acceptance probability of this algorithm is an estimate of the desired quantity, up to a constant factor of $1/4$, establishing DQC1-hardness of our problem. Though this may seem unintuitive, since this reappropriation of the Hamiltonian symmetry testing algorithm takes only a few short steps, it suffices to verify our claims of computational complexity.

\section{A Derivation of Symmetry in the Acceptance Probability}\label{sec:hamaccept}

We have given statements about complexity and circuit construction, but the connection to symmetry may yet still seem aloof. Thus, we will continue no further without addressing the elephant in the room; from the acceptance probability given in \eqref{eq:modifiedacc}, we will derive a relationship with the familiar expression of Hamiltonian symmetry in quantum mechanics, further establishing this as an authentic test of symmetry. 

Begin by expanding the Hamiltonian evolution $e^{i H t}$, under the assumption that $\tau \coloneqq \left\Vert H \right\Vert_\infty t < 1$, where $\left\Vert X\right\Vert_\infty \coloneqq \sup_{\ket{\psi} \neq 0} \frac{\left\Vert X \ket{\psi}\right\Vert_2}{\left\Vert \ket{\psi}\right\Vert_2}$:
\begin{equation}\label{eq:hamtaylorseries}
    e^{iHt}=\mathbb{I} + iHt - \frac{H^2 t^2}{2}  - \frac{iH^3t^3}{6}+ \mathcal{O}(\tau^4).
\end{equation}
(This expansion is simply a truncated Taylor series.) Substituting this relation into the trace argument of \eqref{eq:modifiedacc}, 
we find that
\begin{multline}
    \operatorname{Tr}[U^\dag e^{iHt} U e^{-iHt}] 
     = d + t^2(\operatorname{Tr}[H U^\dag H U] - \operatorname{Tr}[H^2]) \\+ \frac{it^3}{2}(\operatorname{Tr}[U^\dag H^2 UH] - \operatorname{Tr}[U^\dag H UH^2]) + \mathcal{O}(\tau^4),
\end{multline}
where the equality is obtained using the linearity and cyclicity properties of the trace. After summing over all group elements, as in \eqref{eq:modifiedacc}, and using the group property (that $g\in G$ implies $g^{-1} \in G$), we find that $\frac{1}{|G|} \sum _{g \in G}(\operatorname{Tr}[U^\dag(g) H^2 U(g)H] - \operatorname{Tr}[U^\dag(g) H U(g)H^2]) = 0$. Thus, the third-order term of \eqref{eq:modifiedacc} vanishes. We can simplify the second-order term of \eqref{eq:modifiedacc} by using 
\begin{equation}
    \frac{1}{2}\operatorname{Tr}\!\left[|[U,H]|^2\right] =-\operatorname{Tr}[H U^\dag H U] + \operatorname{Tr}[H^2],
\end{equation}
where $|X|^2 \coloneqq X^\dag X$ implies that
\begin{equation}
    |[U,H]|^2 = H^2 - H U^\dag H U - U^\dag H U H + U^\dag H^2 U.
\end{equation}

Putting these equations together, we can  rewrite the acceptance probability of our first quantum algorithm elegantly as
\begin{equation}
    P_{\text{acc}} = 1 - \frac{t^2}{2 d |G|}\sum_{g \in G} \Big\|[U(g),H]\Big\|^2_2 + \mathcal{O}(\tau^4),
    \label{eq:accept-prob-comm}
\end{equation}
where $\left\| A \right\|_2 \coloneqq \sqrt{\operatorname{Tr}[|A|^2]}$ is the Hilbert--Schmidt norm defined in Section~\ref{sec:normsandlemmas}. Thus, to the first non-vanishing order of time $t$, the acceptance probability is equal to one if and only if
\begin{equation}
    [U(g),H] = 0, \quad \forall g \in G.
\end{equation}
This is exactly the familiar expression for symmetry in quantum mechanics. Furthermore, the expression in \eqref{eq:accept-prob-comm} clarifies that the normalized commutator norm $\frac{1}{ d |G|}\sum_{g \in G} \big\|[U(g),H]\big\|^2_2$ can be estimated efficiently by employing our algorithm. From \eqref{eq:accept-prob-comm}, we can see that the normalized commutator norm is small---equivalently, the Hamiltonian $H$ is approximately symmetric---if and only if the acceptance probability is close to one. Finally, as we show at length in Appendix~\ref{app:exact-expansion}, the acceptance probability has an exact expansion as follows, such that all odd powers in $t$ vanish and the even powers are scaled by normalized nested commutator norms, quantifying higher orders of symmetry:
\begin{equation}
    P_{\text{acc}} =\sum_{n=0}^{\infty}\frac{\left(  -1\right)
^{n}t^{2n}}{ \left(  2n!\right)  }\left(\frac{1}{d\left\vert G\right\vert}\sum_{g\in
G}\Big\Vert \left[  \left(  H\right)  ^{n},U(g)\right]  \Big\Vert _{2}^{2}\right)
\label{eq:acc_prob_beauty}
\end{equation}
where the nested commutator is defined as
\begin{equation}
[(X)^n,Y]  \coloneqq \underbrace{[X,\dotsb[X,[X}_{n \text { times }}, Y]] \dotsb],\quad
\quad [(X)^0,Y]  \coloneqq Y.
\label{eq:def-nested-comm}
\end{equation}
Note that the expansion in \eqref{eq:acc_prob_beauty} is valid for all $t\in\mathbb{R}$.
We also provide an alternative formula for $P_{\text{acc}}$ in Appendix~\ref{app:exact-expansion} using the group twirl. However, that formula lacks the elegant simplicity shown here, so we defer to \eqref{eq:acc_prob_beauty} in an act of blatant favoritism.

\section{Variational Quantum Algorithm for Symmetry Testing}\label{sec:hamvar}

Rather than feeding in the maximally-mixed state to the input of the circuit in Figure~\ref{fig:altcircuit}, we can instead feed in an arbitrary input state $\ket{\psi}$. The acceptance probability when doing so is equal to
\begin{equation}
    \left\Vert \mathcal{T}_G(e^{-iHt})
|\psi\rangle\right\Vert _{2}^{2} =
1-t^{2}\left\langle \mathcal{T}_{G}(H^{2})-\left(  \mathcal{T}
_{G}(H)\right)  ^{2}\right\rangle _{\psi}+O(\tau^{3}),
\label{eq:accept-prob-fixed-state}
\end{equation}
where $
\mathcal{T}_{G}(X)\coloneqq \frac{1}{|G|}\sum_{g\in G}U(g)XU^{\dag}(g)$ is the group twirl.
(Appendix~\ref{app:variational-sym-test} gives the full derivation for this expression.) Note that the bracketed term is non-negative as a consequence of the Kadison--Schwarz inequality \cite[Theorem~2.3.2]{bhatia07positivedefinitematrices}.
If we had the ability to prepare arbitrary quantum states as modeled in \cite{VW15}, we could optimize this acceptance probability over all states, resulting in the following value:
\begin{align}
    \left\Vert \mathcal{T}_G(e^{-iHt})\right\Vert _{\infty}^{2}   & \geq
1-\frac{2}{\left\vert G\right\vert }\sum_{g\in G}\left\Vert \left[
U(g),e^{-iHt}\right]  \right\Vert _{\infty} 
\label{eq:accept-prob-optimized} \\
& \geq 
1-\frac{2t}{\left\vert G\right\vert }\sum_{g\in G}\left\Vert \left[
U(g),H\right]  \right\Vert _{\infty}-4\tau^{2}
\label{eq:lower-bnd-est-small-t}.
\end{align}
These inequalities are proven in Appendix~\ref{app:variational-sym-test}, 
and the second holds under the assumption that $ \tau < 1$. This
demonstrates that the acceptance probability $\left\Vert \mathcal{T}_G(e^{-iHt})\right\Vert _{\infty}^{2}$ can be bounded from below in terms of a familiar expression of Hamiltonian symmetry. Thus, if the commutator norm $\frac{1}{\left\vert G\right\vert }\sum_{g\in G}\left\Vert \left[
U(g),H\right]  \right\Vert _{\infty}$ is small, as is the case when the Hamiltonian is approximately symmetric, then the acceptance probability of this algorithm is close to one. In Appendix~\ref{app:variational-sym-test}, 
we also prove that the acceptance probability satisfies
\begin{equation}
\left\Vert \mathcal{T}_G(e^{-iHt})\right\Vert _{\infty}^{2} \geq
    \left(  1-\sum_{n=1}^{\infty}\frac{t^{n}}{n!}\frac{1}{\left\vert
G\right\vert }\sum_{g\in G}\Big\Vert \left[  \left(  H\right)  ^{n}
,U(g)\right]  \Big\Vert _{\infty}\right)  ^{2}.
\label{eq:qma-nested-comm-bnd}
\end{equation}
(Safe to say that Appendix~\ref{app:variational-sym-test} contains all of the proofs for this section, as they are rather drawn-out mathematically but not particularly enlightening beyond their conclusions.)

Unfortunately, it is physically impossible to optimize over all input states. Instead, we can employ a variational ansatz to do so, in order to arrive at a lower-bound estimate of the acceptance probability on the left-hand side of \eqref{eq:accept-prob-optimized}. These methods have been vigorously pursued in recent years in the quantum computing literature \cite{CABBEFMMYCC20,bharti2021noisy}, and they can be combined with our approach here. In short, the acceptance probability in \eqref{eq:accept-prob-fixed-state} is a reward function that can be estimated by means of the circuit in Figure~\ref{fig:altcircuit} and a parameterized circuit that prepares the state $\ket{\psi}$. Then one can employ gradient ascent on a classical computer to modify the parameters used to prepare the state $\ket{\psi}$. After many iterations, these algorithms typically converge to a value, which in our case provides a lower bound estimate of the acceptance probability on the left-hand side  of~\eqref{eq:accept-prob-optimized}. In practice, it might be difficult in experiments to optimize over all pure states, and one could instead consider a variational product state ansatz, as in \cite{GLXXGLSPXZWZF21}. 

\section{\label{sec:hamexamples}Examples}

To exhibit our first algorithm, we consider three different examples. Namely, we consider the transverse Ising model with a cyclic boundary condition, the weakly $J$-coupled NMR Hamiltonian, and the Heisenberg XY model. In each case, we have employed IBM Quantum's noisy simulator to demonstrate the behavior of a quantum computer. In the near future, we suspect all of these algorithms will be testable on physical systems available to the public.

First, we consider the dynamics given by the transverse Ising model with a cyclic boundary condition. This Hamiltonian is given as $
    H_{\text{TIM}} \coloneqq \sigma^Z_N \otimes \sigma^Z_1 + \sum_{i=1}^{N-1} \sigma^Z_i \otimes \sigma^Z_{i+1} +  \sum_{i=1}^N \sigma^X_i$.
It is permutationally invariant, so that $[H_{\text{TIM}},W^{\pi}] = 0$ for all $\pi \in S_N$, where $W^{\pi}$ is a unitary representation of the permutation $\pi \in S_N$, with $S_N$ denoting the symmetric group on $N$ letters. It also obeys the symmetry $[H_{\text{TIM}},\sigma^X_1 \otimes \cdots \otimes \sigma^X_N] = 0$. Thus, we can use our algorithm to test these symmetries, and we do so in Figure~\ref{fig:ising} for $N=3$ and $N=4$. (Rather than test all permutations, we indicate here that we test for invariance under a cyclic shift, a subgroup of $S_N$.) We find that each respective symmetry test passes with reasonable probability, with deviation from one due to noise added to the simulation. 

\begin{figure}[t!]
\begin{center}
\includegraphics[width=3.4in]{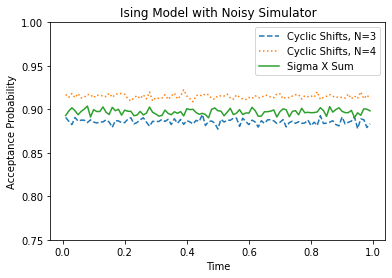}
\end{center}
\caption{Results of symmetry tests for the transverse Ising model for $N=3$ and $N=4$, using IBM Quantum's noisy simulator. The symmetries in question are given by acting simultaneously on all systems by either the cyclic group of order $N$ or a conjugation by~$(\sigma^X)^{\otimes N}$.}
\label{fig:ising}
\end{figure}

\begin{figure}[t!]
\begin{center}
\includegraphics[
width=2.7in
]{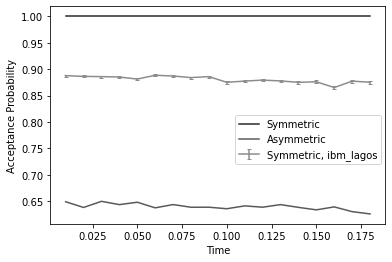}
\end{center}
\caption{Acceptance probability of our Hamiltonian symmetry testing algorithm over time, for the NMR Hamiltonian example and a group for which the Hamiltonian is either symmetric or asymmetric. The acceptance probability decays as time gets larger for the asymmetric case, even ideally. Example calculations using \textit{ibm\_lagos} show a large degree of initial symmetry before noise begins to dominate.}
\label{fig:combined}
\end{figure}

Next, we consider the example dynamics given by a weakly $J$-coupled NMR Hamiltonian \cite{van1996multidimensional}. This Hamiltonian can be expressed as 
\begin{equation}
    H_{\textrm{NMR}} \coloneqq \frac{\omega_1}{2} \sigma^Z_{1} + \frac{\omega_2}{2} \sigma^Z_{2} + \frac{\pi J}{2} \sigma^Z_{1} \otimes \sigma_{2}^Z,
\end{equation}
in units of $\hbar = 1$, where $\omega_i$ is the frequency associated to  spin $i\in\{1,2\}$ and $J$ is the coupling constant. This can be written as a diagonal matrix in the eigenbasis basis of the Pauli-$Z$ matrix; therefore, the time-evolution dynamics are also given by a diagonal matrix, as shown below. Due to this simplicity, this example can be easily simulated on noisy quantum computers for appropriately short times using a two-qubit system. 

It is clear that $H_\textrm{NMR}$ is symmetric with respect to the group generated by taking the Pauli-$Z$ gates on either qubit---this corresponds to a representation of the group $\mathbb{Z}_2 \times \mathbb{Z}_2$. It is not, however, symmetric under the group generated by the CNOT and SWAP gates acting on the two-qubit system---corresponding to $D_3$, the triangular dihedral group. Thus, using these two groups as described in our algorithms, we can visualize examples of both symmetry and asymmetry, as shown in Figure \ref{fig:combined}. To generate this Hamiltonian, we define the terms $\omega_{\textrm{AVG}}=\frac{1}{2}(\omega_1 + \omega_2)$ and $\Delta\omega = \omega_2 - \omega_1$ as is common. Then the Hamiltonian can be written as:
\begin{equation}
    H_{\text{NMR}} =
    \begin{pmatrix}
-\omega_{\text{AVG}} + \frac{\pi J}{2} & 0 & 0 & 0 \\
0 & \frac{\Delta\omega-\pi J}{2} & 0 & 0\\
0 & 0 & -\frac{\Delta\omega+\pi J}{2} & 0\\
0 & 0 & 0 & \omega_{\text{AVG}} + \frac{\pi J}{2}\\
\end{pmatrix}.
\end{equation}

\begin{figure}
\begin{center}
\includegraphics[width=3.4in]{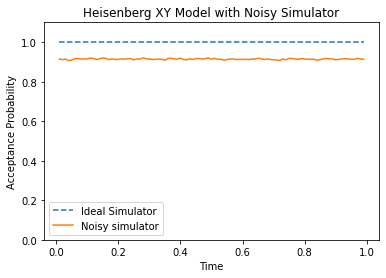}
\end{center}
\caption{Results of testing the Heisenberg XY model on four qubits using IBM Quantum's noisy simulator. The symmetries in question are given by acting simultaneously on all systems by the stated $\sigma_i$ matrix.}
\label{fig:HXY}
\end{figure}

As a final example, we consider the Heisenberg~XY model  \cite{LSM61}. The Hamiltonian under consideration is given by
\begin{equation}
    H_{\text{XY}} \coloneqq J \left( \sum_{i=1}^{N-1} \sigma^X_i \otimes \sigma^X_{i+1} + \sigma^Y_i \otimes \sigma^Y_{i+1}\right) \, ,
\end{equation}
where  $J > 0$ is the antiferromagnetic exchange interaction between spins.
This Hamiltonian showcases symmetry with respect to conjugation by each Pauli matrix acting on all qubits simultaneously. That is,
\begin{align}
    [H_{\text{XY}},\sigma^X_1 \otimes \cdots \otimes \sigma^X_N] = 0 , \\
    [H_{\text{XY}},\sigma^Y_1 \otimes \cdots \otimes \sigma^Y_N] = 0 ,\\
    [H_{\text{XY}},\sigma^Z_1 \otimes \cdots \otimes \sigma^Z_N] = 0 .
\end{align}
The symmetry group to consider in this case is thus $\{\pm 1, \pm i\} \times \{I^{\otimes N}, (\sigma^X)^{\otimes N},(\sigma^Y)^{\otimes N},(\sigma^Z)^{\otimes N}\} $. As the global phase factors are irrelevant in this case, we can test this symmetry by having two control qubits prepared in a uniform superposition, one of which activates $(\sigma^X)^{\otimes N}$ and the other activating $(\sigma^Z)^{\otimes N}$. We tested this symmetry by implementing our algorithm on the IBM Quantum noisy simulator, with $N=4$, and find that the symmetry test passes with reasonable probability, as indicated in Figure~\ref{fig:HXY}. The fact that the acceptance probability is not exactly equal to one has to do with the noise involved in the simulation.

We note here that all computer codes used to generate the examples in the main text and the supplementary material are available online.\footnote{\url{https://github.com/mlabo15/Hamiltonian-Symmetry}}

\section{Conclusion}

In this work, we have specified algorithms to test a Hamiltonian for symmetry with respect to a group. We have demonstrated this construction from a ground-up perspective beginning with channel covariance and ending with an efficient algorithm that relies solely on a maximally-mixed state as input. We show that this test is DQC1-complete and therefore considered classically hard and is furthermore efficiently realizable on quantum computers given similarly efficient Hamiltonian simulations. Even better, the acceptance probability directly relies upon the familiar expression of symmetry from quantum mechanics. We then give a second test employing a variational approach to estimate the commutator norm between the Hamiltonian evolution and group action. We have been successful in showing that these algorithms are useful tools that should be of interest throughout many realms of physics.

\pagebreak
\singlespacing
\chapter{Symmetry Testing of Quantum States}\label{ch:symmstates}
\doublespacing
\section{Introduction}

In many applications, the symmetries of states themselves are the most interesting investigation. For instance, permutation symmetry in the extension of a bipartite quantum state indicates a lack of entanglement in that state \cite{W89a,DPS02,DPS04}. This permutation symmetry limits entanglement, which relates to fundamental principles of quantum information such as the no-cloning theorem \cite{Park1970,D82,nat1982}\ and monogamy \cite{T04}. Additionally, consider a system of two parties when there is a lack of a shared reference frame between them. This implies that a quantum state prepared relative to another party's reference frame respects a certain symmetry and is less useful than one that breaks the same symmetry \cite{RevModPhys.79.555}. In all of these cases, a state respecting a symmetry is less resourceful than one that breaks it. 

Resource theory in quantum information is its own beast to consider. In recent years, quantum resource theories have been proposed for each of the above scenarios. Asymmetry in general has been established via resources theories \cite{MS13,MS14}, and the concept of unextendibility---a quantum state that disobeys permutation symmetry when artificially extended to larger Hilbert spaces---has been proposed as a potential resource as well \cite{KDWW19}. The above example of shared reference frames, or frameness, was discussed in \cite{GS08}. All of these aim to quantify the resourcefulness of quantum states. (For a more thorough review, see \cite{CG18}.) As such, it is useful to be able to test whether a quantum state possesses symmetry and quantify how much symmetry it possesses. 

In this chapter, we show how a quantum computer can test for symmetries of quantum states and channels generated by quantum circuits. In fact, our quantum-computational tests actually quantify how symmetric a state or channel is. Given that asymmetry is a useful resource in a wide variety of contexts, as alluded to above, while being potentially difficult for a classical computer to verify, our tests are helpful in determining how useful a state will be for certain quantum information processing tasks. These tests are in the spirit of the larger research program of using quantum computers to understand fundamental quantum-mechanical properties of high-dimensional quantum states, such as symmetry and entanglement, that are out of reach for classical computers.

We will make use of a general form of symmetry of quantum states, originally introduced in \cite{laborde2021testing}, that captures both the extendibility of bipartite states \cite{W89a,DPS02,DPS04}, as well as symmetries of a single quantum system with respect to a group of unitary transformations \cite{MS13,MS14}. Luckily, we have already introduced such notions of symmetry in Chapter~\ref{ch:intro}! We will draw heavily on the notions of $G$-symmetric extendibility and related definitions given in Section~\ref{sec:statenotions}, which were originally published contiguous with the results of this chapter---however, those notions will serve us well not only here but also in the next chapter (and to some degree, the previous chapter as well) so we have sneakily chosen to introduce them right from the beginning. 

In Section~\ref{sec:tests-o-sym}, we discuss the most important contribution of this chapter---namely, how a quantum computer can test for and estimate quantifiers of $G$-symmetric extendibility. These quantifiers are collectively called \textit{maximum symmetric fidelities}, with more particular names given as appropriate. We prove that our quantum computational tests of symmetry have acceptance probabilities precisely equal to these fidelities, thus endowing these resource-theoretic measures with operational meanings and allowing us to estimate them to arbitrary precision. Using complexity-theoretic language, we demonstrate that these quantum-computational tests of symmetry can be conducted in the form of a QIP(2) system \cite{W09,VW15}. (See Chapter~\ref{ch:intro} for a review of relevant quantum complexity theory, and Section~\ref{qipn} for QIP in particular.) Our results thus generalize previous results in the context of unextendibility and entanglement of bipartite quantum states \cite{HMW13,Hayden:2014:TQI}; additionally, we go on in the next section to clarify the relation between our results and previous ones in the literature. Simpler forms of the test can be conducted without any prover are thus efficiently computable on a quantum computer.

In Section~\ref{sec:specialized-tests}, we show how the various symmetry tests specialize for testing the $k$-extendibility or $k$-Bose extendibility \cite{W89a,DPS02,DPS04} of both bipartite and multipartite states. These serve as tests of separability, and we will be expanding upon them in Chapter~\ref{ch:gensep}. We also show how to test for the covariance symmetry of a quantum channel, which includes testing the symmetries of Hamiltonian evolution as a special case, generalizing the case discussed in the previous chapter.

Finally, in Section~\ref{sec:res-theories}, we review the resource theory of asymmetry \cite{MS13,MS14}. After doing so, we define several generalized resource theories of asymmetry, including both the resource theory of asymmetry and the resource theory of $k$-unextendibility \cite{KDWW19}\ as special cases. We also define resource theories of Bose asymmetry, which is an original contribution of our work in \cite{laborde2021testing}. This development shows that the acceptance probabilities of the aforementioned algorithms, i.e., maximum symmetric fidelities, are resource monotones and thus well-motivated from the resource-theoretic perspective.

In what follows, we proceed in the aforementioned order, and we finally reflect on these concepts in Section~\ref{sec:sym-conclusion}.

\section{Tests of Symmetry \& Extendibility}\label{sec:tests-o-sym}

To begin, recall the definitions given in Section~\ref{sec:statenotions} for $G$-Bose symmetric extendibility and $G$-symmetric extendibility.

\begin{definition}
[$G$-symmetric-extendible]A state $\rho_{S}$ is
$G$-symmetric-extendible if there exists a state $\omega_{RS}$ such that
\begin{enumerate}
\item the state $\omega_{RS}$ is an extension of $\rho_{S}$, i.e.,
\begin{equation}
\operatorname{Tr}_{R}[\omega_{RS}]=\rho_{S}, \label{eq:G-ext-1}
\end{equation}
\item the state $\omega_{RS}$ is $G$-invariant, in the sense that
\begin{equation}
\omega_{RS}=U_{RS}(g)\omega_{RS}U_{RS}(g)^{\dag}\qquad\forall g\in G.
\label{eq:G-ext-2}
\end{equation}
\end{enumerate}
\end{definition}

\begin{definition}
[$G$-Bose-symmetric-extendible]A state $\rho_{S}$ is $G$-Bose-symmetric-extendible (G-BSE) if there exists a state $\omega_{RS}$ such that
\begin{enumerate}
\item the state $\omega_{RS}$ is an extension of $\rho_{S}$, i.e.,
\begin{equation}
\operatorname{Tr}_{R}[\omega_{RS}]=\rho_{R},
\end{equation}
\item the state $\omega_{RS}$ satisfies
\begin{equation}
\omega_{RS}=\Pi_{RS}^{G}\omega_{RS}\Pi_{RS}^{G}.
\label{eq:bose-G-sym-ext-cond}
\end{equation}
\end{enumerate}
where $\Pi^G_{RS}$ is defined as in \eqref{eq:groupprojectorog}.
\end{definition}

Given these definitions, we will progress throughout this section by delineating algorithms to test for these abstract symmetries. We will then progressively narrow our view, moving from abstract tests of any symmetry, to separability tests, to covariance of quantum channels. We thus begin by defining tests of these above definitions, and, building upon them, we will create a managerie of algorithms forthwith. 

We assert that we can use a quantum computer to test for $G$-symmetric extendibility of a quantum state, as well as for other previously discussed forms of symmetry given in Section~\ref{sec:channelnotions}. We assume the following in doing so:
\begin{enumerate}
\item there is a quantum circuit available that prepares a purification $\psi_{S^{\prime}S}^{\rho}$ of the state $\rho_{S}$,
\item there is an efficient implementation of each of the unitary operators in the set $\{U_{RS}(g)\}_{g\in G}$,
\item and there is an efficient implementation of each of the unitary operators in the set $\{\overline{U}_{RS}(g)\}_{g\in G}$.
\end{enumerate}

The first assumption can be made less restrictive by employing the variational, purification-learning procedure from \cite{CSZW20}. That is, given a circuit that prepares the state $\rho_{S}$, the variational algorithm from \cite{CSZW20} outputs a circuit that approximately prepares a purification of $\rho_{S}$. We should note that the convergence of the algorithm from \cite{CSZW20} has not been established, and so the first assumption might be necessary for some applications.

The last assumption can be relaxed by the following reasoning: a standard gate set for approximating arbitrary unitaries in quantum computing consists of the controlled-NOT gate, the Hadamard gate, and the $T$ gate \cite{nielsenchuang}. The first two gates have only real entries while the $T$ gate is a diagonal $2\times2$ unitary gate with the entries $1$ and $e^{i\pi/4}$. Thus, the complex conjugate of this gate is equal to $T^{\dag}$. Thus, if a circuit for $U_{RS}(g)$ is constructed from this standard gate set, then we can generate a circuit for $\overline{U}_{RS}(g)$ by replacing every $T$ gate in the original circuit with $T^{\dag}$.

We now consider various quantum computational tests of symmetry that have increasing complexity. To give insight along the way, we provide an example along with the tests below. For this purpose, we employ the triangular dihedral group $D_3$.  This particular group was previously defined in Section~\ref{sec:grouptheory}, and for $D_3$, it has the lovely property of being isomorphic to the symmetric group on three elements---one of the smallest non-Abelian groups.

Of course, to actually implement this example, we will require a unitary representation that allows the group to be implemented on a quantum computer. To suit this purpose, we employ the two-qubit representation generated by setting the rotation element $r \to CNOT$ and the flip $f \to \text{CNOT} \circ \text{SWAP}$. A quick check will ensure that these generators obey the commutation rules of the group and will generate the full group defined in  Section~\ref{sec:grouptheory}. Throughout the next four sections, this group will be substituted into the presented algorithms to demonstrate their construction.

\subsection{Testing $G$-Bose Symmetry} \label{sec:simple-algorithm}

Let us begin by discussing the simplest version of the problem. Suppose that the state $\rho_{S}$ is pure, so that we can write it as $\rho_{S}=\psi_{S} \equiv \ket{\psi}\bra{\psi}_S$, and that the $R$ system is trivial. We recover the traditional case of $G$-Bose symmetry mentioned in Example~\ref{ex:usual-Bose-symmetry}. Thus, our goal is to decide if
\begin{equation}
\ket{\psi}_{S}=U_{S}(g)\ket{\psi}, \quad\forall g\in G.
\end{equation}
This condition is equivalent to
\begin{equation}
\ket{\psi}_{S}=\Pi_{S}^{G}\ket{\psi}_{S},
\end{equation}
where $\Pi_S^G$ is as given in \eqref{eq:unitaries-and-projs}. This last condition is, in turn, equivalent to the statement
\begin{equation}
\left\Vert \Pi_{S}^{G}\ket{\psi}_{S}\right\Vert_{2}=1.
\label{eq:proj-pure-state}
\end{equation}
The equivalence
\begin{equation}
\ket{\psi}_{S}=\Pi_{S}^{G}\ket{\psi}_{S}\quad\Leftrightarrow
\quad\left\Vert \Pi_{S}^{G}\ket{\psi}_{S}\right\Vert _{2}=1
\end{equation}
holds by expanding the norm in terms of the inner product and using the adjoint property of the projector or, alternatively, by the Pythagorean theorem. Thus, if we have a method to perform the projection onto $\Pi_{S}^{G}$, then we can decide whether \eqref{eq:proj-pure-state} holds.

There is a simple quantum algorithm to do so. This algorithm was originally proposed in \cite[Chapter~8]{harrow2005applications} under the name of ``generalized phase estimation.'' It proceeds as follows and can be summarized as \textquotedblleft performing the quantum phase estimation algorithm with respect to the unitary representation $\{U_{S}(g)\}_{g\in G}$\textquotedblright:
\newline
\begin{Algorithm}
[$G$-Bose symmetry test]\label{alg:simple}The algorithm consists of the following steps:
\begin{enumerate}
\item Prepare an ancillary register $C$ in the state $\ket{0}_{C}$.

\item Act on register $C$ with a quantum Fourier transform or other sequence of gates capable of creating a equal superposition over group elements.

\item Append the state $\ket{\psi}_{S}$ and perform the following controlled unitary:
\begin{equation}
\sum_{g\in G}\ket{g}\bra{g}_{C}\otimes U_{S}(g),
\end{equation}

\item Perform an inverse quantum Fourier transform on register~$C$, measure in the basis $\{\ket{g}\bra{g}_{C}\}_{g\in G}$, and accept if the zero outcome $\ket{0}\bra{0}_{C}$ occurs.
\end{enumerate}
\end{Algorithm}

\begin{figure}[ptb]
\begin{center}
\includegraphics[
width=3.5in
]{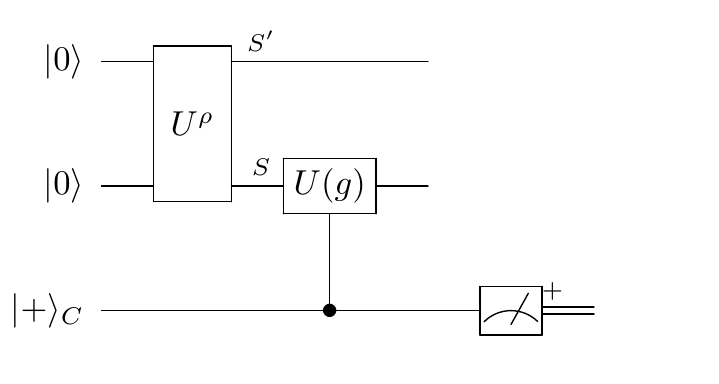}
\end{center}
\caption{Quantum circuit to implement Algorithm~\ref{alg:simple}. The unitary $U^{\rho}$ prepares a purification $\psi_{S^{\prime}S}$ of the state $\rho_{S}$. Algorithm~\ref{alg:simple} tests whether the state $\rho_{S}$ is $G$-Bose symmetric, as defined in Example~\ref{ex:usual-Bose-symmetry}. Its acceptance probability is equal to $\operatorname{Tr}[\Pi_{S}^{G}\rho_{S}]$, where $\Pi_{S}^{G}$ is defined in \eqref{eq:unitaries-and-projs}.}
\label{fig:simple-case}
\end{figure}

Note that the register $C$ has dimension $\left\vert G\right\vert $. Also, we can write the state $\ket{0}_{C}$ as $|e\rangle_{C}$, where $e$ is the identity element of the group. The result of Step~2 of Algorithm~\ref{alg:simple}\ is to prepare the following uniform superposition state:
\begin{equation}
\ket{+}_{C}\coloneqq \frac{1}{\sqrt{\left\vert G\right\vert }}\sum_{g \in G}\ket{g}_{C}. \label{eq:plus-over-group}
\end{equation}
We pause to note that although the quantum Fourier transform is specified in the above algorithm, in fact, any operation which creates the desired superposition state is equally acceptable.

Moving on, the overall state after Step~3 is as follows:
\begin{equation}
\frac{1}{\sqrt{\left\vert G\right\vert }}\sum_{g\in G}\ket{g}_{C} U_{S}(g)\ket{\psi}_{S}.
\end{equation}
The final step of Algorithm~\ref{alg:simple}\ projects the register $C$ onto the state $\ket{+}_{C}$. According to the aforementioned convention, Algorithm~\ref{alg:simple}\ accepts if the identity element outcome $\ket{+}\!\bra{+}$ occurs. The probability that Algorithm~\ref{alg:simple} accepts is equal to
\begin{align}
&  \left\Vert \left(  \langle+|_{C}\otimes I_{S}\right)  \left(  \frac{1}{\sqrt{\left\vert G\right\vert}}\sum_{g\in G}\ket{g}_{C}U_{S} (g)\ket{\psi}_{S}\right) \right\Vert _{2}^{2}\nonumber\\
&  =\left\Vert \frac{1}{\left\vert G\right\vert} \sum_{g\in G} U_{S}(g)\ket{\psi}_{S}\right\Vert_{2}^{2}\label{eq:simple-acc-prob-1}\\
&  =\left\Vert \Pi_{S}^{G}\ket{\psi}_{S}\right\Vert _{2}^{2}.
\label{eq:simple-acc-prob-2}
\end{align}

Figure~\ref{fig:simple-case} depicts this quantum algorithm. Not only does it decide whether the state $\ket{\psi}_{S}$ is symmetric, but it also quantifies how symmetric the state is. Since the acceptance probability is equal to $\left\Vert \Pi_{S}^{G}\ket{\psi}_{S}\right\Vert _{2}^{2}$, and this quantity is a measure of symmetry (see Theorem~\ref{thm:res-mono-G-Bose-sym}), we can repeat the algorithm a large number of times to estimate the acceptance probability to arbitrary precision.

The same quantum algorithm can decide whether a given mixed state $\rho_{S}$ is $G$-Bose symmetric (see Example~\ref{ex:usual-Bose-symmetry}). Similar to the above, it also can estimate how $G$-Bose symmetric the state $\rho_{S}$ is. To see this, consider that the acceptance probability for a pure state can be rewritten as follows:
\begin{equation}
\left\Vert \Pi_{S}^{G}\ket{\psi}_{S}\right\Vert_{2}^{2} = \operatorname{Tr}[\Pi_{S}^{G}\ket{\psi}\bra{\psi}_{S}].
\end{equation}
Then since every mixed state can be written as a probabilistic mixture of pure states, it follows that the acceptance probability of Algorithm~\ref{alg:simple}, when acting on the mixed state $\rho_{S}$, is equal to
\begin{equation}
\operatorname{Tr}[\Pi_{S}^{G}\rho_{S}]. \label{eq:acc-prob-bose-test}
\end{equation}

This acceptance probability is equal to one if and only if $\rho_{S}=\Pi_{S}^{G}\rho_{S}\Pi_{S}^{G}$, and so this test is a faithful test of $G$-Bose symmetry. The equivalence
\begin{equation}
\operatorname{Tr}[\Pi_{S}^{G}\rho_{S}]=1\quad\Leftrightarrow\quad\rho_{S} =\Pi_{S}^{G}\rho_{S}\Pi_{S}^{G} \label{eq:Bose-symmetric-equiv-cond}
\end{equation}
follows as a limiting case of Lemma~\ref{gentleM}, the gentle measurement lemma \cite{itit1999winter,ON07}, and also the positive definiteness of the trace norm. Again, through repetition, we can estimate the acceptance probability $\operatorname{Tr}[\Pi_{S}^{G}\rho_{S}]$ and then employ it as a measure of $G$-Bose symmetry (as we will show in a later section via Theorem~\ref{thm:res-mono-G-Bose-sym}, when it becomes time to discuss such things).

Interestingly, the acceptance probability of Algorithm~\ref{alg:simple} can be expressed as the \textit{maximum $G$-Bose-symmetric fidelity}, defined for a state $\rho_S$ as
\begin{equation}
\max_{\sigma_{S}\in \operatorname{B-Sym}_{G}}F(\rho_{S},\sigma_{S}),
\end{equation}
where
\begin{equation}
\operatorname{B-Sym}_{G}\coloneqq \left\{  \sigma_{S}\in\mathcal{D}(\mathcal{H}_{S}): \sigma_{S}= \Pi_{S}^{G}\sigma_{S}\Pi_{S}^{G}\right\}  ,
\end{equation}
and the fidelity of quantum states $\omega$ and $\tau$  is defined as \cite{U76}
\begin{equation}
F(\omega,\tau)\coloneqq \left\Vert \sqrt{\omega}\sqrt{\tau}\right\Vert_{1}^{2}.
\end{equation}
Thus, we arrive at the following theorem, Theorem~\ref{thm:acc-prob-g-Bose-sym}, and provide a proof in Appendix~\ref{app:acc-prob-g-Bose-sym}. As such, Algorithm~\ref{alg:simple} gives an operational meaning to the maximum $G$-Bose-symmetric fidelity in terms of its acceptance probability, and it can be used to estimate this fundamental measure of symmetry.

\begin{theorem}\label{thm:acc-prob-g-Bose-sym}
For a state $\rho_{S}$, the acceptance probability of Algorithm~\ref{alg:simple} is equal to the maximum $G$-Bose
symmetric fidelity. That is,
\begin{equation}
\operatorname{Tr}[\Pi_{S}^{G}\rho_{S}]=\max_{\sigma_{S}\in
\operatorname{B-Sym}_{G}}F(\rho_{S},\sigma_{S}).
\end{equation}
\end{theorem}

Let us construct this algorithm explicitly for the example of the dihedral group $D_3$. The $\ket{+}_C$ state is a uniform superposition of six elements, and we can achieve this using three qubits and unitary $U_d$ shown in  Figure~\ref{fig:Dihedral_Superposition}:
\begin{equation}\label{eq:Dihedral_Superposition}
    U_d\ket{000} = \frac{1}{\sqrt{6}} (\ket{000} + \ket{001} + \ket{010} + \ket{011} + \ket{100} + \ket{101}).
\end{equation}

\begin{figure}[h]
\begin{center}
\includegraphics[width=0.25\textwidth]{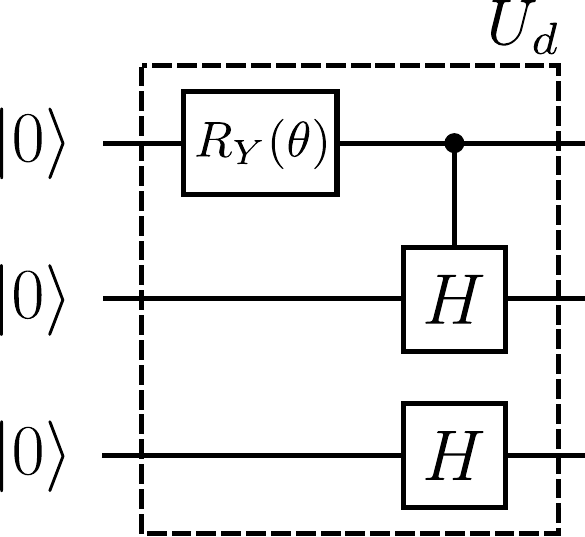}
\end{center}
\caption{Unitary $U_d$, with $\theta = 2 \arctan \left(\frac{1}{\sqrt{2}}\right)$, creates the equal superposition of six elements from \eqref{eq:Dihedral_Superposition}.}.
\label{fig:Dihedral_Superposition}
\end{figure}

These control register states need to be mapped to group elements to be meaningful; thus, we employ the mapping $\{\ket{000} \rightarrow e, \ket{001} \rightarrow fr^2, \ket{010} \rightarrow fr, \ket{011} \rightarrow r, \ket{100} \rightarrow f, \ket{101} \rightarrow r^2\}$. The circuit to test for $D_3$-symmetry is shown in Figure~\ref{fig:Dihedral_GBS_Circuit}.

\begin{figure}[h!]
\begin{center}
\includegraphics[width=0.45\textwidth]{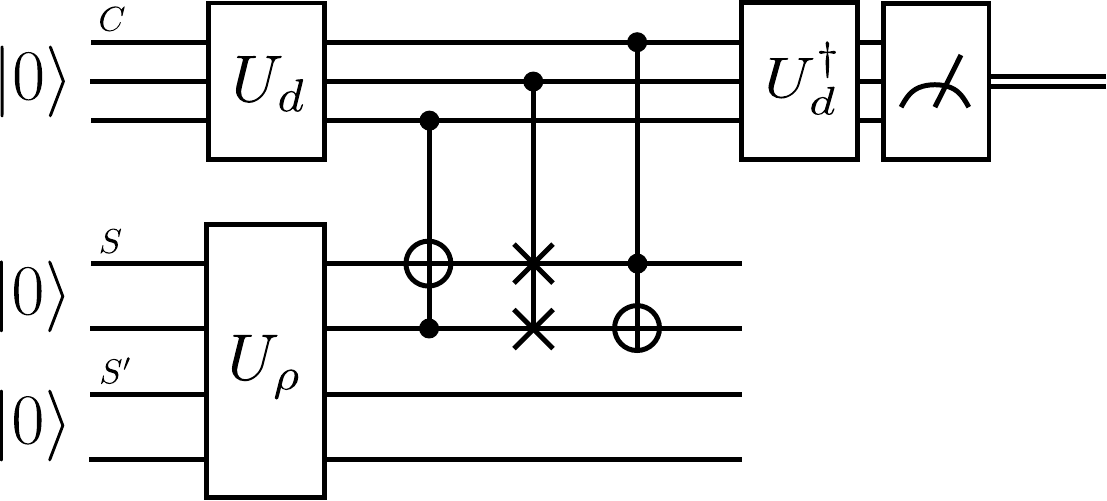}
\end{center}
\caption{ Quantum circuit implementing Algorithm~\ref{alg:simple} to test $G$-Bose symmetry for $D_3$. Compared to Figure~\ref{fig:simple-case}, the systems $S$ and $S'$ are two qubits each, $C$ consists of three qubits, and $\ket{+}_C$ is defined as $U_d\ket{000}$.}
\label{fig:Dihedral_GBS_Circuit}
\end{figure}

\subsection{Testing $G$-Symmetry}
\label{sec:G-sym-single-S-triv-R}

We now discuss how to modify Algorithm~\ref{alg:simple} to one that decides whether a state $\rho_{S}$ is $G$-symmetric (see Example~\ref{ex:usual-symmetry}), i.e., if
\begin{equation}
\rho_{S}=U_{S}(g)\rho_{S}U_{S}(g)^{\dag}\quad\forall g\in G.
\label{eq:rho-symmetric-single-sys}
\end{equation}
There is a subtlety here in the shift from pure states to density matrices---a wider class of quantum states---but also moving from the projector condition to a conjugation condition. We will also prove that the acceptance probability of the modified algorithm is equal to the \textit{maximum }$G$\textit{-symmetric fidelity}, defined as
\begin{equation}
\max_{\sigma\in\text{Sym}_{G}}F(\rho_{S},\sigma_{S}), \label{eq:fid-of-asym}
\end{equation}
where
\begin{equation}
\operatorname{Sym}_{G}\coloneqq \left\{ \sigma_{S}\in \mathcal{D} (\mathcal{H}_{S}):\sigma_{S} = U_{S}(g)\sigma_{S}U_{S}(g)^{\dag}\ \forall g\in G\right\}  ,
\end{equation}
$\mathcal{D}(\mathcal{H}_{S})$ denotes the set of density operators acting on the Hilbert space $\mathcal{H}_{S}$, and the fidelity of quantum states $\omega$ and $\tau$ is defined as \cite{U76}
\begin{equation}
F(\omega,\tau)\coloneqq \left\Vert \sqrt{\omega}\sqrt{\tau}\right\Vert_{1}^{2}.
\end{equation}
Thus, this quantum algorithm gives an operational meaning to the maximum $G$-symmetric fidelity in terms of its acceptance probability, and it can be used to estimate this fundamental measure of symmetry.

In the modified approach, we suppose that the quantum computer (now called the verifier) is equipped with access to a ``quantum prover''---an agent who can perform arbitrarily powerful quantum computations \cite{W09,VW15}. We discussed this situation exactly in Section~\ref{qipn} when we defined the complexity class QIP(n). We suppose that the quantum computer is allowed to exchange two quantum messages with the prover, thus placing our algorithm in QIP(2). We note here that computational problems related to entanglement of bipartite states \cite{HMW13,Hayden:2014:TQI}\ and recoverability of tripartite states \cite{PhysRevA.94.022310}\ were previously shown to be decidable in QIP(2). 

For this next algorithm, let $\ket{\psi}_{S^{\prime}S}$ be a purification of the state $\rho_{S}$, and suppose that the verifier has access to a circuit $U^{\rho}$\ that prepares this purification of $\rho_{S}$. Then proceed as follows:

\begin{Algorithm}
[$G$-symmetry test]\label{alg:g-sym-test} The algorithm consists of the following steps:

\begin{enumerate}
\item The verifier uses the circuit $U^{\rho}$ to prepare the state $\ket{\psi}_{S^{\prime}S}$.

\item The verifier transmits the purifying system $S^{\prime}$ to the prover.

\item The prover appends an ancillary register $E$ in the state $\ket{0}_{E}$ and performs a unitary $V_{S^{\prime}E\rightarrow\hat{S}E^{\prime}}$.

\item The prover sends the system $\hat{S}$ back to the verifier.

\item The verifier prepares a register $C$ in the state $\ket{0}_{C}$.

\item The verifier acts on register $C$ with a quantum Fourier transform or equivalent circuit.

\item The verifier performs the following controlled unitary:
\begin{equation}
\sum_{g\in G}\ket{g}\bra{g}_{C}\otimes U_{S}(g)\otimes\overline{U}_{\hat{S}}(g).
\end{equation}

\item The verifier performs an inverse quantum Fourier transform on register~$C$, measures in the basis $\{\ket{g}\bra{g}_{C}\}_{g\in G}$, and accepts if and only if the zero outcome $\ket{0}\bra{0}_{C}$ occurs.
\end{enumerate}
\end{Algorithm}

Figure~\ref{fig:case-2} depicts this quantum algorithm. The overall state after Step~3 of Algorithm~\ref{alg:g-sym-test}\ is
\begin{equation}
V_{S^{\prime}E\rightarrow\hat{S}E^{\prime}}\ket{\psi}_{S^{\prime}S} \ket{0}_{E}.
\end{equation}
The result of Step~6 is to prepare the uniform superposition state $\ket{+}_{C}$, which is defined in \eqref{eq:plus-over-group}. After Step~7, the overall state is
\begin{equation}
\frac{1}{\sqrt{\left\vert G\right\vert }}\sum_{g\in G}\ket{g}_{C}\left(U_{S}(g) \otimes \overline{U}_{\hat{S}}(g)\right)  V_{S^{\prime}E \rightarrow \hat{S}E^{\prime}}\ket{\psi}_{S^{\prime}S}\ket{0}_{E}.
\end{equation}

For a fixed unitary $V_{S^{\prime}E\rightarrow\hat{S}E^{\prime}}$, the probability of accepting, by following the same reasoning in \eqref{eq:simple-acc-prob-1}--\eqref{eq:simple-acc-prob-2}, is equal to
\begin{equation}
\left\Vert \Pi_{S\hat{S}}^{G}V_{S^{\prime}E\rightarrow\hat{S}E^{\prime}}\ket{\psi}_{S^{\prime}S}\ket{0}_{E}\right\Vert_{2}^{2},
\end{equation}
where
\begin{equation}
\Pi_{S\hat{S}}^{G}\coloneqq \frac{1}{\left\vert G\right\vert }\sum_{g\in G} U_{S}(g)\otimes\overline{U}_{\hat{S}}(g).
\end{equation}
Since the goal of the prover in a quantum interactive proof is to convince the verifier to accept \cite{W09,VW15}, the prover optimizes over every unitary $V_{S^{\prime}E \rightarrow \hat{S}E^{\prime}}$ and the acceptance probability of Algorithm~\ref{alg:g-sym-test}\ is given by
\begin{equation}
\max_{V_{S^{\prime}E\rightarrow\hat{S}E^{\prime}}}\left\Vert \Pi_{S\hat{S}}^{G} V_{S^{\prime}E \rightarrow \hat{S}E^{\prime}}\ket{\psi}_{S^{\prime}S}\ket{0}_{E}\right\Vert _{2}^{2}.
\end{equation}

\begin{figure}[ptb]
\begin{center}
\includegraphics[width=3.4in]{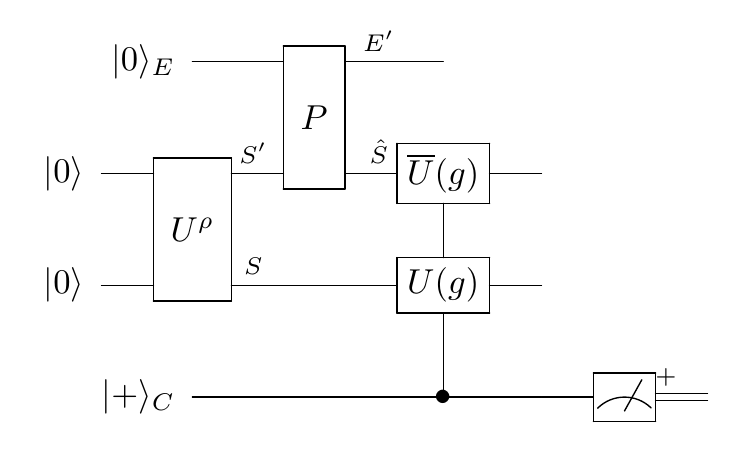}
\end{center}
\caption{Quantum circuit to implement Algorithm~\ref{alg:g-sym-test}. The unitary $U^{\rho}$ prepares a purification $\psi_{S^{\prime}S}$ of the state $\rho_{S}$. Algorithm~\ref{alg:g-sym-test} tests whether the state $\rho_{S}$ is $G$-symmetric, as defined in Example~\ref{ex:usual-symmetry}. Its acceptance probability is equal to the maximum $G$-symmetric fidelity, as defined in \eqref{eq:fid-of-asym}.}
\label{fig:case-2}
\end{figure}

The main idea behind Algorithm~\ref{alg:g-sym-test} is that if the state $\rho_{S}$ possesses the symmetry in \eqref{eq:rho-symmetric-single-sys}, then Theorem~\ref{thm:Bose-sym-purify} (with trivial reference system $R$) guarantees the existence of a purification $\phi_{S\hat{S}}$ of $\rho_{S}$ such that
\begin{equation}
\ket{\phi}_{S\hat{S}}=\Pi_{S\hat{S}}^{G}\ket{\phi}_{S\hat{S}}.
\label{eq:proj-cond-purification}
\end{equation}
Since all purifications of a quantum state are related by a unitary acting on the purifying system (see, e.g., \cite{wildebook}), the prover should be able to apply a unitary taking the purification $\ket{\psi}_{S^{\prime}S}$ to the purification $\ket{\psi}_{S\hat{S}}$. After the prover sends back the system $\hat{S}$, the verifier then performs a quantum-computational test to determine if the condition in \eqref{eq:proj-cond-purification} holds.

\begin{theorem}\label{thm:max-acc-prob-g-sym}
The acceptance probability of Algorithm~\ref{alg:g-sym-test}\ is equal to the maximum $G$-symmetric fidelity in \eqref{eq:fid-of-asym}, i.e.,
\begin{equation}
   \max_{V_{S^{\prime}E\rightarrow\hat{S}E^{\prime}}}\left\Vert \Pi_{S\hat{S}
}^{G}V_{S^{\prime}E\rightarrow\hat{S}E^{\prime}}\ket{\psi}_{S^{\prime}
S}\ket{0}_{E}\right\Vert _{2}^{2}
=\max_{\sigma_{S}\in\operatorname{Sym}_{G}}F(\rho_{S},\sigma_{S}). 
\end{equation}
\end{theorem}

\begin{proof}
Recall the following property of the norm of an arbitrary vector
$\ket{\varphi}$:
\begin{equation}
\left\Vert \ket{\varphi}\right\Vert _{2}^{2} = \max_{\ket{\phi}: \left\Vert \ket{\phi}\right\Vert _{2}=1}\left\vert \langle\phi\ket{\varphi} \right\vert ^{2}. \label{eq:euclidean-norm-opt}
\end{equation}
This follows from the Cauchy--Schwarz inequality and the conditions for saturating it. The formula in \eqref{eq:euclidean-norm-opt} implies that
\begin{multline}
\max_{V_{S^{\prime}E\rightarrow\hat{S}E^{\prime}}} \left\Vert \Pi_{S\hat{S}}^{G} V_{S^{\prime} E \rightarrow \hat{S} E^{\prime}} \ket{\psi}_{S^{\prime}S}\ket{0}_{E} \right\Vert_{2}^{2}\label{eq:max-acc-prob-proof-step-1}\\
=\max_{V_{S^{\prime}E\rightarrow\hat{S}E^{\prime}},\ket{\phi}_{S\hat{S} E^{\prime}}} \left\vert \bra{\phi}_{S\hat{S}E^{\prime}}\Pi_{S\hat{S}} ^{G}V_{S^{\prime}E\rightarrow\hat{S}E^{\prime}}\ket{\psi}_{S^{\prime}S}\ket{0}_{E}\right\vert ^{2}.
\end{multline}
For positive semi-definite operators $\omega_{A}$ and $\tau_{A}$ and corresponding rank-one operators $\psi_{RA}^{\omega}$ and $\psi_{RA}^{\tau}$ satisfying
\begin{align}
\operatorname{Tr}_{R}[\psi_{RA}^{\omega}]  &  =\omega_{A},\\
\operatorname{Tr}_{R}[\psi_{RA}^{\tau}]  &  =\tau_{A},
\end{align}
Uhlmann's theorem \cite{U76} states  that
\begin{equation}
\left\Vert \sqrt{\omega_{A}}\sqrt{\tau_{A}}\right\Vert _{1}^{2}=\max_{V_{R}}\left\vert \bra{\psi^{\omega}}_{RA}\left(  V_{R}\otimes I_{A}\right)
|\psi^{\tau}\rangle_{RA}\right\vert ^{2},
\end{equation}
where the optimization is over every unitary $V_{R}$ acting on the reference
system $R$. Applying this theorem to \eqref{eq:max-acc-prob-proof-step-1} with
the identifications $R\leftrightarrow\hat{S}E^{\prime}\simeq S^{\prime}E$ and
$S\leftrightarrow A$ and noting that
\begin{align}
\operatorname{Tr}_{S^{\prime}E}[\ket{\psi}\!\bra{\psi}_{S^{\prime}S}\otimes\ket{0}\!\bra{0}_{E}]  &  =\rho_{S},\\
\operatorname{Tr}_{\hat{S}E^{\prime}}[\Pi_{S\hat{S}}^{G}\ket{\phi}\!\bra{\phi}_{S\hat{S} E^{\prime}}\Pi_{S\hat{S}}^{G}]  &  =\operatorname{Tr}_{\hat{S}}[\Pi_{S\hat{S}}^{G}\sigma_{S\hat{S}^{\prime}}\Pi_{S\hat{S}}^{G}],
\end{align}
where $\sigma_{S\hat{S}^{\prime}}$ is a quantum state satisfying $\sigma_{S\hat{S}^{\prime}}=\operatorname{Tr}_{E^{\prime}}[\ket{\phi}\!\bra{\phi}_{S\hat{S}E^{\prime}}]$, we conclude that
\begin{equation}\label{eq:1st-fid-formula-pf}
\max_{V_{S^{\prime}E\rightarrow \hat{S}E^{\prime}}, \ket{\phi}_{S\hat{S}E^{\prime}}}\left\vert \bra{\phi}_{S\hat{S}E^{\prime}} \Pi_{S\hat{S}}^{G}V_{S^{\prime}E \rightarrow \hat{S}E^{\prime}} \ket{\psi}_{S^{\prime}S}\ket{0}_{E}\right\vert ^{2} = \max_{\sigma_{S\hat{S}^{\prime}}} F(\rho_{S},\operatorname{Tr}_{\hat{S}}[\Pi_{S\hat{S}}^{G}\sigma_{S\hat{S}^{\prime}}\Pi_{S\hat{S}}^{G}]),
\end{equation}
with the optimization in the last line over every quantum state $\sigma_{S\hat{S}^{\prime}}$.

We finally prove that
\begin{equation}
\max_{\sigma_{S\hat{S}^{\prime}}} F(\rho_{S}, \operatorname{Tr}_{\hat{S}} [\Pi_{S\hat{S}}^{G} \sigma_{S\hat{S}^{\prime}}\Pi_{S\hat{S}}^{G}]) = \max_{\sigma_{S}\in\operatorname{Sym}_{G}} F(\rho_{S},\sigma_{S}).
\label{eq:final-eq-steps-1}
\end{equation}
To prove this equality, we will first show that the left-hand side of \eqref{eq:final-eq-steps-1} is greater than or equal to the right-hand side, and then we shall show it must also be less than or equal to the right-hand side. To justify the forward direction of \eqref{eq:final-eq-steps-1}, let $\sigma_{S}\in\operatorname{Sym}_{G}$, and pick $\sigma_{S\hat{S}}$ to be the purification $\varphi_{S\hat{S}}$ of $\rho_{S}$ from Theorem~\ref{thm:Bose-sym-purify} (with systems $R\hat{R}$ trivial) that satisfies
\begin{equation}
\Pi_{S\hat{S}}^{G}\varphi_{S\hat{S}}\Pi_{S\hat{S}}^{G}=\varphi_{S\hat{S}}.
\end{equation}
Observe that
\begin{equation}
\operatorname{Tr}_{\hat{S}}[\Pi_{S\hat{S}}^{G}\varphi_{S\hat{S}}\Pi_{S\hat{S}}^{G}]=\operatorname{Tr}_{\hat{S}}[\varphi_{S\hat{S}}]=\sigma_{S},
\end{equation}
and so, given that $\sigma_{S}\in\operatorname{Sym}_{G}$ is arbitrary, it follows that
\begin{equation}
\max_{\sigma_{S\hat{S}^{\prime}}}F(\rho_{S},\operatorname{Tr}_{\hat{S}} [\Pi_{S\hat{S}}^{G}\sigma_{S\hat{S}^{\prime}}\Pi_{S\hat{S}}^{G}] ) \geq \max_{\sigma_{S}\in\operatorname{Sym}_{G}}F(\rho_{S},\sigma_{S})\, .
\end{equation}

To justify the reverse direction in \eqref{eq:final-eq-steps-1}, let $\sigma_{S\hat{S}}$ be an arbitrary state. If $\sigma_{S\hat{S}^{\prime}}$ is outside of the subspace onto which $\Pi_{S\hat{S}}^{G}$ projects, then $\Pi_{S\hat{S}}^{G} \sigma_{S\hat{S}^{\prime}} \Pi_{S\hat{S}}^{G}=0$ and the fidelity in \eqref{eq:1st-fid-formula-pf} is equal to zero. Let us suppose that this is not the case, and let us define
\begin{align}
\sigma_{S\hat{S}}^{\prime}  &  \coloneqq \frac{1}{p}\Pi_{S\hat{S}}^{G} \sigma_{S\hat{S}^{\prime}}\Pi_{S\hat{S}}^{G},\\
p  &  \coloneqq \operatorname{Tr}[\Pi_{S\hat{S}}^{G}\sigma_{S\hat{S}^{\prime}}].
\end{align}
Then consider the following fidelity,
\begin{align}
F(\rho_{S},\operatorname{Tr}_{\hat{S}}[\Pi_{S\hat{S}}^{G}\sigma_{S\hat{S}^{\prime}}\Pi_{S\hat{S}}^{G}])  &  =p F(\rho_{S},\tau_{S})\\
&  \leq F(\rho_{S},\tau_{S}),
\end{align}
where
\begin{equation}
\tau_{S}\coloneqq \operatorname{Tr}_{\hat{S}}[\sigma_{S\hat{S}}^{\prime}],
\end{equation}
and we used the fact that $p\leq1$. If $\tau_{S} \in \operatorname{Sym}_{G}$, we will have completed our argument. To see that this is true, we can perform a series of manipulations using the definition of $\tau_{S}$ as follows
\begin{align}
\tau_{S}  &  =\operatorname{Tr}_{\hat{S}}[\sigma_{S\hat{S}}^{\prime}]\\
&  =\operatorname{Tr}_{\hat{S}} [\Pi_{S\hat{S}}^{G}\sigma_{S\hat{S}}^{\prime} \Pi_{S\hat{S}}^{G}] \\
&  =\operatorname{Tr}_{\hat{S}}[\left(  U_{S}\otimes\overline{U}_{\hat{S}}\right)  \Pi_{S\hat{S}}^{G}\sigma_{S\hat{S}}^{\prime}\Pi_{S\hat{S}}^{G}\left(  U_{S}\otimes\overline{U}_{\hat{S}}\right)  ^{\dag}]\\
&  =U_{S}\operatorname{Tr}_{\hat{S}}[\overline{U}_{\hat{S}}\Pi_{S\hat{S}}^{G}\sigma_{S\hat{S}}^{\prime}\Pi_{S\hat{S}}^{G}\overline{U}_{\hat{S}}^{\dag}]U_{S}^{\dag}\\
&  =U_{S}\operatorname{Tr}_{\hat{S}}[\overline{U}_{\hat{S}}^{\dag}\overline{U}_{\hat{S}}\Pi_{S\hat{S}}^{G}\sigma_{S\hat{S}}^{\prime}\Pi_{S\hat{S}}^{G}]U_{S}^{\dag}\\
&  =U_{S}\operatorname{Tr}_{\hat{S}}[\Pi_{S\hat{S}}^{G}\sigma_{S\hat{S}}^{\prime}\Pi_{S\hat{S}}^{G}]U_{S}^{\dag}\\
&  =U_{S}(g)\operatorname{Tr}_{\hat{S}}[\sigma_{S\hat{S}}^{\prime}]U_{S}^{\dag}(g)\\
&  =U_{S}(g)\tau_{S}U_{S}^{\dag}(g).
\end{align}
where we have used the shorthand $U_{S}\equiv U_{S}(g)$ and $\overline {U}_{\hat{S}} \equiv \overline{U}_{\hat{S}}(g)$. Since the equality $\tau_{S}=U_{S}(g)\tau_{S}U_{S}^{\dag}(g)$ holds for all $g\in G$, it follows that
\begin{equation}
\max_{\sigma_{S\hat{S}^{\prime}}}F(\rho_{S},\operatorname{Tr}_{\hat{S}} [\Pi_{S\hat{S}}^{G}\sigma_{S\hat{S}^{\prime}}\Pi_{S\hat{S}}^{G}])\leq \max_{\tau_{S}\in\operatorname{Sym}_{G}}F(\rho_{S},\sigma_{S}).
\label{eq:final-eq-steps-last}
\end{equation}
\end{proof}

Now let us return to our example of $D_3$, and use this construction to decompress from proofs for a bit. For this circuit, we use the same unitary $U_d$ to prepare the superposition $\ket{+}_C$, and the same mapping of control states to group elements. Then the circuit to test for $G$-symmetry (or $D_3$-symmetry, as it were) is shown in Figure~\ref{fig:Dihedral_GS_Circuit}.

\begin{figure}[h]
\begin{center}
\includegraphics[width=0.5\textwidth]{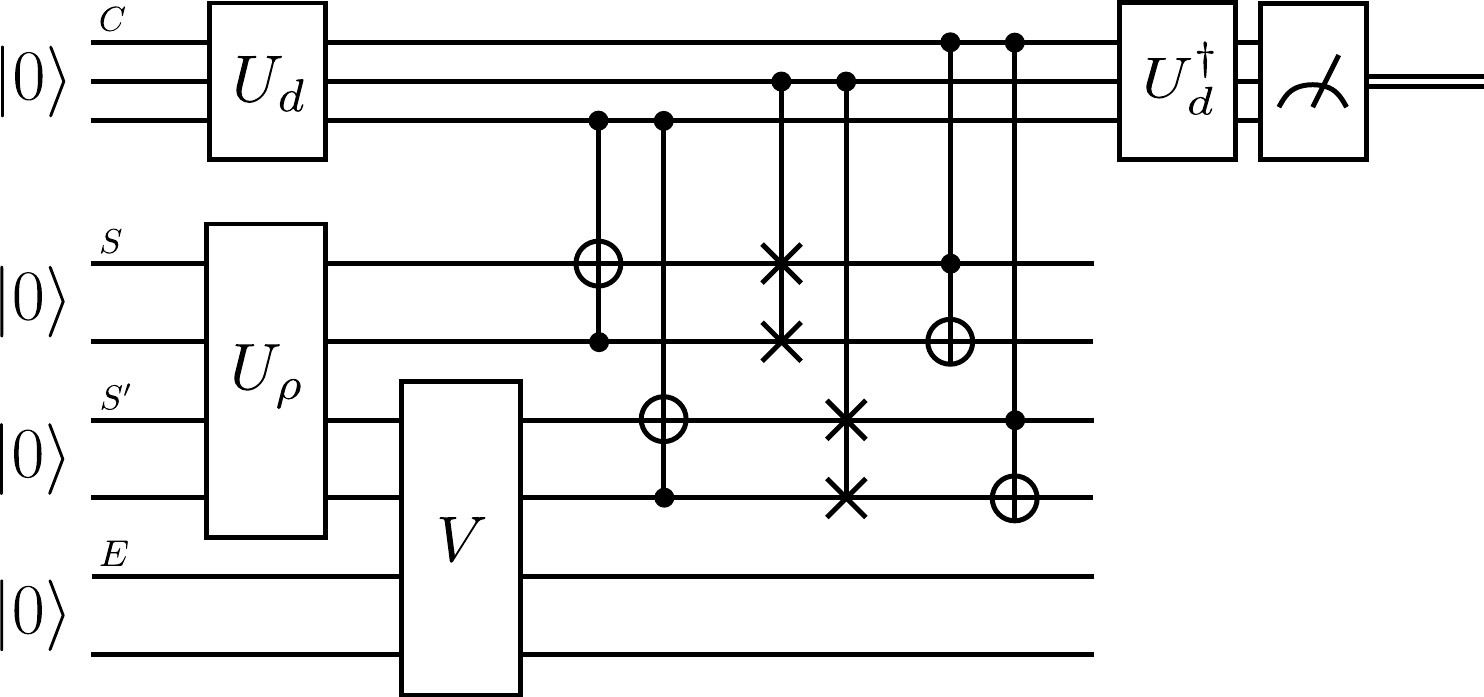}
\end{center}
\caption{Quantum circuit implementing Algorithm~\ref{alg:g-sym-test} to test $G$-symmetry in the case that the group $G$ is the triangular dihedral group. Compared to Figure~\ref{fig:case-2}, the systems $S$ and $S'$ are two qubits each, $C$ consists of three qubits, and $\ket{+}_C$ is defined as $U_d\ket{000}$. Both the SWAP and CNOT gates have no imaginary entries, and thus they are equal to their own complex conjugates.}
\label{fig:Dihedral_GS_Circuit}
\end{figure}

Now, in Figure~\ref{fig:Dihedral_GS_Circuit}, note that we have introduced a prover in the form of a unitary $V$. This prover has access both to the purifying system of our state as well as the environment. As discussed before, provers in QIP(2) have unbounded computational power, and so it may feel like we have reached a wall impeding our progress. Clearly, there is no way incorporate any physical implementation of $V$ as required by complexity theory---so instead, we must bend the rules around this definition. For practical implementations of such algorithms, a variational quantum algorithm (VQA) can stand in as a lower-bound version of the prover. In \cite{laborde2021testing}, the behaviors of these algorithms are shown explicitly. (These VQA results are surpressed from this text but are nonetheless impressive in demonstrating the power of machine learning to approximate all-powerful wizards as the need may be.)

\begin{remark}[Testing incoherence]
\label{rem:incoherence}
We would now like to discuss a particular application of our algorithm: testing for incoherence. Testing the incoherence of a quantum state, in the sense of \cite{BCP14,SAP17}, is a special case of testing $G$-symmetry. To see this, pick $G$ to be the cyclic group over $d$ elements ($C_d$) with unitary representation $\{Z(z)\}_z$, where $Z(z)$ is the generalized Pauli phase-shift unitary, defined as
\begin{equation}
    Z(z) \coloneqq \sum_{j=0}^{d-1} e^{2 \pi i j z /d} \ket{j}\bra{j}.
\end{equation}
A state is symmetric with respect to this group if the condition in \eqref{eq:rho-symmetric-single-sys} holds. This condition is equivalent to 
\begin{equation}
    \rho_S = \frac{1}{|G|} \sum_{g \in G} U_S(g) \rho_S U_S(g)^\dag.
\end{equation}
For the choice mentioned above, this condition holds if and only if the state $\rho_S$ is diagonal in the incoherent basis, i.e., if it can be written as $\rho_S = \sum_j p(j)  \ket{j}\bra{j}$, where $p(j)$ is a probability distribution. Thus, Algorithm~\ref{alg:g-sym-test} can be used to test the incoherence of quantum states.
\end{remark}

\subsection{Testing $G$-Bose Symmetric Extendibility}
\label{sec:G-Bose-sym-ext-test}

We now describe an algorithm for testing $G$-Bose symmetric extendibility of a quantum state $\rho_{S}$, as defined in Definition~\ref{def:g-bose-sym-ext}. The algorithm bears some similarities with Algorithms~\ref{alg:simple} and ~\ref{alg:g-sym-test}. Like Algorithm~\ref{alg:g-sym-test}, it involves an interaction between a verifier and a prover. We prove that its acceptance probability is equal to the maximum $G$-BSE fidelity:
\begin{equation}
\max_{\sigma_{S}\in\operatorname*{BSE}_{G}}F(\rho_{S},\sigma_{S}),
\label{eq:fid-bose-asym-3}
\end{equation}
where BSE$_{G}$ is the set of $G$-Bose symmetric extendible states:
\begin{equation}
\text{BSE}_{G}\coloneqq 
\left\{\begin{array}[c]{c}
\sigma_{S}:\exists\omega_{RS}\in\mathcal{D}(\mathcal{H}_{RS}),\ \operatorname{Tr}_R [\omega_{RS}] =\sigma_S, \\ 
\omega_{RS}=U_{RS}(g)\omega_{RS}\text{, }\forall g\in G
\end{array}
\right\}  .
\label{eq:G-BSE-states-set}
\end{equation}

\noindent Thus, the algorithm endows the maximum $G$-BSE fidelity with an operational meaning. Note that the condition $\omega_{RS}=U_{RS}(g)\omega_{RS}$ $\forall g\in G$ is equivalent to
\begin{equation}
\omega_{RS}=\Pi_{RS}^{G}\omega_{RS}\Pi_{RS}^{G},
\end{equation}
where
\begin{equation}
\Pi_{RS}^{G}\coloneqq\frac{1}{\left\vert G\right\vert }\sum_{g\in G}U_{RS}(g)\, ,
\label{eq:Pi_RS-proj-again}
\end{equation}
as before.

The algorithm is highly similar to Algorithm~\ref{alg:g-sym-test}, but we list it here for completeness. Let $\ket{\psi}_{S^{\prime}S}$ be a purification of the state $\rho_{S}$, and suppose that the circuit $U^{\rho}$ prepares this purification of $\rho_{S}$.

\begin{Algorithm}
[$G$-BSE test]\label{alg:G-BSE-test} The algorithm proceeds as follows:

\begin{enumerate}
\item The verifier uses the circuit provided to prepare the state $\ket{\psi}_{S^{\prime}S}$.

\item The verifier transmits the purifying system $S^{\prime}$ to the prover.

\item The prover appends an ancillary register $E$ in the state $\ket{0}_{E}$ and performs a unitary $V_{S^{\prime}E\rightarrow RE^{\prime}}$.

\item The prover sends the system $R$ back to the verifier.

\item The verifier prepares a register $C$ in the state $\ket{0}_{C}$.

\item The verifier acts on register $C$ with a quantum Fourier transform or equivalent sequence of gates.

\item The verifier performs the following controlled unitary:
\begin{equation}
\sum_{g\in G}\ket{g}\bra{g}_{C}\otimes U_{RS}(g),
\end{equation}

\item The verifier performs an inverse quantum Fourier transform on register~$C$, measures in the basis $\{\ket{g}\!\bra{g}_{C}\}_{g\in G}$, and accepts if and only if the zero outcome $\ket{0}\!\bra{0}_{C}$ occurs.
\end{enumerate}
\end{Algorithm}

\begin{figure}[ptb]
\begin{center}
\includegraphics[width=3.5in]{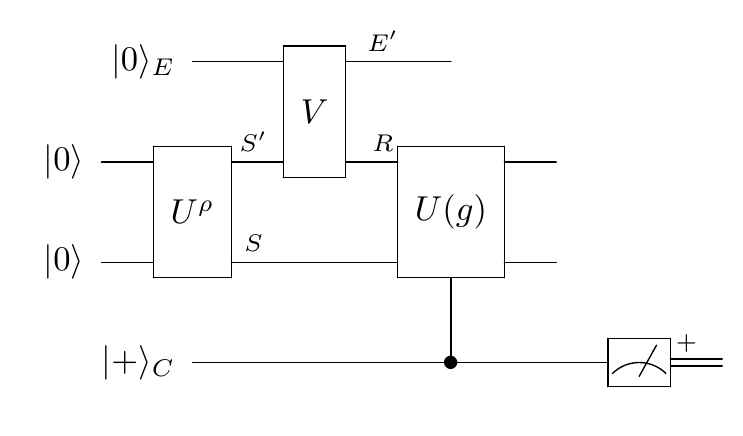}
\end{center}
\caption{Quantum circuit to implement Algorithm~\ref{alg:G-BSE-test}. The unitary $U^{\rho}$ prepares a purification $\psi_{S^{\prime}S}$ of the state $\rho_{S}$. Algorithm~\ref{alg:G-BSE-test} tests whether the state $\rho_{S}$ is $G$-Bose symmetric extendible, as defined in Definition~\ref{def:g-bose-sym-ext}. Its acceptance probability is equal to the maximum $G$-BSE fidelity, as defined in~\eqref{eq:fid-bose-asym-3}.}
\label{fig:case-3}
\end{figure}

Figure~\ref{fig:case-3} depicts this quantum algorithm. The overall state after Step~3 is
\begin{equation}
V_{S^{\prime}E\rightarrow RE^{\prime}}\ket{\psi}_{S^{\prime}S}\ket{0}_{E}.
\end{equation}
Step~6 prepares the uniform superposition state $\ket{+}_{C}$, which is defined in \eqref{eq:plus-over-group}. After Step~7, the overall state is
\begin{equation}
\frac{1}{\sqrt{\left\vert G\right\vert }} \sum_{g\in G}\ket{g}_{C}\left(U_{RS}(g)\right)  V_{S^{\prime}E\rightarrow RE^{\prime}}\ket{\psi}_{S^{\prime}S}\ket{0}_{E}.
\end{equation}
The last step can be understood as the verifier projecting the register $C$
onto the state $\ket{+}_{C}$.

The probability of accepting, following the same reasoning as before, is equal to
\begin{equation}
\left\Vert \Pi_{RS}^{G}V_{S^{\prime}E\rightarrow RE^{\prime}}\ket{\psi}_{S^{\prime}S} \ket{0}_{E}\right\Vert _{2}^{2},
\end{equation}
where $\Pi_{RS}^{G}$ is defined in \eqref{eq:Pi_RS-proj-again}. As before, the goal of the prover in a quantum interactive proof is to convince the verifier to accept \cite{W09,VW15}, and so the prover optimizes over every unitary $V_{S^{\prime}E\rightarrow\hat{S}E^{\prime}}$. The acceptance probability of Algorithm~\ref{alg:G-BSE-test} is then given by
\begin{equation}
\max_{V_{S^{\prime}E\rightarrow RE^{\prime}}}\left\Vert \Pi_{RS}^{G}V_{S^{\prime}E\rightarrow RE^{\prime}}\ket{\psi}_{S^{\prime}S}\ket{0}_{E}\right\Vert _{2}^{2}.
\end{equation}

For completeness, we provide our proof in Appendix~\ref{app:proof-thm-g-bse}, but do not include it here as it is highly similar to the proof given for Theorem~\ref{thm:max-acc-prob-g-sym}; 

\begin{theorem}
\label{thm:G-BSE-acc-prob}
The maximum acceptance probability of Algorithm~\ref{alg:G-BSE-test} is equal to the maximum $G$-BSE\ fidelity in \eqref{eq:fid-bose-asym-3}, i.e.,
\begin{equation}
\max_{V_{S^{\prime}E\rightarrow RE^{\prime}}}\left\Vert \Pi_{RS} ^{G}V_{S^{\prime}E\rightarrow RE^{\prime}}\ket{\psi}_{S^{\prime}S} \ket{0}_{E}\right\Vert _{2}^{2} = \max_{\sigma_{S}\in\operatorname*{BSE}_{G} }F(\rho_{S},\sigma_{S}),
\end{equation}
where the set $\operatorname*{BSE}_{G}$ is defined in \eqref{eq:G-BSE-states-set}.
\end{theorem}

\subsection{Testing $G$-Symmetric Extendibility}

\label{sec:test-g-sym-ext}The final algorithm that we introduce tests whether a state $\rho_{S}$ is $G$-symmetric extendible (recall Definition~\ref{def:g-sym-ext}). Similar to the algorithms in the previous sections, not only does it decide whether $\rho_{S}$ is $G$-symmetric extendible, but it also quantifies how similar it is to a state in the set of $G$-symmetric extendible states. The acceptance probability is equal to the \textit{maximum }$G$\textit{-symmetric extendible fidelity}:
\begin{equation}
\max_{\sigma_{S}\in\operatorname{SymExt}_{G}}F(\rho_{S},\sigma_{S}
),\label{eq:max-g-sym-ext-fid}
\end{equation}
where
\begin{equation}
\operatorname{SymExt}_{G}\coloneqq\left\{
\begin{array}
[c]{c}
\sigma_{S}:\exists\omega_{RS}\in\mathcal{D}(\mathcal{H}_{RS}),\operatorname{Tr}_R[\omega_{RS}]=\sigma_S\\
\omega_{RS}=U_{RS}(g)\omega_{RS}U_{RS}(g)^{\dag}\ \forall g\in G
\end{array}
\right\}  .
\label{eq:g-sym-ext-set}
\end{equation}
We again operate in the model of QIP(2), in which a verifier interacts with a prover via two messages.

We again list the algorithm for completeness, noting its similarity to the previous algorithms. Let $\ket{\psi}_{S^{\prime}S}$ be a purification of the state $\rho_{S}$, and suppose that the circuit $U^{\rho}$ prepares this purification of $\rho_{S}$.

\begin{Algorithm}
\label{alg:sym-ext}The algorithm proceeds as follows:

\begin{enumerate}
\item The verifier uses the circuit $U^{\rho}$ to prepare the state $\ket{\psi}_{S^{\prime}S}$, which is a purification of the state $\rho_{S}$.

\item The verifier transmits the purifying system $S^{\prime}$ to the prover.

\item The prover appends an ancillary register $E$ in the state $\ket{0}_{E}$ and performs a unitary $V_{S^{\prime}E\rightarrow R\hat{R}\hat{S}E^{\prime}}$.

\item The prover sends the systems $R\hat{R}\hat{S}$ back to the verifier.

\item The verifier prepares a register $C$ in the state $\ket{0}_{C}$.

\item The verifier acts on register $C$ with a quantum Fourier transform.

\item The verifier performs the following controlled unitary:
\begin{equation}
\sum_{g\in G}\ket{g}\bra{g}_{C}\otimes U_{RS}(g)\otimes\overline {U}_{\hat{R}\hat{S}}(g),
\end{equation}

\item The verifier performs an inverse quantum Fourier transform on register~$C$, measures in the basis $\{\ket{g}\bra{g}_{C}\}_{g\in G}$, and accepts if and only if the zero outcome $\ket{0}\!\bra{0}_{C}$ occurs.
\end{enumerate}
\end{Algorithm}

Figure~\ref{fig:case-4} depicts this quantum algorithm. After Step~3, the
overall state is
\begin{equation}
V_{S^{\prime}E\rightarrow R\hat{R}\hat{S}E^{\prime}}\ket{\psi}_{S^{\prime} S}\ket{0}_{E}.
\end{equation}
Step~5 prepares the uniform superposition state $\ket{+}_{C}$, which is defined in \eqref{eq:plus-over-group}. After Step~7, the overall state is
\begin{equation}
\frac{1}{\sqrt{\left\vert G\right\vert }}\sum_{g\in G}\ket{g}_{C}\left(U_{RS}(g)\otimes\overline{U}_{\hat{R}\hat{S}}(g)\right)  V\ket{\psi}_{S^{\prime}S}\ket{0}_{E},
\end{equation}
where $V\equiv V_{S^{\prime}E\rightarrow R\hat{R}\hat{S}E^{\prime}}$. The last step can be understood as the verifier projecting the register $C$ onto the state $\ket{+}_{C}$.

The probability of accepting is equal to
\begin{equation}
\left\Vert \Pi_{RS\hat{R}\hat{S}}^{G}V_{S^{\prime}E \rightarrow R\hat{R} \hat{S}E^{\prime}}\ket{\psi}_{S^{\prime}S} \ket{0}_{E}\right\Vert _{2}^{2},
\end{equation}
where $\Pi_{RS\hat{R}\hat{S}}^{G}$ is defined in \eqref{eq:projector-ref-unitaries}. As before, the prover optimizes over every unitary $V_{S^{\prime}E\rightarrow R\hat{R} \hat{S}E^{\prime}}$. The acceptance probability of Algorithm~\ref{alg:sym-ext} is then given by
\begin{equation}
\left\Vert \Pi_{RS\hat{R}\hat{S}}^{G}V_{S^{\prime}E\rightarrow R\hat{R}\hat {S}E^{\prime}}\ket{\psi}_{S^{\prime}S}\ket{0}_{E}\right\Vert _{2}^{2}.
\end{equation}

\begin{figure}[ptb]
\begin{center}
\includegraphics[width=3.5in]{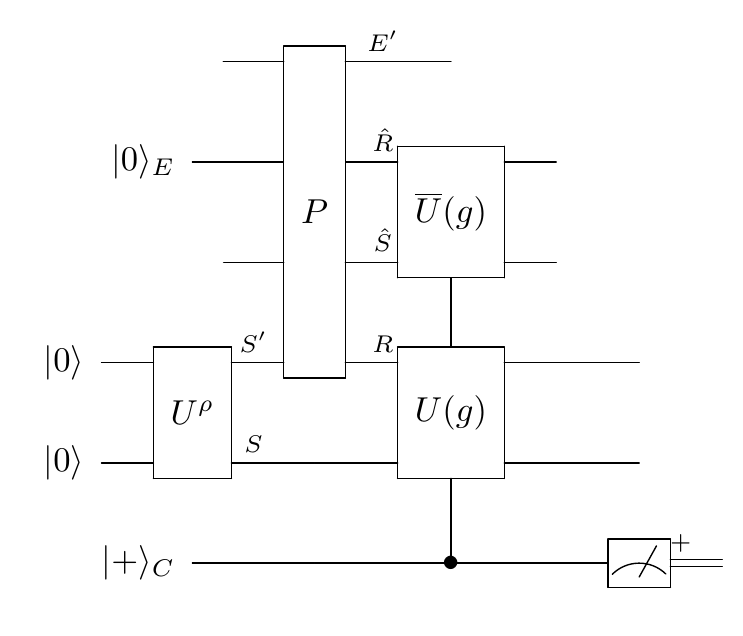}
\end{center}
\caption{Quantum circuit to implement Algorithm~\ref{alg:sym-ext}. The unitary $U^{\rho}$ prepares a purification $\psi_{S^{\prime}S}$ of the state $\rho_{S}$. Algorithm~\ref{alg:sym-ext} tests whether the state $\rho_{S}$ is $G$-symmetric extendible, as defined in Definition~\ref{def:g-sym-ext}. Its acceptance probability is equal to the maximum $G$-symmetric extendible fidelity, as defined in~\eqref{eq:max-g-sym-ext-fid}.}
\label{fig:case-4}
\end{figure}

Our proof of the following theorem is similar to the proof given for Theorem~\ref{thm:max-acc-prob-g-sym}. For completeness, we provide our proof in Appendix~\ref{app:proof-thm-g-se}.

\begin{theorem}
\label{thm:G-SE-acc-prob}The maximum acceptance probability of Algorithm~\ref{alg:sym-ext} is equal to the maximum $G$-symmetric extendible fidelity in \eqref{eq:max-g-sym-ext-fid}, i.e.,
\begin{equation}
\max_{V_{S^{\prime}E\rightarrow R\hat{R}\hat{S}E^{\prime}}}\left\Vert \Pi_{RS\hat{R}\hat{S}}^{G}V_{S^{\prime}E\rightarrow R\hat{R}\hat{S}E^{\prime}}\ket{\psi}_{S^{\prime}S}\ket{0}_{E}\right\Vert _{2}^{2} =\max_{\sigma_{S}\in\operatorname*{SymExt}_{G}}F(\rho_{S},\sigma_{S}),
\end{equation}
where the set $\operatorname*{SymExt}_{G}$ is defined in \eqref{eq:g-sym-ext-set}.
\end{theorem}

Once again, we return to our familiar friend, the triangular dihedral group, to demonstrate the construction of this algorithm. As before, we use the same unitary $U_d$ to prepare the superposition $\ket{+}_C$ and the same mapping of control states to group elements. Then the circuit to test for $G$-Bose symmetric extendibility is shown in Figure~\ref{fig:Dihedral_GBSE_Circuit}. Like it's predecessor, it employs a unitary $V$ to indicate the presense of a prover (or VQA, as the case may be.)

\begin{figure}[h]
\begin{center}
\includegraphics[width=0.5\textwidth]{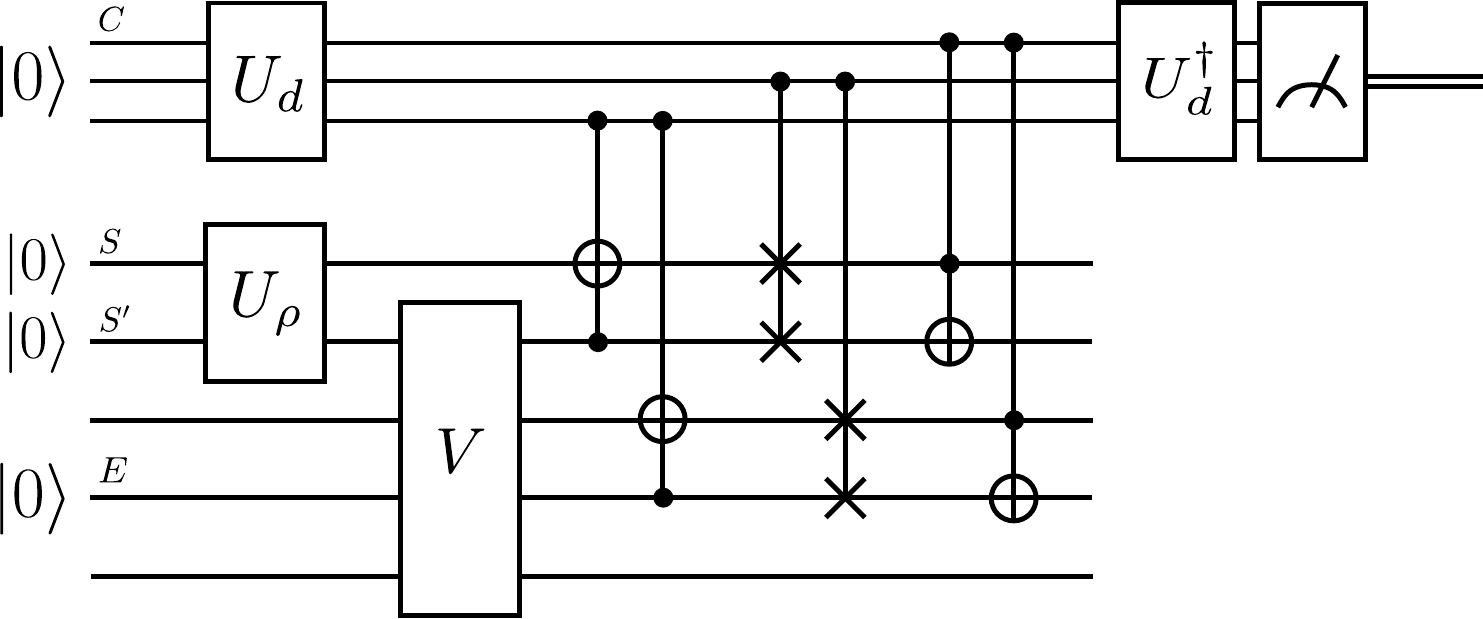}
\end{center}
\caption{Quantum circuit implementing Algorithm~\ref{alg:sym-ext} to test $G$-symmetric extendibility in the case that  the group $G$ is the triangular dihedral group. Compared to Figure~\ref{fig:case-4}, the systems $S$ and $S'$ are one qubit each, $C$ consists of three qubits, and $\ket{+}_C$ is defined as $U_d\ket{000}$. Both the SWAP and CNOT gates have no imaginary entries and thus are equal to their own complex conjugates.}
\label{fig:Dihedral_GBSE_Circuit}
\end{figure}

\begin{table*}[h]
\caption{Summary of the various symmetry tests proposed in Section~\ref{sec:tests-o-sym} and their acceptance probabilities. For more details, see Theorems~\ref{thm:acc-prob-g-Bose-sym}, \ref{thm:max-acc-prob-g-sym}, \ref{thm:G-BSE-acc-prob}, and \ref{thm:G-SE-acc-prob}.}
\begin{center}
\begin{tabular}
[c]{c|c|c}\hline\hline
\multicolumn{1}{c}{Test} & \multicolumn{1}{|c|}{Algorithm} &
\multicolumn{1}{c}{Acceptance Probability}\\\hline\hline
$G$-Bose symmetry & \ref{alg:simple} & $\max_{\sigma\in\text{B-Sym}_{G}}F(\rho,\sigma
)$\\\hline
$G$-symmetry & \ref{alg:g-sym-test} & $\max_{\sigma\in\text{Sym}_{G}}F(\rho,\sigma)$\\\hline
$G$-Bose symmetric extendibility & \ref{alg:G-BSE-test} & $\max_{\sigma\in\text{BSE}_{G}}
F(\rho,\sigma)$\\\hline
$G$-symmetric extendibility & \ref{alg:sym-ext} & $\max_{\sigma\in\text{SymExt}_{G}}
F(\rho,\sigma)$\\\hline\hline
\end{tabular}
\end{center}
\label{tbl:theory-summary}
\end{table*}

\section{Tests of $k$-Extendibility of States and Covariance Symmetry of Channels} 
\label{sec:specialized-tests}

The theory developed in Section~\ref{sec:tests-o-sym} is rather abstract, despite the peppering of $D_3$ as a demonstrable. It is natural to wonder whether these algorithms may find practical use in physics. In the forthcoming subsections, we address this concern by applying our algorithm to test for extendibility of bipartite and multipartite quantum states and to test for covariance symmetry of a quantum channel. The former two cases are used in tests of separability, and the latter we have already employed in Chapter~\ref{ch:symmham}.

\subsection{Separability Test for Pure Bipartite States}\label{sec:separability1}

We illustrate the $G$-Bose symmetry test from Section~\ref{sec:simple-algorithm} on a case of interest: deciding whether a pure bipartite state is entangled. This problem is known to be BQP-complete \cite{GHMW15}, and one can decide it by means of the SWAP test as considered in \cite{HM10}. The SWAP test as a quantum computational method of quantifying entanglement has been further studied in recent work \cite{FKS21,BGCC21}.

First, let us establish that the above notions of symmetry generalize both $k$-extendibility of
bipartite states and $G$-symmetry of unipartite states by introducing both as examples of our algorithm.

\begin{example}
[$k$-extendible]\label{ex:k-ext} Recall that a bipartite state $\rho_{AB}$ is $k$-extendible \cite{W89a,DPS02,DPS04} if there exists an extension state $\omega_{AB_{1}\cdots B_{k}}$ such that
\begin{equation}
\operatorname{Tr}_{B_{2}\cdots B_{k}}[\omega_{AB_{1}\cdots B_{k}}]=\rho_{AB}
\end{equation}
and
\begin{equation}
\omega_{AB_{1}\cdots B_{k}}=W_{B_{1}\cdots B_{k}}(\pi)\omega_{AB_{1}\cdots B_{k}}W_{B_{1}\cdots B_{k}}(\pi)^{\dag},
\end{equation}
for all $\pi\in S_{k}$, where each system $B_{1}$, \ldots, $B_{k}$ is isomorphic to the system $B$ and $W_{B_{1}\cdots B_{k}}(\pi)$ is a unitary representation of the permutation $\pi\in S_{k}$, with $S_{k}$ the symmetric group. Then the established notion of $k$-extendibility is a special case of $G$-symmetric extendibility, in which we set
\begin{align}
S  &  =AB_{1},\label{eq:ident-k-to-g-1}\\
R  &  =B_{2}\cdots B_{k},\\
G  &  =S_{k},\\
U_{RS}(g)  &  =I_{A}\otimes W_{B_{1}\cdots B_{k}}(\pi).
\label{eq:ident-k-to-g-4}
\end{align}
\end{example}

\begin{example}
[$k$-Bose-extendible]\label{ex:k-bose-ext}
A bipartite state $\rho_{AB}$ is $k$-Bose-extendible if there exists an extension state $\omega_{AB_{1}\cdots B_{k}}$ such that 
\begin{equation}
\operatorname{Tr}_{B_{2}\cdots B_{k}}[\omega_{AB_{1}\cdots B_{k}}]=\rho_{AB} \, ,
\end{equation}
and
\begin{equation}
\omega_{AB_{1}\cdots B_{k}}=\Pi_{B_{1}\cdots B_{k}}^{\operatorname{Sym}} \omega_{AB_{1}\cdots B_{k}}\Pi_{B_{1}\cdots B_{k}}^{\operatorname{Sym}},
\end{equation}
where
\begin{equation}
\Pi_{B_{1}\cdots B_{k}}^{\operatorname{Sym}}\coloneqq\frac{1}{k!}\sum_{\pi\in S_{k}}W_{B_{1}\cdots B_{k}}(\pi) 
\label{eq:sym-subspace-proj}
\end{equation}
is the projection onto the symmetric subspace. Thus, $k$-Bose-extendibility is a special case of $G$-Bose-symmetric extendibility under the identifications in \eqref{eq:ident-k-to-g-1}--\eqref{eq:ident-k-to-g-4}.
\end{example}

Let $\psi_{AB}$ be a pure bipartite state, and let $\psi_{AB}^{\otimes k}$ denote $k$ copies of it. Then we can consider the permutation unitaries $W_{B_{1}\cdots B_{k}}(\pi)$ from Example~\ref{ex:k-ext}. This example is a special case of Bose $G$-symmetry with the identifications
\begin{align}
S  &  \leftrightarrow A_{1}B_{1}\cdots A_{k}B_{k},\\
U_{S}(g)  &  \leftrightarrow I_{A_{1}\cdots A_{k}}\otimes W_{B_{1}\cdots B_{k}}(\pi).
\end{align}
The acceptance probability of Algorithm~\ref{alg:simple}\ is equal to
\begin{equation}
\operatorname{Tr}[\Pi_{B_{1}\cdots B_{k}}^{\text{Sym}}\rho_{B}^{\otimes k}],
\end{equation}
where the projection $\Pi_{B_{1}\cdots B_{k}}^{\text{Sym}}$ onto the symmetric subspace is defined in \eqref{eq:sym-subspace-proj} and $\rho_{B} \coloneqq \operatorname{Tr}_{A}[\psi_{AB}]$.\ For $k=2$, this reduces to the well known SWAP\ test with acceptance probability
\begin{equation}
p_{\text{acc}}^{(2)}\coloneqq \frac{1}{2}\left(  1+\operatorname{Tr}[\rho_{B}^{2}]\right)  .
\end{equation}
For $k=3$, the acceptance probability is
\begin{equation}
p_{\text{acc}}^{(3)}\coloneqq \frac{1}{6}\left(  1 + 3 \operatorname{Tr}[\rho_{B}^{2}] + 2 \operatorname{Tr}[\rho_{B}^{3}]\right)  .
\end{equation}
For $k=4$, the acceptance probability is
\begin{equation}
p_{\text{acc}}^{(4)}\coloneqq \frac{1}{24}\left(  1+6\operatorname{Tr}[\rho_{B}^{2}]+3\left(\operatorname{Tr}[\rho_{B}^{2}]\right)^{2}+8\operatorname{Tr}[\rho_{B}^{3}]+6\operatorname{Tr}[\rho_{B}^{4}]\right)  .
\end{equation}
We conclude that
\begin{equation}
p_{\text{acc}}^{(2)}\geq p_{\text{acc}}^{(3)}\geq p_{\text{acc}}^{(4)},
\label{eq:acc-prob-decreasing}
\end{equation}
because $\operatorname{Tr}[\rho^{k}]=\sum_{j}\lambda_{j}^{k}$, where the eigenvalues of $\rho$ are $\{\lambda_{j}\}_{j}$, and for all $x,y\in\left[0,1\right]  $,
\begin{align}
\frac{1}{2}\left(  x+x^{2}\right)   &  \geq\frac{1}{6}\left(  x+3x^{2} + 2x^{3}\right) \\
&  \geq\frac{1}{24}\left(  x+6x^{2}+3x^{2}y+8x^{3}+6x^{4}\right)  .
\end{align}
The inequalities in \eqref{eq:acc-prob-decreasing} imply that the tests become more difficult to pass, or stringent, as $k$ increases. It may be expected that this trend of decreasing acceptance probability continues as $k$ increases. Indeed, in \cite{bradshaw2022cycle}, we showed that this was the case; however, the proof requires additional machinery not yet developed, and so we will readdress this concern in Chapter~\ref{ch:gensep}.

We can interpret these findings in two different ways. For each $k$, the rejection probability $1-p_{\text{acc}}^{(k)}$ can be understood as an entanglement measure for pure bipartite states, similar to how the linear entropy $1-\operatorname{Tr}[\rho_{B}^{2}]$ is interpreted as an entanglement measure of the reduced state $\rho_B$. Indeed, these quantities are non-increasing under local operations and classical communication that take pure states to pure states, as every R\'{e}nyi entropy of the reduced state is an entanglement measure for pure bipartite states \cite{HHHH09}. Another interpretation is that, if using these tests to decide if a given pure state is product or entangled, a decision can be determined with fewer repetitions of the basic test by using tests with higher values of $k$.

\subsection{Separability test for Pure Multipartite States}

We can generalize the test from the previous section to one for pure multipartite entanglement. Let $\psi_{A_{1}\cdots A_{m}}$ be a multipartite pure state, and let $\psi_{A_{1}\cdots A_{m}}^{\otimes k}$ denote $k$ copies of it. For $i\in\left\{  1,\ldots,m\right\}  $ and $\pi_{i}\in S_{k}$, let $W_{A_{i,1}\cdots A_{i,k}}(\pi_{i})$ denote a permutation unitary, where $i$ is an index for the $i$-th party, and the notation $A_{i,j}$ for $j\in\left\{1,\ldots,k\right\}  $ indicates the $j$th system of the $i$th party. This
example is a special case of $G$-Bose symmetry with the identifications:
\begin{align}
S &  \leftrightarrow A_{1,1}\cdots A_{1,k}\cdots A_{m,1}\cdots A_{m,k}, \label{eq:pure-multipartite-idents-1}\\
U_{S}(g) &  \leftrightarrow\bigotimes\limits_{i=1}^{m}W_{A_{i,1}\cdots A_{i,k}}(\pi_{i}),\\
G &  \leftrightarrow\overset{m\text{ times}}{\overbrace{S_{k}\times \cdots\times S_{k}}},\\
g &  \leftrightarrow(\pi_{1},\ldots,\pi_{m}), \label{eq:pure-multipartite-idents-4}
\end{align}
where $\times$ denotes the direct product of groups. The $G$-Bose symmetry test from Section~\ref{sec:simple-algorithm} has the following acceptance probability in this case:
\begin{equation}
\operatorname{Tr}\!\left[  \bigotimes\limits_{i=1}^{m}\Pi_{A_{i,1}\cdots A_{i,k}}^{\operatorname{Sym}}\psi_{A_{1}\cdots A_{m}}^{\otimes k}\right]\,  .
\end{equation}
For $k=2$, this test is known to be a test of multipartite pure-state entanglement \cite{HM10}, which has been considered in more recent works \cite{FKS21,BGCC21}. As far as we aware, the test proposed above, for larger values of $k$, has not been considered previous to our work in \cite{laborde2021testing}. Presumably, as was the case for the bipartite entanglement test mentioned above, the multipartite test is such that it becomes easier to detect an entangled state as $k$ increases. 

\subsection{$k$-Bose Extendibility Test for Bipartite States}
\label{sec:k-bose-ext-test}
We now demonstrate how the test for $G$-Bose symmetric extendibility from Section~\ref{sec:G-Bose-sym-ext-test} can realize a test for $k$-Bose extendibility of a bipartite state. Since every separable state is $k$-Bose extendible, this test is then indirectly a test for separability. To see this in detail, recall that a bipartite state $\sigma_{AB}$ is separable if it can be written as a convex combination of pure product states \cite{HHHH09,KW20book}:
\begin{equation}
\sigma_{AB}=\sum_{x}p_{X}(x)\psi_{A}^{x}\otimes\phi_{B}^{x},
\end{equation}
where $p_{X}$ is a probability distribution and $\{\psi_{A}^{x}\}_{x}$ and $\{\phi_{B}^{x}\}_{x}$ are sets of pure states. A $k$-Bose extension for this state is as follows:
\begin{equation}
\omega_{AB_{1}\cdots B_{k}}=\sum_{x}p_{X}(x)\psi_{A}^{x}\otimes\phi_{B_{1}}^{x}\otimes\cdots \otimes\phi_{B_{k}}^{x}.
\end{equation}
By making the identifications discussed in Example~\ref{ex:k-bose-ext}, it follows from Theorem~\ref{thm:G-BSE-acc-prob}\ that the test from Section~\ref{sec:G-Bose-sym-ext-test} is a test for $k$-Bose extendibility. For an input state $\rho_{AB}$, the acceptance probability of Algorithm~\ref{alg:G-BSE-test} is equal to the maximum $k$-Bose extendible fidelity
\begin{equation}
\max_{\omega_{AB}\in k\text{-BE}}F(\rho_{AB},\omega_{AB}),
\end{equation}
where $k$-BE denotes the set of $k$-Bose extendible states, as defined in Example~\ref{ex:k-bose-ext}.

\subsection{$k$-Extendibility Test for Bipartite States}

In this section, we discuss how the test for $G$-symmetric extendibility from Section~\ref{sec:test-g-sym-ext}\ can realize a test for $k$-extendibility of a bipartite state. Due to the known connections between $k$-extendibility and separability \cite{CKMR08,BCY11,BCY11a,BH12}, this test is an indirect test for separability of a bipartite state, a case we introduce here and expand upon in Chapter~\ref{ch:gensep}. Since every separable state is $k$-Bose extendible, as discussed in Section~\ref{sec:k-bose-ext-test}, and every $k$-Bose extendible state is $k$-extendible, it follows that every separable state is $k$-extendible.

By making the identifications discussed in Example~\ref{ex:k-ext}, it follows from Theorem~\ref{thm:G-SE-acc-prob}\ that the test from Section~\ref{sec:test-g-sym-ext} is a test for $k$-extendibility. For an input state $\rho_{AB}$, the acceptance probability of Algorithm~\ref{alg:sym-ext}\ is equal to the maximum $k$-extendible fidelity
\begin{equation}
\max_{\omega_{AB}\in k\text{-E}}F(\rho_{AB},\omega_{AB}),
\end{equation}
where $k$-E denotes the set of $k$-extendible states, as defined in Example~\ref{ex:k-ext}.

As far as we are aware, this quantum computational test for $k$-extendibility is original to our work in \cite{laborde2021testing}; however, inspired by the approach from \cite{HMW13,Hayden:2014:TQI}. It was argued in \cite{HMW13,Hayden:2014:TQI} that the acceptance probability of the test there is bounded from above by the maximum $k$-extendible fidelity, which is consistent with the fact that the set of $k$-Bose extendible states is contained in the set of $k$-extendible states and our observation here that the acceptance probability of the test in \cite{HMW13,Hayden:2014:TQI} is equal to the maximum $k$-Bose extendible fidelity.

\subsection{Extendibility Tests for Multipartite States}

We now would like to discuss briefly how the tests from Sections~\ref{sec:G-Bose-sym-ext-test} and \ref{sec:test-g-sym-ext} apply to the multipartite case, using identifications similar to those in \eqref{eq:pure-multipartite-idents-1}--\eqref{eq:pure-multipartite-idents-4}.

First, let us recall the definition of multipartite extendibility \cite{DPS05}. Let $\sigma_{A_{1}\cdots A_{m}}$ be a multipartite state. Such a state is $(k_{1},\ldots,k_{m})$-extendible if there exists a state $\omega_{A_{1,1}\cdots A_{1,k_{1}}\cdots A_{m,1}\cdots A_{m,k_{m}}}$ such that
\begin{equation}\label{eq:multipartite-ext-cond}
\sigma_{A_{1}\cdots A_{m}}= \operatorname{Tr}_{A_{1,2} \cdots A_{1,k_{1}}\cdots A_{m,2}\cdots A_{m,k_{m}}}[\omega_{A_{1,1}\cdots A_{1,k_{1}}\cdots A_{m,1}\cdots A_{m,k_{m}}}]
\end{equation}
and
\begin{multline}
\omega_{A_{1,1}\cdots A_{1,k_{1}}\cdots A_{m,1}\cdots A_{m,k_{m}}}\\
=W_{_{A_{1,1}\cdots A_{1,k_{1}}\cdots A_{m,1}\cdots A_{m,k_{m}}}}^{\mathbf{\pi}}\omega_{A_{1,1} \cdots A_{1,k_{1}}\cdots A_{m,1}\cdots A_{m,k_{m}}} (W_{_{A_{1,1}\cdots A_{1,k_{1}}\cdots A_{m,1}\cdots A_{m,k_{m}}}}^{\mathbf{\pi}})^{\dag},
\end{multline}
for all $\mathbf{\pi}$, where $\mathbf{\pi}=(\pi_{1},\ldots,\pi_{m})\in S_{k_{1}}\times\cdots\times S_{k_{m}}$ and
\begin{equation}
W_{_{A_{1,1}\cdots A_{1,k_{1}}\cdots A_{m,1}\cdots A_{m,k_{m}}}}^{\mathbf{\pi }}\coloneqq  \bigotimes\limits_{i=1}^{m}W_{A_{i,1}\cdots A_{i,k_{i}}}^{\pi_{i}}.
\end{equation}
A multipartite state is $(k_{1},\ldots,k_{m})$-Bose extendible if there exists a state $\omega_{A_{1,1}\cdots A_{1,k_{1}}\cdots A_{m,1}\cdots A_{m,k_{m}}}$ such that \eqref{eq:multipartite-ext-cond} holds and
\begin{multline}
\omega_{A_{1,1}\cdots A_{1,k_{1}}\cdots A_{m,1}\cdots A_{m,k_{m}}}=\\
\Pi_{A_{1,1}\cdots A_{1,k_{1}}\cdots A_{m,1}\cdots A_{m,k_{m}}}\omega_{A_{1,1}\cdots A_{1,k_{1}}\cdots A_{m,1}\cdots A_{m,k_{m}}} \Pi_{A_{1,1}\cdots A_{1,k_{1}}\cdots A_{m,1}\cdots A_{m,k_{m}}},
\end{multline}
where
\begin{align}
\Pi_{A_{1,1}\cdots A_{1,k_{1}}\cdots A_{m,1}\cdots A_{m,k_{m}}}  & \coloneqq \bigotimes\limits_{i=1}^{m}\Pi_{A_{i,1}\cdots A_{i,k_{i}}}^{\operatorname{Sym}},\\
\Pi_{A_{i,1}\cdots A_{i,k_{i}}}^{\operatorname{Sym}}  & \coloneqq \frac{1}{k_{i}!} \sum_{\pi_{i}\in S_{k_{i}}}W_{A_{i,1}\cdots A_{i,k_{i}}}^{\pi_{i}}.
\end{align}

By making the identifications
\begin{align}
S &  \leftrightarrow A_{1,1}\cdots A_{m,1},\\
R &  \leftrightarrow A_{1,2}\cdots A_{1,k_{1}}\cdots A_{m,2}\cdots A_{m,k_{m}}\\
U_{S}(g) &  \leftrightarrow\bigotimes\limits_{i=1}^{m}W_{A_{i,1}\cdots A_{i,k_{i}}}(\pi_{i})\\
G &  \leftrightarrow S_{k_{1}}\times\cdots\times S_{k_{m}},\\
g &  \leftrightarrow(\pi_{1},\ldots,\pi_{m}),
\end{align}
it follows that Algorithm~\ref{alg:G-BSE-test} is a test for multipartite $(k_{1},\ldots,k_{m})$-Bose extendibility of a state $\rho_{A_{1}\cdots A_{m}}$, with acceptance probability equal to
\begin{equation}
\max_{\omega_{A_{1}\cdots A_{m}}\in(k_{1},\ldots,k_{m})\text{-BE}} F(\rho_{A_{1}\cdots A_{m}},\omega_{A_{1}\cdots A_{m}}),
\end{equation}
and Algorithm~\ref{alg:sym-ext} is a test for multipartite $(k_{1},\ldots,k_{m})$-extendibility of a state $\rho_{A_{1}\cdots A_{m}}$, with acceptance probability equal to
\begin{equation}
\max_{\omega_{A_{1}\cdots A_{m}}\in(k_{1},\ldots,k_{m})\text{-E}}F(\rho_{A_{1}\cdots A_{m}},\omega_{A_{1}\cdots A_{m}}),
\end{equation}
where $(k_{1},\ldots,k_{m})$-BE and $(k_{1},\ldots,k_{m})$-E denote the sets of $(k_{1},\ldots,k_{m})$-Bose extendible and $(k_{1},\ldots,k_{m})$-extendible states, respectively.

\subsection{Testing Covariance Symmetry of a Quantum Channel}

\label{sec:cov-sym-ch-test}

We can also use the test from Algorithm~\ref{alg:g-sym-test} to test for
covariance symmetry of a quantum channel. We have already given a definition for covariance symmetry in Chapter~\ref{ch:intro}, Section~\ref{sec:channelnotions}, and we have further already demonstrated the usefulness of channel symmetry in Chapter~\ref{ch:symmham} in Section~\ref{sec:hamcov}. Therefore, we will skip its redefinition here and instead focus on generalizing the approach given in Chapter~\ref{ch:symmham}. 

We know from previous sections that  a channel is covariant in the sense above if and only if
its Choi state is invariant in the following sense \cite[Eq.~(59)]{CDP09}:
\begin{equation} 
\Phi_{RB}^{\mathcal{N}}=(\overline{\mathcal{U}}_{R}(g)\otimes\mathcal{V}_{B}(g))(\Phi_{RB}^{\mathcal{N}})\quad\forall g\in G,
\end{equation}
as in \eqref{eq:choicondition}. (All relevant terms in this definition are given therein.) We will once again employ this definition to give a test for the covariance symmetry of a quantum channel.

Suppose now that a circuit is available that generates the channel $\mathcal{N}_{A\rightarrow B}$. Similar to the first assumption in Section~\ref{sec:tests-o-sym}, we suppose that the circuit realizes a unitary channel $\mathcal{W}_{AE^{\prime}\rightarrow BE}$\ that extends the original channel, in the sense that
\begin{equation}
\mathcal{N}_{A\rightarrow B}(\omega_{A})=(\operatorname{Tr}_{E}\circ
\mathcal{W}_{AE^{\prime}\rightarrow BE})(\omega_{A}\otimes\ket{0}\!\bra{0}_{E^{\prime}}).
\end{equation}
Then to decide whether the channel is covariant, we send in one share of a maximally-entangled state to the unitary extension channel, such that the overall state is
\begin{equation}
\mathcal{W}_{AE^{\prime}\rightarrow BE}(\Phi_{RA}\otimes\ket{0}\!\bra{0}_{E^{\prime}}).
\end{equation}
Now making the identifications
\begin{align}
E  & \leftrightarrow S^{\prime},\\
RB  & \leftrightarrow S,\\
\overline{U}_{R}(g)\otimes V_{B}(g)  & \leftrightarrow U_{S}(g),
\end{align}
we apply Algorithm~\ref{alg:g-sym-test}, and as a consequence of Theorem~\ref{thm:max-acc-prob-g-sym}, the acceptance probability is equal to
\begin{equation}
\max_{\sigma_{RB}\in\operatorname{Sym}_G}F(\Phi_{RB}^{\mathcal{N}},\sigma_{RB}),
\end{equation}
where
\begin{equation}
\operatorname{Sym}_G\coloneqq \left\{
\begin{array}
[c]{c}
\sigma_{RB}\in\mathcal{D}(\mathcal{H}_{RB}):\\
\sigma_{RB}=(\overline{\mathcal{U}}_{R}(g)\otimes\mathcal{V}_{B}(g))(\sigma_{RB})\ \forall g\in G
\end{array}
\right\}  .
\end{equation}
Thus, the test accepts with probability equal to one if and only if the channel is covariant in the sense of \eqref{eq:ch-cov-def}. Indeed, this is very similar to the construction used in Chapter~\ref{ch:symmham}, but builds on the structure of symmetry we have been working with throughout this chapter.

\section{Resource Theories}
\label{sec:res-theories}

In this section, we prove that the various maximum symmetric fidelities proposed in Section~\ref{sec:tests-o-sym} are proper resource-theoretic monotones, in the sense reviewed in~\cite{CG18}. By fulfilling this requirement, we can assert that these fidelities are indeed measures of symmetry as claimed.

To begin with, let us recall the basics of a resource theory (see \cite[Definition~1]{CG18}). Let $\mathcal{F}$ be a mapping that assigns a unique set of quantum channels to any arbitrary input and output systems $A$ and $B$, respectively. We require that $\mathcal{F}$ include the identity channel ($\mathcal{F}(A\rightarrow A) = \mathbb{I}_A$) and  that, for any three physical systems $A$, $B$, and $C$, any two maps $\mathcal{N}_{A\rightarrow B} \in \mathcal{F}(A \rightarrow B)$ and $\mathcal{M}_{B\rightarrow C}\in\mathcal{F}(B \rightarrow C)$ have the transitive property
\begin{equation}\label{eq:trans-prop}
 \mathcal{M}_{B\rightarrow C}\circ\mathcal{N}_{A\rightarrow B} \in\mathcal{F}(A\rightarrow C).
\end{equation}
If $\mathcal{F}$ obeys above criteria, then the mapping $\mathcal{F}$ defines the resource theory. Each instance of $\mathcal{F}$ defines a map between two mathematical spaces, and the spaces themselves define the characteristics of the resulting map. The set $\mathcal{F}(\mathbb{C}\rightarrow A)$ defines the set of free states---that is, channels from the trivial space ($\mathbb{C}$) to system $A$ should be quantum states. The set $\mathcal{F}(A \rightarrow B)$ defines the set of free channels from system $A$ to system $B$. 

How can this be thought of in a more intuitive way? A resource theory defines a set of channels and rules by which they behave. Free states are, in some sense, the ``allowed" states in a resource theory; they are the states that can be generated by the channels in the set $\mathcal{F}$. A state that cannot be generated this way is considered a ``resource". Conversely, free channels are channels that cannot generate resources that were not already present. In this sense, the cost of creating one of the allowed states is ``free", hence the name.

An intuitive picture, which may be helpful to keep in mind, is to think of the resource section of a library. Here, the valuable resource is the knowledge contained within the books. An example of a free channel would be to permute the order of books on the shelf---say we change from the Dewey decimal system to an alphabetic system. We have not gained any resources by this operation, and so it is free. In fact, we could even pile all of the texts on the floor, and, while certainly inconvenient, the knowledge we have access to will not have changed. However, if instead we were to utilize the inter-library loan system to gain new books, clearly this is not a free operation as the amount of resources have increased. In the following section, we will discuss how asymmetry ties in with this idea of quantifying resources.

\subsection{Resource Theory of Asymmetry}
\label{sec:rt-asym}
The resource theory of asymmetry is well established by now \cite{MS13}, but to the best of our knowledge, the resource theory of Bose asymmetry had not been defined prior to \cite{laborde2021testing}. We will begin by recalling the resource theory of asymmetry. Afterwards, we establish the resource theory of Bose asymmetry as well as two other generalizations involving unextendibility, which are in turn generalizations of the resource theory of unextendibility proposed in \cite{KDWW19}.

Let $G$ be a group, and let $\{U_{A}(g)\}_{g\in G}$ and $\{V_{B}(g)\}_{g\in G}$ denote projective unitary representations of $G$. A channel $\mathcal{N}_{A\rightarrow B}$ is a free channel in the resource theory of asymmetry if the following $G$-covariance symmetry condition holds
\begin{equation}
\mathcal{N}_{A\rightarrow B} \circ \mathcal{U}_{A}(g)= \mathcal{V}_{B}(g) \circ \mathcal{N}_{A\rightarrow B} \qquad\forall g\in G,
\end{equation}
where the unitary channels $\mathcal{U}_{A}(g)$ and $\mathcal{V}_{B}(g)$ are respectively defined from $U_{A}(g)$ and $V_{B}(g)$ as in \eqref{eq:uandv}. It then follows that a state $\sigma_{A}$ is free in this resource theory if it is $G$-symmetric such that
\begin{equation}
\sigma_{A}=\mathcal{U}_{A}(g)(\sigma_{A})\qquad\forall g\in G,
\end{equation}
with a similar definition for the $B$ system; furthermore, the free channels take free states to free states \cite{MS13}, in the sense that $\mathcal{N}_{A\rightarrow B}(\sigma_{A})$ is a free state if $\mathcal{N}_{A\rightarrow B}$ is a free channel and $\sigma_{A}$ is a free state.

The maximum $G$-symmetric fidelity is a resource monotone in the following sense:
\begin{equation}
\max_{\sigma_{A}\in\operatorname*{Sym}_{G}}F(\rho_{A},\sigma_{A}) \leq \max_{\sigma_{B}\in\operatorname*{Sym}_{G}}F(\mathcal{N}_{A\rightarrow B} (\rho_{A}),\sigma_{B}).
\end{equation}
This follows from the facts that the fidelity does not decrease under the action of a quantum channel and that free channels take free states to free states.

\subsection{Resource Theory of Bose Asymmetry}

Now we define the resource theory of Bose asymmetry and prove that the acceptance probability $\operatorname{Tr}[\Pi_{A}^{G}\rho_{A}]$ of Algorithm~\ref{alg:simple} is a resource monotone in this resource theory. This demonstrates that $\operatorname{Tr}[\Pi_{A}^{G}\rho_{A}]$ is a legitimate quantifier of Bose symmetry of a state.

Following the same notation as in Section~\ref{sec:rt-asym}, recall that a state $\sigma_{A}$ is Bose symmetric if the following condition holds
\begin{equation}
\sigma_{A}=\Pi_{A}^{G}\sigma_{A}\Pi_{A}^{G},
\end{equation}
where $\Pi_{A}^{G}$ is given by \eqref{eq:unitaries-and-projs}. Similarly, a state $\tau_{B}$ is Bose symmetric if it obeys the same conditions but for the projector $\Pi_{B}^{G}$ specified by $\{V_B(g)\}_{g\in G}$. These are the free states in the resource theory of Bose asymmetry.

To define the resource theory, we need to specify the free channels.
\begin{definition}
[Bose symmetric channel]We define a channel $\mathcal{N}_{A\rightarrow B}$ to
be a Bose symmetric channel (i.e., free channel) if the following condition
holds
\begin{equation}
\left(  \mathcal{N}_{A\rightarrow B}\right)  ^{\dag}(\Pi_{B}^{G}) \geq \Pi_{A}^{G}, 
\label{eq:Bose-sym-channel}
\end{equation}
where $\left( \mathcal{N}_{A\rightarrow B}\right)  ^{\dag}$ is the Hilbert--Schmidt adjoint of $\mathcal{N}_{A\rightarrow B}$
\cite{wildebook,KW20book}.
\end{definition}

\begin{proposition}
Bose symmetric channels include the identity channel and they obey the transitive property in \eqref{eq:trans-prop}. Additionally, Bose symmetric states are a special case of Bose symmetric channels when the input space is trivial.
\end{proposition}
\begin{proof}
When the input and output systems are the same, as well as the unitary representations, it follows that $\Pi_{B}^{G}=\Pi_{A}^{G}$. Since the identity channel is its own adjoint, we then conclude that \eqref{eq:Bose-sym-channel} holds for the identity channel.

Suppose that $\mathcal{N}_{A\rightarrow B}$ is a quantum channel that obeys the condition in
\eqref{eq:Bose-sym-channel}. Let $\{W_{C}(g)\}_{g\in G}$ be a projective unitary representation of $G$, and suppose that $\mathcal{M}_{B\rightarrow C}$ is a Bose symmetric channel satisfying
\begin{equation}
\left(  \mathcal{M}_{B\rightarrow C}\right)  ^{\dag}(\Pi_{C}^{G})\geq\Pi_{B}^{G},
\end{equation}
where $\Pi_{C}^{G}\coloneqq \frac{1}{\left\vert G\right\vert }\sum_{g\in G}W_{C}(g)$. Consider that
\begin{align}
\left(  \mathcal{M}_{B\rightarrow C}\circ\mathcal{N}_{A\rightarrow B}\right)  ^{\dag}(\Pi_{C}^{G})\nonumber &  =(\mathcal{N}_{A\rightarrow B})^{\dag}\left[  \left(  \mathcal{M}_{B\rightarrow C}\right)  ^{\dag}(\Pi_{C}^{G})\right] \\
&  \geq(\mathcal{N}_{A\rightarrow B})^{\dag}\left[  \Pi_{B}^{G}\right] \\
&  \geq\Pi_{A}^{G}.
\end{align}
The first equality follows by exploiting the identity $\left(  \mathcal{M}_{B\rightarrow C} \circ \mathcal{N}_{A\rightarrow B}\right)  ^{\dag}=(\mathcal{N}_{A\rightarrow B})^{\dag} \circ \left(  \mathcal{M}_{B\rightarrow C}\right)  ^{\dag}$ for adjoints. The first inequality follows from the assumption that $\mathcal{M}_{B\rightarrow C}$ is a Bose symmetric channel and from the fact that $\mathcal{N}_{A\rightarrow B}$ is completely positive, so that $(\mathcal{N}_{A\rightarrow B})^{\dag}$ is also. We thus conclude that $\mathcal{M}_{B\rightarrow C} \circ \mathcal{N}_{A\rightarrow B}$ is a Bose symmetric channel, so that the transitive property in \eqref{eq:trans-prop} holds.

Finally, suppose that the input system $A$ of a Bose symmetric channel $\mathcal{N}_{A\rightarrow B}$ is trivial. Then each group element $g$ is trivially represented by the number one. It follows that $\Pi_{A}^{G}=1$. Then the channel $\mathcal{N}_{A\rightarrow B}$ is really just a state $\omega_{B}$ \cite{wildebook} with some spectral decomposition $\omega_{B}=\sum_{x}p(x)|x\rangle\!\langle x|_{B}$; furthermore, the associated Kraus operators are given by $\{\sqrt{p(x)}|x\rangle_{B}\}_{x}$. Then the condition
\begin{equation}
\left(  \mathcal{N}_{A\rightarrow B}\right)  ^{\dag}(\Pi_{B}^{G})\geq\Pi_{A}^{G}
\end{equation}
reduces to
\begin{equation}
\sum_{x}p(x)\langle x|_{B}\Pi_{B}^{G}|x\rangle_{B}\geq 1,
\end{equation}
which is the same as
\begin{equation}
\operatorname{Tr}[\Pi_{B}^{G}\omega_{B}]\geq1.
\end{equation}
Since $\omega_{B}$ is a state and $\Pi_{B}^{G}$ is a projection, it follows that $\operatorname{Tr}[\Pi_{B}^{G}\omega_{B}]\leq1$. Combining these inequalities, we conclude that $\operatorname{Tr}[\Pi_{B}^{G}\omega_{B}]=1$. Finally, we apply \eqref{eq:Bose-symmetric-equiv-cond} to conclude that $\omega_{B}$ is a Bose symmetric state.
\end{proof}

\begin{theorem}
Suppose that a quantum channel $\mathcal{N}_{A\rightarrow B}$ obeys the condition in \eqref{eq:Bose-sym-channel}. Let $\sigma_{A}$ be a Bose symmetric state. Then $\mathcal{N}_{A\rightarrow B}(\sigma_{A})$ is a Bose symmetric state.
\end{theorem}

\begin{proof}
Recall from \eqref{eq:Bose-symmetric-equiv-cond}\ that a state $\sigma_{A}$ is Bose symmetric if and only if $\operatorname{Tr}[\Pi_{A}^{G}\sigma_{A}]=1$. Then consider that
\begin{align}
1  &  \geq\operatorname{Tr}[\Pi_{B}^{G}\mathcal{N}_{A\rightarrow B}(\sigma_{A})]\\
&  =\operatorname{Tr}[\left(  \mathcal{N}_{A\rightarrow B}\right)  ^{\dag} (\Pi_{B}^{G})\sigma_{A}]\\
&  \geq\operatorname{Tr}[\Pi_{A}^{G}\sigma_{A}]\\
&  =1.
\end{align}
It follows that $\operatorname{Tr}[\Pi_{B}^{G}\mathcal{N}_{A\rightarrow B}(\sigma_{A})]=1$, and, by applying \eqref{eq:Bose-symmetric-equiv-cond} again, that $\mathcal{N}_{A\rightarrow B}(\sigma_{A})$ is Bose symmetric.
\end{proof}

By essentially the same proof, it follows that the measure $\operatorname{Tr}[\Pi_{A}^{G}\rho_{A}]$ from \eqref{eq:acc-prob-bose-test}\ is non-decreasing under the action of a Bose symmetric channel $\mathcal{N}_{A\rightarrow B}$. Thus, the acceptance probability $\operatorname{Tr}[\Pi_{A}^{G}\rho_{A}]$ of a Bose symmetry test is a resource monotone in the resource theory of Bose asymmetry.

\begin{theorem}
\label{thm:res-mono-G-Bose-sym}Let $\rho_{A}$ be a state, and let $\mathcal{N}_{A\rightarrow B}$ be a Bose symmetric channel. Then $\operatorname{Tr}[\Pi_{A}^{G}\rho_{A}]$ is a resource monotone in the following sense:
\begin{equation}
\operatorname{Tr}[\Pi_{B}^{G}\mathcal{N}_{A\rightarrow B} (\rho_{A})] \geq \operatorname{Tr}[\Pi_{A}^{G} \rho_{A}].
\end{equation}

\end{theorem}

\begin{proof}
Consider that
\begin{align}
\operatorname{Tr}[\Pi_{B}^{G}\mathcal{N}_{A\rightarrow B}(\rho_{A})]  &
=\operatorname{Tr}[\left(  \mathcal{N}_{A\rightarrow B}\right)  ^{\dag} (\Pi_{B}^{G})\rho_{A}]\\
&  \geq\operatorname{Tr}[\Pi_{A}^{G}\rho_{A}],
\end{align}
which follows from \eqref{eq:Bose-sym-channel}.
\end{proof}

Throughout this section, we have adopted the perspective that Bose symmetric channels are defined by the condition in \eqref{eq:Bose-sym-channel}. It then follows as a consequence that $\operatorname{Tr}[\Pi_{A}^{G}\rho_{A}]$ is a resource monotone. We can adopt a different perspective and conclude consistency between them. Let us instead suppose that $\operatorname{Tr}[\Pi_{A}^{G}\rho_{A}]$ is non-decreasing under the action of a free channel $\mathcal{N}_{A\rightarrow B}$. That is, suppose that the following inequality holds for every state~$\rho_{A}$:
\begin{equation}
\operatorname{Tr}[\Pi_{B}^{G}\mathcal{N}_{A \rightarrow B}(\rho_{A})] \geq \operatorname{Tr}[\Pi_{A}^{G}\rho_{A}].
\end{equation}
Then by rewriting this inequality as
\begin{equation}
\operatorname{Tr}[(\left(  \mathcal{N}_{A\rightarrow B}\right)  ^{\dag} (\Pi_{B}^{G})-\Pi_{A}^{G})\rho_{A}]\geq0\quad\forall\rho_{A} \in\mathcal{D} (\mathcal{H}_{A}),
\end{equation}
we conclude that $\left(  \mathcal{N}_{A\rightarrow B}\right)^{\dag}(\Pi_{B}^{G})-\Pi_{A}^{G}$ is a positive semi-definite operator, which is equivalent to the condition in \eqref{eq:Bose-sym-channel}. Thus, $\mathcal{N}_{A\rightarrow B}$ is a Bose symmetric channel if and only if $\operatorname{Tr}[\Pi_{A}^{G}\rho_{A}]$ is a resource monotone.

\subsection{Resource Theory of Asymmetric Unextendibility}
\label{sec:rt-asym-unext}

We now give a resource theory that generalizes that proposed in \cite{KDWW19}, just as the set of $G$-symmetric extendible states generalizes the set of $k$-extendible states (recall Example~\ref{ex:k-ext}). One of the main ideas is to use the notion of channel extension introduced in \cite{KDWW19}; additionally, this resource theory allows us to conclude that the acceptance probability of Algorithm~\ref{alg:sym-ext} (i.e., the maximum $G$-symmetric extendible fidelity) is a resource monotone and thus well motivated in this sense.

For the following definitions, let $G$ be a group, and let $\{U_{RS}(g)\}_{g\in G}$ and $\{V_{R^{\prime}S^{\prime}}(g)\}_{g\in G}$ be projective unitary representations of $G$ acting on $\mathcal{H}_{R}\otimes\mathcal{H}_{S}$ and $\mathcal{H}_{R^{\prime}}\otimes \mathcal{H}_{S^{\prime}}$ respectively.

\begin{definition}[$G$-symmetric extendible channel]\label{def:GSE-channels}
A channel $\mathcal{N}_{S\rightarrow S^{\prime}}$ is $G$-symmetric extendible if there exists a bipartite channel $\mathcal{M}_{RS\rightarrow R^{\prime}S^{\prime}}$ such that

\begin{enumerate}
\item $\mathcal{M}_{RS\rightarrow R^{\prime}S^{\prime}}$ is a channel extension of $\mathcal{N}_{S\rightarrow S^{\prime}}$:
\begin{equation}
\operatorname{Tr}_{R^{\prime}}\circ\mathcal{M}_{RS \rightarrow R^{\prime}S^{\prime}}=\mathcal{N}_{S\rightarrow S^{\prime}}\circ\operatorname{Tr}_{R},
\label{eq:channel-ext}
\end{equation}
\item $\mathcal{M}_{RS\rightarrow R^{\prime}S^{\prime}}$ is covariant with respect to $\{U_{RS}(g)\}_{g\in G}$ and $\{V_{R^{\prime}S^{\prime}}(g)\}_{g\in G}$:
\begin{equation}
\mathcal{M}_{RS\rightarrow R^{\prime}S^{\prime}}\circ\mathcal{U}_{RS}(g) = \mathcal{V}_{R^{\prime}S^{\prime}}(g) \circ \mathcal{M}_{RS\rightarrow R^{\prime}S^{\prime}}\quad\forall g\in G, 
\label{eq:covariance-G-sym-ext}
\end{equation}
where $\mathcal{U}_{RS}(g)(\cdot)$ and $\mathcal{V}_{R^{\prime}S^{\prime}}(g)(\cdot)$ are defined similarly to \eqref{eq:uandv}.
\end{enumerate}
\end{definition}

The condition in \eqref{eq:channel-ext} implies that the extension channel $\mathcal{M}_{RS\rightarrow R^{\prime}S^{\prime}}$ is non-signaling from $R$ to $S^{\prime}$ \cite{BGNP01,ESW02,PHHH06}, in the sense that
\begin{equation}
\operatorname{Tr}_{R^{\prime}} \circ \mathcal{M}_{RS\rightarrow R^{\prime} S^{\prime}} = \operatorname{Tr}_{R^{\prime}} \circ \mathcal{M}_{RS\rightarrow R^{\prime}S^{\prime}} \circ \mathcal{R}_{R}^{\pi}, 
\label{eq:non-sig-1}
\end{equation}
where $\mathcal{R}_{R}^{\pi}(\cdot)\coloneqq \operatorname{Tr}[\cdot]\pi_{R}$ is a replacer channel that traces out its input and replaces with the maximally-mixed state $\pi_{R}$. This follows because
\begin{align}
\operatorname{Tr}_{R^{\prime}}\circ\mathcal{M}_{RS\rightarrow R^{\prime} S^{\prime}} \circ \mathcal{R}_{R}^{\pi}  &  = \mathcal{N}_{S\rightarrow S^{\prime}}\circ\operatorname{Tr}_{R}\circ\mathcal{R}_{R}^{\pi}\\
&  =\mathcal{N}_{S\rightarrow S^{\prime}}\circ\operatorname{Tr}_{R}\\
&  =\operatorname{Tr}_{R^{\prime}}\circ\mathcal{M}_{RS\rightarrow R^{\prime}S^{\prime}}, 
\label{eq:non-sig-4}
\end{align}
where we have exploited the identity in \eqref{eq:channel-ext} in the first and last lines, and in the second line used the fact that $\operatorname{Tr}_{R}\circ\mathcal{R}_{R}^{\pi}=\operatorname{Tr}_{R}$.

Definition~\ref{def:GSE-channels}\ leads to a consistent resource theory of $G$-asymmetric unextendibility, in the sense that the free states are $G$-symmetric extendible states and the output of a $G$-symmetric extendible channel acting on a $G$-symmetric extendible state is a $G$-symmetric extendible state.

\begin{proposition}
\label{prop:g-sym-ext-ch-triv-input-state}A $G$-symmetric extendible channel $\mathcal{N}_{S\rightarrow S^{\prime}}$\ with trivial input system is a $G$-symmetric extendible state.
\end{proposition}

\begin{proof}
If the input system $S$ of $\mathcal{N}_{S\rightarrow S^{\prime}}$ is trivial, then it follows that $\mathcal{N}_{S\rightarrow S^{\prime}}$ is a state (call it $\rho_{S^{\prime}}$); furthermore, we can choose the input system $R$ of the extension channel $\mathcal{M}_{RS\rightarrow R^{\prime}S^{\prime}}$ to be trivial, in which case $\mathcal{M}_{RS\rightarrow R^{\prime}S^{\prime}}$ is a state (call it $\omega_{R^{\prime}S^{\prime}}$) that extends $\rho_{S^{\prime}}$. The condition in \eqref{eq:covariance-G-sym-ext}\ then collapses to $\omega_{R^{\prime}S^{\prime}}=\mathcal{V}_{R^{\prime}S^{\prime}} (g)(\omega_{R^{\prime}S^{\prime}})$ for all $g\in G$. It follows by Definition~\ref{def:g-bose-sym-ext} that $\rho_{S^{\prime}}$ is a $G$-symmetric extendible state.
\end{proof}

\begin{proposition}
\label{prop:golden-rule-G-SE}Let $\mathcal{N}_{S\rightarrow S^{\prime}}$ be a $G$-symmetric extendible channel, and let $\rho_{S}$ be a $G$-symmetric extendible state. Then $\mathcal{N}_{S\rightarrow S^{\prime}}(\rho_{S})$ is a $G$-symmetric extendible state.
\end{proposition}

\begin{proof}
Since $\rho_{S}$ is a $G$-symmetric extendible state, by Definition~\ref{def:g-sym-ext}, there exists an extension state $\omega_{RS}$ satisfying the conditions stated there. Since $\mathcal{N}_{S\rightarrow S^{\prime}}$ is a $G$-symmetric extendible channel, by Definition~\ref{def:GSE-channels}, there exists an extension channel  $\mathcal{M}_{RS\rightarrow R^{\prime}S^{\prime}}$ satisfying the conditions stated there. It follows that $\mathcal{M}_{RS\rightarrow R^{\prime}S^{\prime}}(\omega_{RS})$ is an extension of $\mathcal{N}_{S\rightarrow S^{\prime}}(\rho_{S})$ as
\begin{align}
\operatorname{Tr}_{R^{\prime}}[\mathcal{M}_{RS\rightarrow R^{\prime}S^{\prime}} (\omega_{RS})]  &  =\mathcal{N}_{S\rightarrow S^{\prime}}(\operatorname{Tr}_{R}[\omega_{RS}])\\
&  =\mathcal{N}_{S\rightarrow S^{\prime}}(\rho_{S}),
\end{align}
where the first equality follows from \eqref{eq:channel-ext}. Also, consider that the following holds for all $g\in G$:
\begin{align}
(\mathcal{V}_{R^{\prime}S^{\prime}}(g)\circ\mathcal{M}_{RS\rightarrow R^{\prime}S^{\prime}})(\omega_{RS})\nonumber &  =(\mathcal{M}_{RS\rightarrow R^{\prime}S^{\prime}}\circ\mathcal{U}_{RS}(g))(\omega_{RS})\\
&  =\mathcal{M}_{RS\rightarrow R^{\prime}S^{\prime}}(\omega_{RS}),
\end{align}
where the first equality follows from \eqref{eq:covariance-G-sym-ext} and the second from \eqref{eq:G-ext-2}.
\end{proof}

As a consequence of Proposition~\ref{prop:golden-rule-G-SE} and the data-processing inequality for fidelity, the maximum $G$-symmetric extendible fidelity is a resource monotone.

\begin{corollary}
Let $\rho_{S}$ be a state, and let $\mathcal{N}_{S\rightarrow S^{\prime}}$ be a $G$-symmetric extendible channel. Then the maximum $G$-symmetric extendible fidelity is a resource monotone,
\begin{equation}
\max_{\sigma_{S}\in\operatorname*{SymExt}_{G}}F(\rho_{S},\sigma_{S}) \leq \max_{\sigma_{S^{\prime}} \in \operatorname*{SymExt}_{G}}F(\mathcal{N}_{S\rightarrow S^{\prime}}(\rho_{S}),\sigma_{S^{\prime}}).
\end{equation}
\end{corollary}

\begin{example}[$k$-unextendibility] The resource theory of $k$-unextendibility, proposed in \cite{KDWW19}, is a special case of the resource theory of $G$-asymmetric unextendibility. To see this, recall that a bipartite channel $\mathcal{N}_{AB\rightarrow A^{\prime}B^{\prime}}$ is $k$-extendible if there exists an extension channel $\mathcal{M}_{AB_{1}\cdots B_{k}\rightarrow A^{\prime} B_{1}^{\prime} \cdots B_{k}^{\prime}}$ satisfying
\begin{equation}
\operatorname{Tr}_{B_{2}^{\prime}\cdots B_{k}^{\prime}}\circ\mathcal{M}_{AB_{1}\cdots B_{k}\rightarrow A^{\prime}B_{1}^{\prime}\cdots B_{k}^{\prime}} =\mathcal{N}_{AB\rightarrow A^{\prime}B^{\prime}}\circ\operatorname{Tr}_{B_{2}\cdots B_{k}}
\end{equation}
and
\begin{equation}
\mathcal{W}_{B_{1}^{\prime}\cdots B_{k}^{\prime}}^{\pi}\circ\mathcal{M}_{AB_{1}\cdots B_{k}\rightarrow A^{\prime}B_{1}^{\prime}\cdots B_{k}^{\prime}} =\mathcal{M}_{AB_{1}\cdots B_{k}\rightarrow A^{\prime}B_{1}^{\prime}\cdots B_{k}^{\prime}}\circ\mathcal{W}_{B_{1}\cdots B_{k}}^{\pi},
\end{equation}
for all $\pi \in S_k$, where $\mathcal{W}_{B_{1}\cdots B_{k}}^{\pi}$ and $\mathcal{W}_{B_{1}^{\prime}\cdots B_{k}^{\prime}}^{\pi}$ are unitary permutation channels. Thus, by setting
\begin{align}
S  &  =AB,\\
R  &  =B_{2}\cdots B_{k},\\
S^{\prime}  &  =A^{\prime}B^{\prime},\\
R^{\prime}  &  =B_{2}^{\prime}\cdots B_{k}^{\prime},\\
U_{RS}(g)  &  =I_{A}\otimes W_{B_{1}\cdots B_{k}}(\pi),\\
V_{R^{\prime}S^{\prime}}(g)  &  =I_{A^{\prime}}\otimes W_{B_{1}^{\prime}\cdots
B_{k}^{\prime}}(\pi),
\end{align}
we see that a $k$-extendible channel is a special case of a $G$-symmetric
extendible channel.
\end{example}

\subsection{Resource Theory of Bose Asymmetric Unextendibility}

We finally consider the resource theory of Bose asymmetric unextendibility, with the goal being similar to that of the previous sections; we want to justify the acceptance probability of Algorithm~\ref{alg:G-BSE-test} (i.e., the maximum $G$-BSE fidelity) as a resource monotone. At the same time, we establish a novel resource theory that could have further applications in quantum information.

Once again, let $G$, $\{U_{RS}(g)\}_{g\in G}$, and $\{V_{R^{\prime}S^{\prime}}(g)\}_{g\in G}$ be defined the same way as in Section~\ref{sec:rt-asym-unext}.

\begin{definition}[$G$-BSE channel]\label{def:G-BSE-channels}
A channel $\mathcal{N}_{S\rightarrow S^{\prime}}$ is $G$-Bose symmetric extendible ($G$-BSE) if there exists a bipartite channel $\mathcal{M}_{RS\rightarrow R^{\prime}S^{\prime}}$ such that
\begin{enumerate}
\item $\mathcal{M}_{RS\rightarrow R^{\prime}S^{\prime}}$ is a channel extension of $\mathcal{N}_{S\rightarrow S^{\prime}}$:
\begin{equation}
\operatorname{Tr}_{R^{\prime}}\circ\mathcal{M}_{RS\rightarrow R^{\prime}S^{\prime}}=\mathcal{N}_{S\rightarrow S^{\prime}}\circ\operatorname{Tr}_{R},
\label{eq:channel-ext-bose}
\end{equation}

\item $\mathcal{M}_{RS\rightarrow R^{\prime}S^{\prime}}$ is Bose symmetric:
\begin{equation}
(\mathcal{M}_{RS\rightarrow R^{\prime}S^{\prime}})^{\dag}(\Pi_{R^{\prime}S^{\prime}}^{G})\geq\Pi_{RS}^{G}, \label{eq:BSE-bose-condition}
\end{equation}
where $\Pi_{RS}^{G}$ and $\Pi_{R^{\prime}S^{\prime}}^{G}$ are defined as in \eqref{eq:Pi_RS-proj-again} as sums over $U_{RS}(g)$ and $V_{R^{\prime}S^{\prime}}(g)$ respectively.
\end{enumerate}
\end{definition}

As discussed in \eqref{eq:non-sig-1}--\eqref{eq:non-sig-4}, the condition in \eqref{eq:channel-ext-bose} can be understood as imposing a no-signaling constraint, from $R$ to $S^{\prime}$. Now, with the same line of reasoning given in the proof of Proposition~\ref{prop:g-sym-ext-ch-triv-input-state}, we conclude the following:
\begin{proposition}
A $G$-BSE channel $\mathcal{N}_{S\rightarrow S^{\prime}}$\ with trivial input system is a $G$-BSE state.
\end{proposition}

The following proposition demonstrates that the resource theory delineated by Definition~\ref{def:G-BSE-channels} is indeed a consistent resource theory.
\begin{proposition}
\label{prop:golden-rule-G-BSE}Let $\mathcal{N}_{S\rightarrow S^{\prime}}$ be a $G$-BSE channel, and let $\rho_{S}$ be a $G$-BSE state. Then $\mathcal{N}_{S\rightarrow S^{\prime}}(\rho_{S})$ is a $G$-BSE state.
\end{proposition}

As this proof is similar to that of Proposition~\ref{prop:golden-rule-G-SE}, we include it in Appendix~\ref{app:proof-prop-G-BSE}. As a consequence of Proposition~\ref{prop:golden-rule-G-BSE}\ and the data-processing inequality for fidelity, it follows that the maximum $G$-BSE fidelity is a resource monotone.

\begin{corollary}
Let $\rho_{S}$ be a state, and let $\mathcal{N}_{S\rightarrow S^{\prime}}$ be a $G$-BSE channel. Then the maximum $G$-BSE\ fidelity is a resource monotone in the following sense:
\begin{equation}
\max_{\sigma_{S}\in\operatorname*{BSE}_{G}}F(\rho_{S},\sigma_{S}) \leq \max_{\sigma_{S^{\prime}}\in\operatorname*{BSE}_{G}}F(\mathcal{N}_{S\rightarrow S^{\prime}}(\rho_{S}),\sigma_{S^{\prime}}).
\end{equation}

\end{corollary}

To define the resource theory of $k$-Bose unextendibility, we establish the notion of a free channel (i.e., a $k$-Bose extendible bipartite channel) and discuss it in the following example. 
\begin{example}[$k$-Bose unextendibility]
We say that a bipartite channel $\mathcal{N}_{AB\rightarrow A^{\prime}B^{\prime}}$ is $k$-Bose-extendible if there exists an extension channel $\mathcal{M}_{AB_{1}\cdots B_{k}\rightarrow A^{\prime}B_{1}^{\prime} \cdots B_{k}^{\prime}}$ satisfying
\begin{equation}
\operatorname{Tr}_{B_{2}^{\prime}\cdots B_{k}^{\prime}}\circ\mathcal{M}_{AB_{1} \cdots B_{k}\rightarrow A^{\prime}B_{1}^{\prime}\cdots B_{k}^{\prime}}=\mathcal{N}_{AB\rightarrow A^{\prime}B^{\prime}}\circ\operatorname{Tr}_{B_{2}\cdots B_{k}}
\end{equation}
and
\begin{equation}
(\mathcal{M}_{AB_{1}\cdots B_{k}\rightarrow A^{\prime}B_{1}^{\prime}\cdots B_{k}^{\prime}})^{\dag}(\Pi_{B_{1}^{\prime}\cdots B_{k}^{\prime}} ^{\operatorname{Sym}})\geq\Pi_{B_{1}\cdots B_{k}}^{\operatorname{Sym}},
\end{equation}
where $\Pi_{B_{1}^{\prime}\cdots B_{k}^{\prime}}^{\operatorname{Sym}}$ and $\Pi_{B_{1}\cdots B_{k}}^{\operatorname{Sym}}$ are projections onto symmetric subspaces,
\begin{align}
\Pi_{B_{1}\cdots B_{k}}^{\operatorname{Sym}}  &  \coloneqq \frac{1}{k!}\sum_{\pi\in S_{k}}W_{B_{1}\cdots B_{k}}^{\pi},\\
\Pi_{B_{1}^{\prime}\cdots B_{k}^{\prime}}^{\operatorname{Sym}}  &  \coloneqq \frac{1}{k!}\sum_{\pi\in S_{k}}W_{B_{1}^{\prime}\cdots B_{k}^{\prime}}^{\pi},
\end{align}
and $W_{B_{1}\cdots B_{k}}^{\pi}$ and $W_{B_{1}^{\prime}\cdots B_{k}^{\prime}}^{\pi}$ are unitary representations of the permutation $\pi\in S_{k}$. Thus, by setting
\begin{align}
S  &  =AB,\\
R  &  =B_{2}\cdots B_{k},\\
S^{\prime}  &  =A^{\prime}B^{\prime},\\
R^{\prime}  &  =B_{2}^{\prime}\cdots B_{k}^{\prime},\\
U_{RS}(g)  &  =I_{A}\otimes W_{B_{1}\cdots B_{k}}(\pi),\\
V_{R^{\prime}S^{\prime}}(g)  &  =I_{A^{\prime}}\otimes W_{B_{1}^{\prime}\cdots
B_{k}^{\prime}}(\pi),
\end{align}
we see that a $k$-Bose-extendible channel is a special case of a $G$-Bose symmetric extendible channel.
\end{example}

\section{Conclusion}

\label{sec:sym-conclusion}

In summary, we have proposed various quantum computational tests of symmetry, as well as various notions of symmetry like $G$-symmetric extendibility and $G$-Bose symmetric extendibility, which include previous notions of symmetry from \cite{MS13,MS14,W89a,DPS02,DPS04} as a special case. These tests have acceptance probabilities equal to various maximum symmetric fidelities, thus endowing these measures with operational meaning. We have also established resource theories of asymmetry beyond that proposed in \cite{MS13}, which put the maximum symmetric fidelities on firm ground in a resource-theoretic sense. 

Going forward from here, we will now discuss an expansion of a particular part of this chapter. Namely, we will address unanswered questions in the derivation of acceptance probabilities of bipartite separability tests. Futhermore, we will go on to show that an entire class of separability tests can be generated and compared beyond the full symmetric test prior literature has prepared for us.

\pagebreak
\singlespacing
\chapter{Generalized Separability Tests for Bipartite Pure States}\label{ch:gensep}
\doublespacing
\section{Introduction}\label{introduction}

In the previous chapter, we described a specific kind of symmetry test that acts as a separability test for pure bipartite states. Separability and entanglement of quantum states are a topic of high interest throughout quantum information, as might well be expected. (See, for example, quantifiers and measures in \cite{Per96,HHH96,ZHSL98,EP99,Vidal2002,W89a,DPS02,KDWW19,WWW19,KDWW21}.) In this chapter, we will give a recipe for developing separability tests in the vein of $k$-extendibility \cite{W89a,DPS02}.

The most common quantum computational test of separability of pure states is the swap test, introduced in \cite{barenco1997stabilization} and used in quantum fingerprinting \cite{buhrman2001quantum}.
To understand it, first recall that a pure bipartite state $\ket{\psi}_{AB}$ is separable if it can be written as a tensor product of two states, as
\begin{align}
    \ket{\psi}_{AB} =\ket{\phi}_A\otimes\ket{\varphi}_B.
    \label{eq:sep-pure-state}
\end{align}
Now, if we take two copies of this separable state, it has the following form:
\begin{align}
    \ket{\psi}_{A_1 B_1} \otimes \ket{\psi}_{A_2 B_2} =\ket{\phi}_{A_1} \otimes \ket{\varphi}_{B_1} \otimes \ket{\phi}_{A_2} \otimes \ket{\varphi}_{B_2}.
\end{align}
This state is invariant under a swap of systems $A_1$ and $A_2$, as well as a swap of systems $B_1$ and $B_2$. Thus, the swap test accepts with certainty in this case; however, if a pure bipartite state is not separable, the two-copy state does not possess the above swap invariance, and the swap test can detect this lack of invariance by means of the well-known phase kickback trick

In Chapter~\ref{ch:symmstates} and \cite{laborde2021testing}, we proposed a generalization of the swap test as a method for detecting entanglement, based on the observation that multiple copies of the separable state in \eqref{eq:sep-pure-state} are invariant under arbitrary permutations of both the $A$ systems and $B$ systems. Indeed, by writing such a  state down explicitly  as
\begin{equation}
    \bigotimes_{i=1}^{k} \ket{\psi}_{A_i B_i} = \bigotimes_{i=1}^{k} \ket{\phi}_{A_i} \otimes \ket{\varphi}_{B_i},
\end{equation}
it is clear that such a state is invariant as mentioned above; however, if the state $\ket{\psi}_{A B}$ is not separable, then checking for various kinds of permutation invariance of the state $\bigotimes_{i=1}^{k} \ket{\psi}_{A_i B_i}$ leads to more fine-grained tests of entanglement with alternative mathematical expressions for the acceptance probability of the test. This chapter endeavors to understand these expressions in more detail.

For our approach here, we will make liberal use of the framework of $G$-Bose symmetry tests, as discussed in Section~\ref{sec:simple-algorithm}. This method facilitates our discussions of the separability of a pure bipartite state. Specifically, in the framework we consider, we take a pure state $\ket{\psi}_{AB}$ and conduct an $S_k$-Bose symmetry test on the tensor-power state $\ket{\psi}_{AB}^{\otimes k}$, where $S_k$ denotes the symmetric group on $k$ letters. This tensor-power state realizes the case where we have access to $k$ copies of our state under test. The swap test is recovered as a special case in which $k=2$.

The $G$-Bose symmetry tests allow for a generalization of the swap test to more copies of a state of interest and higher-order groups. These algorithms exchange simplicity for certainty, analogous to how fingerprinting is both more accurate and complicated when greater numbers of prints are taken. In choosing to investigate group symmetries, rather than merely the swap test, the separability of a state can be determined more quickly and accurately.

The natural question is, when do these more complex tests merit performing? In \cite{bradshaw2022cycle}, the paper this chapter is based upon, we derive the acceptance probability of a generalized separability test, and the end result is included here. In doing so, we present an inherent reliance on the cycle index polynomial, a particularly important polynomial in Pólya theory \cite{polya,combinatorics} that encodes the structure of a permutation group by storing the number of elements of a given cycle type as its coefficients. This allows us to compare separability tests generated from various groups, as well as investigate the mathematical relationships inherently present in these tests. We directly show that an arbitrary finite group generates a separability test with its acceptance probability given by the cycle index polynomial of that group. We supplement this by then giving explicit quantum circuit descriptions for groups of interest and counting the number of gates needed to realize each test. Combining our acceptance probability results with resource counting gives us a metric to compare when the relative strictness of the test is outweighed by the benefit of fewer gate resources, and we discuss this factor in more detail in Section~\ref{sec:comparison}.

In Section~\ref{formula}, we revisit the algorithm for the  bipartite pure-state separability test in Section~\ref{sec:separability1}. We show that the acceptance probability of this algorithm is given by the cycle index polynomial \cite{polya,combinatorics} of the symmetric group $S_k$, which is itself related to the complete Bell polynomials \cite{Roman}. 
In Section~\ref{sec:conjecture}, we prove our conjecture from Section~\ref{sec:separability1}, that the acceptance probability of such algorithms does not increase as $k\to\infty$. In fact, we show that it strictly decreases and converges to zero whenever $\rho_B \coloneqq \tr_A[\ket{\psi}\bra{\psi}_{AB}]$ is not a pure state.

In Section~\ref{generalization}, we generalize the  bipartite pure-state separability test to an algorithm involving any group $G$, in which a $G$-Bose symmetry test is performed on the tensor-power state $\ket{\psi}_{AB}^{\otimes k}$. By identifying $G$ with a subgroup of $S_k$, which is guaranteed to exist by Cayley's theorem, we show, by the same reasoning as in Section~\ref{formula}, that the acceptance probability of the algorithm is given by the cycle index polynomial of the group $G$. We discuss how these generalized tests are in fact separability tests for pure, bipartite states, and they have an interesting connection to combinatorics via the cycle index polynomial. 

In Section~\ref{sec:comparison}, we analyze the resources needed to implement these tests on quantum computers; in doing so, we show that simpler groups can give comparable performance for fewer resources. We also give constructions for tests with respect to the cyclic group $C_k$ and the more typical $S_k$ group, and we compare the resouce costs of these test with respect to their rejection probability.

Finally, we conclude in Section~\ref{sec:conclusion} with a summary.

\section{Bipartite Pure-State Separability Test}\label{formula}

Let us begin by reviewing the construction of the bipartite pure-state separability test in Chapter~\ref{ch:symmstates}, which can be viewed as a $G$-Bose symmetry test. For convenience, we now recall the definition of a $G$-Bose symmetric state.
\begin{definition}
Let $G$ be a group with a unitary representation $U_S:G\to U(\mathcal{H})$, where $U(\mathcal{H})$ denotes the set of all unitaries that act on a Hilbert space $\mathcal{H}$. Then a state $\rho_S$ is called $G$-Bose symmetric if
\begin{align}
    \label{PiG}
    \Pi^G_S\rho_S\Pi^G_S=\rho_S,
\end{align}
where $\Pi^G_S\coloneqq\frac{1}{|G|}\sum_{g\in G}U_S(g)$ is the group representation projection.
\end{definition}

We discussed in Section~\ref{sec:separability1} how this framework can be used to test for the separability of a bipartite pure state $\psi_{AB}$. To do so, we suppose that $k$ copies of the state $\psi_{AB}$ are available, which we write as $\psi_{AB}^{\otimes k}$. We also identify the $A$ systems by $A_1 \cdots A_k$ and the $B$ systems by $B_1 \cdots B_k$. We then perform an $S_k$-Bose symmetry test on the state $\psi_{AB}^{\otimes k}$ by identifying $S$ with $A_1B_1\cdots A_kB_k$ and $U_S(\pi)$ with $I_{A_1\cdots A_k}\otimes W_{B_1\cdots B_k}(\pi)$, where $\pi\in S_k$ and $W_{B_1\cdots B_k}:S_k\to U(\mathcal{H}_{B_1\cdots B_k})$ is the standard unitary representation of $S_k$ that acts on $\mathcal{H}_{B_1\cdots B_k} \equiv \mathcal{H}_{B_1} \otimes \cdots \otimes \mathcal{H}_{B_k}$ by permuting the Hilbert spaces according to the corresponding permutation. Define $\rho_B\coloneqq \tr_A[\psi_{AB}]$. By applying \eqref{PiG}, the acceptance probability for the bipartite pure-state separability algorithm is given by
\begin{align}
p^{(k)} \coloneqq \tr[\Pi_{B_1\cdots B_k}\rho_B^{\otimes k}] \, ,
\label{eq:def-acc-prob}
\end{align}
where
\begin{align}
\Pi_{B_1\cdots B_k}\coloneqq \frac{1}{k!}\sum_{\pi\in S_k}W_{B_1\cdots B_k}(\pi)\, .
\end{align}

\begin{figure}[ptb]
\begin{center}
\includegraphics[width=4in]{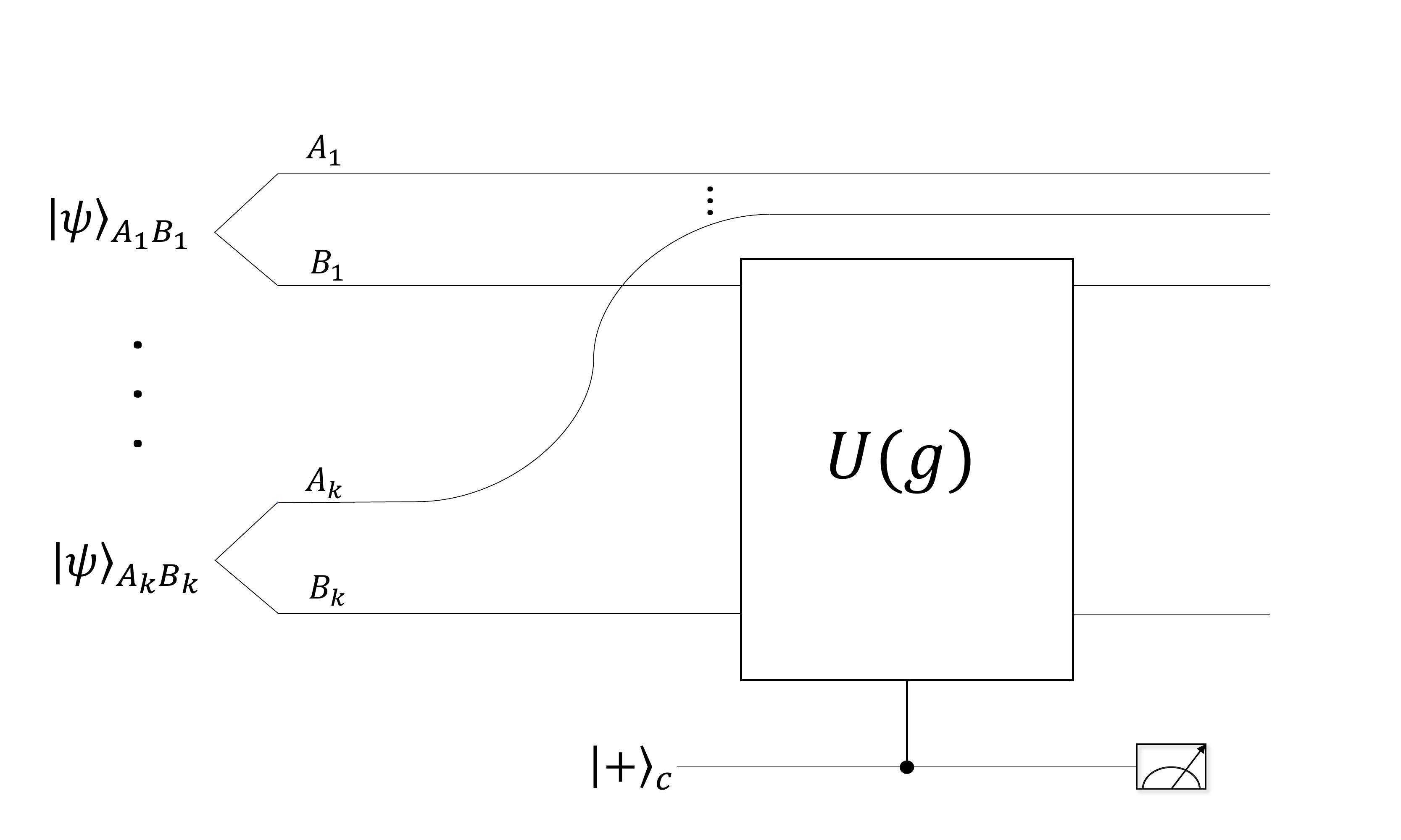}
\end{center}
\caption{Quantum circuit to implement a $G$-Bose symmetry test. We take $k$ copies of an initial bipartite state $\ket{\psi}_{AB}$ and consider the reduced state $\rho_B = \operatorname{Tr}_A [\ket{\psi}\!\bra{\psi}_{AB}]$. The collection of these reduced states are then subjected to the separability test determined by the group, where $\ket{+}_C $ is defined in \eqref{eq:plus-over-group}, and $U(g)$ is an element of the group representation.}
\label{fig:circuit}
\end{figure}

Figure~\ref{fig:circuit} reviews the $G$-Bose symmetry test. The circuit begins with $k$ copies of an initial bipartite state $\ket{\psi}_{AB}$. The $B_i$ subsystems are collected and subject to a controlled unitary gate whose mathematical description involves each unitary $U(g)$. The control register is initialized to the state $\ket{+}_C$, as defined in \eqref{eq:plus-over-group}. Under the separability test circuit, $G = S_k$, and the unitary representation is a permutation of the $B_i$ subsystems.

Our first main result is a formula for the acceptance probability $p^{(k)}$ in \eqref{eq:def-acc-prob} as a sum over the partitions of $k$ of a product of traces of $\rho_B$ and its powers and certain scaling factors. 

\begin{theorem}
\label{formulathm}
Let $\psi_{AB}$ denote a pure bipartite state and define $\rho_B\coloneqq \tr_A[\psi_{AB}]$. Then the acceptance probability $p^{(k)}$ for the bipartite pure-state separability test is given by
\begin{align}
p^{(k)}=\sum_{a_1+2a_2+\cdots +ka_k=k}\prod_{j=1}^k\frac{(\tr[\rho_B^j])^{a_j}}{j^{a_j}a_j!}\, ,
\label{eq:main-thm-cycle-index-woohoo}
\end{align}
where the sum is taken over the partitions of $k$.
\end{theorem}

\begin{proof}
Let $\pi\coloneqq (1\ 2\ \cdots\ k)$ and consider the representation $W_{B_1\cdots B_k}(\pi)$. It was shown in \cite{Ekert} that $\tr[W_{B_1\cdots B_k}(\pi)\rho_B^{\otimes k}]=\tr[\rho_B^k]$, but we include a proof here for completeness. Expanding $\rho$ in the standard basis as $\rho = \sum_{i,j} p_{i,j} |i\rangle\!\langle j|$, we have
\begin{align}
& \tr[W_{B_1\cdots B_k}(\pi)\rho_B^{\otimes k}] \notag \\
&= \tr\bigg[W_{B_1\cdots B_k}(\pi) \sum_{\substack{i_1,\ldots,i_k \\ j_1,\ldots,j_k}}p_{i_1j_1}\cdots p_{i_kj_k}|i_1\rangle\!\langle j_1|\otimes|i_2\rangle\!\langle j_2|\otimes\cdots\otimes|i_{k}\rangle\!\langle j_k|\bigg]\\
&= \tr\bigg[\sum_{\substack{i_1,\ldots,i_k \\ j_1,\ldots,j_k}}p_{i_1j_1}\cdots p_{i_kj_k}|i_k\rangle\!\langle j_1|\otimes|i_1\rangle\!\langle j_2|\otimes\cdots\otimes|i_{k-1}\rangle\!\langle j_k|\bigg]\\
&=\sum_{\substack{i_1,\ldots,i_k \\ j_1,\ldots,j_k \\ t_1,\ldots,t_k}}p_{i_1j_1}\cdots p_{i_kj_k}\delta_{t_1i_k}\delta_{j_1t_1}\cdots\delta_{t_ki_{k-1}}\delta_{j_kt_k}\\
&=\sum_{t_1,\ldots,t_k}p_{t_2t_1}p_{t_3t_2}\cdots p_{t_kt_{k-1}}p_{t_1t_k}.
\end{align}
Meanwhile, 
\begin{align} \tr[\rho^k]&=\tr\bigg[\sum_{i_1,\ldots,i_k,j_1,\ldots,j_k}p_{i_1j_1}\cdots p_{i_kj_k}|i_1\rangle\!\langle j_1|i_2\rangle\!\langle j_2|\cdots|i_k\rangle\!\langle j_k|\bigg]\\
&=\sum_{i_1,i_2,\ldots,i_k}p_{i_1i_2}p_{i_2i_3}\cdots p_{i_{k-1}i_k}p_{i_ki_1}.
\end{align}
Thus, by relabeling the indices, we see that
\begin{align} \tr[W_{B_1\cdots B_k}(\pi)\rho_B^{\otimes k}]=\tr[\rho^k].
\end{align}
Similarly, we can show for every $m$-cycle $\pi_m\in S_k$ that
\begin{align}
\tr[W_{B_1\cdots B_k}(\pi_m)\rho_B^{\otimes k}]=\tr[\rho^m].
\end{align}
Now suppose $\pi_m$ and $\pi_n$ are disjoint $m$- and $n$-cycles, respectively. Then they act on different Hilbert spaces and so the trace of the product of their representations acting on $\rho_B^{\otimes k}$ splits into the product of traces. That is,
\begin{align} \tr[W_{B_1\cdots B_k}(\pi_m)W_{B_1\cdots B_k}(\pi_n)\rho^{\otimes k}]=\tr[\rho^m]\tr[\rho^n].
\end{align}
Now, since every $m$-cycle yields a factor of $\tr[\rho^m]$ and products of disjoint cycles split the trace, we have
\begin{align} p^{(k)}&=\tr[\Pi_{B_1\cdots B_k}\rho_B^{\otimes k}]\\
&=\tr\bigg[\frac{1}{k!}\sum_{\pi\in S_k} W_{B_1\cdots B_k}(\pi)\rho_B^{\otimes k}\bigg]\\
&=\frac{1}{k!}\sum_{\pi\in S_k}\tr[W_{B_1\cdots B_k}(\pi)\rho_B^{\otimes k}]\\
&=\frac{1}{k!}\sum_{a_1+2a_2+\cdots+ka_k=k}c(a_1,\ldots,a_k)\prod_{j=1}^k\tr[\rho_B^j]^{a_j}
\end{align}
where $c(a_1,\ldots,a_k)$ is the number of cycles in $S_k$ with cycle type $(a_1,\ldots,a_k)$, which is known to be $\frac{k!}{\prod_{j=1}^kj^{a_j}a_j!}$ (see \cite[Eq.~(13.3)]{van2001}). Thus, the equality in \eqref{eq:main-thm-cycle-index-woohoo} follows.
\end{proof}

We assert now that the formula is identical to that of the cycle index polynomial of the symmetric group $S_k$, with each variable $x_j$ taking the value $\tr[\rho^j]$. The cycle index polynomial of a permutation group $G$ is defined by
\begin{align} \label{cycleindex}
    Z(G)(x_1,\ldots,x_n)\coloneqq\frac{1}{|G|}\sum_{g\in G}x_1^{c_1(g)}\cdots x_n^{c_n(g)}\, ,
\end{align}
where $c_j(g)$ denotes the number of cycles of length $j$ in the disjoint cycle decomposition of $g$. Setting $x_j=\tr[\rho_B^j]$, we see that the acceptance probability of the separability test is given by the cycle index polynomial of the symmetric group $S_k$ (see \cite[Chapter 37, pg.~526]{van2001}), so that it satisfies the recurrence relation
\begin{align} \label{recurrence}
    p^{(k)}&=\frac{1}{k}\sum_{j=1}^k\tr[\rho^j]p^{(k-j)}\, .
\end{align} Furthermore, the cycle index polynomial of the symmetric group $S_k$ is equivalent to
\begin{align}\label{bellpolyAP} 
\frac{1}{k!}B_k(x_1,x_2,2!\, x_3,3!\, x_4,\ldots,(k-1)!\, x_k)\, ,
\end{align}
where $B_k(x_1,\ldots,x_k)$ is the complete Bell polynomial \cite{comtet1974}. From this perspective, the acceptance probability can be interpreted as the $k^{\text{th}}$ raw moment of a probability distribution with the first $k$ cumulants given by $1$, $\tr[\rho^2]$, \ldots, $\tr[\rho^k]$. See \cite[Chapter~1, Section~2]{macdonald1995} for more information on \eqref{recurrence}.

Thus we have addressed one of the two open questions purposed in Section~\ref{sec:separability1}. Now we move on to the second.

\subsection{Strictly Decreasing Acceptance Probability}\label{sec:conjecture}

In Chapter~\ref{ch:symmstates}, we made reference to a conjecture from \cite{laborde2021testing} that the acceptance probability of the bipartite pure-state separability test is monotone non-increasing in $k$. We answer this conjecture in the affirmative as a corollary of the following lemma about complete Bell polynomials. In fact, this inequality is strict, and the acceptance probability approaches zero in the limit $k\to\infty$ whenever $\rho_B$ is not a pure state.

The full proof of this is given in \cite{bradshaw2022cycle} and is beyond the scope of this thesis, but we list the results here. By first proving a relevant relationship between Bell polynomials, we are able to show the following theorem.

\begin{theorem}\label{th:decreasing} 
The acceptance probability $p^{(k)}$ is strictly decreasing and  $\lim_{k\to\infty}p^{(k)}=0$ when $\rho_B$ is not a pure state.
\end{theorem}

These results indicate that as $k$ goes to infinity, fewer repetitions of the test are needed to determine whether a given pure state is entangled. There is a trade-off, however, between increasing $k$ and the computational resources needed to conduct a single test. As $k$ increases, one might suspect that the resources needed will increase in such a way that a large enough $k$ is not feasible. Indeed, as one of our results, we discuss the scaling in this claim in Section~\ref{sec:comparison}.

\section{Generalization of the Algorithm}
\label{generalization}

The previous sections address results and conjectures of the archetypal $k$-Bose extendibility tests, but we now present a generalization of the bipartite pure-state separability algorithm to groups other than the symmetric group $S_k$. Furthermore, these algorithms are also separability tests. As always, let $G$ be a finite group, and let $\psi_{AB}$ be a pure state. Recall that Cayley's theorem \cite{Dummit_Foote} guarantees that every finite group is isomorphic to a subgroup of a permutation group. Thus, there exists a representation of $G$ such that every $g\in G$ is mapped to an element $\pi \in S_k$ for some $k \in \mathbb{N}$. In turn, there exists a map from $\pi$ to the operator that permutes the Hilbert spaces in the composite Hilbert space $\mathcal{H}^{\otimes k}$. Then a generalization of the bipartite pure-state separability algorithm is given by performing a $G$-Bose symmetry test on the state $\psi_{AB}^{\otimes k}$.

By the argument in the proof of Theorem~\ref{formulathm}, we see that one simply has to count the number of cycles of any given cycle type in the permutation subgroup isomorphic to $G$ to obtain a formula for the acceptance probability of the algorithm. That is, the argument in Theorem~\ref{formulathm}, combined with Cayley's theorem, proves the following theorem:

\begin{theorem}
\label{generalformula}
Let $p_G$ denote the acceptance probability with respect to the group $G$ for the generalization of the bipartite pure-state separability algorithm. Then
\begin{align} \label{pG} 
p_G=Z(G)(1,\ldots,\tr[\rho^k])\, .
\end{align}
\sloppy That is, the acceptance probability $p_G$ is given by the cycle index polynomial \eqref{cycleindex} of~$G$ evaluated at $x_j=\tr[\rho^j]$ for $j \in \{1,\ldots,k\}$.
\end{theorem}

As an aside, we note that \eqref{pG} has an interesting combinatorial meaning. Let $\{\lambda_i\}_{i=1}^r$ denote the eigenvalues of $\rho$. By Pólya's enumeration theorem \cite{brualdi2010,tucker1995}, we can interpret \eqref{pG} as a generating function for the number of nonequivalent colorings of a set $S$ with the $r$ colors $\{\lambda_i\}_{i=1}^r$. The role of $G$ here is to define the equivalence between colorings through its action on $S$.

We would now like to give two relevant examples to demonstrate the above result. Specifically, we show the already discussed case of $S_k$ but also the cyclic group $C_k$, as these two examples will be of interest to us in later sections.

\begin{example}
We consider the example already discussed in Theorem~\ref{formulathm}. Let $G=S_k$ be the symmetric group, which is already a permutation group. Then the acceptance probability is given by the cycle index polynomial of $S_k$. That is,
\begin{align}
p_{\text{sym}}^{(k)}=\sum_{a_1+2a_2+\cdots +ka_k=k} \prod_{j=1}^k \frac{(\tr[\rho_B^j])^{a_j}}{j^{a_j}a_j!}.
\end{align}
\qed
\end{example}

\begin{example} 
In this example, we generalize the cyclic test to products of cyclic groups. Let $G=\mathbb{Z}_m^k$ be the product of $k$ copies of the group $\mathbb{Z}_m$. We represent $G$ as a permutation subgroup by labeling its elements and letting them act on the group to construct a permutation. For example, if $k=1$, then $G=\{0,1,\ldots,m-1\}$. Since 0 has no effect on any element of the group, we map it to the identity element $e$. Meanwhile, 1 acts on each element of the group by sending 0 to 1, 1 to 2, and so on. So we identify 1 with the cycle $(1\ \cdots\ m)$. The remaining permutations are defined similarly. 

Returning to the more general setting, we see that the elements of each order $n$ correspond to products of $n$-cycles. Now, for an element of $G$ to have order $n$, each component must contain an element of an order that divides $n$, with at least one component filled by an element of order $n$. So the number of elements of order $n$ is given by
\begin{align} 
\sum_{i=1}^k\binom{k}{i}(\phi(n))^i \bigg(\sum_{\substack{l|n\\ l<n}}\phi(l)\bigg)^{k-i} & = \bigg(\phi(n) + \sum_{\substack{l|n\\ l<n}}\phi(l)\bigg)^k-\bigg(\sum_{\substack{l|n\\ l<n}}\phi(l)\bigg)^k\\
& =n^k-(n-\phi(n))^k,
\end{align} 
where $\phi$ denotes the Euler $\phi$-function. The acceptance probability given by the cycle index polynomial of $\mathbb{Z}_m^k$ is then
\begin{align} 
p_{\mathbb{Z}_m^k}^{(k)}=\frac{1}{m^k}\sum_{n|m}(n^k-(n-\phi(n))^k)(\tr[\rho_B^n])^{\frac{m^k}{n}}.
\end{align}
\qed
\end{example}

Finally, we claim that the above nontrivial examples, as well as any other example involving a nontrivial finite group, are tests for separability of a pure bipartite state. Thus, we have produced an entire class of separability tests.

\begin{theorem}
Let $\psi_{AB}$ denote a pure bipartite state. Then the generalized bipartite pure-state separability algorithm is, in fact, a faithful test for separability of $\psi_{AB}$ for any nontrivial finite group $G$, meaning that the acceptance probability is equal to one if and only if the pure state is a separable state.
\end{theorem}

Once again, the proof of this theorem is given in \cite{bradshaw2022cycle}, but we include it here for completeness.

\begin{proof}
Suppose $\psi_{AB}$ is separable. That is, $\ket{\psi}_{AB} = \ket{\phi}_A \otimes \ket{\varphi}_B$ for some states $\ket{\phi}_A \in \mathcal{H}_A$ and $\ket{\varphi}_B\in\mathcal{H}_B$. Then
\begin{align} 
\rho_B &\coloneqq \tr_A[\psi_{AB}]\\
&=\tr_A[\ket{\phi}\!\!\bra{\phi}_A\otimes\ket{\varphi}\!\!\bra{\varphi}_B]\\
&=\ket{\varphi}\!\!\bra{\varphi}_B.
\end{align}
That is, $\rho_B$ is a pure state. From Theorem~\ref{generalformula}, the acceptance probability of the algorithm is given by the cycle index polynomial evaluated at the traces of increasing powers of $\rho_B$. But since $\rho_B$ is pure, $\tr[\rho_B^j]=1$ for all $j \in \{1,\ldots,n\}$. Then the acceptance probability is equal to the cycle index polynomial at $x_j=1$ for all $j\in\{1,\ldots,n\}$. That is,
\begin{align}
p_G &=Z(G)(1,\ldots,\tr[\rho^n])\\
&=Z(G)(1,\ldots,1)\\
&=\frac{1}{|G|}\sum_{g\in G} 1^{c_1(g)}\cdots1^{c_n(g)}\\
&=\frac{1}{|G|}\sum_{g\in G} 1\\
&=1
\end{align}
where $c_i(g)$ denotes the number of cycles of length $i$ in the disjoint cycle decomposition of $g$. Thus, $\psi_{AB}$ separable implies that the acceptance probability is identically one.

Now suppose $\rho_B$ is a mixed state. Then $\tr[\rho_B^j]<1$ for all $j>1$ and we have
\begin{align} 
p_G &=Z(G)(1,\tr[\rho_B^2],\ldots,\tr[\rho_B^n])\\
&=\frac{1}{|G|}\sum_{g\in G} 1^{c_1(g)}(\tr[\rho_B^2])^{c_2(g)}\cdots(\tr[\rho_B^n])^{c_n(g)}\\
&<\frac{1}{|G|}\sum_{g\in G} 1^{c_1(g)}\cdots1^{c_n(g)}\\
&=\frac{1}{|G|}\sum_{g\in G} 1\\
&=1,
\end{align}
where we have used the assumption that $G$ is nontrivial to guarantee that at least one of the $c_j(g)$ is nonzero so that the inequality holds. Thus, the test is faithful.
\end{proof}

\section{Resource Comparison of Symmetry Tests}

\label{sec:comparison}

Given the generalization in Section~\ref{generalization}, we can now compare the performance of these separability tests. There are two practical concerns to consider when implementing such a test: the rate at which the acceptance probability decays and the resources required to construct it. The cycle index polynomial results described above allow for direct analysis of the former topic, but the latter requires additional consideration before it can be adequately addressed. First, we will specify how resources are counted for each algorithm. Then we compare the resource cost for each algorithm given this framework. We accompany this with a discussion of the acceptance probability of the compared methods.

We now clarify what is meant by resources in this context. For the $G$-Bose symmetry test described in Chapter~\ref{ch:symmstates} and tests of that nature, the two primary resources are the number of gates used to construct the test and how many qubits are needed in the control register. An alternate consideration is the depth of the circuit, which we will mention where appropriate. We begin with a discussion of gate counting.

\subsection{Resource Counting of Quantum Gates}\label{sec:resourcecounting}

The unitary representation in this context is always formed from a collection of SWAP gates used to permute the subsystems. SWAP gates can be realized by a sequence of three CNOT gates in alternating direction. Often, the literature commonly counts the number of CNOT gates used as a resource (see, e.g., \cite{grassl2000cyclic}); however, particular architectures may have more efficient realizations of the SWAP gate. Furthermore, this algorithm actually calls for controlled-SWAP gates, which may have vastly different compilations between architectures. For the purposes of this discussion, we will be counting the necessary number of controlled-SWAPs alone. Additionally, we do not restrict to swapping between consecutive Hilbert spaces although in principle this could be a limitation of particular systems.

Here, we give an explicit construction for two example groups. The first is the cyclic group test, which is a simple Abelian subgroup of the symmetric group and therefore of interest as a point of comparison. Although constructions of cyclic shifts exist in the literature, our construction follows binary encoding procedures \cite{babbush2018encoding,low2019hamiltonian} and uses fewer gates than a naïve implementation and thus warrants discussion. The second construction given describes a recursive implementation of the full permutation test.  Similarly, although the quantum Schur transform \cite{BCH06,BCH07,krovi2019efficient} gives a recipe for implementing the symmetric group in principle, the gate construction is abstract and thus difficult to use for accurate gate counts compared to other approaches. As such, we utilize the construction given in \cite{barenco1997stabilization}. In the following two subsections, we show that a cyclic group test can be implemented with $\mathcal{O}(k \log (k))$ controlled-SWAP gates and a full symmetric group test (also known as a permutation test) with $\mathcal{O}(k^2)$ controlled-SWAP gates.

\subsubsection{Cyclic Group}

Analysis of the cyclic group benefits from established literature. Any cyclic permutation can be achieved in constant depth with $k-1$ gates, where $k$ is the order of the cycle \cite{grassl2000cyclic}. We will now show that any cyclic group test can be generated by implementing solely the elements in that cycle that are powers of two. This means that the resource cost of implementing the cyclic test of order $k$ is $(k-1)\log_2 (k)$, and the constant depth condition above from \cite{grassl2000cyclic} gives a corresponding depth of $\mathcal{O}(\log_2(k))$ in the separability test.

To see this, first recall that the $k$-order cyclic group is isomorphic to the set $\mathbb{Z}_k$ of integers modulo $k$ under addition. This will allow us to symbolically represent each element by a single number, understood in this context to be modulo $k$. 

Since the case of $k=1$ is trivial, let us first consider the base case of $k=2$. This example illustrates the general construction of cyclic tests and recreates the well-established swap test \cite{barenco1997stabilization,buhrman2001quantum}. The controlled-SWAP element corresponds to the element $1 = 2^0$, and is the sole gate needed, and the identity element is naturally $0$. (Note that $1$ is the sole power of two in $\mathbb{Z}_2 = \{0,1\}$.) The control state for this test is given by a single qubit state of
\begin{equation}
    \ket{+}_{C(2)} =\frac{1}{\sqrt{2}}(\ket{0} + \ket{1})\, ,
\end{equation}
where we employ the computational basis. It is clear that each element in the ancillary basis will give rise to its corresponding group element with this test.

How does this construction generalize? For each given $k$, we follow a similar recipe as above. As $C_k$ is isomorphic to $\mathbb{Z}_k$, start by identifying each cycle in $C_k$ with a number in $\mathbb{Z}_k$. If we always map the first $k$-cycle to one, then this map follows simply by mapping cycle composition to integer addition by one. Consider, for instance, the case of $C_5$. Then the first cycle is $(1\ 2\ 3\ 4\ 5)$. Map this to $1$. Then the next element, $(1\ 3\ 5\ 2\ 4)=(1\ 2\ 3\ 4\ 5)(1\ 2\ 3\ 4\ 5)$ maps to $1+1=2$. After we have identified each element of $C_k$ with an element of $\mathbb{Z}_k$, we can always rewrite these numbers in binary. The beauty of binary construction, as is well appreciated in computer science, is that only elements corresponding to powers of two need to be individually defined, and every other number can be generated from combinations of them. Thus, after this second rewrite, we have a mapping between every cycle in $C_k$ and a binary number. Now to construct the circuit, we only need to implement controlled gates that correspond to cycles that have mapped to a power of two. For $C_5$, this would be gates that have mapped to $001$, $010$, and $100$ (in decimal: 1, 2, and 4, respectively). 

To show how this construction grows, it is most convenient to denote the gates by which power of two they implement.  In Figure~\ref{fig:cyclic}, we label gates as $2^j$ where $j$ ranges from 0 to $\lfloor \log_2(k-1)\rfloor$. To see why $\lfloor \log_2(k-1)\rfloor$ is the final gate, recall the convention that $\mathbb{Z}_k$ always contains $0$ instead of $k$. Then the bound falls out from inspection. Revisiting our above example of $k=5$, the gates we identified as necessary can be equivalently represented as $001=2^0$, $010=2^1$, and $100=2^2$.

This construction can also be achieved by considering the labeling of the control state. If the computational basis is read as a number in binary, we can clearly define the relationship between the computational basis and the group element construction as $\ket{g}=\ket{g_{\textrm{binary}}}=\ket{g_{\textrm{decimal}}}$, where the abstract construction is equivalent to a computational basis in binary, which equivalently realizes the familiar group element in decimal. For example, following the above convention, the basis state for $k=5$ given by $\ket{(1\ 3\ 5\ 2\ 4)}=\ket{10}=\ket{2}$ indicates that the element $(1\ 3\ 5\ 2\ 4)$ can be labeled as the $2$ element of the group. As $2$ is obviously a power of $2$, this group element must be encoded in the circuit. This construction is shown generally in Figure~\ref{fig:cyclic} and for our specific example of $k=5$ in Figure~\ref{fig:cyclicexample}. Note that all elements of $C_k$ will take at most $k-1$ SWAP gates to implement. 

Furthermore, note that cyclic permutations can be implemented in a constant depth of two \cite{grassl2000cyclic}. To maintain this depth even for the controlled gates, a GHZ state, $\frac{1}{\sqrt{2}}(\ket{0}^{\otimes m} +\ket{1}^{\otimes m})$, where $m=\lfloor k/2\rfloor$, can be used instead of a single plus state, similar to the approach employed in \cite{yihui}. Then the controls can act on different qubits of the state, and the final measurement is taken by projecting back to the GHZ state. This state preparation and the corresponding measurement may add complexity to the ancilla register; however, since the circuit to prepare a $k$-qubit GHZ state has depth $\mathcal{O}(\log_2(k))$ (with the circuit to project onto it being its inverse), this gives the cyclic group test a depth that grows as $\mathcal{O}(\log_2 (k))$.

To see that this circuit is capable of generating every element of the cyclic group, we again refer to the isomorphism between $C_k$ and $\mathbb{Z}_k$. Writing every element of $\mathbb{Z}_k$ in binary, it becomes obvious that every element can be written as an addition of powers of 2 that form the basis of binary numbers. As such, only elements corresponding to new ``digits" need to be considered.

\begin{figure}[ptb]
\begin{center}
\includegraphics[
width=4.5 in
]{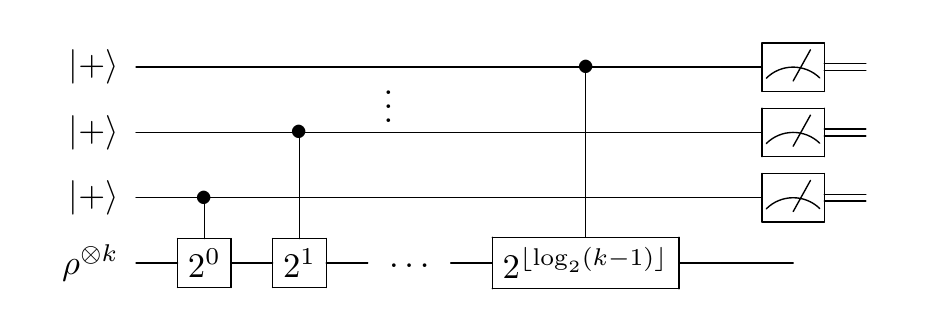}
\end{center}
\caption{Figure demonstrating how to systematically generate a test for the cyclic group of order $k$. The notation $(2^j)$ indicates the unitary representation of the element in $C_k$ labeled by the $j$-th power of two. Alternatively, this element is obtained by the full $k$-cycle $(1,2,\ldots,k)$ acting on itself $2^j$ times. Note that the final power is always given by $\lfloor \log_2(k-1)\rfloor.$ Also, $|+\rangle = (\ket{0}+\ket{1})/\sqrt{2}$ in the circuit diagram above and the final measurements are performed in the Hadamard basis, accepting if all $+1$ outcomes occur.}
\label{fig:cyclic}
\end{figure}

\begin{figure}[ptb]
\begin{center}
\includegraphics[
width=5 in
]{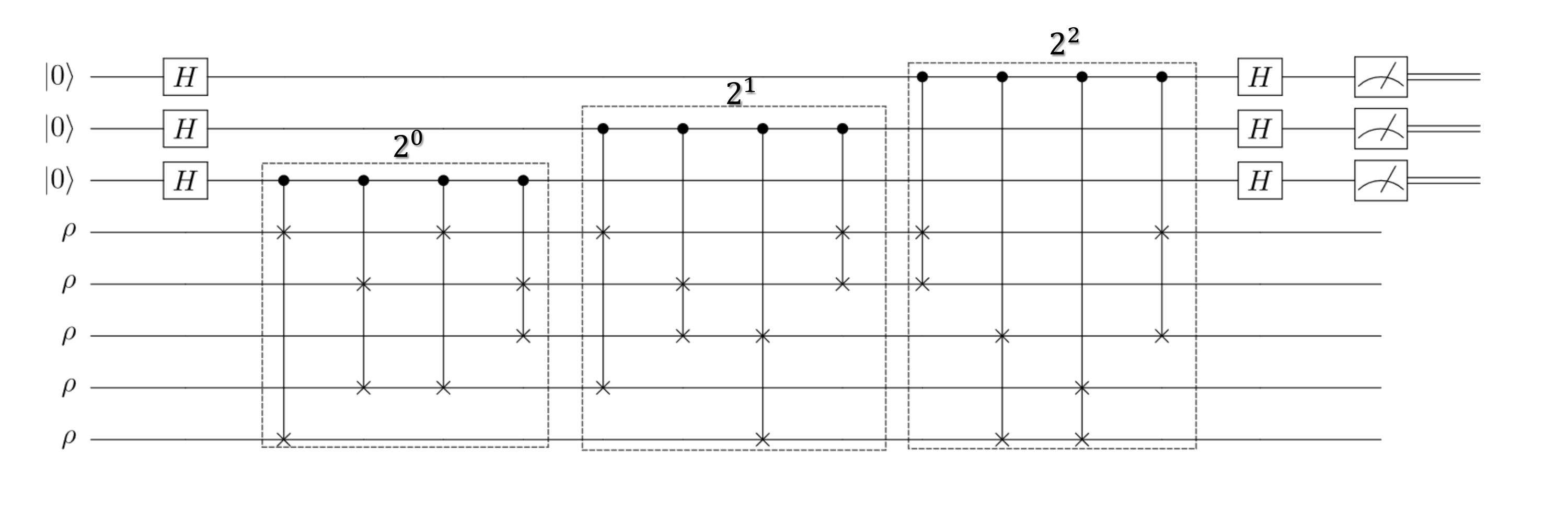}
\end{center}
\caption{An example of the cyclic group test for $k=5$. The notation $(2^j)$ indicates the unitary representation of the element in $C_k$ labeled by the $j$-th power of two. For this case, only the elements corresponding to $2^0$, $2^1$, and $2^2$ contribute. Notice that, if the gates are not controlled on the same qubit, each individual cycle collapses to a depth of two with $k-1$ gates.}
\label{fig:cyclicexample}
\end{figure}

\subsubsection{Symmetric Group}

We now review a recursive algorithm for the construction of the symmetric group.  Necessary to this construction is the proof that the entire group $S_k$ can be generated in a convenient way, using solely transpositions. This construction is equivalent to that given in \cite{barenco1997stabilization}, but we explain it here for convenience.

Observe that $S_2$ can be generated by the element $(1\ 2)$. To generate $S_3$, we need only act on this element from the left by $(1\ 3)$ and $(2\ 3)$. Indeed, the remaining elements of $S_3$ are given by $(1\ 2\ 3)=(1\ 3)(1\ 2)$ and $(1\ 3\ 2)=(2\ 3)(1\ 2)$. This serves as our base case, and we now proceed by induction. Suppose we can generate every element of $S_{k-1}$ in this way. We must show that the remaining elements of $S_k$ are given by acting on $S_{k-1}$ from the left by the transpositions of the form $(i\ k)$ for $i\in \{1,2,\ldots,k-1\}$. To see this, let $(i_1\ i_2\ \cdots\ i_m)$ be an arbitrary $m$-cycle in $S_{k-1}$. Then acting from the left by $(i_j\ k)$ for some $j\in \{1,\ldots,m\}$ yields $(i_j\ k)(i_1\ i_2\ \cdots\ i_m)=(i_1\ i_2\ \cdots\ i_{j-1}\ k\ i_j\ \cdots\ i_m)$. In this way, we can generate every cycle in $S_k$. Since every element of $S_k$ can be decomposed into a product of disjoint cycles, we can now generate every element of $S_k$ recursively by appending only transpositions of the form $(i\ k)$. We can visualize this construction by the circuit given in Figure~\ref{fig:permutation4} for an example when $k =4$. 

Given a way to generate $S_k$, we now need an appropriate control state to implement these elements. By supposition, the identity can always be implemented via the state $\ket{0}$ tensored with itself to some power. What then for the remaining states? Consider only one 'layer' of the recursive construction of $S_k$. It suffices to only ever use one transposition at a time. Thus the control state for every $i$-th layer of transpositions should take the form
\begin{equation}
    \ket{+}_{S_i}=\frac{1}{\sqrt{i+1}}(\ket{0}^{\otimes i} + \ket{10\cdots 0} + \ket{01\cdots 0}+\cdots +\ket{00\cdots 1})\, ,
\end{equation}
as given in \cite{barenco1997stabilization}. These individual control states should be concatenated together to form the control register for the entire algorithm. For a quick sanity check, when considering the tensor product of such states as $i$ ranges from 1 to $k$, the normalization constant out in front becomes $\sqrt{k!}=\sqrt{|S_k|}$.

However, a question remains; can the control register for such a circuit also be generated recursively? Observe, in Figure~\ref{fig:permutation4}, that we denote a series of gates $A_j$ that act on the control register to create superpositions. Furthermore, notice that we have arranged the transpositions in a consistent manner such that each gate is appended in ascending order of transposition. Then we define the gate $A_j$ to act as such:
\begin{equation}
    A_j \ket{0}^{\otimes j-1} =\frac{1}{\sqrt{j}} (\ket{0}^{\otimes j-1} + \ket{W_{j-1}}) \, ,
\end{equation}
where $\ket{W_{j-1}}=\frac{1}{\sqrt{j-1}}\sum_{i=1}^{j-1}\ket{2^i}$ is the $W$-state on $j-1$ qubits. Here $\ket{2^i}$ is the state with a one in the $i$-th component and a zero elsewhere. (This is equivalent to the one-hot encoding commonly used in literature.) We can observe by inspection that this action, when taken recursively from $j=2$ to $j=k$, will generate a superposition over $k!$ basis elements. An example of this construction can be seen in Figure~\ref{fig:permutation4} for $k=4$. 

There are several choices available to construct these $A_j$ gates. We review two here. One recursive approach is to begin by designing the circuit for $A_i$; then the next gate $A_{i+1}$ is given by adding $i+1$ control qubits, initializing the first qubit to a superposition of $(\frac{1}{\sqrt{i}} \ket{0}+\frac{\sqrt{i-1}}{\sqrt{i}}\ket{1})$, then controlling off of this state, implement $A_i$ on the remaining new qubits. However, this na\"ive approach will use numerous gates and quickly grow in size. In \cite{barenco1997stabilization}, they assume the first $i$ qubits are initialized, then add $i+1$ qubits for the recursive step. The $(i+1)$-th qubit can be acted on by a one-qubit gate $U_i$ given by
\begin{align}
    U_i\coloneqq\frac{1}{\sqrt{i+1}}
    \begin{pmatrix} 
    1 & -\sqrt{i}\\
    \sqrt{i} & 1
    \end{pmatrix}\, .
\end{align}
Following this, act simultaneously on the $i+1$ qubit and the remaining qubits with a series of two-qubit gates given by 
\begin{align}
    T_{j,j+1}\coloneqq\frac{1}{\sqrt{i-j+1}}
    \begin{pmatrix} 
    \sqrt{i-j+1} & 0 & 0 & 0\\
    0 & 1 & \sqrt{i-j} & 0\\
    0 & -\sqrt{i-j} & 1 & 0\\
    0 & 0 & 0 & \sqrt{i-j+1}
    \end{pmatrix}\, ,
\end{align}
where $j$ ranges from 1 to $i-1$. This will give the desired control state. In all likelihood, there are even more ways to generate the desired control register. Whichever approach is chosen, the control state should remain the same. Note that the ancilla cost of the control state should be at least $\mathcal{O}(k \log_2 k)$ ancilla qubits regardless simply from the magnitude of the symmetric group, $|S_k|=k!$. 

Given this construction, it is easy to see the number of controlled-SWAP gates needed to perform the symmetric group test. Indeed, from Figure~\ref{fig:permutation4}, we see that the number of controlled-SWAPs needed when $k=4$ is $1+2+3=6$, where the 1 corresponds to the permutation $(1\ 2)$ needed to generate $S_2$, the 2 corresponds to the permutations $(2\ 3)$ and $(1\ 3)$ needed to generate $S_3$ from $S_2$, and the 3 corresponds to the permutations $(3\ 4)$, $(2\ 4)$, and $(1\ 4)$ needed to generate $S_4$ from $S_3$. By induction, the number of controlled-SWAP's needed to perform the $k$-th symmetric group test is the sum of the first $k$ integers, or $k(k-1)/2$, thus leading to the claimed $\mathcal{O}(k^2)$ gate complexity.

\begin{figure}[ptb]
\begin{center}
\includegraphics[width=4.0in]{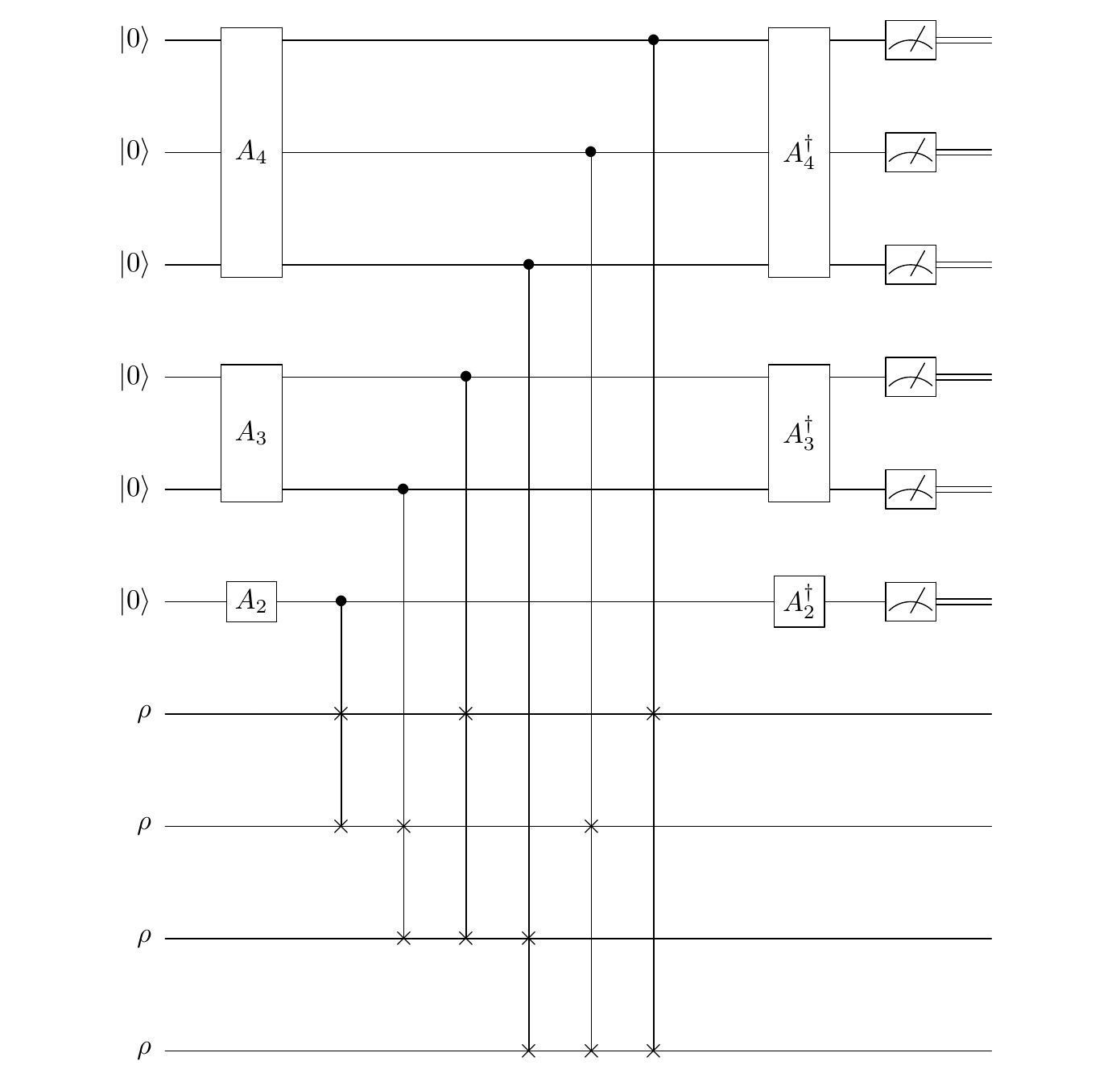}
\end{center}
\caption{Figure demonstrating how to systematically generate a test for the symmetric group of order four. }
\label{fig:permutation4}
\end{figure}

\subsubsection{Dihedral Group}
The dihedral group, $D_k$ is isomorphic to the semi-direct product of $\mathbb{Z}_k$ with $\mathbb{Z}_2$, with $\mathbb{Z}_2$ acting on $\mathbb{Z}_k$ by inversion. As such, it can be formed in a faithful way using a cyclic group generator and a non-commuting action that squares to identity. Using just the generators of the group, it is clear that the unitary flip action adds a factor of two to the number of cyclic gates needed, plus the additional instance of the flip element acting alone. In this manner, the full dihedral group requires at most $2k \log_2(k)$ gates to implement.

\subsection{Comparison between Subgroups of the Symmetric Group}

Now that we have given a method to count the number of quantum gates necessary for these separability tests, we consider if there is any advantage to using a simpler group as $k$ increases. Essentially, when is the trade-off between additional gates and acceptance probability favorable towards the various tests? 

The inherent motivation behind increasing $k$ is to obtain a smaller acceptance probability, prompting the need for Theorem~\ref{th:decreasing}. Clearly, the symmetric test provides the most stringent bound (see Figure~\ref{fig:acceptance}), yet it grows quickly in terms of gate resources needed (see Figure~\ref{fig:resource}). The cyclic group, however, benefits from the simplest construction but does not decay as quickly as the full symmetric group. 

\begin{figure}[ptb]
\begin{center}
\includegraphics[
width=4.0in
]{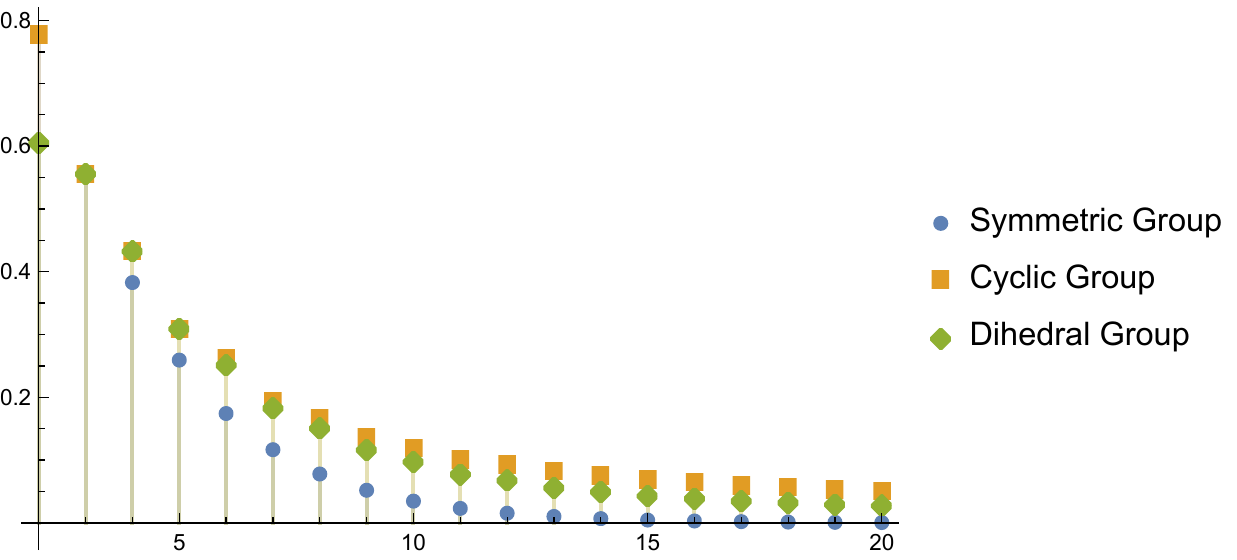}
\end{center}
\caption{Plot of the acceptance probabilities of each separability test as $k$ increases, for the symmetric group $S_k$, the cyclic group $C_k$, and the dihedral group $D_k$. For this example, we use a reduced $W$-state as an example to illustrate the algorithmic scaling for an unextendible state. For a separable state, all acceptance probabilities are equal to one.}
\label{fig:acceptance}
\end{figure}

\begin{figure}[ptb]
\begin{center}
\includegraphics[
width=4.0in
]{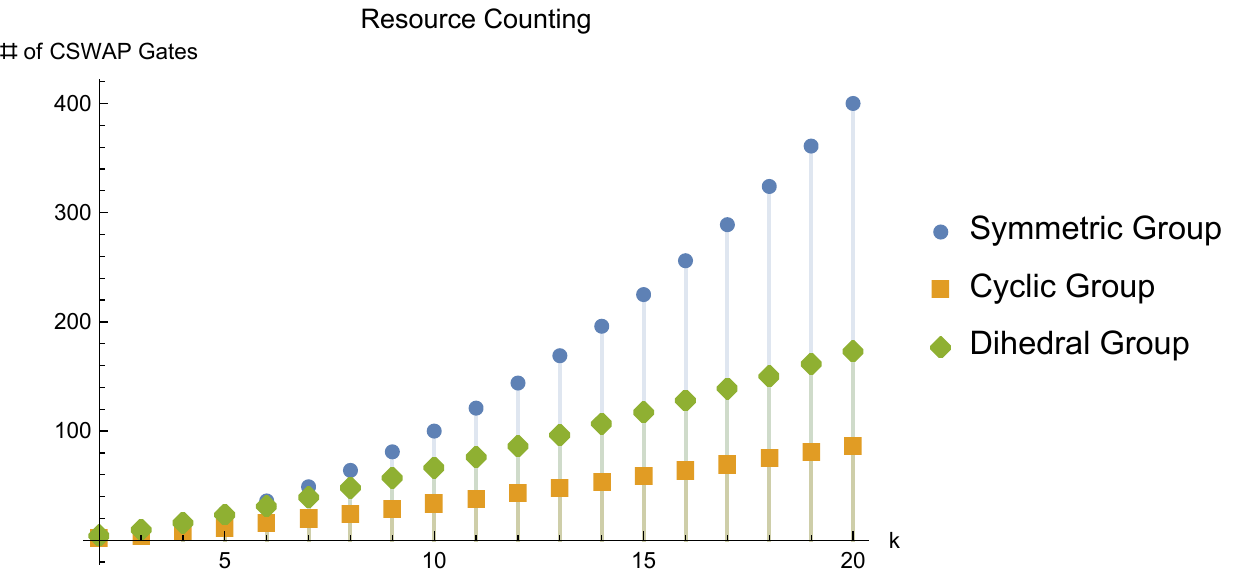}
\end{center}
\caption{Plot of the resource scaling in terms of the number of controlled-SWAP gates used for each group test as $k$ increases. We consider the symmetric group $S_k$, the cyclic group $C_k$, and the dihedral group $D_k$, and we use the gate counting methods described in the text.}
\label{fig:resource}
\end{figure}

To visualize this trade-off, we consider the quantity $ \frac{R_{\textrm{test}}}{1-P_{\textrm{acc}}}$, where $R_{\textrm{test}}$ is the number of resources needed to perform the test via the counting methods described above and $P_{\textrm{acc}}$ is the acceptance probability of the test. We employ the quantity $1-P_{\textrm{acc}}$ in the denominator, as we would like the test to have a lower acceptance probability for non-separable states, and thus the denominator will converge to one for better algorithms. This quantifier is very closely aligned with the expected runtime of the algorithm until getting a failure (it would be exactly equal to the expected runtime if we instead used circuit depth over $1-P_{\textrm{acc}}$ as the figure of merit). Thus, in comparing this quantity for the various tests, smaller values correspond to more ideal behavior from the algorithms. In Figure~\ref{fig:ratio}, we show this quantity for algorithms generated by the cyclic group and the symmetric group.

\begin{figure}[ptb]
\begin{center}
\includegraphics[
width=4.0in
]{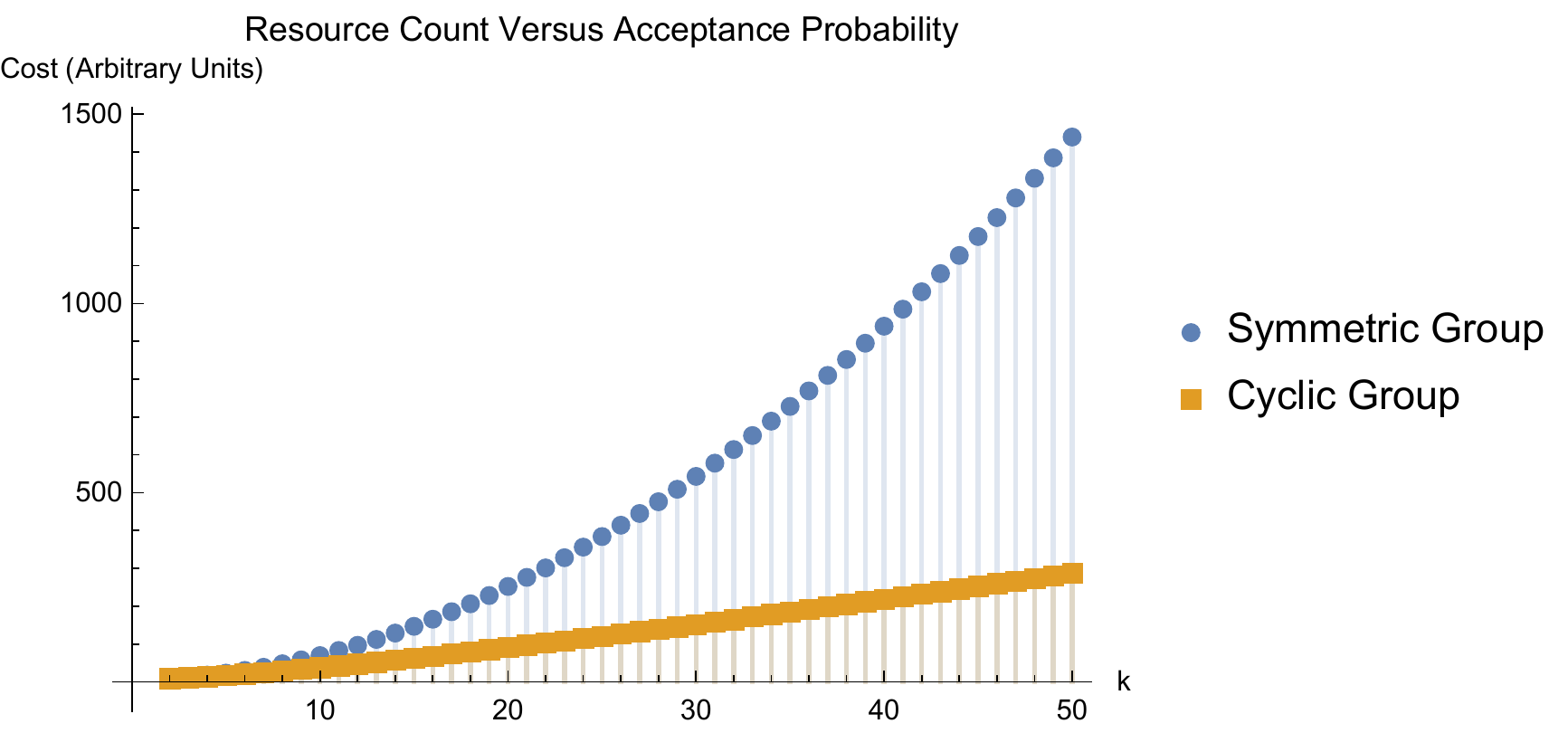}
\end{center}
\caption{We show the ratio of the resources required to rejection probability as $k$ increases. We consider here the cyclic group of $k$ elements as an example of a simple Abelian group and show it gives an advantage in terms of the resources-to-rejection metric over the test generated by the full symmetric group.}
\label{fig:ratio}
\end{figure}

Examining Figure~\ref{fig:ratio}, we see a clear difference in the performance between the tests generated by the cyclic group and the symmetric group. The plotted ratio can be thought of as resources-to-rejection, in the sense that $1-P_{\textrm{acc}}$ is the probability that a non-separable state is correctly identified---or rather, the failure rate of the algorithm for such a state. Although Figure~\ref{fig:acceptance} makes it appear that the standard test generated from $S_k$ would always be preferable, we determine from this comparison that the $C_k$ algorithm gives more benefit per gate resource. 

From this analysis, we can assert that the simpler test for separability is more cost-efficient than the full permutation test. We show in Figure~\ref{fig:acceptance} that both tests show a decrease in acceptance probability as $k$ increases, a desirable trait. However, Figure~\ref{fig:resource} shows how quickly circuit sizes grow as $k$ increases, particularly for $S_k$, which can be considered the standard test. Figure~\ref{fig:ratio} bridges these notions to show that the comparative growth in gate resources of $S_k$ outweighs the relative decrease in acceptance probability given over the $C_k$ test. We thus determine that the cyclic group $C_k$ suffices as a separability test of this nature.

\section{Conclusion}

\label{sec:conclusion}

In this chapter, we have presented several separability tests for bipartite pure states, and we have established analytical expressions for their acceptance probabilities. These expressions invariably rely on the cycle index polynomial of the group. Indeed, from a mathematical point of view, this relationship seems natural, due to the inherent combinatorics present in the algorithms. Nonetheless, these expressions give us direct insight into the performance of any separability tests generated from a finite group---which we have shown can be feasibly constructed. Using this perspective, we demonstrate that when utilizing more copies of the state under test, these tests become more stringent. Additionally, we observe that the full symmetric test using a representation of the symmetric group gives a quickly decreasing acceptance probability for an entangled state; however, for the given implementations of these algorithms, other tests can use fewer resources and still show great efficiency. 

Here, we have limited ourselves to pure bipartite states; however, we believe multipartite tests may yield results in a similar vein. For instance, a trivial implementation would be to separate all parties into individual tests and then multiply the results. There is a question, however, if more elegant algorithms exist for multipartite cases, and if interesting mathematics arise in the study of such systems. 

\subsection*{Data Availability Statement}
The datasets generated during and/or analysed during the current study are available in the GitHub repository, \sloppy \url{https://github.com/mlabo15/GeneralizedSeparability}.

\pagebreak
\singlespacing
\chapter{Lagniappe}
\doublespacing
\section{Introduction}
In Louisiana French, ``lagniappe" (pronounced ``lan-yap") means some extra inclusion or gift beyond what is strictly necessary. It is often translated as ``a little something extra", and this chapter aims to be exactly that. In what follows, we represent three results not yet published or presented in another medium, for whatever reason. The first investigates the use of density matrix exponentiation to obtain a normalized commutator. Second, we consider the impact of using all measurement outcomes in the Hamiltonian symmetry test when the group is Abelian. This is then supplemented with an alternate construction of the Hamiltonian symmetry test from Chapter~\ref{ch:symmham} using block encoding. These findings, all resulting from joint work with Dr. Mark Wilde, are included below as lagniappe.

\section{Density Matrix Exponentiation and Symmetry Testing}
We begin with an approach combining two concepts previously presented in this work---namely, testing the symmetries of quantum states from Chapter~\ref{ch:symmstates} and the nested commutator result procured in Chapter~\ref{ch:symmham}. We do so by employing density exponentiation \cite{lloyd2014quantum} as an alternative to the algorithm presented in Section~\ref{sec:efficient} previously used to test Hamiltonian symmetry.

Here, we do not delve into the details of density matrix exponentiation but will instead view it as a `black box' that can be called upon. Following Theorem~1 of \cite{kimmel2017hamiltonian}, the process of density matrix exponentiation can be summarized thusly: by using $\mathcal{O}(\frac{t^2}{\delta})$ copies of a state $\rho$, one can simulate the unitary channel $(\cdot) \to e^{-i \rho t}(\cdot)e^{i \rho t}$ up to $\delta$-error in diamond distance. That is, there exists a quantum algorithm described by the channel $\mathcal{A}_{SB_{1}\cdots B_{n}\rightarrow S}$ such that
\begin{equation}\label{eq:diamond-dist-bnd}
\sup_{\sigma_{RS}}\frac{1}{2}\left\Vert \mathcal{A}_{SB_{1}\cdots B_{n}\rightarrow S}(\sigma_{RS}\otimes\rho_{B_{1}}\otimes\cdots\otimes \rho_{B_{n}})-\left(  I_{R}\otimes e^{-i\rho t}\right)  \sigma_{RS}\left(I_{R}\otimes e^{-i\rho t}\right)  ^{\dag}\right\Vert _{1}\leq\delta,
\end{equation}
where the optimization is over every state $\sigma_{RS}$, with reference system $R$ arbitrarily large. Observe that this is a method for Hamiltonian simulation where the Hamiltonian in this case is the quantum state $\rho$. (Indeed, every state is a legitimate Hamiltonian as every state is guaranteed to be Hermitian.)

Given this computational tool, we can now discuss how density matrix exponentiation can be utilized to test a state $\rho$ for symmetry with respect to a group $G$ with projective unitary representation $\{U(g)\}_{g\in G}$. Recall the algorithm from Chapter~\ref{ch:symmham}, Section~\ref{sec:efficient}. This algorithm had an acceptance probability, via \eqref{eq:modifiedacc}, of
\begin{equation}
\operatorname{Tr}[\Pi^{G}\left(  I_{R}\otimes e^{iHt}\right)  \Phi_{RS}\left( I_{R}\otimes e^{-iHt}\right)  ^{\dag}]=\frac{1}{d\left\vert G\right\vert }\sum_{g\in G}\operatorname{Tr}[U^{\dag}(g)e^{iHt}U(g)e^{-iHt}],
\end{equation}
where $\tau=\left\Vert H\right\Vert _{\infty}t$. Our goal now is to replace the operator $e^{i Ht}$ with the density matrix exponentiation black box previously described, as shown in Figure~\ref{fig:hamcircuitagain}.

\begin{figure}[t]
\begin{center}
\includegraphics[width=\columnwidth]{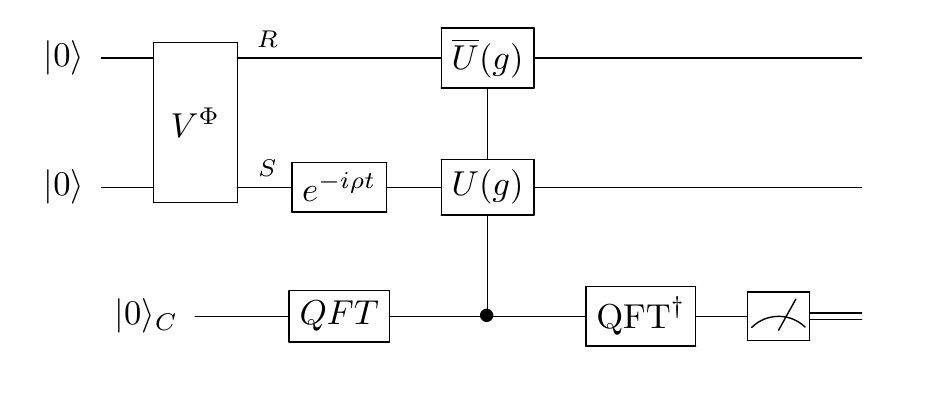}
\end{center}
\caption{Quantum circuit to test for the symmetry of a state via density matrix exponentiation channel. The unitary $V^\Phi$ generates the state $\ket{\Phi}_{RS}$, the maximally-entangled state on $RA$. The exponentiation is achieved via by $e^{-i \rho t} $, and the $U(g)$ gates are controlled on a superposition over all of the elements $g \in G$, as in \eqref{eq:superpose-gr-elems}.}
\label{fig:hamcircuitagain}
\end{figure}

 First, assume that the density matrix exponentiation is ideal. Then the acceptance probability is found by simply plugging $e^{i \rho t}$ into the above equation, giving
\begin{equation}
    \operatorname{Tr}[\Pi^{G}\left(  I_{R}\otimes e^{i\rho t}\right)  \Phi_{RS}\left(I_{R}\otimes e^{-i\rho t}\right)  ^{\dag}] \, .
\end{equation}
Invoking \eqref{eq:diamond-dist-bnd}, we find that
\begin{equation}
\left\vert \operatorname{Tr}[\Pi^{G}\left(  I_{R}\otimes e^{i\rho t}\right) \Phi_{RS}\left(  I_{R}\otimes e^{-i\rho t}\right)  ^{\dag}]-\widetilde {P}_{\text{acc}}\right\vert \leq\delta \, ,
\end{equation}
where
\begin{equation}
\widetilde{P}_{\text{acc}}=\operatorname{Tr}[\Pi^{G}\mathcal{A}_{SB_{1}\cdots B_{n}\rightarrow S}(\Phi_{RS}\otimes\rho_{B_{1}}\otimes\cdots\otimes\rho_{B_{n}})] \, .
\end{equation}
Then, by plugging in the above, employing the expansion in \eqref{eq:hamtaylorseries}, and choosing $\delta=t^{4}$ such that the error from density matrix exponentiation is of fourth order, we conclude that
\begin{equation}\label{eq:DME-expansion}
\widetilde{P}_{\text{acc}}=1-\frac{t^{2}}{2d\left\vert G\right\vert } \sum_{g\in G}\left\Vert \left[  U(g),\rho\right]  \right\Vert _{2}^{2}+O(t^{4}),
\end{equation}
given that $\tau=\left\Vert H\right\Vert _{\infty}t\leq t$ and $H=\rho$ in this case, so that $\left\Vert H\right\Vert _{\infty}\leq1$. 

Thus, the acceptance probability, in this case, includes the measure of asymmetry, the normalized commutator norm $\frac{1}{d\left\vert G\right\vert }\sum_{g\in G}\left\Vert \left[  U(g),\rho\right]  \right\Vert _{2}^{2}$. Furthermore, the number of copies of $\rho$ needed to realize the expression in \eqref{eq:DME-expansion} is $O(t^{2}/\delta)=O(1/t^{2})$. 

\section{Hamiltonian Symmetry Measurement with Abelian Groups}
In all of our algorithms, we take as a success only one potential outcome and reject all others. A natural question is whether we can make any use of these other measurement outcomes. To investigate this situation, consider the case of Hamiltonian symmetry described in Chapter~\ref{ch:symmham}. 

Imagine we construct a control register in the manner we have become accustomed to, as a superposition over group elements $g \in G$, and let us consider as input the maximally-entangled state $\ket{\Phi}$ of dimension $d$. Then the initial state of the system is 
\begin{equation}
    \frac{1}{\sqrt{|G|}}\sum_{g \in G} \ket{g}\ket{\Phi} \, .
\end{equation}
Now perform the controlled unitary $\ket{g}\bra{g} \otimes U^{\dag}(g)e^{-i H t}U(g)$ to create the state
\begin{equation}\label{eq:inputstate}
    \frac{1}{\sqrt{|G|}}\sum_{g \in G}  \ket{g} (U^{\dag}(g)e^{-i H t}U(g) \otimes \mathbb{I})\ket{\Phi} \, .
\end{equation}
(As for how to implement $e^{-iHt}$, allow us to borrow arguments given in Chapter~\ref{ch:symmham}.) So far, we have not ventured far from our original algorithm; let's remedy that. Instead of the usual measurement, we will measure in the Fourier basis given by
\begin{equation}\label{eq:fouriermmt}
    \ket{\Tilde{g}} \coloneqq \frac{1}{\sqrt{|G|}}\sum_{g\in G}e^{\frac{2\pi i g \Tilde{g}}{|G|}}\ket{g} \,.
\end{equation}

Now we must pause to address a difficulty. We have unceremoniously assumed that placing group elements into exponents is an unproblematic step. In fact, this step needs to be approached cautiously. Case in point, those familiar with Lie groups and Lie algebras may think by the factor $e^{\frac{2\pi i g \Tilde{g}}{|G|}}$ we are loosely appropriating that structure here---emphatically, this is not the case as we do not assume any sort of Lie group structure. There is an inherent assumption we make in declaring this a measurement, however, that our final outcome will be a real number. We would like to use \eqref{eq:fouriermmt} to employ some phase estimation approach, but this requires that $e^{\frac{2\pi ig\Tilde{g}}{|G|}}$ not be an operator. To facilitate this, we must require that $G$ is Abelian. Then there exists a one-dimensional representation of $G$ as discussed in Section~\ref{sec:representin}. We require that the group operation denoted by addition $g + \Tilde{g}$ be an equivalent representation to the unitary representation $\{U(g)\}_{g \in G}$ such that
\begin{equation}\label{eq:ftrepresentations}
    g + \Tilde{g} \cong U(g)U(\Tilde{g}) \, .
\end{equation}
Since we have assumed a one-dimensional representation in the exponent, the product $g\Tilde{g}$ will arise from the structure of the complex numbers. 

The probability of measuring the state in \eqref{eq:inputstate} to be in the state $\ket{\Tilde{g}}$ is given by
\begin{align}
    & \left \Vert (\bra{\Tilde{g}}_C \otimes \mathbb{I})\left ( \frac{1}{\sqrt{|G|}} \sum_{g \in G} \ket{g}_C (U^{\dag}(g)e^{-iHt}U(g) \otimes \mathbb{I})\ket{\Phi} \right) \right \Vert_2^2 \notag \\
    & =  \left \Vert \left (\frac{1}{\sqrt{|G|}}\sum_{g^{\prime} \in G} e^{-2\pi ig^{\prime}\Tilde{g}/|G|} \bra{g^{\prime}}_C \otimes \mathbb{I} \right)\left ( \frac{1}{\sqrt{|G|}} \sum_{g \in G} \ket{g}_C (U^{\dag}(g)e^{-iHt}U(g) \otimes \mathbb{I})\ket{\Phi} \right) \right \Vert_2^2 \\
    & = \frac{1}{|G|^2} \left \Vert \sum_{g,g^{\prime} \in G} e^{-2\pi ig^{\prime}\Tilde{g}/|G|} \langle g^{\prime}\ket{g}_C  \left ( (U^{\dag}(g)e^{-iHt}U(g) \otimes \mathbb{I}) \ket{\Phi} \right) \right \Vert_2^2 \\
    & = \frac{1}{|G|^2} \left \Vert \sum_{g \in G} e^{-2\pi ig\Tilde{g}/|G|}  \left ( (U^{\dag}(g)e^{-iHt}U(g) \otimes \mathbb{I}) \ket{\Phi} \right) \right \Vert_2^2 \\
    & = \frac{1}{|G|^2} \left ( \sum_{g^\prime \in G} e^{2\pi ig^{\prime}\Tilde{g}/|G|} \bra{\Phi} \left (U^{\dag}(g^\prime)e^{iHt}U(g^\prime) \otimes \mathbb{I} \right) \right ) \left ( \sum_{g \in G} e^{-2\pi ig\Tilde{g}/|G|}  \left ( U^{\dag}(g)e^{-iHt}U(g) \otimes \mathbb{I} \right ) \ket{\Phi} \right ) \\
    & = \frac{1}{|G|^2} \sum_{g,g^\prime \in G} e^{2\pi i \Tilde{g}(g^\prime - g)/|G|} \bra{\Phi} \left ( U^\dag(g^\prime)e^{iHt}U(g^\prime)U^{\dag}(g)e^{-iHt}U(g) \otimes \mathbb{I} \right ) \ket{\Phi} \\
    & = \frac{1}{d|G|^2} \sum_{g,g^\prime \in G} e^{2\pi i \Tilde{g}(g^\prime - g)/|G|} \tr[U^\dag(g^\prime)e^{iHt}U(g^\prime)U^{\dag}(g)e^{-iHt}U(g)] \\
    & = \frac{1}{d|G|^2} \sum_{g,g^\prime \in G} e^{2\pi i \Tilde{g}(g^\prime - g)/|G|} \tr[U^\dag(g^\prime-g)e^{iHt}U(g^\prime- g)e^{-iHt}] \label{eq:ftmessy}
\end{align}
Now we use the assumptions in \eqref{eq:ftrepresentations} to allow us to rewrite \eqref{eq:ftmessy} in terms of $h \coloneqq g^\prime - g$, giving us
\begin{align}
    & \frac{1}{d|G|^2} \sum_{g,h \in G} e^{2\pi i \Tilde{g}h/|G|} \tr[U^\dag(h)e^{iHt}U(h)e^{-iHt}] \notag \\
    & =\frac{1}{d|G|} \sum_{h \in G} e^{2\pi i \Tilde{g}h/|G|} \tr[U^\dag(h)e^{iHt}U(h)e^{-iHt}] \\
    & = \frac{1}{d|G|} \sum_{g \in G} e^{2\pi i \Tilde{g}g/|G|} \tr[U^\dag(g)e^{iHt}U(g)e^{-iHt}] \, .
\end{align}

\noindent Thus, we have derived the probability of observing outcome $\ket{\Tilde{g}}$ as 
\begin{equation}
    \operatorname{Pr}[\Tilde{g}]= \frac{1}{d|G|} \sum_{g \in G} e^{2\pi i \Tilde{g}g/|G|} \tr[U^\dag(g)e^{iHt}U(g)e^{-iHt}] \, .
\end{equation}
This result bears some resemblance to the case where we only consider the outcome $\ket{0}$, which serves as a nice sanity check. We can also note that the inverse Fourier transform of $\operatorname{Pr}[\Tilde{g}]$ is given by
\begin{equation}\label{eq:iftotoc}
    \left \{ \frac{1}{d} \tr[U^{\dag}(h)e^{iHt}U(h)e^{-i H t}] \right \}_{h \in G} \, ,
\end{equation}
where each element of \eqref{eq:iftotoc} is a group-averaged out-of-time-order correlator (OTOC) \cite{de2019spectral,swingle2016measuring,hashimoto2017out}.

Now, how can such a result be used practically? Suppose we perform this measurement many times, say for $N$ trials, and keep track of the number of times each $\ket{\Tilde{g}}$ detection as $N(\Tilde{g})$. Then the empirical distribution of $\frac{N(\Tilde{g})}{N}$ will converge to $\operatorname{Pr}[\Tilde{g}]$ as $N$ becomes large. Of course, we need to determine how large $N$ should be for this to happen. For this purpose, allow us to assign a random variable $Y^h_j$, taken as a value $ e^{-2 \pi i \Tilde{g}h/|G|}$, if the outcome of the $j$-th trial is equal to $\Tilde{g}$. This generates a set of $|G|$ random variables $\{Y^h_j\}_{h \in G}$ for each trial, $N|G|$ in total. The expectation of the random variable $Y^h_j$ is 
\begin{align}
    \mathbb{E}[Y^h_j] & = \sum_{\Tilde{g} \in G}\operatorname{Pr}[\Tilde{g}]e^{-2\pi i\Tilde{g}h/|G|} \notag \\
    & = \frac{1}{d|G|}\sum_{g,\Tilde{g} \in G}e^{-2\pi i\Tilde{g}g/|G|}\tr[U^\dag(g)e^{iHt}U(g)e^{-iHt}] e^{-2\pi i\Tilde{g}h/|G|} \\
    & = \frac{1}{d}\sum_{g \in G}\tr[U^\dag(g)e^{iHt}U(g)e^{-iHt}]\frac{1}{|G|} \sum_{\Tilde{g} \in G} e^{-2\pi i\Tilde{g}(g-h)/|G|} \\
    & = \frac{1}{d}\sum_{g \in G}\tr[U^\dag(g)e^{iHt}U(g)e^{-iHt}] \delta_{g,h} \\
    & = \frac{1}{d} \tr[U^\dag(h)e^{iHt}U(h)e^{-iHt}] \, .
\end{align}
Following the same procedure as in Section~\ref{sec:hamaccept}, this can be expanded explicitly in terms of the commutator to show that
\begin{align}
    \mathbb{E}[Y^h_j] &= 1 + \frac{t^2}{d} (\operatorname{Tr}[H U(h) ^\dag H U(h)] - \operatorname{Tr}[H^2]) \notag \\
    & \ \ \ \ \ + \frac{it^3}{2}(\operatorname{Tr}[U(h)^\dag H^2 U(h) H] - \operatorname{Tr}[U(h) ^\dag H U(h) H^2]) + \mathcal{O}(\tau^4) \\
    & = 1 - \frac{t^2}{2d}\left\Vert [U(h), H]\right\Vert_2^2 + \mathcal{O}(\tau^4) \, ,
\end{align}
with $\tau \coloneqq \left\Vert H \right\Vert_\infty t < 1$.
Then the average $\overline{Y^h_N} \coloneqq \frac{1}{N} \sum_{j=1}^N Y_j^h$ is an unbiased estimator of the corresponding OTOC from \eqref{eq:iftotoc}. The Hoeffding Bound \cite{hoeffding1994probability} tells us that to obtain 
\begin{equation} \label{eq:lagnestimator}
    \operatorname{Pr}[|\overline{Y^h_N}-\mathbb{E}[Y^h_j]| \leq \epsilon ] \geq 1 - \delta
\end{equation}
$N$ must satisfy
\begin{equation} \label{eq:Nconstraint}
    N \geq \frac{4}{\epsilon^2}\ln(\frac{4}{\delta}) \, .
\end{equation}

In such a manner, we can estimate the variables $Y^h_j$. Unlike the results presented in Chapter~\ref{ch:symmham}, this procedure requires a further restriction to Abelian groups, as that allows for a natural one-dimensional representation. However, an interesting future question would be to examine if certain non-Abelian groups (for example, the dihedral groups) might allow for an analogous procedure. 

\section{Block-Encoded Hamiltonian Symmetry}

Once again, let us revisit the Hamiltonian symmetry test given in Chapter~\ref{ch:symmham}. Suppose, instead of a Trotterization, we have a block-encoding of a Hamiltonian into a unitary matrix. This formalism is a useful tool used to generalize how matrices can be implemented for use in quantum algorithms \cite{low2019hamiltonian,gilyen2019quantum}, and so we review it briefly here. 

Block-encoding allows a complex matrix $A$ with $\left\Vert A \right\Vert_\infty \leq 1$ to be represented as the upper-left block of a unitary matrix. That is,
\begin{equation}
        U =
    \begin{pmatrix} 
    A & \cdot\\
    \cdot & \cdot
    \end{pmatrix}\, ,
\end{equation}
or, equivalently, $A = (\bra{0}\otimes \mathbb{I}) U (\ket{0} \otimes \mathbb{I})$. Then $U$ is called a block-encoding of $A$.
The unitary U can be thought of as a probabilistic implementation of the linear map realized by $A$. Suppose $A$ acts on $a$ qubits; then, given an
$a$-qubit input state $\ket{\phi}$, acting with U on the state $\ket{0}\ket{phi}$ and post-selecting on a measurement of $\ket{0}$ on the first system will guarantee that the second system contains a state proportional to $A\ket{\psi}$.

Now, let us consider a unitary $B$ a block-encoding of a Hamiltonian $H$ of the form
\begin{align}
    B =
    \begin{pmatrix} 
    H & \cdot\\
    \cdot & \cdot
    \end{pmatrix}\, ,
\end{align}
such that 
\begin{equation}\label{eq:lblockdef}
    (\bra{0}_A \otimes \mathbb{I})B(\ket{0}_A \otimes H)\, ,
\end{equation}
where we require that $\left\Vert H \right\Vert_\infty \leq 1$.

How can we test the symmetry of H using this construction? Let us begin by creating an analogous block encoding of our unitary representation. This can be done via the new representation $\hat{U}(g)$ such that
\begin{equation}
    \hat{U}(g) \coloneqq \ket{0}\!\bra{0}_A \otimes U(g) + (\mathbb{I}_A - \ket{0}\!\bra{0}_A) \otimes \mathbb{I} \, .
\end{equation}
Next, promote $\hat{U}(g)$ to a controlled unitaries $V$, with the control state $\ket{+}_{C} \coloneqq \frac{1}{\sqrt{|G|}} \sum_{g \in G}\ket{g}$ such that
\begin{equation}\label{eq:lVpromote}
    V \coloneqq \sum_{g \in G} \ket{g}\!\bra{g}_C \otimes (\hat{U}(g)) \, .
\end{equation}

Now we can construct a quantum algorithm to test the symmetry of this approach in much the same vein as in Chapter~\ref{ch:symmham} and achieve familiar results. Once again, we consider the initial state of the system to consist of a control register initialized to $\ket{+}_C$ and an input of the maximally-entangled state $\ket{\Phi}$, but we augment this with an ancillary state $\ket{0}_A$. Then the initial state of the system is
\begin{equation}
    \ket{+}_C \ket{0}_A \ket{\Phi} \, .
\end{equation}
Then we mimic the action of the algorithm in Chapter~\ref{ch:symmham} replacing $U(g)$ with $V$ and $e^{iHt}$ with $\mathbb{I}_C \otimes B$. That is, we act first with $V$ then  $\mathbb{I}_C \otimes B$ then $V^\dag$. Then the state of the system is
\begin{equation}
    V^\dag(\mathbb{I}_C \otimes B) V\ket{+}_C \ket{0}_A \ket{\Phi} \, .
\end{equation}
Finally, measure $C$ and $A$, and accept only if the outcome $\ket{+}_C\ket{0}_A$ is observed. This means the acceptance probability is given by
\begin{align}
    & \left\Vert \bra{+}_C\bra{0}_A V^\dag(\mathbb{I}_C \otimes B) V\ket{+}_C \ket{0}_A \ket{\Phi} \right\Vert_2^2 \notag \\
    & = \left\Vert \frac{1}{\sqrt{|G|}} \sum_{g^\prime \in G} \bra{g^\prime}_C\bra{0}_A V^\dag(\mathbb{I}_C \otimes B) V \left( \frac{1}{\sqrt{|G|}} \sum_{g \in G}\ket{g}_C \ket{0}_A \ket{\Phi} \right) \right\Vert_2^2 \\
    & = \frac{1}{|G|^2} \left\Vert \sum_{g,g^\prime \in G} \bra{g^\prime}_C\bra{0}_A V^\dag(\mathbb{I}_C \otimes B) V\ket{g}_C \ket{0}_A \ket{\Phi} \right\Vert_2^2 \label{eq:lblockstep1} \, .
\end{align}
Now, in \eqref{eq:lblockstep1}, we can substitute in the definition of $V$ from \eqref{eq:lVpromote}. (Respectively, we can also easily implement the Hermitian conjugate of \eqref{eq:lVpromote} for $V^\dag$.) In doing so, we will gain two more summations over the group elements of $G$; however, using the assumption of orthogonality in our basis states $\ket{g}$, these all can be combined and simplified into a single sum over $g$. Furthermore, we can easily collapse the measurement over the $A$ system using \eqref{eq:lblockdef}. This procedure will thus simplify \eqref{eq:lblockstep1} into
\begin{align}
\frac{1}{|G|^2} &  \left\Vert \sum_{g,g^\prime \in G} \bra{g^\prime}_C\bra{0}_A V^\dag(\mathbb{I}_C \otimes B) V\ket{g}_C \ket{0}_A \ket{\Phi}\right\Vert_2^2 \notag \\
    & =  \frac{1}{|G|^2} \left\Vert \sum_{g \in G} U^\dag (g) H U(g) \ket{\Phi} \right\Vert_2^2  \\
    & = \frac{1}{|G|^2} \bra{\Phi} \left( \sum_{g^\prime \in G}U^\dag (g^\prime) H U(g^\prime) \right) \left( \sum_{g \in G}U^\dag (g) H U(g) \right) \ket{\Phi} \\
    & = \frac{1}{|G|^2} \sum_{g,g^\prime \in G} \bra{\Phi} U^\dag (g^\prime) H U(g^\prime)U^\dag (g) H U(g) \ket{\Phi} \\
    & =  \frac{1}{d|G|^2} \sum_{g,g^\prime \in G} \tr[U^\dag (g^\prime) H U(g^\prime)U^\dag (g) H U(g)] \\
    & =\frac{1}{d|G|^2} \sum_{g,g^\prime \in G} \tr[U(g)U^\dag(g^\prime) H U(g^\prime)U^\dag (g) H ] \, ,
\end{align}
where in the last step we use cyclicity of trace. Next, we will use the group homomorphism property of representations and the group operation $g^\prime g^{-1} = h\ \in G$ to continue.
\begin{align}
    & =\frac{1}{d|G|^2} \sum_{g,g^\prime \in G} \tr[U^\dag(g^\prime g^{-1}) H U(g^\prime g^{-1}) H ] \\
    & =\frac{1}{d|G|^2} \sum_{g,h \in G} \tr[U^\dag(h) H U(h) H ] \\
    & =\frac{1}{d|G|} \sum_{h \in G} \tr[U^\dag(h) H U(h) H ] \, .
\end{align}
Finally, we arrive at the result that our acceptance probability $P_{\textrm{acc}}$ is given by
\begin{align}
    P_{\textrm{acc}} &= \frac{1}{d|G|} \sum_{g \in G} \tr[U^\dag(g) H U(g) H ] \label{eq:block-enc-prob1} \, ,
\end{align}
where $d$ is the dimension as usual. This result in \eqref{eq:block-enc-prob1} should look very familiar, as it bares great resemblance to our original result in \eqref{eq:modifiedacc}; however, without the exponential present, this equation only provides the second order term of \ref{eq:accept-prob-comm}, which is the lowest-order term in which the commutator appears. As an additional note, whenever $H$ is unitary, no block-encoding is necessary. Thus, we have given an alternate approach to our algorithm in Chapter~\ref{ch:symmham} when the Hamiltonian simulation is available via block encoding.

\section{Conclusion}

\doublespacing

With this, we conclude the lagniappe section of this work. In this chapter, we have given an additional test for asymmetry of a quantum state and two expansions on the work of Chapter~\ref{ch:symmham}. While not major contributions independently, these results are nonetheless interesting in their own right, and we hope to continue expanding upon these findings in the future. For the time being, however, we can now progress to the conclusion of this work.

\pagebreak
\singlespacing
\chapter{Conclusion}
\doublespacing
     Throughout this thesis, we have demonstrated the intersection of symmetry property testing and quantum computational algorithms. We have introduced various types of symmetry and given their relevant algorithmic tests. We have shown that these tasks can be computationally difficult for classical computers and sometimes quantum computers as well. Nonetheless, we maintain that these results are of interest to both quantum information applications in particular and the field of physics in general.
     
     In Chapter 1 of this thesis, we began with an introduction to background terminology in mathematics and quantum computation. We defined relevant terms in group theory, representation theory, and quantum information in order to make this work self-contained. We further defined some notions of symmetry to be used in future chapters.
     
     In Chapter 2, we gave a quantum computing algorithm to test Hamiltonian dynamics for symmetry with respect to a group. The acceptance probability of this algorithm depended explicitly on the familiar expression of symmetry from quantum mechanics. Furthermore, we were able to show that this algorithm is DQC1-complete, making it unlikely to be efficiently calculable on classical computers. We further expand on these results by giving examples of relevant Hamiltonians calculated using the IBM quantum simulator.
     
     In Chapter 3, we proposed algorithms to test for various notions of symmetry, including $G$-symmetric extendibility and $G$-Bose extendibility. The acceptance probabilities of these algorithms are equal to the maximum symmetric fidelity of their respective symmetry, thus endowing these measures with operational meaning. Furthermore, we established resource theories of asymmetry corresponding to the symmetries we have tested. 
     
     In Chapter 4, we followed up on a specific subject introduced in the previous chapter---that of separability tests for pure bipartite states. We established acceptance probabilities for a general class of separability tests derived from $G$-Bose symmetric extendible tests, specifically showing a reliance on the cycle index polynomial of the group. This result then allowed us to compare the traditional separability test to a simpler cyclic group test, which was shown to be more resource efficient in terms of both depth and number of gates.
     
     Finally, in Chapter 5, we included three additional but related results. We invoked density matrix exponentiation to test the symmetry of a quantum state. We discussed reconsidered the test of Hamiltonian symmetry in which an Abelian group is being used and all measurement outcomes are considered. For the former, we showed that a simple alteration to the permutation symmetry test would suffice to account for this change. For the latter, we demonstrated how such an approach could be used to estimate out-of-time-order correlators. Finally, we gave an alternate approach to the algorithm in Chapter 2 when a block encoding of a Hamiltonian is available.
     
Further detailed calculations for the primary chapters can be found in the appendices, and all code used for the various projects contained within can be found at the appropriate links at the end of the respective chapters.

\pagebreak
\singlespacing
\appendix
\chapter{Supplementary Material for Chapter 2}
\doublespacing

The appendices serve primarily to imprison long proofs. Importantly, these proofs give context and credence to the results presented therein. As such, we include them here.

\section{Acceptance Probability of the First Hamiltonian Symmetry Test}

\label{app:accept}

To see that the acceptance probability of the first Hamiltonian symmetry test in Figure~\ref{fig:originalcircuit}
is given by $\operatorname{Tr}[\Pi^{G}\Phi_{RB}^{t}]$, consider that the state
just before the measurement is as follows:
\begin{equation}
\frac{1}{\sqrt{\left\vert G\right\vert }}\sum_{g\in G}\ket{g}_{C}\left(
\overline{U}_{R}(g)\otimes U_{B}(g)\right)  \ket{\Phi^{t}}_{RB}.
\end{equation}
Then the acceptance probability is given by
\begin{align}
& \left\Vert
\begin{array}
[c]{c}
\left(  \bra{+}_{C}\otimes \mathbb{I}_{RB}\right)  \times\\
\left(  \frac{1}{\sqrt{\left\vert G\right\vert }}\sum_{g\in G}\ket{g}
_{C}\left(  \overline{U}_{R}(g)\otimes U_{B}(g)\right)  \ket{\Phi^{t}}
_{RB}\right)
\end{array}
\right\Vert _{2}^{2}\nonumber\\
& =\left\Vert
\begin{array}
[c]{c}
\left(  \frac{1}{\sqrt{\left\vert G\right\vert }}\sum_{g^{\prime}\in G}\langle
g^{\prime}|_{C}\otimes \mathbb{I}_{RB}\right)  \times\\
\left(  \frac{1}{\sqrt{\left\vert G\right\vert }}\sum_{g\in G}\ket{g}
_{C}\left(  \overline{U}_{R}(g)\otimes U_{B}(g)\right)  \ket{\Phi^{t}}
_{RB}\right)
\end{array}
\right\Vert _{2}^{2}\\
& =\left\Vert \frac{1}{\left\vert G\right\vert }\sum_{g^{\prime},g\in
G}\langle g^{\prime}\ket{g}_{C}\left(  \overline{U}_{R}(g)\otimes
U_{B}(g)\right)  \ket{\Phi^{t}}_{RB}\right\Vert _{2}^{2}\\
& =\left\Vert \frac{1}{\left\vert G\right\vert }\sum_{g\in G}\left(
\overline{U}_{R}(g)\otimes U_{B}(g)\right)  \ket{\Phi^{t}}_{RB}\right\Vert
_{2}^{2}\\
& =\left\Vert \Pi^{G}\ket{\Phi^{t}}_{RB}\right\Vert _{2}^{2}\\
& =\operatorname{Tr}[\Pi^{G}\Phi_{RB}^{t}].
\end{align}

Now we show that 
\begin{equation}\label{eq:start}
    P_{\text{acc}} = \operatorname{Tr}[\Pi^G \Phi_{RB}^{t}]
\end{equation}
is equal to the following expression:
\begin{equation}
    P_{\text{acc}}=\frac{1}{d |G|}\sum_{g \in G}\operatorname{Tr}[U^\dag (g) e^{i H t} U(g) e^{-i H t}]\, .
\end{equation}
To see this, we begin with equation \eqref{eq:start} and note that, using the cyclicity of the trace, it can be rewritten as
\begin{equation}
    \operatorname{Tr}[(\mathbb{I}_R \otimes e^{iHt}) \Pi^G (\mathbb{I}_R \otimes e^{-iHt}) \Phi_{RA}].
\end{equation}
We can now substitute the definition of the projector given in \eqref{eq:projector_app}, giving
\begin{multline}
    \frac{1}{|G|}\operatorname{Tr}[(\mathbb{I}_R \otimes e^{iHt}) \sum_{g \in G} \overline{U}_R(g) \otimes U_A(g) (\mathbb{I}_R \otimes e^{-iHt}) \Phi_{RA}], \\
    = \frac{1}{|G|}\sum_{g \in G} \operatorname{Tr}[(\mathbb{I}_R \otimes e^{iHt}) \overline{U}_R(g) \otimes U_A(g) (\mathbb{I}_R \otimes e^{-iHt}) \Phi_{RA}],
\end{multline}
where the second equality follows from the linearity of the trace. 

We now want to employ the transpose trick:
\begin{equation}
    \mathbb{I}_R \otimes U_A \ket{\Phi}_{RA} = U^{T}_R \otimes \mathbb{I}_A \ket{\Phi}_{RA},
\end{equation}
where $T$ denotes the transpose. The description of this action can be easily interpreted through the language of tensor networks \cite{biamonte2017tensor}. Using this relation, we can rewrite the above as
\begin{equation}
    \frac{1}{|G|}\sum_{g \in G} \operatorname{Tr}[U^\dag_A(g)(\mathbb{I}_R \otimes e^{iHt}) U_A(g) (\mathbb{I}_R \otimes e^{-iHt}) \Phi_{RA}].
\end{equation}

We can now evaluate the trace as a composition of partial traces ($\operatorname{Tr} [\cdot] =\operatorname{Tr}_{RA}[\cdot] = \operatorname{Tr}_A [\operatorname{Tr}_R[\cdot]]$). Computing the trace on $R$ first, we obtain
\begin{equation}
    \frac{1}{d |G|}\sum_{g \in G} \operatorname{Tr}[U^\dag_A(g) e^{iHt} U_A(g) e^{-iHt} ],
\end{equation}
which is exactly  \eqref{eq:modifiedacc}. 

\section{Exact Expansion of the Acceptance Probability of the First (Efficient) Hamiltonian Symmetry Test}

\label{app:exact-expansion}

Here we first prove that the following equality holds:
\begin{equation}
\frac{1}{d\left\vert G\right\vert }\sum_{g\in G}\operatorname{Tr}[U^{\dag}(g)e^{iHt}U(g)e^{-iHt}] = \frac{1}{d}\sum_{n=0}^{\infty}\frac{\left(  -1\right)  ^{n}}{\left(2n \right)  !}t^{2n} f(n,k,H,G),
\end{equation}
where
\begin{equation}
f(n,k,H,G)\coloneqq \sum_{k=0}^{n}\binom{2n}{k}\left(  2-\delta_{k,n}\right)  \left(  -1\right)
^{k}\operatorname{Tr}[\mathcal{T}_{G}(H^{2n-k})H^{k}]
\end{equation}
and the group twirl $\mathcal{T}_{G}$ is defined as
\begin{equation}
\mathcal{T}_{G}(X)\coloneqq \frac{1}{|G|} \sum_{g\in G}U(g)XU^{\dag}(g).
\end{equation}
After that, we establish the expansion in \eqref{eq:acc_prob_beauty}.

Consider that
\begin{align}
\frac{1}{\left\vert G\right\vert }\sum_{g}\operatorname{Tr}[U^{\dag}(g)e^{iHt}U(g)e^{-iHt}] &  =\operatorname{Tr}[\mathcal{T}_{G}(e^{iHt})e^{-iHt}]\\
&  =\sum_{\ell=0}^{\infty}\frac{\left(  it\right)  ^{\ell}}{\ell
!}\operatorname{Tr}[\mathcal{T}_{G}(H^{\ell})e^{-iHt}]\\
&  =\sum_{\ell,m=0}^{\infty}\frac{\left(  it\right)  ^{\ell}\left(
-it\right)  ^{m}}{\ell!m!}\operatorname{Tr}[\mathcal{T}_{G}(H^{\ell})H^{m}]\\
&  =\sum_{\ell,m=0}^{\infty}\frac{\left(  it\right)  ^{\ell+m}\left(
-1\right)  ^{m}}{\ell!m!}\operatorname{Tr}[\mathcal{T}_{G}(H^{\ell})H^{m}]\\
&  =\sum_{n=0}^{\infty}\sum_{k=0}^{n}\frac{\left(  it\right)  ^{n}\left(
-1\right)  ^{k}}{n-k!k!}\operatorname{Tr}[\mathcal{T}_{G}(H^{n-k})H^{k}]\\
&  =\sum_{n=0}^{\infty}\left(  it\right)  ^{n}\sum_{k=0}^{n}\frac{\left(
-1\right)  ^{k}}{n-k!k!}\operatorname{Tr}[\mathcal{T}_{G}(H^{n-k})H^{k}]
\end{align}

Let us consider the term
\begin{equation}
\sum_{k=0}^{n}\frac{\left(  -1\right)  ^{k}}{n-k!k!}\operatorname{Tr}[\mathcal{T}_{G}(H^{n-k})H^{k}].
\end{equation}
Suppose that $n$ is odd. Then consider, with the substitution $\ell=n-k$, that
\begin{align}
  \sum_{k=0}^{n}\frac{\left(  -1\right)  ^{k}}{n-k!k!}\operatorname{Tr} [\mathcal{T}_{G}(H^{n-k})H^{k}] &  =\sum_{k=0}^{\left(  n-1\right)  /2}\frac{\left(  -1\right)  ^{k}}
{n-k!k!}\operatorname{Tr}[\mathcal{T}_{G}(H^{n-k})H^{k}]\nonumber\\
&  \qquad+\sum_{k=\left(  n+1\right)  /2}^{n}\frac{\left(  -1\right)  ^{k}
}{n-k!k!}\operatorname{Tr}[\mathcal{T}_{G}(H^{n-k})H^{k}]\\
&  =\sum_{k=0}^{\left(  n-1\right)  /2}\frac{\left(  -1\right)  ^{k}}
{n-k!k!}\operatorname{Tr}[\mathcal{T}_{G}(H^{n-k})H^{k}]\nonumber\\
&  \qquad+\sum_{\ell=0}^{\left(  n-1\right)  /2}\frac{\left(  -1\right)
^{n-\ell}}{\ell!n-\ell!}\operatorname{Tr}[\mathcal{T}_{G}(H^{\ell})H^{n-\ell
}]\\
&  =\sum_{k=0}^{\left(  n-1\right)  /2}\frac{\left(  -1\right)  ^{k}}
{n-k!k!}\operatorname{Tr}[\mathcal{T}_{G}(H^{n-k})H^{k}]\nonumber\\
&  \qquad+\sum_{\ell=0}^{\left(  n-1\right)  /2}\frac{\left(  -1\right)
^{n-\ell}}{\ell!n-\ell!}\operatorname{Tr}[\mathcal{T}_{G}(H^{n-\ell})H^{\ell
}]\\
&  =\sum_{k=0}^{\left(  n-1\right)  /2}\frac{\left(  -1\right)  ^{k}}
{n-k!k!}\operatorname{Tr}[\mathcal{T}_{G}(H^{n-k})H^{k}]\nonumber\\
&  \qquad+\sum_{k=0}^{\left(  n-1\right)  /2}\frac{\left(  -1\right)  ^{n-k}
}{k!n-k!}\operatorname{Tr}[\mathcal{T}_{G}(H^{n-k})H^{k}]\\
&  =\sum_{k=0}^{\left(  n-1\right)  /2}\frac{\left(  -1\right)  ^{k}+\left(
-1\right)  ^{n-k}}{n-k!k!}\operatorname{Tr}[\mathcal{T}_{G}(H^{n-k})H^{k}]\\
&  =0.
\end{align}
The second-to-last line follows from the fact that the twirl is its own
adjoint and from cyclicity of trace. For the last line, consider that $\left(
-1\right)  ^{k}+\left(  -1\right)  ^{n-k}=0$ for all $k\in\{0,\ldots, (n-1)/2\}$ when $n$ is odd.

Suppose instead that $n$ is even. Then setting $\ell=n-k$ we find that
\begin{align}
 &\sum_{k=0}^{n}\frac{\left(  -1\right)  ^{k}}{n-k!k!}\operatorname{Tr}[\mathcal{T}_{G}(H^{n-k})H^{k}]\nonumber\\
 &=\sum_{k=0}^{n/2}\frac{\left(  -1\right)  ^{k}}{n-k!k!}\operatorname{Tr}[\mathcal{T}_{G}(H^{n-k})H^{k}] +\sum_{k=n/2+1}^{n}\frac{\left(  -1\right)  ^{k}}{n-k!k!} \operatorname{Tr}[\mathcal{T}_{G}(H^{n-k})H^{k}]\\
&  =\sum_{k=0}^{n/2}\frac{\left(  -1\right)  ^{k}}{n-k!k!}\operatorname{Tr}
[\mathcal{T}_{G}(H^{n-k})H^{k}] +\sum_{\ell=0}^{n/2-1}\frac{\left(  -1\right)  ^{n-\ell}}{\ell
!n-\ell!}\operatorname{Tr}[\mathcal{T}_{G}(H^{\ell})H^{n-\ell}]\\
&  =\sum_{k=0}^{n/2}\frac{\left(  -1\right)  ^{k}}{n-k!k!}\operatorname{Tr}
[\mathcal{T}_{G}(H^{n-k})H^{k}] +\sum_{k=0}^{n/2-1}\frac{\left(  -1\right)  ^{n-k}}{k!n-k!}
\operatorname{Tr}[\mathcal{T}_{G}(H^{k})H^{n-k}]\\
&  =\sum_{k=0}^{n/2}\frac{\left(  -1\right)  ^{k}}{n-k!k!}\operatorname{Tr}
[\mathcal{T}_{G}(H^{n-k})H^{k}] +\sum_{k=0}^{n/2-1}\frac{\left(  -1\right)  ^{n-k}}{k!n-k!}
\operatorname{Tr}[\mathcal{T}_{G}(H^{n-k})H^{k}]\\
&  =\frac{\left(  -1\right)  ^{n/2}}{\left(  n/2!\right)  ^{2}}
\operatorname{Tr}[\mathcal{T}_{G}(H^{n/2})H^{n/2}] +\sum_{k=0}^{n/2-1}\frac{\left(  -1\right)  ^{k}}{n-k!k!}\operatorname{Tr}[\mathcal{T}_{G}(H^{n-k})H^{k}]\nonumber\\
&  \qquad+\sum_{k=0}^{n/2-1}\frac{\left(  -1\right)  ^{n-k}}{k!n-k!}
\operatorname{Tr}[\mathcal{T}_{G}(H^{n-k})H^{k}]\\
&  =\frac{\left(  -1\right)  ^{n/2}}{\left(  n/2!\right)  ^{2}}
\operatorname{Tr}[\mathcal{T}_{G}(H^{n/2})H^{n/2}] +\sum_{k=0}^{n/2-1}\frac{\left(  -1\right)  ^{k}+\left(  -1\right)^{n-k}}{n-k!k!}\operatorname{Tr}[\mathcal{T}_{G}(H^{n-k})H^{k}]\\
&  =\frac{\left(  -1\right)  ^{n/2}}{\left(  n/2!\right)  ^{2}} \operatorname{Tr}[\mathcal{T}_{G}(H^{n/2})H^{n/2}] +\sum_{k=0}^{n/2-1}\frac{2\left(  -1\right)  ^{k}}{n-k!k!}
\operatorname{Tr}[\mathcal{T}_{G}(H^{n-k})H^{k}] \\
&=\sum_{k=0}^{n/2}\frac{\left(  2-\delta_{k,n/2}\right)  \left(  -1\right)^{k}}{n-k!k!}\operatorname{Tr}[\mathcal{T}_{G}(H^{n-k})H^{k}]\, .
\end{align}
Then the overall formula is given by
\begin{align}
 & \frac{1}{d\left\vert G\right\vert }\sum_{g\in G}\operatorname{Tr}[U^{\dag}(g)e^{iHt}U(g)e^{-iHt}] \nonumber \\
 &  =\frac{1}{d}\sum_{n=0}^{\infty}\left(  -1\right)  ^{n}t^{2n}
  \sum_{k=0}^{n}\frac{\left(  2-\delta_{k,n}\right)  \left(  -1\right) ^{k}}{2n-k!k!}\operatorname{Tr}[\mathcal{T}_{G}(H^{2n-k})H^{k}]\\
&  =\frac{1}{d}\sum_{n=0}^{\infty}\frac{\left(  -1\right)  ^{n}}{\left(2n\right)  !}t^{2n} \sum_{k=0}^{n}\binom{2n}{k}\left(  2-\delta_{k,n}\right)  \left( -1\right)  ^{k}\operatorname{Tr}[\mathcal{T}_{G}(H^{2n-k})H^{k}].
\end{align}

Now let us establish the expansion in \eqref{eq:acc_prob_beauty}. By applying the Baker--Campbell--Hausdorff formula and the nested commutator
in \eqref{eq:def-nested-comm}, consider that
\begin{align}
\operatorname{Tr}[U^{\dag}(g)e^{iHt}U(g)e^{-iHt}] & =\operatorname{Tr}\left[  U^{\dag}(g)\sum_{n=0}^{\infty}\frac{\left[ \left(  iHt\right)  ^{n},U(g)\right]  }{n!}\right]  \\
& =\sum_{n=0}^{\infty}\frac{\left(  it\right)  ^{n}}{n!}\operatorname{Tr}
\left[  U^{\dag}(g)\left[  \left(  H\right)  ^{n},U(g)\right]  \right]  
\end{align}
As derived above, it is only necessary to consider even powers in $t$ when including the sum over $g\in G$, and so we consider the following:
\begin{equation}
 =\sum_{n=0}^{\infty}\frac{\left(  it\right)  ^{2n}}{2n!}\operatorname{Tr}
\left[  U^{\dag}(g)\left[  \left(  H\right)  ^{2n},U(g)\right]  \right]  
 =\sum_{n=0}^{\infty}\frac{\left(  -1\right)  ^{n}t^{2n}}{2n!}
\operatorname{Tr}\left[  U^{\dag}(g)\left[  \left(  H\right)  ^{2n}
,U(g)\right]  \right]  .
\end{equation}

Then we find that
\begin{align}
 &\operatorname{Tr}\left[  U^{\dag}(g)\left[  \left(  H\right)  ^{2n} ,U(g)\right]  \right] \nonumber  \\
& =\operatorname{Tr}\left[  U^{\dag}(g)\left[  H,\left[  \left(  H\right)
^{2n-1},U(g)\right]  \right]  \right]  \\
& =\operatorname{Tr}\left[  U^{\dag}(g)\left(  H\left[  \left(  H\right)
^{2n-1},U(g)\right]  -\left[  \left(  H\right)  ^{2n-1},U(g)\right]  H\right)
\right]  \\
& =\operatorname{Tr}\left[  \left(  U^{\dag}(g)H-HU^{\dag}(g)\right)  \left[
\left(  H\right)  ^{2n-1},U(g)\right]  \right]  \\
& =\operatorname{Tr}\left[  \left[  U^{\dag}(g),H\right]  \left[  \left(
H\right)  ^{2n-1},U(g)\right]  \right]  \\
& =\operatorname{Tr}\left[  \left[  U^{\dag}(g),H\right]  \left(  H\left[
\left(  H\right)  ^{2n-2},U(g)\right]  -\left[  \left(  H\right)
^{2n-2},U(g)\right]  H\right)  \right]  \\
& =\operatorname{Tr}\left[  \left[  \left[  U^{\dag}(g),H\right]  ,H\right]
\left[  \left(  H\right)  ^{2n-2},U(g)\right]  \right]  \\
& =\operatorname{Tr}\left[  \left[  \left[  U^{\dag}(g),H\right]  ,H\right]
\left(  H\left[  \left(  H\right)  ^{2n-3},U(g)\right]  -\left[  \left(
H\right)  ^{2n-3},U(g)\right]  H\right)  \right]  \\
& =\operatorname{Tr}\left[  \left[  \left[  \left[  U^{\dag}(g),H\right]
,H\right]  ,H\right]  \left[  \left(  H\right)  ^{2n-3},U(g)\right]  \right]
\\
& =\operatorname{Tr}\left[  \left[  U^{\dag}(g),\left(  H\right)  ^{n}\right]
\left[  \left(  H\right)  ^{n},U(g)\right]  \right]  \\
& =\operatorname{Tr}\left[  \left(  \left[  \left(  H\right)  ^{n}
,U(g)\right]  \right)  ^{\dag}\left[  \left(  H\right)  ^{n},U(g)\right]
\right]  \\
& =\Big\Vert \left[  \left(  H\right)  ^{n},U(g)\right]  \Big\Vert _{2}
^{2}.
\end{align}

The third-to-last line follows from induction and the second-to-last from the
fact that
\begin{equation}
\left[  Y^{\dag},\left(  X\right)  ^{n}\right]  ^{\dag}=\left[  \left(
X\right)  ^{n},Y\right]  ,
\label{eq:cascade-dagger}
\end{equation}
for Hermitian $X$ and by using the convention that
\begin{align}
\left[  Y^{\dag},\left(  X\right)  ^{n}\right]    & \equiv[\cdots
\lbrack\lbrack Y^{\dag},\underbrace{X],X]\cdots,X}_{n\text{ times}}],\\
\left[  Y^{\dag},(X)^{0}\right]    & \equiv Y^{\dag}.
\end{align}
Eq.~\eqref{eq:cascade-dagger} follows from applying $[A,B]^\dag = [B^\dag,A^\dag]$ inductively.
Plugging back in above, we find that
\begin{equation}
\frac{1}{d\left\vert G\right\vert }\sum_{g\in G}\operatorname{Tr}[U^{\dag
}(g)e^{iHt}U(g)e^{-iHt}]
=\sum_{n=0}^{\infty}\frac{\left(  -1\right)  ^{n}t^{2n}}{d\left\vert
G\right\vert \left(  2n!\right)  }\sum_{g\in G}\Big\Vert \left[  \left(
H\right)  ^{n},U(g)\right]  \Big\Vert _{2}^{2}.
\end{equation}

\section{Derivation of Acceptance Probability of the Second (Variational) Hamiltonian Symmetry Test}

\label{app:variational-sym-test}

Here we present an alternative derivation of \eqref{eq:modifiedacc}, as well as a derivation of
\eqref{eq:accept-prob-optimized} and \eqref{eq:lower-bnd-est-small-t}. Suppose that the input to the circuit in Figure~\ref{fig:altcircuit} is a pure state
$\ket{\psi}$, rather than the maximally mixed state. Then the initial state
of the algorithm is given by
\begin{equation}
\frac{1}{\sqrt{\left\vert G\right\vert }}\sum_{g\in G}\ket{g}_{C}
\ket{\psi}.
\end{equation}
After the first controlled unitary, the Hamiltonian evolution $e^{-iHt}$, and
the second controlled unitary, the state becomes
\begin{equation}
\frac{1}{\sqrt{\left\vert G\right\vert }}\sum_{g\in G}\ket{g}_{C}
U(g)e^{-iHt}U^{\dag}(g)\ket{\psi}.
\end{equation}
The acceptance probability is then given by
\begin{align}
& \left\Vert
\begin{array}
[c]{c}
\left(  \bra{+}_{C}\otimes \mathbb{I}\right)  \times\\
\left(  \frac{1}{\sqrt{\left\vert G\right\vert }}\sum_{g\in G}\ket{g}
_{C}U(g)e^{-iHt}U^{\dag}(g)\ket{\psi}\right)
\end{array}
\right\Vert _{2}^{2}\nonumber\\
& =\left\Vert
\begin{array}
[c]{c}
\left(  \frac{1}{\sqrt{\left\vert G\right\vert }}\sum_{g^{\prime}\in G}\langle
g^{\prime}|_{C}\otimes \mathbb{I}\right)  \times\\
\left(  \frac{1}{\sqrt{\left\vert G\right\vert }}\sum_{g\in G}\ket{g}
_{C}U(g)e^{-iHt}U^{\dag}(g)\ket{\psi}\right)
\end{array}
\right\Vert _{2}^{2}\\
& =\left\Vert \frac{1}{\left\vert G\right\vert }\sum_{g^{\prime},g\in
G}\langle g^{\prime}\ket{g}_{C}U(g)e^{-iHt}U^{\dag}(g)\ket{\psi}
\right\Vert _{2}^{2}\\
& =\left\Vert \frac{1}{\left\vert G\right\vert }\sum_{g\in G}U(g)e^{-iHt}
U^{\dag}(g)\ket{\psi}\right\Vert _{2}^{2}.
\end{align}

First let us suppose that the maximally mixed state is input. This is
equivalent to picking a pure state $|\psi_{x}\rangle$ from an orthonormal
basis, with probability $1/d$. Then in this case, the acceptance probability
is given by
\begin{align}
& \frac{1}{d}\sum_{x=1}^{d}\left\Vert \frac{1}{\left\vert G\right\vert }
\sum_{g\in G}U(g)e^{-iHt}U^{\dag}(g)|\psi_{x}\rangle\right\Vert _{2}
^{2}\nonumber\\
& =\frac{1}{d}\sum_{x=1}^{d}\left(  \frac{1}{\left\vert G\right\vert }
\sum_{g^{\prime}\in G}\langle\psi_{x}|U(g^{\prime})e^{iHt}U^{\dag}(g^{\prime
})\right)  \left(  \frac{1}{\left\vert G\right\vert }\sum_{g\in G}U(g)e^{-iHt}
U^{\dag}(g)|\psi_{x}\rangle\right)  \\
& =\frac{1}{d\left\vert G\right\vert ^{2}}\sum_{x=1}^{d}\sum_{g^{\prime},g\in
G}\langle\psi_{x}|U(g^{\prime})e^{iHt}U^{\dag}(g^{\prime}) U(g)e^{-iHt}U^{\dag}(g)|\psi_{x}\rangle\\
& =\frac{1}{d\left\vert G\right\vert ^{2}}\sum_{x=1}^{d}\sum_{g^{\prime},g\in
G}\operatorname{Tr}[U^{\dag}(g)|\psi_{x}\rangle\langle\psi_{x}| U(g^{\prime})e^{iHt}U^{\dag}(g^{\prime})U(g)e^{-iHt}]\\
& =\frac{1}{d\left\vert G\right\vert ^{2}}\sum_{g^{\prime},g\in G}
\operatorname{Tr}\left[  U^{\dag}(g)U(g^{\prime})e^{iHt}U^{\dag}(g^{\prime
})U(g)e^{-iHt}\right]  \\
& =\frac{1}{d\left\vert G\right\vert ^{2}}\sum_{g^{\prime},g\in G}
\operatorname{Tr}\left[  U^{\dag}(g^{\prime-1}\circ g)e^{iHt}U(g^{\prime
-1}\circ g)e^{-iHt}\right]  \\
& =\frac{1}{d\left\vert G\right\vert }\sum_{g\in G}\operatorname{Tr}\left[
U^{\dag}(g)e^{iHt}U(g)e^{-iHt}\right]  .
\end{align}
The second-to-last equality follows from the group property and the fact that $U(g)$ is a representation of $g$. So this provides
an alternate proof of \eqref{eq:modifiedacc}.

Now let us prove the expansion in \eqref{eq:accept-prob-fixed-state}. Consider that
\begin{align}
&  \left\Vert \frac{1}{\left\vert G\right\vert }\sum_{g\in G}U(g)e^{-iHt}
U^{\dag}(g)\ket{\psi}\right\Vert _{2}^{2}\nonumber\\
&  =\left\Vert \mathcal{T}_{G}(e^{-iHt})\ket{\psi}\right\Vert _{2}^{2}\\
&  =\langle\psi|\mathcal{T}_{G}(e^{iHt})\mathcal{T}_{G}(e^{-iHt})|\psi
\rangle\\
&  =\langle\psi|\mathcal{T}_{G}\left(  \mathbb{I}+iHt-H^{2}t^{2}/2+O(\tau^{3})\right)
\mathcal{T}_{G}\left(  \mathbb{I}-iHt-H^{2}t^{2}/2+O(\tau^{3})\right)
\ket{\psi}\\
&  =\langle\psi|\left(  \mathbb{I}+it\mathcal{T}_{G}(H)-\left(  t^{2}/2\right)
\mathcal{T}_{G}(H^{2})+O(\tau^{3})\right)  \left(  \mathbb{I}-it\mathcal{T}_{G}(H)-\left(  t^{2}/2\right)  \mathcal{T}
_{G}(H^{2})+O(\tau^{3})\right)  \ket{\psi}\\
&  =1+t^{2}\langle\psi|\left(  \mathcal{T}_{G}(H)\right)  ^{2}|\psi
\rangle-t^{2}\langle\psi|\mathcal{T}_{G}(H^{2})\ket{\psi}+O(\tau^{3})\\
&  =1-t^{2}\langle\psi|\left(  \mathcal{T}_{G}(H^{2})-\left(  \mathcal{T}
_{G}(H)\right)  ^{2}\right)  \ket{\psi}+O(\tau^{3})\\
&  =1-t^{2}\left\langle \mathcal{T}_{G}(H^{2})-\left(  \mathcal{T}
_{G}(H)\right)  ^{2}\right\rangle _{\psi}+O(\tau^{3}).
\end{align}
The Kadison--Schwarz inequality \cite[Theorem~2.3.2]{bhatia07positivedefinitematrices} implies the following operator inequality:
\begin{equation}
    \mathcal{T}_{G}(H^{2}) \geq \left(  \mathcal{T}
_{G}(H)\right)  ^{2}.
\end{equation}
As a consequence, the following inequality holds for every state $\ket{\psi}$:
\begin{equation}
    \left\langle \mathcal{T}_{G}(H^{2})-\left(  \mathcal{T}
_{G}(H)\right)  ^{2}\right\rangle_{\psi} \geq 0.
\end{equation}

If we perform a maximization of the acceptance probability over every
input state $\ket{\psi}$, then it is equal to
\begin{align}
& \max_{\ket{\psi}: \left\Vert \ket{\psi}\right\Vert_2=1}\left\Vert \frac{1}{\left\vert G\right\vert }\sum_{g\in
G}U(g)e^{-iHt}U^{\dag}(g)\ket{\psi}\right\Vert _{2}^{2}\nonumber\\
& =\left\Vert \frac{1}{\left\vert G\right\vert }\sum_{g\in G}U(g)e^{-iHt}
U^{\dag}(g)\right\Vert _{\infty}^{2}\\
& =\left\Vert \frac{1}{\left\vert G\right\vert }\sum_{g\in G}\left(  \left[
U(g),e^{-iHt}\right]  +e^{-iHt}U(g)\right)  U^{\dag}(g)\right\Vert _{\infty
}^{2}\\
& =\left\Vert \frac{1}{\left\vert G\right\vert }\sum_{g\in G}\left(  \left[
U(g),e^{-iHt}\right]  U^{\dag}(g)+e^{-iHt}\right)  \right\Vert _{\infty}
^{2}\\
& =\left\Vert e^{-iHt}+\frac{1}{\left\vert G\right\vert }\sum_{g\in G}\left[
U(g),e^{-iHt}\right]  U^{\dag}(g)\right\Vert _{\infty}^{2}\\
& \geq\left(  \left\Vert e^{-iHt}\right\Vert _{\infty}-\left\Vert \frac
{1}{\left\vert G\right\vert }\sum_{g\in G}\left[  U(g),e^{-iHt}\right]
U^{\dag}(g)\right\Vert _{\infty}\right)  ^{2}\\
& =\left(  1-\left\Vert \frac{1}{\left\vert G\right\vert }\sum_{g\in G}\left[
U(g),e^{-iHt}\right]  U^{\dag}(g)\right\Vert _{\infty}\right)  ^{2}\\
& \geq\left(  1-\frac{1}{\left\vert G\right\vert }\sum_{g\in G}\left\Vert
\left[  U(g),e^{-iHt}\right]  U^{\dag}(g)\right\Vert _{\infty}\right)  ^{2}\\
& =\left(  1-\frac{1}{\left\vert G\right\vert }\sum_{g\in G}\left\Vert \left[
U(g),e^{-iHt}\right]  \right\Vert _{\infty}\right)  ^{2}\\
& \geq1-\frac{2}{\left\vert G\right\vert }\sum_{g\in G}\left\Vert \left[
U(g),e^{-iHt}\right]  \right\Vert _{\infty}.
\end{align}
The first inequality follows from the reverse triangle inequality. The next
equality follows because $\left\Vert e^{-iHt}\right\Vert _{\infty}=1$. The
second inequality follows from the triangle inequality. The final equality
follows from the unitary invariance of the spectral norm. Thus we have established \eqref{eq:accept-prob-optimized}.

Now suppose that $\left\Vert H\right\Vert _{\infty}t<1$. Then we find that
\begin{align}
& \left\Vert \left[  U(g),e^{-iHt}\right]  \right\Vert _{\infty}\notag \\
& =\left\Vert \left[  U(g),\mathbb{I}-iHt+\sum_{n=2}^{\infty}\frac{\left(  -iHt\right)
^{n}}{n!}\right]  \right\Vert _{\infty}\\
& =\left\Vert -it\left[  U(g),H\right]  +\left[  U(g),\sum_{n=2}^{\infty}
\frac{\left(  -iHt\right)  ^{n}}{n!}\right]  \right\Vert _{\infty}\\
& \leq t\left\Vert \left[  U(g),H\right]  \right\Vert _{\infty}+\left\Vert
\left[  U(g),\sum_{n=2}^{\infty}\frac{\left(  -iHt\right)  ^{n}}{n!}\right]
\right\Vert _{\infty}\\
& \leq t\left\Vert \left[  U(g),H\right]  \right\Vert _{\infty}+2\left\Vert
\sum_{n=2}^{\infty}\frac{\left(  -iHt\right)  ^{n}}{n!}\right\Vert _{\infty
}\\
& \leq t\left\Vert \left[  U(g),H\right]  \right\Vert _{\infty}+2\sum
_{n=2}^{\infty}\frac{\left(  \left\Vert H\right\Vert _{\infty}t\right)  ^{n}
}{n!}\\
& \leq t\left\Vert \left[  U(g),H\right]  \right\Vert _{\infty}+2\left(
\left\Vert H\right\Vert _{\infty}t\right)  ^{2}\sum_{n=2}^{\infty}\frac{1}
{n!}\\
& =t\left\Vert \left[  U(g),H\right]  \right\Vert _{\infty}+2\left(
\left\Vert H\right\Vert _{\infty}t\right)  ^{2}\left(  e-2\right)  \\
& \leq t\left\Vert \left[  U(g),H\right]  \right\Vert _{\infty}+2\left\Vert
H\right\Vert _{\infty}^{2}t^{2},
\end{align}
where the second-to-last inequality follows from the assumption that $\left\Vert H\right\Vert _{\infty}t<1$.
This implies that
\begin{equation}
\frac{1}{\left\vert G\right\vert ^{2}}\max_{\ket{\psi}}\left\Vert \sum_{g\in
G}U(g)e^{-iHt}U^{\dag}(g)\ket{\psi}\right\Vert _{2}^{2}
\geq1-\frac{2t}{\left\vert G\right\vert }\sum_{g\in G}\left\Vert \left[
U(g),H\right]  \right\Vert _{\infty}-4\left\Vert H\right\Vert _{\infty}
^{2}t^{2},
\end{equation}
thus establishing \eqref{eq:lower-bnd-est-small-t}.

We now prove \eqref{eq:qma-nested-comm-bnd}. Consider that
\begin{align}
& \left\Vert \frac{1}{\left\vert G\right\vert }\sum_{g\in G}U(g)e^{-iHt}
U^{\dag}(g)\right\Vert _{\infty}^{2}\nonumber\\
& =\left\Vert \frac{1}{\left\vert G\right\vert }\sum_{g\in G}e^{iHt}
U(g)e^{-iHt}U^{\dag}(g)\right\Vert _{\infty}^{2}\\
& =\left\Vert \frac{1}{\left\vert G\right\vert }\sum_{g\in G}\sum
_{n=0}^{\infty}\frac{\left[  \left(  iHt\right)  ^{n},U(g)\right]  }
{n!}U^{\dag}(g)\right\Vert _{\infty}^{2}\\
& =\left\Vert \sum_{n=0}^{\infty}\frac{\left(  it\right)  ^{n}}{n!}\frac
{1}{\left\vert G\right\vert }\sum_{g\in G}\left[  \left(  H\right)
^{n},U(g)\right]  U^{\dag}(g)\right\Vert _{\infty}^{2}\\
& =\left\Vert \mathbb{I}+\sum_{n=1}^{\infty}\frac{\left(  it\right)  ^{n}}{n!}\frac
{1}{\left\vert G\right\vert }\sum_{g\in G}\left[  \left(  H\right)
^{n},U(g)\right]  U^{\dag}(g)\right\Vert _{\infty}^{2}\\
& \geq\left(  \left\Vert \mathbb{I}\right\Vert _{\infty}-\left\Vert \sum_{n=1}^{\infty
}\frac{\left(  it\right)  ^{n}}{n!}\frac{1}{\left\vert G\right\vert }
\sum_{g\in G}\left[  \left(  H\right)  ^{n},U(g)\right]  U^{\dag
}(g)\right\Vert _{\infty}\right)  ^{2}\\
& =\left(  1-\left\Vert \sum_{n=1}^{\infty}\frac{\left(  it\right)  ^{n}}
{n!}\frac{1}{\left\vert G\right\vert }\sum_{g\in G}\left[  \left(  H\right)
^{n},U(g)\right]  U^{\dag}(g)\right\Vert _{\infty}\right)  ^{2}\\
& \geq\left(  1-\sum_{n=1}^{\infty}\frac{t^{n}}{n!}\frac{1}{\left\vert
G\right\vert }\sum_{g\in G}\left\Vert \left[  \left(  H\right)  ^{n}
,U(g)\right]  U^{\dag}(g)\right\Vert _{\infty}\right)  ^{2}\\
& =\left(  1-\sum_{n=1}^{\infty}\frac{t^{n}}{n!}\frac{1}{\left\vert
G\right\vert }\sum_{g\in G}\left\Vert \left[  \left(  H\right)  ^{n}
,U(g)\right]  \right\Vert _{\infty}\right)  ^{2}.
\end{align}
In the above, we employed unitary invariance of the spectral norm, the Baker--Campbell--Hausdorff formula, and the triangle inequality.

\pagebreak
\chapter{Supplementary Material for Chapter 3}
\doublespacing
\section{Proof of Theorem~\ref{thm:acc-prob-g-Bose-sym}}

\label{app:acc-prob-g-Bose-sym}

Let $\psi_{RS}$ be an arbitrary purification of $\rho_{S}$, and consider that
\begin{align}
\operatorname{Tr}[\Pi_{S}^{G}\rho_{S}]  & =\operatorname{Tr}[(I_{R}\otimes \Pi_{S}^{G})\psi_{RS}]\\
& =\left\Vert \left(  I_{R}\otimes\Pi_{S}^{G}\right) \ket{\psi}_{RS}\right\Vert _{2}^{2}.
\end{align}
Recall the following property of the norm of an arbitrary vector $\ket{\varphi}$:
\begin{equation}
\left\Vert \ket{\varphi}\right\Vert _{2}^{2}= \max_{\ket{\phi}:\left\Vert \ket{\phi}\right\Vert _{2}=1}\left\vert \langle\phi\ket{\varphi} \right\vert ^{2}.
\end{equation}
This follows from the Cauchy--Schwarz inequality and the conditions for saturating it. This implies that
\begin{equation}
 \left\Vert \left(  I_{R}\otimes\Pi_{S}^{G}\right)  \ket{\psi}_{RS}\right\Vert _{2}^{2} =\max_{\ket{\phi}:\left\Vert \ket{\phi}\right\Vert _{2}=1}\left\vert \bra{\phi}_{RS}\left(  I_{R}\otimes\Pi_{S}^{G}\right)  \ket{\psi}_{RS}\right\vert ^{2}
\end{equation}
Let us also recall Uhlmann's theorem \cite{U76}: For positive semi-definite operators $\omega_{A}$ and $\tau_{A}$ and corresponding rank-one operators $\psi_{RA}^{\omega}$ and $\psi_{RA}^{\tau}$ satisfying
\begin{align}
\operatorname{Tr}_{R}[\psi_{RA}^{\omega}]  &  =\omega_{A},\label{eq:uhlmann-thm-1}\\
\operatorname{Tr}_{R}[\psi_{RA}^{\tau}]  &  =\tau_{A},
\end{align}
Uhlmann's theorem \cite{U76} states that
\begin{align}
F(\omega_A,\tau_A) & = \left\Vert \sqrt{\omega_{A}}\sqrt{\tau_{A}}\right\Vert _{1}^{2}\\
& =\max_{V_{R}}\left\vert \bra{\psi^{\omega}}_{RA}\left(  V_{R}\otimes I_{A}\right)
\ket{\psi^{\tau}}_{RA}\right\vert ^{2},
\label{eq:uhlmann-thm-last}
\end{align}
where the optimization is over every unitary $V_{R}$ acting on the reference system $R$. We also implicitly defined fidelity more generally for positive semi-definite operators. Considering that
\begin{equation}
\rho_S  = \operatorname{Tr}_R[\psi_{RS}], \qquad  \sigma_S  \coloneqq \operatorname{Tr}_R[  \phi_{RS} ],
\end{equation}
so that 
\begin{equation}
    \Pi^G_S \sigma_S\Pi^G_S = \operatorname{Tr}_R[ \Pi^G_S \phi_{RS} \Pi^G_S],
\end{equation}
we conclude that
\begin{align}
 \max_{\ket{\phi}:\left\Vert \ket{\phi}\right\Vert _{2}=1}\left\vert \bra{\phi}_{RS}\left(  I_{R}\otimes\Pi_{S}^{G}\right)  \ket{\psi}_{RS}\right\vert ^{2} \notag & =  \max_{\ket{\phi}:\left\Vert \ket{\phi}\right\Vert _{2}=1}\max_{U_R}\left\vert \bra{\phi}_{RS}\left(  U_{R}\otimes\Pi_{S}^{G}\right)  \ket{\psi}_{RS}\right\vert ^{2} \\
& =\max_{\sigma_{S}\in\mathcal{D}(\mathcal{H}_{S})}F(\rho_{S},\Pi_{S}^{G}\sigma_{S}\Pi_{S}^{G}).
\end{align}
where the last equality follows from Uhlmann's theorem with the identifications $\ket{\psi^{\omega}} \leftrightarrow (I \otimes \Pi^G ) \ket{\phi}$ and $\ket{\psi^{\tau}} \leftrightarrow  \ket{\psi}$. Clearly, we have that
\begin{align}
\max_{\sigma_{S}\in\mathcal{D}(\mathcal{H}_{S})}F(\rho_{S},\Pi_{S}^{G}
\sigma_{S}\Pi_{S}^{G})  \notag & \geq\max_{\sigma\in\text{B-Sym}_{G}}F(\rho_{S},\Pi_{S}^{G}\sigma_{S}\Pi_{S}^{G})\\
& =\max_{\sigma\in\text{B-Sym}_{G}}F(\rho_{S},\sigma_{S}),
\end{align}
because B-Sym$_{G}\subset\mathcal{D}(\mathcal{H})$. Now let us consider showing the opposite inequality. Let $\sigma\in\mathcal{D}(\mathcal{H})$. If $\Pi^{G}\sigma\Pi^{G}=0$, then this is a suboptimal choice as it follows that the objective function $F(\rho_{S},\Pi_{S}^{G}\sigma_{S}\Pi_{S}^{G})=0$ in this case. So, let us suppose this is not the case. Then define
\begin{align}
\sigma^{\prime}  & \coloneqq \frac{1}{p}\Pi^{G}\sigma\Pi^{G},\\
p  & \coloneqq \operatorname{Tr}[\Pi^{G}\sigma],
\end{align}
and observe that $\sigma_{S}^{\prime}\in$B-Sym$_{G}$. Consider that
\begin{align}
F(\rho_{S},\Pi_{S}^{G}\sigma_{S}\Pi_{S}^{G})  & =pF(\rho_{S},\sigma_{S}^{\prime})\\
& \leq F(\rho_{S},\sigma_{S}^{\prime})\\
& \leq\max_{\sigma_{S}\in\text{B-Sym}_{G}}F(\rho_{S},\sigma_{S}).
\end{align}
We have thus proved the opposite inequality, concluding the proof.

\section{Proof of Theorem~\ref{thm:G-BSE-acc-prob}}

\label{app:proof-thm-g-bse}Following the same reasoning given in \eqref{eq:euclidean-norm-opt}--\eqref{eq:1st-fid-formula-pf}, by using Uhlmann's theorem, we conclude that
\begin{equation}
\max_{V_{S^{\prime}E\rightarrow RE^{\prime}}}\left\Vert \Pi_{RS}^{G}V_{S^{\prime}E\rightarrow RE^{\prime}}\ket{\psi}_{S^{\prime}S} \ket{0}_{E}\right\Vert _{2}^{2} = \max_{\sigma_{RS}} F(\rho_{S},\operatorname{Tr}_{R}[\Pi_{RS}^{G}\sigma_{RS} \Pi_{RS}^{G}]),
\end{equation}
where the optimization is over every state $\sigma_{RS}$ and $\Pi_{RS}^{G}$ is defined in \eqref{eq:Pi_RS-proj-again}. The next part of the proof shows that
\begin{equation}
\max_{\sigma_{RS}}F(\rho_{S},\operatorname{Tr}_{R}[\Pi_{RS}^{G}\sigma_{RS} \Pi_{RS}^{G}])=\max_{\sigma_{S}\in\operatorname*{BSE}_{G}}F(\rho_{S},\sigma_{S})
\end{equation}
and is similar to \eqref{eq:final-eq-steps-1}--\eqref{eq:final-eq-steps-last}. To justify the inequality$~\geq$, let $\sigma_{S}$ be an arbitrary state in $\operatorname*{BSE}_{G}$. Then by Definition~\ref{def:g-bose-sym-ext}, this means that there exists a state $\omega_{RS}$ such that $\operatorname{Tr}_{R}[\omega_{RS}]=\sigma_{S}$ and $\Pi_{RS}^{G}\omega_{RS}\Pi_{RS}^{G} =\omega_{RS}$. This means that
\begin{align}
F(\rho_{S},\sigma_{S})  &  =F(\rho_{S},\operatorname{Tr}_{R}[\omega_{RS}])\\
&  =F(\rho_{S},\operatorname{Tr}_{R}[\Pi_{RS}^{G}\omega_{RS}\Pi_{RS}^{G}])\\
&  \leq\max_{\sigma_{RS}}F(\rho_{S},\operatorname{Tr}_{R}[\Pi_{RS}^{G} \sigma_{RS}\Pi_{RS}^{G}]),
\end{align}
which implies that
\begin{equation}
\max_{\sigma_{RS}}F(\rho_{S},\operatorname{Tr}_{R}[\Pi_{RS}^{G}\sigma_{RS}\Pi_{RS}^{G}]) \geq \max_{\sigma_{S}\in\operatorname*{BSE}_{G}}F(\rho_{S},\sigma_{S})
\end{equation}
To justify the inequality~$\leq$, let $\sigma_{RS}$ be an arbitrary state. If $\Pi_{RS}^{G}\sigma_{RS}\Pi_{RS}^{G}=0$, then the desired inequality trivially follows. Supposing then that this is not the case, let us define
\begin{align}
\sigma_{RS}^{\prime}  &  \coloneqq \frac{1}{p}\Pi_{RS}^{G}\sigma_{RS}\Pi_{RS}^{G},\\
p  &  \coloneqq \operatorname{Tr}[\Pi_{RS}^{G}\sigma_{RS}].
\end{align}
We then find that
\begin{align}
F(\rho_{S},\operatorname{Tr}_{R}[\Pi_{RS}^{G}\sigma_{RS}\Pi_{RS}^{G}])  & =pF(\rho_{S}, \operatorname{Tr}_{R}[\sigma_{RS}^{\prime}])\\
&  \leq F(\rho_{S},\operatorname{Tr}_{R}[\sigma_{RS}^{\prime}]).
\end{align}
Consider that $\sigma_{S}^{\prime}\coloneqq \operatorname{Tr}_{R}[\sigma_{RS}^{\prime}]$ is $G$-Bose symmetric extendible because $\sigma_{RS}^{\prime}$ is an extension of it that satisfies $\Pi_{RS}^{G} \sigma_{RS}^{\prime} \Pi_{RS}^{G} = \sigma_{RS}^{\prime}$. We conclude that
\begin{equation}
F(\rho_{S},\operatorname{Tr}_{R}[\Pi_{RS}^{G}\sigma_{RS}\Pi_{RS}^{G}])\leq \max_{\sigma_{S}\in\operatorname*{BSE}_{G}}F(\rho_{S},\sigma_{S}).
\end{equation}
Since this inequality holds for every state $\sigma_{RS}$, we surmise the desired result
\begin{equation}
\max_{\sigma_{RS}}F(\rho_{S},\operatorname{Tr}_{R}[\Pi_{RS}^{G}\sigma_{RS} \Pi_{RS}^{G}])\leq\max_{\sigma_{S}\in\operatorname*{BSE}_{G}}F(\rho_{S},\sigma_{S}).
\end{equation}

\section{Proof of Theorem~\ref{thm:G-SE-acc-prob}}

\label{app:proof-thm-g-se}Following the same reasoning given in \eqref{eq:euclidean-norm-opt}--\eqref{eq:1st-fid-formula-pf}, by using Uhlmann's theorem, we conclude that
\begin{equation}
\max_{V_{S^{\prime}E\rightarrow R\hat{R}\hat{S}E^{\prime}}}\left\Vert \Pi_{RS\hat{R}\hat{S}}^{G}V_{S^{\prime}E\rightarrow R \hat{R} \hat{S}E^{\prime}} \ket{\psi}_{S^{\prime}S}\ket{0}_{E}\right\Vert _{2}^{2} = \max_{\sigma_{R\hat{R}S\hat{S}}}F(\rho_{S},\operatorname{Tr}_{R\hat{R}\hat{S}}[\Pi_{RS\hat{R}\hat{S}}^{G}\sigma_{R\hat{R}S\hat{S}}\Pi_{RS\hat{R}\hat{S} }^{G}]),
\end{equation}
where the optimization is over every state $\sigma_{RS\hat{R}\hat{S}}$ and $\Pi_{RS\hat{R} \hat{S}}^{G}$ is defined in \eqref{eq:projector-ref-unitaries}. The next part of the proof shows that
\begin{equation}
\max_{\sigma_{R\hat{R}S\hat{S}}}F(\rho_{S},\operatorname{Tr}_{R\hat{R}\hat{S}
}[\Pi_{RS\hat{R}\hat{S}}^{G}\sigma_{R\hat{R}S\hat{S}}\Pi_{RS\hat{R}\hat{S}
}^{G}]) = \max_{\sigma_{S}\in\operatorname*{SymExt}_{G}}F(\rho_{S},\sigma_{S})
\end{equation}
and is similar to \eqref{eq:final-eq-steps-1}--\eqref{eq:final-eq-steps-last}. To justify the inequality $\geq$, let $\sigma_{S}$ be a state in $\operatorname*{SymExt}_{G}$. Then by Theorem~\ref{thm:Bose-sym-purify}, there exists a purification $\varphi_{RS\hat{R}\hat{S}}$ of $\sigma_{S}$\ satisfying $\varphi_{RS\hat{R}\hat{S}} = \Pi_{RS\hat{R}\hat{S}}^{G} \varphi_{RS\hat{R} \hat{S}} \Pi_{RS\hat{R}\hat{S}}^{G}$. We find that
\begin{align}
  F(\rho_{S},\sigma_{S})
&  =F(\rho_{S},\operatorname{Tr}_{R\hat{R}\hat{S}}[\varphi_{RS\hat{R}\hat{S} }])\\
&  =F(\rho_{S},\operatorname{Tr}_{R\hat{R}\hat{S}}[\Pi_{RS\hat{R}\hat{S}} ^{G} \varphi_{RS\hat{R}\hat{S}} \Pi_{RS\hat{R}\hat{S}}^{G}])\\
&  \leq\max_{\sigma_{R\hat{R}S\hat{S}}}F(\rho_{S},\operatorname{Tr}_{R\hat{R}\hat{S}}[\Pi_{RS\hat{R}\hat{S}}^{G}\sigma_{R\hat{R}S\hat{S}}\Pi_{RS\hat{R}\hat{S}}^{G}]).
\end{align}
Since the inequality holds for all $\sigma_{S}\in\operatorname*{SymExt}_{G}$, we conclude that
\begin{equation}
\max_{\sigma_{S}\in\operatorname*{SymExt}_{G}}F(\rho_{S},\sigma_{S})
\leq\max_{\sigma_{R\hat{R}S\hat{S}}}F(\rho_{S},\operatorname{Tr}_{R\hat{R} \hat{S}}[\Pi_{RS\hat{R}\hat{S}}^{G}\sigma_{R\hat{R}S\hat{S}}\Pi_{RS\hat{R} \hat{S}}^{G}]).
\end{equation}
To justify the inequality $\leq$, let $\sigma_{R\hat{R}S\hat{S}}$ be an arbitrary state. If $\Pi_{RS\hat{R}\hat{S}}^{G}\sigma_{R\hat{R}S\hat{S}} \Pi_{RS\hat{R}\hat{S}}^{G}=0$, then the desired inequality follows trivially. Suppose this is not the case, then define
\begin{align}
\sigma_{R\hat{R}S\hat{S}}^{\prime}  &  \coloneqq \frac{1}{p}\Pi_{RS\hat{R}\hat{S}}^{G}\sigma_{R\hat{R}S\hat{S}}\Pi_{RS\hat{R}\hat{S}}^{G},\\
p  &  \coloneqq \operatorname{Tr}[\Pi_{RS\hat{R}\hat{S}}^{G}\sigma_{R\hat{R}S\hat{S}}].
\end{align}
Then we find that
\begin{align}
 F(\rho_{S},\operatorname{Tr}_{R\hat{R}\hat{S}}[\Pi_{RS\hat{R}\hat{S}} ^{G}\sigma_{R\hat{R}S\hat{S}}\Pi_{RS\hat{R}\hat{S}}^{G}])
&  =pF(\rho_{S},\operatorname{Tr}_{R\hat{R}\hat{S}}[\sigma_{R\hat{R}S\hat{S}}^{\prime}])\\
&  \leq F(\rho_{S},\operatorname{Tr}_{R\hat{R}\hat{S}}[\sigma_{R\hat{R}S\hat{S}}^{\prime}])\\
&  =F(\rho_{S},\tau_{S}),
\end{align}
where $\tau_{S}\coloneqq \operatorname{Tr}_{R\hat{R}\hat{S}}[\sigma_{R\hat {R}S\hat{S}}^{\prime}]$. We now aim to show that $\tau_{S}\in \operatorname*{SymExt}_{G}$. To do so, it suffices to prove that $\sigma _{RS}^{\prime} = U_{RS}(g) \sigma_{RS}^{\prime}U_{RS}(g)^{\dag}$ for all $g\in G$. Abbreviating $U\otimes\overline{U}\equiv U_{RS}(g)\otimes\overline {U}_{\hat{R}\hat{S}}(g)$, consider that
\begin{align}
\sigma_{RS}^{\prime} 
&  =\operatorname{Tr}_{\hat{R}\hat{S}}[\sigma_{RS\hat{R}\hat{S}}^{\prime}]\\
&  =\operatorname{Tr}_{\hat{R}\hat{S}}[\Pi_{RS\hat{R}\hat{S}}^{G} \sigma_{RS\hat{R}\hat{S}}^{\prime}\Pi_{RS\hat{R}\hat{S}}^{G}]\\
&  =\operatorname{Tr}_{\hat{R}\hat{S}}[(U\otimes\overline{U})\Pi_{RS\hat {R}\hat{S}}^{G}\sigma_{RS\hat{R}\hat{S}}^{\prime}\Pi_{RS\hat{R}\hat{S}}^{G}(U\otimes\overline{U})^{\dag}]\\
&  =U\operatorname{Tr}_{\hat{R}\hat{S}}[\overline{U}\Pi_{RS\hat{R}\hat{S}} ^{G}\sigma_{RS\hat{R}\hat{S}}^{\prime}\Pi_{RS\hat{R}\hat{S}}^{G}\overline {U}^{\dag}]U^{\dag}\\
&  =U\operatorname{Tr}_{\hat{R}\hat{S}}[\overline{U}^{\dag}\overline{U} \Pi_{RS\hat{R}\hat{S}}^{G}\sigma_{RS\hat{R}\hat{S}}^{\prime}\Pi_{RS\hat{R} \hat{S}}^{G}]U^{\dag}\\
&  =U\operatorname{Tr}_{\hat{R}\hat{S}}[\Pi_{RS\hat{R}\hat{S}}^{G} \sigma_{RS\hat{R}\hat{S}}^{\prime}\Pi_{RS\hat{R}\hat{S}}^{G}]U^{\dag}\\
&  =U\operatorname{Tr}_{\hat{R}\hat{S}}[\sigma_{RS\hat{R}\hat{S}}^{\prime}]U^{\dag}\\
&  =U_{RS}(g)\sigma_{RS}^{\prime}U_{RS}(g)^{\dag}.
\end{align}
It follows that $\tau_{S}\in\operatorname*{SymExt}_{G}$, and we conclude that
\begin{equation}
F(\rho_{S},\operatorname{Tr}_{R\hat{R}\hat{S}}[\Pi_{RS\hat{R}\hat{S}}^{G}\sigma_{R\hat{R}S\hat{S}}\Pi_{RS\hat{R}\hat{S}}^{G}]) \leq \max_{\sigma_{S}\in\operatorname*{SymExt}_{G}}F(\rho_{S},\sigma_{S}).
\end{equation}
Since the inequality holds for every state $\sigma_{R\hat{R}S\hat{S}}$, we conclude that
\begin{equation}
\max_{\sigma_{R\hat{R}S\hat{S}}}F(\rho_{S},\operatorname{Tr}_{R\hat{R}\hat{S}}[\Pi_{RS\hat{R} \hat{S}}^{G}\sigma_{R\hat{R}S\hat{S}}\Pi_{RS\hat{R}\hat{S}}^{G}]) \leq \max_{\sigma_{S}\in\operatorname*{SymExt}_{G}}F(\rho_{S},\sigma_{S}).
\end{equation}

\section{Proof of Proposition~\ref{prop:golden-rule-G-BSE}}
\label{app:proof-prop-G-BSE}

The idea of the proof is similar to that for Proposition~\ref{prop:golden-rule-G-SE}. Since $\rho_{S}$ is a $G$-BSE state, by Definition~\ref{def:g-bose-sym-ext}, there exists an extension state $\omega_{RS}$ satisfying the conditions stated there. Since $\mathcal{N}_{S\rightarrow S^{\prime}}$ is a $G$-BSE channel, by Definition~\ref{def:G-BSE-channels}, there exists an extension channel $\mathcal{M}_{RS\rightarrow R^{\prime}S^{\prime}}$ satisfying the conditions stated there. It follows that $\mathcal{M}_{RS\rightarrow R^{\prime}S^{\prime}}(\omega_{RS})$ is an extension of $\mathcal{N}_{S\rightarrow S^{\prime}}(\rho_{S})$ because
\begin{align}
\operatorname{Tr}_{R^{\prime}}[\mathcal{M}_{RS\rightarrow R^{\prime}S^{\prime}}(\omega_{RS})]  &  =\mathcal{N}_{S\rightarrow S^{\prime}}(\operatorname{Tr}_{R}[\omega_{RS}])\\
&  =\mathcal{N}_{S\rightarrow S^{\prime}}(\rho_{S}),
\end{align}
where the first equality follows from \eqref{eq:channel-ext-bose}. Also, consider that the following holds
\begin{align}
1  &  \geq\operatorname{Tr}[\Pi_{R^{\prime}S^{\prime}}^{G}\mathcal{M}_{RS\rightarrow R^{\prime}S^{\prime}}(\omega_{RS})]\nonumber\\
&  =\operatorname{Tr}[(\mathcal{M}_{RS\rightarrow R^{\prime}S^{\prime}})^{\dag}(\Pi_{R^{\prime}S^{\prime}}^{G})\omega_{RS}]\\
&  \geq\operatorname{Tr}[\Pi_{RS}^{G}\omega_{RS}]\\
&  =1.
\end{align}
The first inequality follows because $\mathcal{M}_{RS\rightarrow R^{\prime}S^{\prime}}(\omega_{RS})$ is a state and $\Pi_{R^{\prime}S^{\prime}}^{G}$ is projection. The first equality follows from the definition of channel adjoint. The second inequality follows from \eqref{eq:BSE-bose-condition}. where the first equality follows from \eqref{eq:covariance-G-sym-ext} and the second from \eqref{eq:G-ext-2}. We conclude that $\operatorname{Tr} [\Pi_{R^{\prime}S^{\prime}}^{G} \mathcal{M}_{RS\rightarrow R^{\prime}S^{\prime}}(\omega_{RS})]=1$, which by \eqref{eq:Bose-symmetric-equiv-cond}, implies that $\mathcal{M}_{RS\rightarrow R^{\prime}S^{\prime}}(\omega_{RS})$ is a $G$-Bose symmetric state. It then follows that $\mathcal{N}_{S\rightarrow S^{\prime}}(\rho_{S})$ is $G$-Bose symmetric extendible.

\pagebreak
\singlespacing
\pagebreak
\chapter{Copyright Information}
\doublespacing
This dissertation uses material from a previously published article in which the dissertation writer was an author. 

\section{Chapter 2}
Chapter 2 contains material from the article ``Quantum Algorithms for Testing Hamiltonian Symmetry" by Margarite L. LaBorde and Mark M. Wilde. This article is published in \textit{Physical Review Letter} by publisher American Physical Society (APS) and is copyrighted by American Physical Society as of 2022. 

Below is the response by APS to a request for written permission to include this work in the dissertation:

\begin{figure*}[h]
\begin{center}
\includegraphics[width=\textwidth]{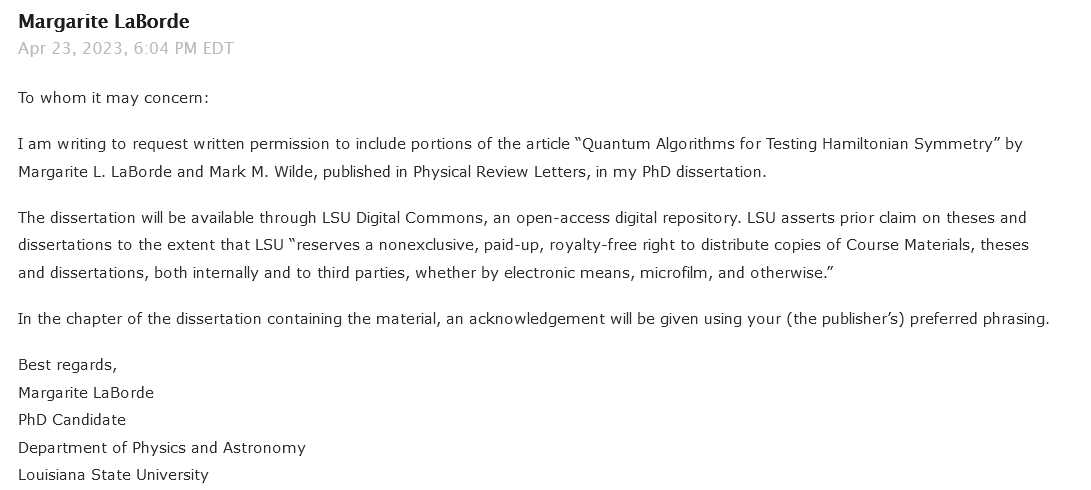}
\end{center}
\caption{}
\label{apscopyrightask}
\end{figure*}
\begin{figure*}[h]
\begin{center}
\includegraphics[]{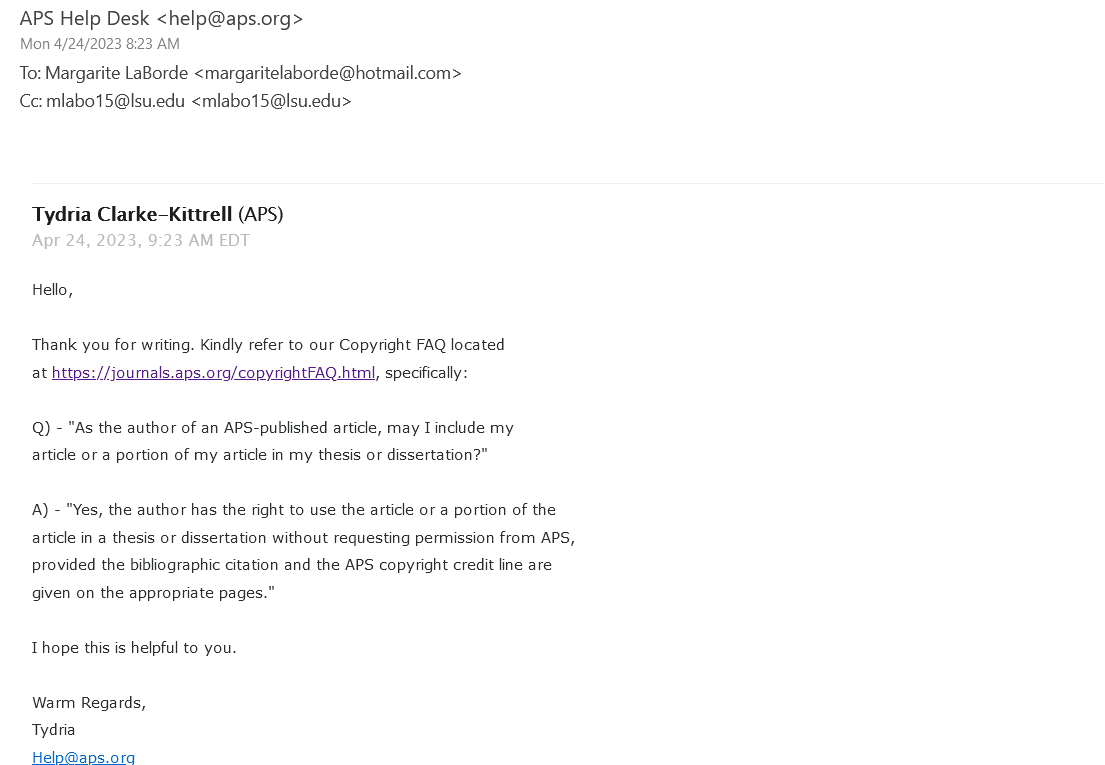}
\end{center}
\caption{}
\label{apscopyrightresponse}
\end{figure*}
\pagebreak
\singlespacing

\backmatter

\bibliographystyle{alpha}
\bibliography{main}
\pagebreak

\chapter{Vita}
\doublespacing
Margarite L. LaBorde was born in 1995 in Baton Rouge, Louisiana. She attended Port Allen High School in Baton Rouge and graduated in May of 2014. She obtained dual degrees in physics and mathematics at Louisiana State University in May 2018, obtaining Latin honors in both and earning the College Honors distinction for her physics degree. She received her Masters of Science in Physics in 2022 and is currently a candidate for the degree of Doctorate of Philosophy in Physics, which is to be awarded in May 2023.
\end{document}